\PassOptionsToPackage{svgnames}{xcolor}

\documentclass[a4paper,UKenglish,cleveref,autoref,thm-restate]{lipics-v2021}

\title{Fearless Asynchronous Communications with Timed Multiparty Session Protocols} %

\author{Ping Hou}{
University of Oxford, UK
}{%
ping.hou@cs.ox.ac.uk
}{%
https://orcid.org/0000-0001-6899-9971
}{}
\author{Nicolas Lagaillardie}{Imperial College London, UK}{%
n.lagaillardie19@imperial.ac.uk
}{%
https://orcid.org/0000-0002-6431-4100
}{}
\author{Nobuko Yoshida}{University of Oxford, UK}{%
nobuko.yoshida@cs.ox.ac.uk
}{%
https://orcid.org/0000-0002-3925-8557
}{}

\authorrunning{P. Hou, N. Lagaillardie, N. Yoshida}


\relatedversion{} %

\nolinenumbers %

\hideLIPIcs  %


\usepackage{hyperref}
\usepackage{hyperxmp}

\usepackage{booktabs}   %
\usepackage{subcaption} %
\usepackage{xspace}
\usepackage{listings-rust} %
\usepackage{listings-scribble} %
\usepackage{cleveref} %
\usepackage{amsthm}
\usepackage{mathtools}
\usepackage{amsmath}
\usepackage{textcomp}
\usepackage{tabularx}
\usepackage{xr}
\usepackage{makecell}
\usepackage{tikz}
\usepackage[clock]{ifsym}

\usepackage{soulutf8}
\setul{0.4ex}{0.4ex}
\setulcolor{gray}

\usepackage{pifont}%
\usetikzlibrary{
  arrows,
  automata,
  calc,
  fit,
  positioning,
  shapes,
  arrows.meta,
  shapes.multipart,
  tikzmark,
  backgrounds
}

\newcommand{\lbl}[1]{\ensuremath{l}}
\definecolor{dkblue}{rgb}{0,0.1,0.5}
\definecolor{dkgreen}{rgb}{0,0.4,0.1}
\definecolor{dkred}{rgb}{0.4,0,0}

\definecolor{linkColor}{rgb}{0,0,0.5}
\definecolor{pblue}{rgb}{0.13,0.13,1}
\definecolor{purple}{rgb}{0.5,0,0.5}
\definecolor{beige}{rgb}{0,0.5,0.5}
\definecolor{pred}{rgb}{0.9,0,0}
\definecolor{WhiteSmoke}{rgb}{0.96, 0.96, 0.96}
\definecolor{mygray}{gray}{0.6}
\definecolor{DarkRed}{rgb}{0.55, 0.0, 0.0}

\definecolor{types}{rgb}{0.282, 0.714, 0.627}
\definecolor{functions}{rgb}{0.463, 0.463, 0.373}
\definecolor{typetype}{rgb}{0.318, 0.576, 0.792}
\definecolor{let}{rgb}{0.137, 0.412, 0.639}
\definecolor{cond}{rgb}{0.169, 0.075, 0.165}
\definecolor{str}{rgb}{0.392, 0.227, 0.157}
\definecolor{var}{rgb}{0.1, 0.1, 0.5}
\definecolor{primitive}{rgb}{0, 0.5, 0}

\newcommand{\CODESIZE}{\small}

\newcommand{\CODESTYLE}{\ttfamily\bfseries}
\newcommand{\CODE}[1]{\sloppy{\CODESIZE\CODESTYLE\lstinline[language=Rust,style=colouredRust]{#1}}}

\newcommand{\lstscribble}[1]{\sloppy{\ttfamily\lstinline[language=Scribble,style=colouredScribble]{#1}}}

\newcommand{\norole}[1]{\textbf{\color{teal}\emph{#1}}}

\newcommand{\kmc}{$k$-MC}

\definecolor{darkgreen}{rgb}{0, 0.5, 0}

\newcommand{\cmark}{\ding{51}}%
\newcommand{\xmark}{\ding{55}}%

\newcommand{\labSet}{\mathcal{M}}

{\bfseries}{\itshape}
{\bfseries}{\itshape}
{\bfseries}{\itshape}
{\bfseries}{\itshape}
{\bfseries}{\itshape}
{\bfseries}{\itshape}
{\bfseries}{\itshape}
{\bfseries}{\itshape}
{\bfseries}{\itshape}

\definecolor{modify}{rgb}{0.0, 0.0, 0.0}%
\definecolor{modified}{rgb}{0.0, 0.0, 0.0}%

\makeatletter
\newsavebox{\@brx}
\newcommand{\llangle}[1][]{\savebox{\@brx}{\(\m@th{#1\langle}\)}%
  \mathopen{\copy\@brx\kern-0.5\wd\@brx\usebox{\@brx}}}
\newcommand{\rrangle}[1][]{\savebox{\@brx}{\(\m@th{#1\rangle}\)}%
  \mathclose{\copy\@brx\kern-0.5\wd\@brx\usebox{\@brx}}}
\makeatother

\newcommand{\elip}{\ifmmode\mathinner{\ldotp\kern-0.2em\ldotp\kern-0.2em\ldotp}\else.\kern-0.13em.\kern-0.13em.\fi}
\newcommand{\elipc}{\ifmmode\mathinner{\cdotp\kern-0.2em\cdotp\kern-0.2em\cdotp}\else.\kern-0.13em.\kern-0.13em.\fi}

\newcommand{\roleS}{{\color{roleColor}\roleFmt{s}}}

\newcommand{\timedmulticrusty}[0]{{$\texttt{MultiCrusty}^{T}$}\!\xspace}
\newcommand{\rumpsteak}[0]{{\small\texttt{Rumpsteak}}\xspace}

\newcommand{\ferrite}[0]{{\small\texttt{Ferrite}}\xspace}
\newcommand{\typestate}[0]{{\small\texttt{Typestate}}\xspace}
\newcommand{\multicrusty}[0]{{\texttt{MultiCrusty}}\xspace}

\newcommand{\ATMP}[0]{{\textsf{ATMP}}\xspace}
\newcommand{\AMPST}[0]{{\textsf{AMPST}}\xspace}%
\newcommand{\MPST}[0]{{\textsf{MPST}}\xspace}%
\newcommand{\Scribble}[0]{{\textsc{Scribble}}\xspace}%
\newcommand{\nuscr}[0]{{$\nu$\textsc{Scr}}\xspace}%
\newcommand{\timednuscr}[0]{{$\nu$\textsc{Scr}$^{T}$}\xspace}%
\newcommand{\crosschan}[0]{{\ttfamily{crossbeam\_channel}}\xspace}%
\newcommand{\timeRust}{\xspace{\ttfamily{time}}\xspace}%
\newcommand{\criterionRust}{\xspace{\ttfamily{criterion}}\xspace}%
\newcommand{\rustup}{\xspace{\ttfamily{Rustup}}\xspace}%
\newcommand{\tokioRust}{\xspace{\ttfamily{tokio}}\xspace}%

\newcommand{\gtPing}[1][]{%
\ifempty{#1}{\gtMsgFmt{ping}}{{\color{gtColor}\gtMsgFmt{ping}_{#1}}}%
}%
\newcommand{\gtPong}[1][]{%
\ifempty{#1}{\gtMsgFmt{pong}}{{\color{gtColor}\gtMsgFmt{pong}_{#1}}}%
}%

\newcommand{\dom}[1]{{\color{black}\operatorname{dom}\!\left({#1}\right)}}%
\newcommand{\fv}[1]{\operatorname{fv}\!\left({#1}\right)}%
\newcommand{\fc}[1]{\operatorname{fc}\!\left({#1}\right)}%
\newcommand{\dpv}[1]{\operatorname{dpv}\!\left({#1}\right)}%
\newcommand{\fpv}[1]{\operatorname{fpv}\!\left({#1}\right)}%
\newcommand{\unfoldOne}[1]{%
  {\color{black}\operatorname{unf}\!\left({#1}\right)}}%
\newcommand{\notImplies}{\mathrel{{\kern 0.5em}{\not{\kern -0.5em}\implies}}}%
\newcommand{\notImpliedBy}{\mathrel{{\kern 1em}{\not{\kern -1em}\impliedby}}}%

\newcommand{\coloncolonequals}{\Coloneqq}%
\newcommand{\bnfdef}{\coloncolonequals}%
\newcommand{\bnfsep}{\mathbin{\;\big|\;}}%

\def\eg{e.g.\@\xspace}%
\def\ie{i.e.\@\xspace}%
\def\wrt{w.r.t.\@\xspace}%

\definecolor{ruleColor}{rgb}{0.1, 0.3, 0.1}%
\newcommand{\inferrule}[1]{{\color{ruleColor}\textsc{\scriptsize [#1]}}}%
\newcommand{\inference}[3][]{\infer[\ifempty{#1}{}{\inferrule{#1}}]{#3}{#2}}%
\newcommand{\cinference}[3][]{\infer=[\ifempty{#1}{}{\inferrule{#1}}]{#3}{#2}}%
\newcommand{\inferenceSingle}[2][]{{#2}\ifempty{#1}{}{\ \inferrule{#1}}}%

\newcommand{\setenum}[1]{\mathord{{\color{black}\left\{#1\right\}}}}%
\newcommand{\setcomp}[2]{\mathord{%
  {\color{black}\left\{{#1} \,\middle|\, {#2}\right\}}}}%

\definecolor{groundColor}{rgb}{0.38, 0.25, 0.32}%

\newcommand{\bind}[2]{\nicefrac{#2}{#1}}
\newcommand{\substenum}[1]{\mathord{{\color{black}\left\{{#1}\right\}}}}
\newcommand{\subst}[2]{\substenum{\bind{#1}{#2}}}

\definecolor{roleColor}{rgb}{0.5, 0.0, 0.0}%
\newcommand{\roleCol}[1]{{\color{roleColor}#1}}%
\newcommand{\roleSet}{\roleCol{\mathcal{R}}}%
\newcommand{\roleFmt}[1]{\boldsymbol{\roleCol{\mathtt{#1}}}}%
\newcommand{\roleP}[1][]{%
  \ifempty{#1}{{\color{roleColor}\roleFmt{p}}}{{\color{roleColor}\roleFmt{p}_{#1}}}%
}%
\newcommand{\rolePi}[1][]{%
  \ifempty{#1}{{\color{roleColor}\roleFmt{p}'}}{{\color{roleColor}\roleFmt{p}'_{#1}}}%
}%
\newcommand{\roleQ}[1][]{%
  \ifempty{#1}{{\color{roleColor}\roleFmt{q}}}{{\color{roleColor}\roleFmt{q}_{#1}}}%
}%
\newcommand{\roleQi}[1][]{%
  \ifempty{#1}{{\color{roleColor}\roleFmt{q}'}}{{\color{roleColor}\roleFmt{q}'_{#1}}}%
}%
\newcommand{\roleR}[1][]{%
  \ifempty{#1}{{\color{roleColor}\roleFmt{r}}}{{\color{roleColor}\roleFmt{r}_{\!#1}}}%
}%
\newcommand{\roleT}[1][]{%
  \ifempty{#1}{{\color{roleColor}\roleFmt{t}}}{{\color{roleColor}\roleFmt{t}_{\!#1}}}%
}%

\newcommand{\labFmt}[2][]{\ifempty{#1}{\mathtt{#2}}{\mathtt{#2}\textsubscript{#1}}}%

\definecolor{mpColor}{rgb}{0, 0, 0}%
\newcommand{\mpFmt}[1]{{\color{mpColor}#1}}%

\newcommand{\mpLab}[1][]{%
  \mpFmt{\ifempty{#1}{\labFmt{m}}{{\labFmt{m}}_{\mathnormal #1}}}%
}%
\newcommand{\mpLabi}[1][]{%
  \mpFmt{\ifempty{#1}{\labFmt{m}'}{\labFmt{m}'_{\mathnormal #1}}}%
}%
\newcommand{\mpLabFmt}[1]{\mpFmt{\labFmt{#1}}}%

\newcommand{\mpChanRole}[2]{\mpFmt{{#1}[{#2}]}}%

\newcommand{\mpNil}{\mpFmt{\mathbf{0}}}%
\newcommand{\mpSeq}{\mathbin{\mpFmt{\!.\!}}}%
\newcommand{\mpChoice}[3]{%
  \mpFmt{%
    \mpLabFmt{#1}\ifempty{#2}{}{({#2})}\ifempty{#3}{}{\vphantom{x}\mpSeq {#3}}%
  }%
}%
\newcommand{\mpChoiceNoBind}[3]{%
  \mpFmt{%
    \mpLabFmt{#1}\ifempty{#2}{}{\langle{#2}\rangle}\ifempty{#3}{}{\vphantom{x}\mpSeq {#3}}%
  }%
}%

\newcommand{\mpChoiceRaw}[2]{%
  \mpFmt{%
    \mpLabFmt{#1} \mpSeq {#2}
  }
}

\newcommand{\mpChoiceRawN}[1]{%
  \mpFmt{%
    \mpLabFmt{#1}
  }
}

\newcommand{\mpBranchRawN}[3]{%
  \mpFmt{%
    {#1}[\roleFmt{#2}]%
      {#3}%
  }%
}%

\newcommand{\mpPar}{\mathbin{\mpFmt{\mid}}}%
\newcommand{\mpBigPar}[2]{\mathbin{\mpFmt{\Pi_{#1}}{#2}}}%
\newcommand{\mpRes}[2]{\mpFmt{\left(\mathbf{\nu}{#1}\right){#2}}}%
\newcommand{\mpJustDef}[3]{%
  \mpFmt{{#1}(#2) = {#3}}%
}%
\newcommand{\mpDef}[4]{%
  \mpFmt{\mathbf{def}\;\mpJustDef{#1}{#2}{#3}\;\mathbf{in}\;{#4}}%
}%
\newcommand{\mpDefAbbrev}[2]{%
  \mpFmt{\mathbf{def}\;{#1}\;\mathbf{in}\;{#2}}%
}%
\newcommand{\mpCall}[2]{\mpFmt{{#1}\!\left\langle{#2}\right\rangle}}%
\newcommand{\mpCErr}{\mpFmt{\boldsymbol{\mathsf{cerr}}}}%

\newcommand{\mpCtx}[1][]{\mpFmt{\ifempty{#1}{\mathbb{C}}{\mathbb{C}_{#1}}}}%
\newcommand{\mpCtxHole}{[\,]}%
\newcommand{\mpCtxApp}[2]{{#1}\!\left[{#2}\right]}%

\newcommand{\mpC}[1][]{\mpFmt{\ifempty{#1}{c}{c_{#1}}}}%
\newcommand{\mpCi}[1][]{\mpFmt{\ifempty{#1}{c'}{c'_{#1}}}}%
\newcommand{\mpD}[1][]{\mpFmt{\ifempty{#1}{d}{d_{#1}}}}%

\newcommand{\mpS}[1][]{\mpFmt{\ifempty{#1}{s}{s_{#1}}}}%
\newcommand{\mpSi}[1][]{\mpFmt{\ifempty{#1}{s'}{s'_{#1}}}}%

\newcommand{\mpX}[1][]{\mpFmt{\ifempty{#1}{X}{X_{#1}}}}%
\newcommand{\mpY}[1][]{\mpFmt{\ifempty{#1}{Y}{Y_{#1}}}}%

\newcommand{\mpP}[1][]{\mpFmt{\ifempty{#1}{P}{P_{#1}}}}%
\newcommand{\mpPi}[1][]{\mpFmt{\ifempty{#1}{P'}{P'_{#1}}}}%
\newcommand{\mpPii}[1][]{\mpFmt{\ifempty{#1}{P''}{P''_{#1}}}}%
\newcommand{\mpQ}[1][]{\mpFmt{\ifempty{#1}{Q}{Q_{#1}}}}%
\newcommand{\mpQi}[1][]{\mpFmt{\ifempty{#1}{Q'}{Q'_{#1}}}}%
\newcommand{\mpR}[1][]{\mpFmt{\ifempty{#1}{R}{R_{#1}}}}%
\newcommand{\mpRi}[1][]{\mpFmt{\ifempty{#1}{R'}{R'_{#1}}}}%

\newcommand{\mpDefD}[1][]{\mpFmt{\ifempty{#1}{D}{D_{#1}}}}%
\newcommand{\mpDefDi}[1][]{\mpFmt{\ifempty{#1}{D'}{D'_{#1}}}}%

\newcommand{\mpMove}{\to}%
\newcommand{\mpMoveStar}{\mathrel{\mpMove{}^{\!\!\!*}}}%
\newcommand{\mpNotMoveP}[1]{\mpFmt{#1}\!\not\!\!{\mpMove}}%

\definecolor{gtColor}{rgb}{0.43, 0.21, 0.1}%
\newcommand{\gtFmt}[1]{{\color{gtColor}#1}}%
\newcommand{\gtMsgFmt}[1]{\gtFmt{\labFmt{#1}}}%
\newcommand{\gtLab}[1][]{%
  \ifempty{#1}{\gtMsgFmt{m}}{{\color{gtColor}\gtMsgFmt{m}_{#1}}}%
}%
\newcommand{\gtLabi}[1][]{%
  \ifempty{#1}{\gtMsgFmt{m}'}{{\color{gtColor}\gtMsgFmt{m}'_{#1}}}%
}%
\newcommand{\gtLabii}[1][]{%
  \ifempty{#1}{\gtMsgFmt{m}'}{{\color{gtColor}\gtMsgFmt{m}''_{#1}}}%
}%

\newcommand{\gtG}[1][]{\gtFmt{\ifempty{#1}{G}{G_{#1}}}}%
\newcommand{\gtGi}[1][]{\gtFmt{\ifempty{#1}{G'}{G'_{#1}}}}%
\newcommand{\gtGii}[1][]{\gtFmt{\ifempty{#1}{G''}{G''_{#1}}}}%
\newcommand{\gtGiii}[1][]{\gtFmt{\ifempty{#1}{G'''}{G'''_{#1}}}}%

\newcommand{\gtSeq}{\mathbin{\gtFmt{.}}}%

\newcommand{\gtCommRaw}[3]{%
  \gtFmt{%
    {#1} {\to} {#2}{:}%
    \left\{%
      {#3}%
    \right\}%
  }%
}%
\newcommand{\gtComm}[6]{%
  \gtFmt{%
    \gtCommRaw{#1}{#2}{%
      \gtCommChoice{#4}{#5}{#6}%
    }_{#3}%
  }%
}%
\newcommand{\gtCommChoice}[3]{%
  \gtFmt{%
    \gtMsgFmt{#1}\ifempty{#2}{}{({#2})}%
    \ifempty{#3}{}{\vphantom{x} \gtSeq {#3}}%
  }%
}%

\newcommand{\gtEnd}{\gtFmt{\mathbf{end}}}%

\newcommand{\gtRec}[2]{\gtFmt{\mu{#1}.{#2}}}%
\newcommand{\gtRecVarBase}{\gtFmt{\mathbf{t}}}%
\newcommand{\gtRecVar}[1][]{\gtFmt{\ifempty{#1}{\gtRecVarBase}{\gtRecVarBase_{#1}}}}%
\newcommand{\gtRecVari}[1][]{\gtFmt{\ifempty{#1}{\gtRecVarBase'}{\gtRecVarBase'_{#1}}}}%

\newcommand{\gtRoles}[1]{{\color{roleColor} \operatorname{roles}(\gtFmt{#1})}}%

\newcommand{\gtProj}[2]{%
  {\color{stColor}\gtFmt{#1} {\upharpoonright}\, \roleFmt{#2}}%
}%

\definecolor{stColor}{rgb}{0, 0, 0.9}%
\newcommand{\stFmt}[1]{{\color{stColor}#1}}%

\newcommand{\stOut}[3]{\ifempty{#1}{}{\roleFmt{#1}}\stFmt{\oplus{#2}\ifempty{#3}{}{({#3})}}}%

\newcommand{\stSeq}{\mathbin{\!\stFmt{.}\!}}%
\newcommand{\stIntC}{\mathbin{\stFmt{\oplus}}}%
\newcommand{\stIntSum}[3]{\roleFmt{#1}\stFmt{\oplus\!\left\{#3\right\}_{#2}}}%
\newcommand{\stExtC}{\mathbin{\stFmt{\&}}}%
\newcommand{\stExtSum}[3]{\roleFmt{#1}\stFmt{\&\!\left\{#3\right\}_{#2}}}%
\newcommand{\stRec}[2]{\stFmt{\mu{#1}.{#2}}}%
\newcommand{\stEnd}{\stFmt{\mathbf{end}}}%

\newcommand{\stLab}[1][]{\stFmt{\ifempty{#1}{\labFmt{m}}{\labFmt{m}_{#1}}}}%
\newcommand{\stLabi}[1][]{\stFmt{\ifempty{#1}{\labFmt{m}'}{\labFmt{m}'_{#1}}}}%
\newcommand{\stLabii}[1][]{\stFmt{\ifempty{#1}{\labFmt{m}'}{\labFmt{m}''_{#1}}}}%
\newcommand{\stLabFmt}[1]{\stFmt{\labFmt{#1}}}%

\newcommand{\stS}[1][]{\stFmt{\ifempty{#1}{S}{S_{#1}}}}%
\newcommand{\stSi}[1][]{\stFmt{\ifempty{#1}{S'}{S'_{#1}}}}%
\newcommand{\stSii}[1][]{\stFmt{\ifempty{#1}{S''}{S''_{#1}}}}%
\newcommand{\stSiii}[1][]{\stFmt{\ifempty{#1}{S'''}{S'''_{#1}}}}%

\newcommand{\stT}[1][]{\stFmt{\ifempty{#1}{T}{T_{#1}}}}%
\newcommand{\stTi}[1][]{\stFmt{\ifempty{#1}{T'}{T'_{#1}}}}%
\newcommand{\stTii}[1][]{\stFmt{\ifempty{#1}{T''}{T''_{#1}}}}%

\newcommand{\stU}[1][]{\stFmt{\ifempty{#1}{U}{U_{#1}}}}%
\newcommand{\stUi}[1][]{\stFmt{\ifempty{#1}{U'}{U'_{#1}}}}%

\newcommand{\stRecVarBase}{\stFmt{\mathbf{t}}}%
\newcommand{\stRecVar}[1][]{\stFmt{\ifempty{#1}{\stRecVarBase}{\stRecVarBase_{#1}}}}%

\newcommand{\stMerge}[2]{\stFmt{\bigsqcap_{#1}{#2}}}%
\newcommand{\stBinMerge}{\mathbin{\stFmt{\sqcap}}}%

\newcommand{\stSub}{\mathrel{\stFmt{\leqslant}}}%
\newcommand{\stNotSub}{\mathrel{\stFmt{\not\leqslant}}}%
\newcommand{\stSup}{\mathrel{\stFmt{\geqslant}}}%

\definecolor{ptColor}{rgb}{0.20, 0.29, 0.09}%

\newcommand{\stEnv}[1][]{\stFmt{\ifempty{#1}{\Gamma}{\Gamma_{#1}}}}%
\newcommand{\stEnvi}[1][]{\stFmt{\ifempty{#1}{\Gamma'}{\Gamma'_{#1}}}}%
\newcommand{\stEnvii}[1][]{\stFmt{\ifempty{#1}{\Gamma''}{\Gamma''_{#1}}}}%
\newcommand{\stEnviii}[1][]{\stFmt{\ifempty{#1}{\Gamma'''}{\Gamma'''_{#1}}}}%
\newcommand{\stEnvEmpty}{\stFmt{\emptyset}}%
\newcommand{\stEnvMap}[2]{\stFmt{\mpFmt{#1}\mathbin{\!:\!}{#2}}}%
\newcommand{\stEnvComp}{\mathpunct{\stFmt{,}}}%
\newcommand{\stEnvApp}[2]{\stFmt{#1\!\left(\mpFmt{#2}\right)}}%

\newcommand{\stEnvMove}{\mathrel{\stFmt{\to}}}%
\newcommand{\stEnvAnnotGenericSym}{\stFmt{\alpha}}%
\newcommand{\stEnvAnnotGenericSymi}{\stFmt{\alpha'}}%
\newcommand{\stEnvMoveAnnot}[1]{\mathrel{\stFmt{\xrightarrow{#1}}}}
\newcommand{\stEnvMoveGenAnnot}{\stEnvMoveAnnot{\stEnvAnnotGenericSym}}%
\newcommand{\stEnvMoveP}[1]{{#1}\!\!\stEnvMove}%
\newcommand{\stEnvMoveStar}{\mathrel{\stFmt{\stEnvMove{}^{\!\!\!*}}}}%

\newcommand{\stEnvEndPred}{\operatorname{end}}%
\newcommand{\stEnvEndP}[1]{\stEnvEndPred(\stFmt{#1})}%

\newcommand{\stEnvEntails}[3]{%
  \stFmt{#1} \vdash \stFmt{\mpFmt{#2} \mathbin{\!:\!} {#3}}%
}%

\newcommand{\mpEnv}[1][]{\stFmt{\ifempty{#1}{\Theta}{\Theta_{#1}}}}%
\newcommand{\mpEnvEmpty}{\stFmt{\emptyset}}%
\newcommand{\mpEnvMap}[2]{\stFmt{\mpFmt{#1}{:}\stFmt{#2}}}%
\newcommand{\mpEnvComp}{\mathpunct{\stFmt{,}}}%
\newcommand{\mpEnvApp}[2]{\stFmt{#1}\!\left(\mpFmt{#2}\right)}%

\newcommand{\mpEnvEntails}[3]{%
  \stFmt{#1} \vdash \stFmt{\mpFmt{#2} \mathbin{\!:\!} \stFmt{#3}}%
}%

\newcommand{\stJudge}[3]{%
  \stFmt{\ifempty{#1}{#2}{{#1} \cdot {#2}}%
  \mathrel{\mpFmt{\vdash}} \mpFmt{#3}}%
}%

\newcommand{\mpSessionQueue}[2]{%
  \mpFmt{{#1}\mathrel{\!\blacktriangleright\!}{#2}}}%
\newcommand{\mpQueue}[1][]{\mpFmt{\ifempty{#1}{\sigma}{\sigma_{#1}}}}%
\newcommand{\mpQueuei}[1][]{\mpFmt{\ifempty{#1}{\sigma'}{\sigma'_{#1}}}}%
\newcommand{\mpQueueEmpty}{\mpFmt{\epsilon}}%
\newcommand{\mpQueueCons}[2]{\mpFmt{{#1}\mathrel{\!\cdot\!}{#2}}}%

\newcommand{\stEnvQ}[1][]{\stFmt{\ifempty{#1}{\Gamma}{\Gamma_{#1}}}}%
\newcommand{\stEnvQi}[1][]{\stFmt{\ifempty{#1}{\Gamma'}{\Gamma'_{#1}}}}%
\newcommand{\stEnvQii}[1][]{\stFmt{\ifempty{#1}{\Gamma''}{\Gamma''_{#1}}}}%
\newcommand{\stEnvQiii}[1][]{\stFmt{\ifempty{#1}{\Gamma'''}{\Gamma'''_{#1}}}}%
\newcommand{\stEnvQComp}{\mathpunct{\stFmt{,}}}%
\newcommand{\stQJudge}[4]{%
  \stFmt{\ifempty{#1}{#2}{{#1} \cdot {#2}}%
  \mathrel{\mpFmt{\vdash}_{\mpFmt{#3}}} \mpFmt{#4}}%
}%

\newcommand{\stEnvQAnnotQueueSym}{\stFmt{!}}%
\newcommand{\stEnvQAnnotRecvSym}{\stFmt{,}}%
\newcommand{\stEnvQQueueAnnotSmall}[3]{\stFmt{\mpS{:}{#1}{\stEnvQAnnotQueueSym}{#2}{:}{#3}}}%
\newcommand{\stEnvQRecvAnnotSmall}[3]{\stFmt{\mpS{:}{#1}{\stEnvQAnnotRecvSym}{#2}{:}{#3}}}%

\newcommand{\stQEmpty}{\stFmt{\epsilon}}%
\newcommand{\stQCons}[2]{\stFmt{{#1}\mathbin{\!\cdot\!}{#2}}}%
\newcommand{\stQMsg}[3]{\stFmt{{#1}\mathbin{\!\mathbf{!}\!}{#2}\ifempty{#3}{}{({#3})}}}%

\newcommand{\stQEquiv}{\mathrel{\stFmt{\equiv}}}%

\newcommand{\stMPair}[2]{\stFmt{{#2}{\mathbin{\mathbf{;\,}}}{#1}}}%

\newcommand{\stEnvQMove}{\mathrel{\stFmt{\to}}}%
\newcommand{\stEnvQMoveModQ}[1]{\mathrel{\stFmt{\rightarrow}_{{#1}}}}%
\newcommand{\stEnvQNotMoveModQP}[2]{%
  {#2}/{\kern -0.5em}{\stEnvQMoveModQ{#1}}}%

\newcommand{\muCol}[1]{{\color{red}#1}}%
\newcommand{\muWordEmpty}[1][]{\muCol{\epsilon}}%
\newcommand{\iruleMPRedCall}{R-$\mpX$}%
\newcommand{\iruleMPRedCongr}{R-$\equiv$}%
\newcommand{\iruleMPRedCtx}{R-Ctx}%

\newcommand{\iruleTCtxOut}{$\stEnv$-$\stFmt{\oplus}$}%
\newcommand{\iruleTCtxIn}{$\stEnv$-$\stFmt{\&}$}%
\newcommand{\iruleTCtxRec}{$\stEnv$-$\mu$}%
\newcommand{\iruleSubQueue}{Sub-$\stFmt{;}$}%

\newcommand{\iruleMPEnd}{T-$\stEnvEndPred$}%
\newcommand{\iruleMPX}{T-$\mpFmt{X}$}%
\newcommand{\iruleMPSub}{T-Sub}%
\newcommand{\iruleMPNil}{T-$\mpNil$}%
\newcommand{\iruleMPDef}{T-$\mpFmt{\mathbf{def}}$}%
\newcommand{\iruleMPCall}{T-$\mpX$-Call}%
\newcommand{\iruleMPPar}{T-$\mpPar$}%
\newcommand{\iruleMPBranch}{T-$\mpFmt{\&}$}%
\newcommand{\iruleMPSel}{T-$\mpFmt{\oplus}$}%

\newcommand{\iruleMPQueueEmpty}{T-$\mpQueueEmpty$}%
\newcommand{\iruleMPQueue}{T-$\mpQueue$}%

\newcommand{\proclit}[1]{\ensuremath{\operatorname{\text{\sffamily\bfseries{#1}}}}}
\newcommand{\trycatch}[2]{\mpFmt{\ensuremath{\proclit{try}\ {#1} \ \proclit{catch}\ {#2}}}}
\newcommand{\trycatchB}{\mpFmt{\ensuremath{\proclit{try-catch}}}}
\newcommand{\mpCancel}[2]{\ifempty{#1}{}{\mpFmt{\proclit{cancel}({#1})}{\ifempty{#2}{}{\ \mpSeq \ {#2}}}}}
\newcommand{\kills}[1]{\ensuremath{{#1} \lightning}}

\newcommand{\iruleMPRedCan}{R-Can}%
\newcommand{\iruleMPCCat}{C-Cat}%

\newcommand{\mpEtx}[1][]{\mpFmt{\ifempty{#1}{\mathbb{E}}{\mathbb{E}_{#1}}}}%

\newcommand{\ifempty}[3]{%
  \ifthenelse{\isempty{#1}}{#2}{#3}%
}%
\definecolor{stColor}{rgb}{0, 0, 0.9}%

\newcommand{\iruleStSubEnd}{Sub-$\stEnd$}

\newcommand{\iruleStSubRecL}{Sub-$\stFmt{\mu}$L}
\newcommand{\iruleStSubRecR}{Sub-$\stFmt{\mu}$R}
\newcommand{\iruleStSubOut}{Sub-$\stFmt{\oplus}$}
\newcommand{\iruleStSubIn}{Sub-$\stFmt{\&}$}

\definecolor{green_colour_blind}{HTML}{89CE00}
\definecolor{blue_colour_blind}{HTML}{0073E6}
\definecolor{red_colour_blind}{HTML}{E6308A}
\definecolor{tcColor}{rgb}{0, 0, 0}%
\newcommand{\tcFmt}[1]{\ensuremath{{\color{tcColor}#1}}\xspace}%

\newcommand{\ccst}[1][]{\tcFmt{\ifempty{#1}{\delta}{\delta_{#1}}}}%
\newcommand{\ccsti}[1][]{\tcFmt{\ifempty{#1}{\delta'}{\delta'_{#1}}}}%
\newcommand{\ccstii}[1][]{\tcFmt{\ifempty{#1}{\delta''}{\delta''_{#1}}}}%

\newcommand{\ccstO}[1][]{\tcFmt{\ifempty{#1}{\delta_\text{O}}{{\delta_\text{O}}_{#1}}}}%
\newcommand{\ccstOi}[1][]{\tcFmt{\ifempty{#1}{\delta'_\text{O}}{{\delta'_\text{O}}_{#1}}}}%
\newcommand{\ccstOii}[1][]{\tcFmt{\ifempty{#1}{\delta''_\text{O}}{{\delta''_\text{O}}_{#1}}}}%

\newcommand{\ccstI}[1][]{\tcFmt{\ifempty{#1}{\delta_\textrm{I}}{{\delta_\text{I}}_{#1}}}}%
\newcommand{\ccstIi}[1][]{\tcFmt{\ifempty{#1}{\delta'_\text{I}}{{\delta'_\text{I}}_{#1}}}}%
\newcommand{\ccstIii}[1][]{\tcFmt{\ifempty{#1}{\delta''_\text{I}}{{\delta''_\text{I}}_{#1}}}}%

\newcommand{\crst}[1][]{\tcFmt{\ifempty{#1}{\lambda}{\lambda_{#1}}}}%
\newcommand{\crsti}[1][]{\tcFmt{\ifempty{#1}{\lambda'}{\lambda'_{#1}}}}%
\newcommand{\crstii}[1][]{\tcFmt{\ifempty{#1}{\lambda''}{\lambda''_{#1}}}}%

\newcommand{\crstO}[1][]{\tcFmt{\ifempty{#1}{\lambda_\text{O}}{{\lambda_\text{O}}_{#1}}}}%
\newcommand{\crstOi}[1][]{\tcFmt{\ifempty{#1}{\lambda'_\text{O}}{{\lambda'_\text{O}}_{#1}}}}%
\newcommand{\crstOii}[1][]{\tcFmt{\ifempty{#1}{\lambda''_\text{O}}{{\lambda''_\text{O}}_{#1}}}}%

\newcommand{\crstI}[1][]{\tcFmt{\ifempty{#1}{\lambda_\text{I}}{{\lambda_\text{I}}_{#1}}}}%
\newcommand{\crstIi}[1][]{\tcFmt{\ifempty{#1}{\lambda'_\text{I}}{{\lambda'_\text{I}}_{#1}}}}%
\newcommand{\crstIii}[1][]{\tcFmt{\ifempty{#1}{\lambda''_\text{I}}{{\lambda''_\text{I}}_{#1}}}}%

\newcommand{\cVal}[1][]{\tcFmt{\ifempty{#1}{\mathbb{V}}{\mathbb{V}_{#1}}}}%
\newcommand{\cVali}[1][]{\tcFmt{\ifempty{#1}{\mathbb{V'}}{\mathbb{V'}_{#1}}}}%
\newcommand{\cValii}[1][]{\tcFmt{\ifempty{#1}{\mathbb{V''}}{\mathbb{V''}_{#1}}}}%
\newcommand{\cValiii}[1][]{\tcFmt{\ifempty{#1}{\mathbb{V'''}}{\mathbb{V'''}_{#1}}}}%

\newcommand{\cUnit}[1][]{\tcFmt{\ifempty{#1}{t}{t_{#1}}}}%
\newcommand{\cUniti}[1][]{\tcFmt{\ifempty{#1}{t'}{t'_{#1}}}}%
\newcommand{\cUnitii}[1][]{\tcFmt{\ifempty{#1}{t''}{t''_{#1}}}}%

\newcommand{\rn}{\mathbb{R}}
\newcommand{\cs}{\mathbf{C}}
\newcommand{\cValUpd}[3]{#1 {} [#2 \mapsto #3]}
\newcommand{\ccstSubt}[3]{#1 {} [#2/#3]}

\newcommand{\Rust}[0]{\textsc{Rust}\xspace}

\newcommand{\Python}[0]{\textsc{Python}\xspace}

\newcommand{\Links}[0]{\textsc{Links}\xspace}
\newcommand{\Ensemble}[0]{\textsc{Ensemble}\xspace}
\newcommand{\Esterel}[0]{\textsc{Esterel}\xspace}
\newcommand{\Lucid}[0]{\textsc{Lucid}\xspace}
\newcommand{\Synchrone}[0]{\textsc{Synchrone}\xspace}
\newcommand{\Lustre}[0]{\textsc{Lustre}\xspace}

\newcommand{\stDelegate}[2]{({#1}, {#2})}

\newcommand{\gtCommTransitRaw}[4]{%
  \gtFmt{%
    {#1} {\rightsquigarrow} {#2}{:}{#4}%
    \left\{%
      {#3}%
    \right\}%
  }%
}%

\newcommand{\gtCommTChoice}[4]{%
  \gtFmt{%
    \gtMsgFmt{#1}\ifempty{#2}{}{({#2})}%
    \ifempty{#3}{}{\{#3\}}%
    \ifempty{#4}{}{\vphantom{x} \!\gtSeq\! {#4}}%
  }%
}%

\newcommand{\gtCommTChoiceSmall}[4]{%
  \gtFmt{%
    \gtMsgFmt{#1}\ifempty{#2}{}{({#2})}%
     \ifempty{#3}{}{\{#3\}}%
    \ifempty{#4}{}{\vphantom{x} \!\gtSeq\! {#4}}%
  }%
}%

\newcommand{\gtCommT}[7]{%
  \gtFmt{%
    \gtCommRaw{#1}{#2}{%
      \gtCommTChoice{#4}{#5}{#6}{#7}%
    }_{#3}%
  }%
}%

\newcommand{\gtCommTSmall}[7]{%
  \gtFmt{%
    \gtCommRaw{#1}{#2}{%
      \gtCommTChoiceSmall{#4}{#5}{#6}{#7}%
    }_{#3}%
  }%
}%

\newcommand{\gtCommTTransit}[8]{%
  \gtFmt{%
    \gtCommTransitRaw{#1}{#2}{%
      \gtCommTChoice{#4}{#5}{#6}{#7}%
    }{#8}_{#3}%
  }%
}%

\newcommand{\iruleStSubSort}{Sub-\stFmt{S}}

\newcommand{\stTChoice}[3]{\stLabFmt{#1}\ifempty{#2}{}{\stFmt{({#2})}}\ifempty{#3}{}{\stFmt{\{#3\}}}}%
\definecolor{taColor}{rgb}{0.0, 0.0, 0.0}%
\newcommand{\taFmt}[1]{\ensuremath{{\color{taColor}#1}}\xspace}%

\newcommand{\timeLab}[1][]{\taFmt{\ifempty{#1}{t}{t_{#1}}}}%

\definecolor{roleColor}{rgb}{0.5, 0.0, 0.0}
\newcommand{\ltsSubject}[1]{{\color{roleColor} \operatorname{subject}({#1})}}%
\newcommand{\gtWithTime}[2]{\langle{#1};{#2}\rangle}

\newcommand{\gtMove}[1][\phantom{\stEnvAnnotGenericSym}]{\xrightarrow{#1}} %
\newcommand{\gtMoveStar}[1][]{\ifempty{#1}{\gtMove[]{}^{\!\!\!*}}{\gtMove[]{#1}^{\!*}}} %

\newcommand{\iruleGtMove}[1]{GR-{#1}}
\newcommand{\iruleGtMoveTime}[0]{\iruleGtMove{t}}

\newcommand{\iruleGtMoveOut}[0]{\iruleGtMove{$\oplus$}}
\newcommand{\iruleGtMoveIn}[0]{\iruleGtMove{$\&$}}
\newcommand{\iruleGtMoveRec}[0]{\iruleGtMove{$\mu$}}

\newcommand{\iruleGtMoveCtx}[0]{\iruleGtMove{Ctx-i}}
\newcommand{\iruleGtMoveCtxi}[0]{\iruleGtMove{Ctx-ii}}

\newcommand{\stEnvUpd}[3]{\stFmt{#1 {} [#2 \mapsto #3]}}

\newcommand{\stEnvMoveWithSession}[1][]{
\ifempty{#1}{}{
    \mathrel{\stFmt{\to}_{\!#1}}}}
\newcommand{\stEnvMoveTWithSession}[1][]{
\ifempty{#1}{}{
    \mathrel{\stFmt{\to}_{\!#1}}}}
\newcommand{\stEnvMoveWithSessionStar}[1][]{%
  \ifempty{#1}{}{%
    \mathrel{\stFmt{\to^{\!*}}_{\!\!\!{#1}}}}}

\newcommand{\mpSessionQueueO}[3]{\mpSessionQueue{\mpChanRole{#1}{#2}}{#3}}
\newcommand{\mpQueueOElem}[3]{%
  \mpFmt{{#1}{!}{#2}\ifempty{#3}{}{\langle{#3}\rangle}}}%

\newcommand{\procSubject}[1]{{\color{roleColor} \operatorname{subjP}({#1})}}%

\newcommand{\romanF}[1]{\ensuremath{\operatorname{\text{\rmfamily\bfseries{#1}}}}}
\newcommand{\delay}[2]{\ifempty{#1}{}{\mpFmt{\romanF{delay}({#1})}{\ifempty{#2}{}{\ \mpSeq \ {#2}}}}}
\newcommand{\mpFailedP}[1]{\ifempty{#1}{}\mpFmt{\boldsymbol{\mathtt{timeout}}}\lbrack{#1}\rbrack}%

\newcommand{\mpTBranch}[7]{%
  \mpFmt{%
   {{#1}^{#7}[\roleFmt{#2}]} \mathbin{\!\ifempty{\sum}{\&}{\sum_{#3}}\!}%
    \ifempty{#3}{%
      \mpChoice{#4}{#5}{#6}%
    }{%
      \mpChoice{#4}{#5}{#6}%
    }%
  }%
}%

\newcommand{\mpTSel}[6]{%
  \mpFmt{%
    {#1}^{#6}[\roleFmt{#2}] \mathbin{\!\oplus\!}%
    \mpChoiceNoBind{#3}{#4}{#5}%
  }%
}%

\newcommand{\mpSTSel}[5]{%
  \mpFmt{%
    {#1}^{#5}[\roleFmt{#2}] \mathbin{\!\oplus\!}%
    \mpChoiceRaw{#3}{#4}%
  }%
}%

\newcommand{\mpSTBranch}[5]{%
  \mpFmt{%
    {#1}^{#5}[\roleFmt{#2}]%
    \mpChoiceRaw{#3}{#4}%
  }%
}%

\newcommand{\mpSTBranchOut}[4]{%
  \mpFmt{%
    {#1}^{#4}[\roleFmt{#2}]%
    \mpChoiceRawN{#3}%
  }%
}%

\newcommand{\mpSTSelN}[4]{%
  \mpFmt{%
    {#1}^{#4}[\roleFmt{#2}] \mathbin{\!\oplus\!}%
    \mpChoiceRawN{#3}%
  }%
}%

\newcommand{\mpnonTMove}{\rightarrowtail}%
\newcommand{\mpMoveTime}{\rightharpoonup}
\newcommand{\mpMoveTimeStar}{\mathrel{\mpMoveTime{}^{\!\!\!*}}}%
\newcommand{\mpnonTMoveStar}{\mathrel{\mpnonTMove{}^{\!\!\!*}}}%
\newcommand{\mpNotMoveTimeP}[1]{\mpFmt{#1}\!\not\!\!{\mpMoveTime}}%
\newcommand{\mpnonTNotMoveP}[1]{\mpFmt{#1}\!\not\!\!{\mpnonTMove}}%

\newcommand{\stEnvAssoc}[3]{\tcFmt{{#1} \mathrel{\stFmt{\sqsubseteq}_{#3}} {#2}}}

\newcommand{\iruleMPRedInstant}{R-Ins}%
\newcommand{\iruleMPRedTimeConsume}{R-TC}%

\newcommand{\iruleMPRedOut}{R-Out}
\newcommand{\iruleMPRedIn}{R-In}
\newcommand{\iruleMPRedErr}{R-Err}
\newcommand{\iruleMPRedDet}{R-Det}
\newcommand{\iruleMPRedDelay}{R-Time}

\newcommand{\iruleMPRedCongrTime}{R-$\equiv$T}%
\newcommand{\iruleMPRedCanIn}{C-In}%
\newcommand{\iruleMPRedCanQ}{C-Queue}%
\newcommand{\iruleMPRedTryFail}{R-Fail}%
\newcommand{\iruleMPRedFailCatch}{R-FailCat}%

\newcommand{\mpEtxApp}[2]{{#1}\!\left[{#2}\right]}%

\newcommand{\timePass}[2]{\mpFmt{\Psi_{#1}(#2)}}

\newcommand{\stSpecf}[1][]{\stFmt{\ifempty{#1}{\tau}{\tau_{#1}}}}%
\newcommand{\stSpecfi}[1][]{\stFmt{\ifempty{#1}{\tau'}{\tau'_{#1}}}}%

\newcommand{\stCPair}[2]{\stFmt{({#1}{\mathbin{\mathbf{,\,}}}{#2})}}%

\newcommand{\gtQProj}[2]{%
  {\color{stColor}\gtFmt{#1}{\upharpoonright^{\stFmt{Q}}_{#2}}}}%

\definecolor{lightblue}{rgb}{0.68, 0.85, 0.9}
\definecolor{lightcyan}{rgb}{0.88, 1.0, 1.0}

\newcommand{\highlight}[2][\highlightColour]{\mathchoice%
  {\setlength{\fboxsep}{0pt}\colorbox{#1}{$\displaystyle#2$}}%
  {\setlength{\fboxsep}{0pt}\colorbox{#1}{$\textstyle#2$}}%
  {\setlength{\fboxsep}{0pt}\colorbox{#1}{$\scriptstyle#2$}}%
  {\setlength{\fboxsep}{0pt}\colorbox{#1}{$\scriptscriptstyle#2$}}}%

\definecolor{purple}{HTML}{D62728}
\definecolor{bordeaux}{HTML}{9467BD}
\newcommand{\highlightbox}[2]{\mathchoice%
  {\setlength{\fboxsep}{0pt}\colorbox{#1}{$\displaystyle#2$}}%
  {\setlength{\fboxsep}{0pt}\colorbox{#1}{$\textstyle#2$}}%
  {\setlength{\fboxsep}{0pt}\colorbox{#1}{$\scriptstyle#2$}}%
  {\setlength{\fboxsep}{0pt}\colorbox{#1}{$\scriptscriptstyle#2$}}}%

\newcommand{\iruleMPTry}{T-Try}%
\newcommand{\iruleMPKill}{T-Kill}%
\newcommand{\iruleMPCancel}{T-Cancel}

\newcommand{\iruleMPClock}{T-$\ccst$}%
\newcommand{\iruleMPTime}{T-$\cUnit$}%

\definecolor{stTColor}{rgb}{0.09, 0.45, 0.27}%
\newcommand{\stTFmt}[1]{\stFmt{#1}}%

\newcommand{\stTEnvAnnotGenericSymT}{\stFmt{\alpha}}%

\newcommand{\stEnvQTAnnotQueueSym}{\stTFmt{!}}%
\newcommand{\stEnvQTAnnotRecvSym}{\stTFmt{,}}%
\newcommand{\stEnvQTQueueAnnotSmall}[3]{\stTFmt{\mpS{:}{#1}{\stEnvQTAnnotQueueSym}{#2}{:}{#3}}}%
\newcommand{\stEnvQTRecvAnnotSmall}[3]{\stFmt{\mpS{:}{#1}{\stEnvQTAnnotRecvSym}{#2}{:}{#3}}}%
\newcommand{\stEnvQTMoveAnnot}[1]{\mathrel{\stTFmt{\xrightarrow{#1}}}}
\newcommand{\stEnvQTMoveGenAnnotT}{\stEnvQTMoveAnnot{\stTEnvAnnotGenericSymT}}%
\newcommand{\stEnvQTMoveTimeAnnot}{\stEnvQTMoveAnnot{\timeLab}}%
\newcommand{\stEnvQTMoveQueueAnnot}[3]{%
  \stEnvQTMoveAnnot{\stEnvQTQueueAnnotSmall{#1}{#2}{#3}}%
}%
\newcommand{\stEnvQTMoveRecvAnnot}[3]{%
  \stEnvQTMoveAnnot{\stEnvQTRecvAnnotSmall{#1}{#2}{#3}}%
}%

\newcommand{\iruleTCtxSend}{$\stEnv$-$\oplus$}%
\newcommand{\iruleTCtxRcv}{$\stEnv$-$\&$}%
\newcommand{\iruleTCtxTime}{$\stEnv$-,T}%
\newcommand{\iruleTCtxTimeSession}{$\stEnv$-Ts}
\newcommand{\iruleTCtxTimeQ}{$\stEnv$-Tq}%
\newcommand{\iruleTCtxTimeCombined}{$\stEnv$-Tc}
\newcommand{\iruleTCtxStruct}{$\stEnv$-struct}%

\newcommand{\iruleMPFailed}{T-Failed}%

\newcommand{\stQType}[1][]{\stFmt{\ifempty{#1}{\mathcal{M}}{\mathcal{M}_{#1}}}}
\newcommand{\stQTypei}[1][]{\stFmt{\ifempty{#1}{\mathcal{M}'}{\mathcal{M}'_{#1}}}}
\newcommand{\stQTypeii}[1][]{\stFmt{\ifempty{#1}{\mathcal{M}''}{\mathcal{M}''_{#1}}}}

\newcommand{\stQEmptyType}{\stFmt{\oslash}}

\newcommand{\iruleSubEQueueType}{Sub-$\stFmt{\mathcal{M}}$}%
\newcommand{\iruleSubESessionType}{Sub-\stFmt{T}}%
\newcommand{\iruleSubEQueueTypeEmpty}{Sub-\stFmt{$\stQEmptyType$}}%

\newcommand{\iruleTCtxCongX}{$\stEnv$-,x}%
\newcommand{\iruleTCtxCongCombined}{$\stEnv$-,$\tau$}%

\newcommand{\iruleMPResPropG}{T-$\nu$-$\gtG$}%

\newcommand{\stEnvQTypeEndPred}{\operatorname{End}}%
\newcommand{\stEnvQTypeEndP}[1]{\stEnvQTypeEndPred(\stFmt{#1})}%
\newcommand{\stEquiv}{\mathrel{\stFmt{\equiv}}}%

\def\aka{a.k.a.\@\xspace}%
\def\eg{e.g.\@\xspace}%
\def\ie{i.e.\@\xspace}%
\def\wrt{w.r.t.\@\xspace}%
\def\etal{et al.\xspace}%

\usepackage{etoolbox}%

\usepackage[nomargin,inline,index,%
  status=draft %
]{fixme} %
\fxusetheme{colorsig}
\FXRegisterAuthor{fxAS}{anfxAS}{AS}%
\FXRegisterAuthor{fxAB}{anfxAB}{AB}%
\FXRegisterAuthor{fxNY}{anfxNY}{NY}%
\FXRegisterAuthor{fxFZ}{anfxFZ}{FZ}
\FXRegisterAuthor{fxPH}{anfxPH}{PH}

\makeatletter
\DeclareFontFamily{OMX}{MnSymbolE}{}
\DeclareSymbolFont{MnLargeSymbols}{OMX}{MnSymbolE}{m}{n}
\SetSymbolFont{MnLargeSymbols}{bold}{OMX}{MnSymbolE}{b}{n}
\DeclareFontShape{OMX}{MnSymbolE}{m}{n}{
    <-6>  MnSymbolE5
   <6-7>  MnSymbolE6
   <7-8>  MnSymbolE7
   <8-9>  MnSymbolE8
   <9-10> MnSymbolE9
  <10-12> MnSymbolE10
  <12->   MnSymbolE12
}{}
\DeclareFontShape{OMX}{MnSymbolE}{b}{n}{
    <-6>  MnSymbolE-Bold5
   <6-7>  MnSymbolE-Bold6
   <7-8>  MnSymbolE-Bold7
   <8-9>  MnSymbolE-Bold8
   <9-10> MnSymbolE-Bold9
  <10-12> MnSymbolE-Bold10
  <12->   MnSymbolE-Bold12
}{}

\let\llangle\@undefined
\let\rrangle\@undefined
\DeclareMathDelimiter{\llangle}{\mathopen}%
                     {MnLargeSymbols}{'164}{MnLargeSymbols}{'164}
\DeclareMathDelimiter{\rrangle}{\mathclose}%
                     {MnLargeSymbols}{'171}{MnLargeSymbols}{'171}
\makeatother

\usepackage[inline]{enumitem}%

\usepackage{nicefrac}%

\usepackage{proof}%

\usepackage[most]{tcolorbox}%
\newtcolorbox{myframe}[1][]{
  enhanced,
  arc=0pt,
  outer arc=0pt,
  colback=white,
  boxrule=0.5pt,
  boxsep=0mm,
  left=1mm,
  right=1mm,
  top=0.5mm,
  bottom=0.5mm,
  #1
}

\lstdefinelanguage{Scribble}{%
  basicstyle=\footnotesize\ttfamily,
  stringstyle=\color{Blue},
  showstringspaces=false,
  keywords={nested,new,calls,and,as,at,by,catches,choice,continue,do,from,global,import,instantiates,interruptible,local,module,or,par,protocol,rec,role,sig,throws,to,type,with,int,aux},
  morestring=[b]",
  morestring=[b]',
  morecomment=[l][\color{greencomments}]{//},
}

\lstdefinelanguage{nuScr}{%
  basicstyle=\footnotesize\ttfamily,
  stringstyle=\color{Blue},
  showstringspaces=false,
  keywords={
    nested,new,calls,and,as,at,by,catches,choice,continue,do,from,global,import,instantiates,interruptible,local,module,or,par,protocol,rec,role,sig,throws,to,type,with,int,aux,
    safe
  },
  morestring=[b]",
  morestring=[b]',
  morecomment=[l][\color{greencomments}]{//},
  morecomment=[s][\color{magenta}]{(*}{*)},
}

\usepackage{stmaryrd}
\usepackage{thmtools}%
\usepackage{arydshln}

\usepackage{xifthen}%
\usepackage{xspace}%

\usepackage{mathtools}%

\usepackage{balance}%

\hypersetup{hidelinks}

\usetikzlibrary{automata, positioning, arrows, calc, fit}
\tikzset{
  >=stealth,
  node distance=2cm,
  every state/.style={thick, fill=gray!10},
  initial text=$ $,
}

\usepackage[export]{adjustbox}

\Crefname{section}{\S\!}{\S\!}%
\Crefname{subsection}{\S\!}{\S\!}%
\Crefname{subsubsection}{\S\!}{\S\!}%
\Crefname{appendix}{Appendix \S\!}{Appendix \S\!}
\Crefname{definition}{Def.\@}{Defs.\@}%
\Crefname{figure}{Fig.\@}{Figs.\@}%
\Crefname{example}{Ex.\@}{Exs.\@}%
\Crefname{corollary}{Cor.\@}{Cors.\@}%
\Crefname{theorem}{Thm.\@}{Thms.\@}%
\Crefname{proposition}{Prop.\@}{Props.\@}%

\crefname{section}{\S\!}{\S\!}%
\crefname{subsection}{\S\!}{\S\!}%
\crefname{subsubsection}{\S\!}{\S\!}%
\crefname{appendix}{Appendix \S\!}{Appendix \S\!}
\crefname{definition}{Def.\@}{Defs.\@}%
\crefname{figure}{Fig.\@}{Figs.\@}%
\crefname{example}{Ex.\@}{Exs.\@}%
\crefname{corollary}{Cor.\@}{Cors.\@}%
\crefname{theorem}{Thm.\@}{Thms.\@}%
\crefname{proposition}{Prop.\@}{Props.\@}%

\usepackage{magicvariables}

\newtoggle{full}
\settoggle{full}{true}

\iftoggle{full}
{
	
}{
	
}

\makeatletter
\newcommand{\iftoggleverb}[1]{%
  \ifcsdef{etb@tgl@#1}
    {\csname etb@tgl@#1\endcsname\iftrue\iffalse}
    {\etb@noglobal\etb@err@notoggle{#1}\iffalse}%
}
\makeatother

\begin{document}

\lstset{
	language=Rust,
	style=colouredRust,
	basicstyle=\ttfamily\small,
	columns=fullflexible,
	xleftmargin=\parindent,
	numbers=left,
	moredelim=[is][\uwave]{~}{~}
}

\maketitle
\begin{abstract}
    Session types using~\emph{affinity} and~\emph{exception handling}
mechanisms have been developed to ensure the communication safety
of protocols implemented in concurrent and distributed programming
 languages.
Nevertheless,  
current affine session types are inadequate for specifying real-world asynchronous protocols, 
as they are usually imposed by \emph{time constraints} which enable \emph{timeout exceptions}
 to prevent indefinite blocking while awaiting valid messages.
This paper proposes the first formal integration of~\emph{affinity}, 
\emph{time constraints},
 \emph{timeouts}, and~\emph{time-failure handling} based on multiparty session types
 for supporting reliability in asynchronous distributed systems.
 With this theory, we statically guarantee that
 asynchronous timed communication is deadlock-free, communication safe, while 
 being \emph{fearless} -- never hindered by timeout errors or abrupt terminations.

To implement our theory,
 we introduce~\timedmulticrusty,
 a~\Rust toolchain designed to facilitate the implementation
 of safe affine timed protocols.
 \timedmulticrusty leverages generic types and the~\timeRust library to handle timed communications,
 integrated with optional types for affinity.
 We evaluate our approach by extending diverse examples from the literature
 to incorporate time and timeouts, 
 demonstrating that our solution incurs negligible overhead compared with an untimed implementation.
 We also showcase the~\emph{correctness by construction} of our approach
 by implementing various real-world use cases, including a remote data protocol 
 from the Internet of Remote Things domain, %
 as well as protocols from real-time systems like 
 Android motion sensors and smartwatches.

\end{abstract}

\section{Introduction}
\label{sec:introduction}

\subparagraph*{Background}
The growing prevalence of distributed programming has emphasised the significance
of prioritising~\emph{reliability}
in distributed systems.
Dedicated research efforts focus on enhancing reliability through the study and modelling of failures.
This research enables the design of more resilient distributed systems, 
capable of effectively handling failures
and ensuring reliable operation.

A lightweight, type-based methodology, which
ensures basic reliability -- \emph{safety}
in distributed communication systems, is
\emph{session types}~\cite{DBLP:conf/esop/HondaVK98}.
This type discipline is further advanced by
\emph{Multiparty Session Types}~(\MPST)~\cite{honda2008Multiparty,HYC16},
which enable the specification and verification of communication protocols
among multiple message-passing processes in concurrent and distributed systems.
\MPST ensure that
protocols are designed
to prevent common safety errors,
\ie deadlocks and communication mismatches during interactions among many
participants~\cite{honda2008Multiparty,HYC16,DBLP:journals/pacmpl/ScalasY19}.
By adhering to a specified~\MPST protocol,
participants (\aka end-point programs) can communicate reliably and efficiently. 
From a practical perspective, %
\MPST have been implemented in
various programming
languages~\cite{castro2019Distributed,kouzapas2016Typechecking,
lagaillardie2022Affine,neykova2017timed,viering2021Multiparty,zhouCFSM2021},
facilitating their applications and providing
safety guarantees in
real-world programs.

Nevertheless,
tackling the challenges of unreliability and failures remains
a significant issue for session types.
Most session type systems operate under the assumption of flawless
and reliable communication without failures.
To address this limitation, recent works~\cite{mostrous2018Affine,fowler2019Exceptional,harvey2021Multiparty,lagaillardie2022Affine}
have developed~\emph{affine session types}
 by incorporating the~\emph{affinity} mechanism that explicitly accounts for and handles unreliability
and failures within session type systems.
Unlike linear types that must be used \emph{exactly} once, affine types can be used \emph{at most} once,
enabling the safe dropping of subsequent types and %
the premature termination
of a session in the presence of protocol execution errors.

In most real-life devices and platforms, communications are predominantly asynchronous:
inner tasks and message transfers may take time. When dealing with such communications,
it becomes crucial to incorporate~\emph{time constraints} and implement~\emph{timeout failure handling}
for each operation. This is necessary to avoid potential blockages where a process might wait indefinitely 
for a message from a non-failed process.  %
While various works, as explained later, %
address time conditions and timeouts in session types,
it is surprising that none of the aforementioned works on affine session types
tackles the handling of %
timeout failures during protocol execution.

\vspace{-1em}
\subparagraph*{This Paper}  %
introduces a new framework, \emph{affine timed multiparty session types} (\ATMP),
to address the challenges of timeouts, disconnections
and other failures in asynchronous communications:
\begin{enumerate*}[label=(\arabic*)]
\item We propose \ATMP, an extension of asynchronous \MPST that incorporates time specifications,
affinity, and mechanisms for handling exceptions, thus facilitating effective management of failures, with a particular focus on timeouts.
Additionally, we
demonstrate that properties from~\MPST,
\ie~\emph{type safety}, \emph{protocol conformance},
and~\emph{deadlock-freedom}, are guaranteed for well-typed processes, even
in the presence of timeouts and their corresponding handling mechanism;
\item We present~\timedmulticrusty,
our~\Rust toolchain designed for building asynchronous timed multiparty protocols under~\ATMP:
\timedmulticrusty enables the implementation of protocols adhering to the properties of \ATMP.
\end{enumerate*}

The primary focus of~\ATMP
lies in
effectively~\emph{handling} timeouts during process execution,
in contrast to the approaches in~\cite{DBLP:conf/concur/BocchiYY14, bocchi2019Asynchronous},
which aim to completely \emph{avoid} time failures. Bocchi \etal~\cite{DBLP:conf/concur/BocchiYY14} introduce 
time conditions in~\MPST to ensure precise timing in communication protocols, 
while their subsequent work~\cite{bocchi2019Asynchronous} extends \emph{binary} timed session types
to incorporate timeouts, allowing for more robust handling of time constraints. %
Yet, they
adopt strict requirements to~\emph{prevent} timeouts.
In~\cite{DBLP:conf/concur/BocchiYY14}, \emph{feasibility} and~\emph{wait-freedom}
are required in their protocol design.
Feasibility requires precise time specifications for protocol termination,
while wait-freedom prohibits overlapping time windows for senders and receivers in a protocol,
which is not practical in real-world applications. %
Similarly, in~\cite{bocchi2019Asynchronous}, strong conditions including~\emph{progress}
of an entire set of processes
and~\emph{urgent receive} are imposed.
The progress property is usually \emph{undecidable},
and the urgent receive condition,
which demands immediate message reception upon availability,
is infeasible with asynchronous communication.

Recently,
\cite{ESOP23MAGPi} proposes the inclusion of timeout as the unique failure notation in~\MPST,
offering flexibility in handling failures.
Time also plays a role in \emph{synchronous} communication systems, where~\cite{DBLP:journals/pacmpl/IraciandCHZ2023} develops \emph{rate}-based \emph{binary} session types,
ensuring \emph{synchronous} message exchanges at the same \emph{rate}, \ie within the same time window.
However, in both~\cite{ESOP23MAGPi} and~\cite{DBLP:journals/pacmpl/IraciandCHZ2023},
time constraints are not integrated into types and static type checking,
resulting in the specifications lacking the ability to guide time behaviour.
Additionally,
the model used in~\cite{DBLP:journals/pacmpl/IraciandCHZ2023} assumes that all communications and computations are non-time-consuming, \ie with zero time cost, making it unfeasible in distributed systems.

By the efficient integration of time and failure handling mechanisms in our framework,
none of those impractical requirements outlined in~\cite{DBLP:conf/concur/BocchiYY14, bocchi2019Asynchronous}
is necessary. In~\ATMP, when a process encounters a timeout error,
a mechanism for handling time failures is triggered,
notifying all participants about the timeout,
leading to the termination of those participants and ultimately ending the session.
Such an approach guarantees that participants consistently
reach the~\emph{end of the protocol},
as the communication session is entirely dropped
upon encountering a timeout error.
As a result,
every process can terminate successfully,
reducing the risk of indefinite blockages, even with timeouts.
Additionally, in our system, time constraints over local clocks are incorporated with types to effectively model asynchronous timed communication, addressing the limitations in~\cite{ESOP23MAGPi,DBLP:journals/pacmpl/IraciandCHZ2023}.

Except for~\cite{DBLP:journals/pacmpl/IraciandCHZ2023},
the aforementioned works on timed session types focus more on theory,
lacking implementations.
To bridge this gap on the practical side, 
we provide~\timedmulticrusty,
a~\Rust implementation of~\ATMP %
designed for secure timed communications.
\timedmulticrusty makes use of affine timed meshed channels,
a communication data structure that integrates time constraints and clock utilisation.
Our toolchain relies on macros and native generic types
to ensure that asynchronous protocols are inherently \emph{correct by construction}.
In particular,
\timedmulticrusty performs compile-time verification to guarantee that,
at any given point in the protocol,
each isolated pair of participants comprises one sender and one receiver
with corresponding time constraints.
Additionally,
we employ affine asynchronous primitives and native optional types
to effectively handle \emph{runtime} timeouts and errors. %

To showcase the capabilities and expressiveness of our toolchain,
we evaluate~\timedmulticrusty through examples from the literature, and further 
case studies including a remote data protocol from an Internet of Remote Things (IoRT) network~\cite{DBLP:journals/sensors/ChenZLCJGYAN22}, 
a servo web protocol from a web engine~\cite{servoWebEngineBuggy}, and protocols from real-time systems such as Android motion sensor~\cite{androidMotionSensors}, PineTime smartwatch~\cite{Pine64}, and keyless entry~\cite{DBLP:journals/tches/WoutersMAGP19}. 
Our comparative analysis with 
a~\Rust implementation of affine \MPST without time~\cite{lagaillardie2022Affine}
reveals that~\timedmulticrusty exhibits minimal overhead
while providing significantly strengthened property checks.

\vspace{-1em}
\subparagraph*{Contributions and Structure}
\textbf{\Cref{sec:overview}} offers a comprehensive overview of our theory and toolchain,
employing the~\emph{remote data}
as a running example.
\textbf{\Cref{sec:session-calculus}} provides a session $\pi$-calculus for \ATMP that incorporates timeout, affinity, asynchrony, and failure handling mechanisms.
\textbf{\Cref{sec:aat-mpst-type-system-popl}}
introduces an extended theory of asynchronous multiparty session types
with time annotations, elucidating the behavioural correspondence between timed types 
in~\Cref{lem:sound_proj,lem:comp_proj}. 
Additionally,
we present a typing system for~\ATMP session
$\pi$-calculus, %
and demonstrate the properties of typed processes, including subject reduction~(processes preserve well-typedness),  %
session fidelity~(processes adhere to their types),
and deadlock-freedom~(processes never get stuck),
in~\Cref{lem:sr_global,lem:aat-mpst-session-fidelity-global,lem:aat-mpst-process-df-proj}. %
\textbf{\Cref{sec:implementation:implementation}} 
delves into the design and usage of \timedmulticrusty,
our \Rust implementation of \ATMP. 
\textbf{\Cref{sec:evaluation}} showcases 
the compilation and execution benchmarks of~\timedmulticrusty,
based on selected case studies. 
Finally, \textbf{\Cref{sec:related_work}} concludes the paper
by discussing related work, and offering conclusions and potential future work. %
Full proofs,  auxiliary material,
and more details of~\timedmulticrusty can be found in Appendix.
Additionally, our toolchain and evaluation examples are available in an artifact,
which can be accessed on \href{https://zenodo.org/doi/10.5281/zenodo.11032195}{\color{blue}{Zenodo}}. 
For those interested in the source files, they are available
on \href{https://github.com/NicolasLagaillardie/ECOOP24-Artefact}{\color{blue}{GitHub}}.

\section{Overview}
\label{sec:overview}
In this section, we give an overview of affine timed multiparty session types (\ATMP) and \timedmulticrusty, our toolchain for implementing affine timed asynchronous protocols.
First, 
we share 
a real-world example inspiring our work on
affine asynchronous timed communication.

\Cref{fig:implementation:remote_data} depicts
our running example, \emph{remote data}.
This real-world scenario is sourced from a satellite-enabled Internet of Remote
Things network~\cite{DBLP:journals/sensors/ChenZLCJGYAN22}, and describes
data transmissions among a \emph{Sensor} ({\small $\roleFmt{Sen}$}),
a \emph{Server} ({\small $\roleFmt{Ser}$}), %
and a \emph{Satellite} ({\small $\roleFmt{Sat}$}):  {\small $\roleFmt{Ser}$} aims to periodically retrieve
data gathered by {\small $\roleFmt{Sen}$}
via {\small $\roleFmt{Sat}$}.
The protocol revolves around a loop initiated by {\small $\roleFmt{Ser}$},
which faces a decision:
either retrieve data or end the protocol.
In the former scenario,
{\small $\roleFmt{Ser}$} requests data retrieval
from {\small $\roleFmt{Sen}$} with a message labelled {\small \emph{GetData}} via {\small $\roleFmt{Sat}$}
within the time window of {\small 5} and {\small 6} time units, as indicated by
clock constraints~(\ie~{\small $5 \leq C_{\roleFmt{Ser}} \leq 6$},
where {\small $C_{\roleFmt{Ser}}$} is the clock associated with {\small $\roleFmt{Ser}$}).
Upon receiving this request,
{\small $\roleFmt{Sen}$} responds by sending the data with
a message labelled {\small \emph{Data}}
to {\small $\roleFmt{Ser}$} through {\small $\roleFmt{Sat}$} within
{\small 6} and {\small 7} time units, followed by clock resets denoted
as reset predicates~(\ie~{\small $C_{\roleFmt{Ser}} := 0$}, resetting the clock
to {\small 0}). 
In the alternative branch,
{\small $\roleFmt{Ser}$} sends a {\small \emph{Close}} message to
{\small $\roleFmt{Sat}$}, which is then %
forwarded to {\small $\roleFmt{Sen}$},
between {\small 5} and {\small 6} time units.

Our remote data protocol 
includes internal tasks that consume time,
notably {\small $\roleFmt{Sen}$} requiring {\small 5} time units to gather data  before transmitting.
In cases where our protocol lacks a specified timing strategy~(\ie no time requirements), and {\small $\roleFmt{Sen}$} cannot accomplish the data-gathering tasks, it results in indefinite blocking for {\small $\roleFmt{Sat}$} 
and {\small $\roleFmt{Ser}$} as they await the data. 
This  could lead to undesirable outcomes, including partially processed data, data corruption, or incomplete transmission of processed data.
Therefore, incorporating time constraints into communication protocols is imperative, as it better reflects real-world scenarios and ensures practical viability.

\subsection{\ATMP: Theory Overview}
\label{sec:overview:theory}

\begin{figure}[!t]
\begin{subfigure}{0.35\textwidth}
\centering
\tiny
\begin{tikzpicture}
\node (PROGRESS) [align = center] {Processes with\\Affinity, Time, and Timeouts};%
\node (TLT) [above = 1.94cm of PROGRESS, align = center, xshift = -0.13cm] {Timed Local Types};
\node (TGT) [above = 0.76cm of TLT, align = center, xshift = 0.00cm] {Timed Global Type};
\node (NUSCR) [right = 3.1cm of TGT, align = center] {\timednuscr};
\node (CTA) [below = 0.67cm of NUSCR, align = center] {Communicating\\Timed Automata};
\node (MC) [below = 0.71cm of CTA, align = center] {\timedmulticrusty};
\node (PROGRAMS) [below = 0.71cm of MC, align = center] {Programs written\\with~\timedmulticrusty API};

\draw[->] (TGT) -> (NUSCR); 
\draw[->] (CTA) -> (TLT);
\draw[->] (PROGRAMS) -> (PROGRESS);

\draw[align = center,->] (NUSCR) -> node [pos = .5, anchor = east, xshift = -.1cm] (ProjectionAnon) {${\highlightbox{pink}{\textnormal{Projection}}}$} (CTA);
\draw[align = center,->] (CTA) -> node [pos = .5, anchor = east, xshift = -.1cm] (GenerationAnon) {${\highlightbox{pink}{\textnormal{Generation}}}$} (MC);
\draw[align = center,->] (MC) -> node [pos = .5, anchor = east, xshift = -.1cm] (TypeAnon) {${\highlightbox{pink}{\textnormal{Type}}}$\\${\highlightbox{pink}{\textnormal{Checking}}}$} (PROGRAMS);

\draw[align = center, ->] (TGT) edge node[pos = .5, anchor = east, xshift = -.1cm] (ProjectionTypes) {${\highlightbox{pink}{\textnormal{Projection}}}$} (TLT);
\draw[align = center, ->] (TGT) edge node[pos = .5, anchor = west, xshift = .1cm] (ThmsTypes) {Thms~\ref{lem:comp_proj},~\ref{lem:sound_proj}}(TLT);

\begin{pgfonlayer}{background}
\node[fit = (TGT) (TLT) (ProjectionTypes) (ThmsTypes), draw = green_colour_blind, minimum width = 2.4cm] (green_rectangle) {};
\node[fit = (PROGRESS), draw = blue_colour_blind] (blue_rectangle){};
\node[fit = (NUSCR) (CTA) (ProjectionAnon) (GenerationAnon) (TypeAnon) (MC) (PROGRAMS), draw = red_colour_blind, minimum width = 2.15cm] (red_rectangle) {};
\node[above = 0cm of green_rectangle, anchor = south] {Section~\ref{sec:aat-mpst-type-system-popl}};
\node[below = 0cm of blue_rectangle, anchor = north] {Section~\ref{sec:session-calculus}};
\node[above = 0cm of red_rectangle, anchor = south] {Section~\ref{sec:implementation:implementation}};
\draw[align = center, ->] (green_rectangle) edge node[pos = .5, anchor = east]{
Thms~\ref{lem:sr_global},~\ref{lem:aat-mpst-session-fidelity-global},~\ref{lem:aat-mpst-process-df-proj}
} (blue_rectangle);
\draw[align = center, ->] (green_rectangle) edge node[pos = .5, xshift = .4cm, anchor = west] (Inst) {${\highlightbox{pink}{\textnormal{Type Checking}}}$} (blue_rectangle);
\node[fit = (Inst), draw = black] (black_rectangle) {};
\node[below = .32cm of black_rectangle, anchor = south] {Section~\ref{sec:aat-mpst-type-system-popl}};
\end{pgfonlayer}
\end{tikzpicture}
\caption{Top-down view of~\ATMP (left) and~\timedmulticrusty (right).}
\label{fig:theory_methodology}
\end{subfigure}
\hfill
\begin{subfigure}{0.43\textwidth}
\centering
\includegraphics[width=1\textwidth]{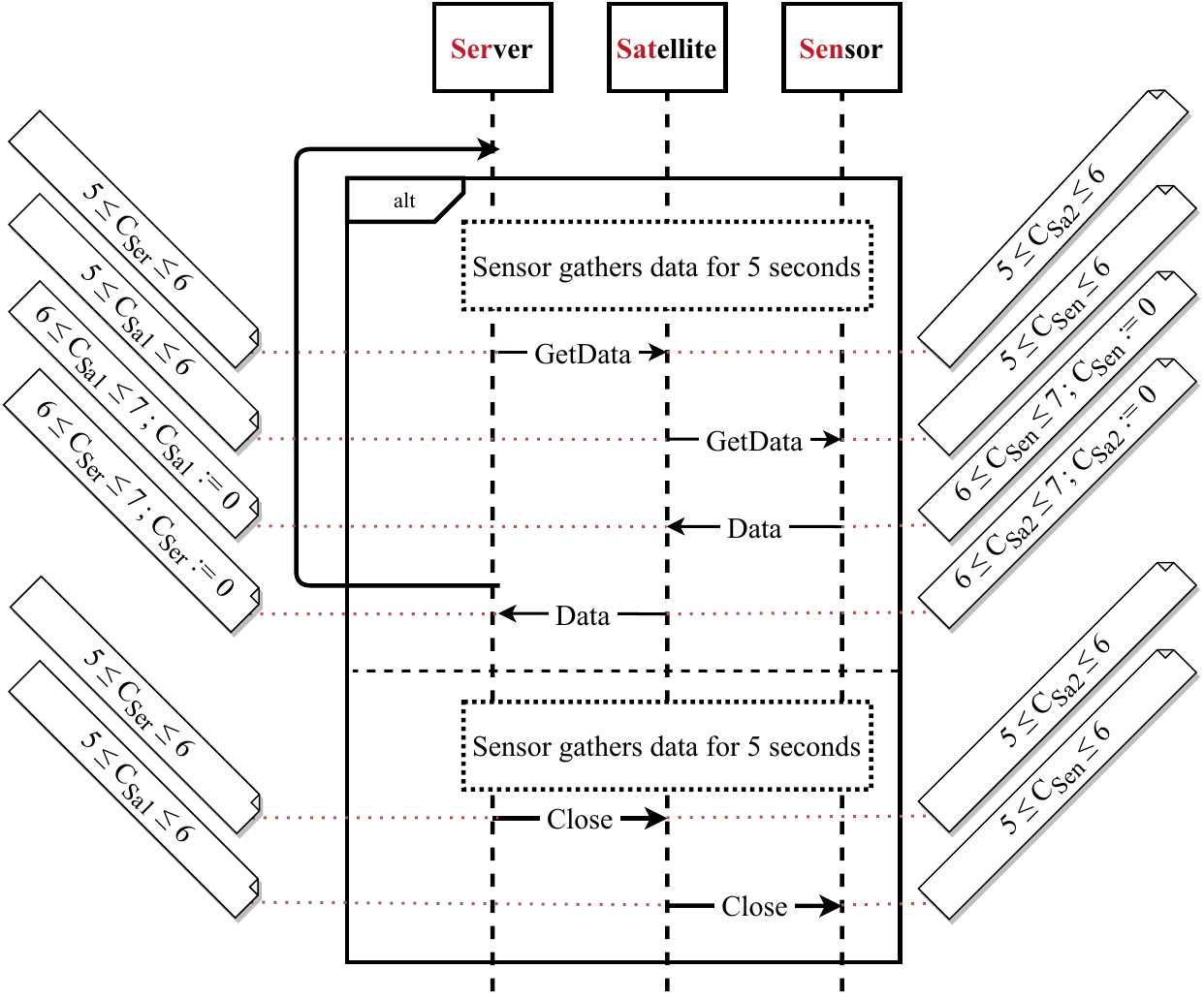}
\caption{Global protocol for \emph{remote data}.}
\label{fig:implementation:remote_data}
\end{subfigure}
\vspace{-.5em}
\caption{Overview of affine asynchronous communication with time.}
\label{fig:general_overview}
\vspace{-1em}
\end{figure}

Our~\ATMP theory follows the \emph{top-down} methodology~\cite{honda2008Multiparty,HYC16}, 
enhancing asynchronous MPST with time features to facilitate timed global and local types.
As shown in~\Cref{fig:theory_methodology}~(left),
we specify %
multiparty
protocols with time as~\emph{timed global types}.
These timed global types are projected into~\emph{timed local types}, which are
then used for type-checking processes
with affine types, time, timeouts, and failure handling, written in a session calculus. %
As an example, we consider a simple communication
scenario derived from %
remote data:
the Satellite~({\small $\roleFmt{Sat}$})
communicates with the server~({\small $\roleFmt{Ser}$})
by sending a~\emph{Data} message~({\small $\gtMsgFmt{Data}$}).
Specifically, {\small $\roleFmt{Sat}$} needs to %
send the message between 6 and 7
time units and reset its clock afterwards, while {\small $\roleFmt{Ser}$} is expected to
receive the message within the same time window and reset its clock accordingly.

\vspace{-1em}
\subparagraph*{Timed Types and Processes}  
This communication behaviour
can be represented by the timed global type $\gtG$: 

\centerline{\(
\small 
   \gtCommT{\roleFmt{Sat}}{\roleFmt{Ser}}{}{
      Data
   }{}{
  \mpFmt{6 \leq C_{\roleFmt{Sat}} \leq 7, C_{\roleFmt{Sat}} := 0, 6 \leq C_{\roleFmt{Ser}} \leq 7, C_{\roleFmt{Ser}} := 0}
   }{
      \gtEnd
   }
   \)}

\noindent
where {\small $C_{\roleFmt{Sat}}$} and {\small $C_{\roleFmt{Ser}}$} denote the clocks of
{\small $\roleFmt{Sat}$} and {\small $\roleFmt{Ser}$}, respectively.
A global type represents a protocol specification involving multiple roles from a global standpoint.

Adhering to the MPST top-down approach,
a timed global type is then \emph{projected} onto timed local
types, which describe communications from the perspective of individual roles.
In our example, $\gtG$ is projected onto two timed local types,
one for each role {\small $\roleFmt{Sat}$} and {\small $\roleFmt{Ser}$}:

\smallskip
\centerline{\(
   \small
   \stT[\roleFmt{Sat}] =
   \roleFmt{Ser}\stFmt{\oplus}
   \stLabFmt{Data}\stFmt{\{ \mpFmt{6 \leq C_{\roleFmt{Sat}} \leq 7, C_{\roleFmt{Sat}} := 0}\}}
   \stSeq
   \stEnd
   \quad
   \stT[\roleFmt{Ser}] =
   \roleFmt{Sat}\stFmt{\&}
   \stLabFmt{Data}\stFmt{\{\mpFmt{6 \leq C_{\roleFmt{Ser}} \leq 7, C_{\roleFmt{Ser}} := 0}\}}
   \stSeq
   \stEnd
   \)}

 \smallskip
   \noindent
Here {\small $\stT[\roleFmt{Sat}]$} indicates
that {\small $\roleFmt{Sat}$} sends~({\small $\stFmt{\oplus}$}) the message {\small $\stLabFmt{Data}$}
to {\small $\roleFmt{Ser}$}
between 6 and 7 time units and then immediately resets its clock $\mpFmt{C_{\roleFmt{Sat}}}$.
Dually,
{\small $\stT[\roleFmt{Ser}]$} denotes {\small $\roleFmt{Ser}$}
receiving~({\small $\stFmt{\&}$}) the message from {\small $\roleFmt{Sat}$}
within the same time frame
and resetting its clock $\mpFmt{C_{\roleFmt{Ser}}}$.

In the final step of the top-down approach, we employ timed local types to conduct type-checking for processes, denoted as {\small $\mpP[i]$}, in the \ATMP session calculus. Our session calculus extends the framework for \emph{affine multiparty session types}~(\AMPST)~\cite{lagaillardie2022Affine}
by incorporating processes that model time, %
timeouts, and asynchrony.
In our example, {\small $\stT[\roleFmt{Sat}]$} and {\small $\stT[\roleFmt{Ser}]$} are used for the type-checking of {\small $\mpChanRole{\mpS}{\roleFmt{Sat}}$} and {\small $\mpChanRole{\mpS}{\roleFmt{Ser}}$}, which respectively represent
the channels (a.k.a. session endpoints) played by roles {\small $\roleFmt{Sat}$} and {\small $\roleFmt{Ser}$} in a multiparty session {\small $\mpS$}, within the processes:

\smallskip
\centerline{\(
   \small
   \mpP[\roleFmt{Sat}] =
   \delay{C_{1} = 6.5}{\mpSTSel{\mpChanRole{\mpS}{\roleFmt{Sat}}}{\roleFmt{Ser}}{\mpLabFmt{Data}}{\mpNil}{0.4}}
   \quad
   \mpP[\roleFmt{Ser}] =
   \delay{C_{2} = 6}{\mpSTBranch{\mpChanRole{\mpS}{\roleFmt{Ser}}}{\roleFmt{Sat}}{\mpLabFmt{Data}}{\mpNil}{0.3}}
   \)
}

\smallskip
\noindent
The Satellite process {\small $\mpP[\roleFmt{Sat}]$} waits for exactly
{\small 6.5} time units~({\small{$\mpFmt{\romanF{delay}(C_1 = 6.5)}$}}),
then sends the message {\small $\mpLabFmt{Data}$} with a timeout of {\small 0.4} time units~({\small{$\mpSTSelN{\mpChanRole{\mpS}{\roleFmt{Sat}}}{\roleFmt{Ser}}{\mpLabFmt{Data}}{0.4}$}}),
and becomes inactive~({\small{$\mpNil$}}).
Meanwhile,
the Server process {\small $\mpP[\roleFmt{Ser}]$}
waits for {\small 6} time units~({\small{$\mpFmt{\romanF{delay}(C_2 = 6)}$}}), then receives the message
with a timeout of  {\small 0.3} time units~({\small{$\mpSTBranchOut{\mpChanRole{\mpS}{\roleFmt{Ser}}}{\roleFmt{Sat}}{\mpLabFmt{Data}}{0.3}$}}), %
subsequently becoming inactive.

\vspace{-1em}
\subparagraph{%
Solution to Stuck Processes Due to Time Failures} %
It appears that the parallel execution of {\small $\mpP[\roleFmt{Sat}]$} and {\small $\mpP[\roleFmt{Ser}]$},
{\small $\mpP[\roleFmt{Sat}] \mpPar \mpP[\roleFmt{Ser}]$},
cannot proceed further due to the disparity in timing requirements.
Specifically, using the same session {\small $\mpS$}, {\small $\roleFmt{Sat}$}
sends the message  {\small $\mpLabFmt{Data}$} to {\small $\roleFmt{Ser}$}
between {\small 6.5} and {\small 6.9} time units,
while {\small $\roleFmt{Ser}$} must
receive it from {\small $\roleFmt{Sat}$} between {\small 6} and {\small 6.3} time units.
This results in a stuck situation, as {\small $\roleFmt{Ser}$} %
cannot meet the required timing condition to receive the message.

Fortunately, in our system, timeout failures are allowed, which can be addressed by leveraging affine session types and
their associated failure handling mechanisms. Back to our example, when {\small $\mpChanRole{\mpS}{\roleFmt{Ser}}$} waits for {\small 6} time units and cannot receive {\small $\mpLabFmt{Data}$} within {\small 0.3} time units, a timeout failure is raised~({\small
$\mpFailedP{\mpBranchRawN{{\mpChanRole{\mpS}{\roleFmt{Ser}}}^{0.3}}{\roleFmt{Sat}}{\mpLabFmt{Data} \mpSeq \mpNil}}$}).
Furthermore, we apply our time-failure handling approach to manage this timeout failure, initiating the termination of
the channel {\small $\mpChanRole{\mpS}{\roleFmt{Ser}}$} and triggering the cancellation process of the session {\small $\mpS$ ($\kills{\mpS}$)}.
As a result, the process will successfully terminate by canceling (or killing) all usages of {\small $\mpS$} within it.

Conversely, the system introduced in~\cite{DBLP:conf/concur/BocchiYY14} enforces strict requirements, including~\emph{feasibility} and~\emph{wait-freedom}, on timed global types
to prevent time-related failures in well-typed processes,
thus preventing them from becoming blocked due to unsolvable timing constraints.
Feasibility ensures the successful termination of each allowed partial execution,
while wait-freedom guarantees that receivers do not have to wait if senders follow their time constraints. %
In our example, we start with a timed global type that is neither feasible nor wait-free,
showcasing how our system effectively handles time failures
and ensures successful process termination
without imposing additional conditions on timed global types.
In essence,
reliance on feasibility and wait-freedom becomes
unnecessary in our system, thanks to the inclusion of affinity
and time-failure handling mechanisms.

\subsection{\timedmulticrusty: Toolchain Overview}
\label{subsec:implementation:overview:framework}

\begin{figure}[!t]
\footnotesize
\centering
\begin{subfigure}[t]{0.3\textwidth}
\centering
\begin{tabular}{c}
\begin{rustlisting}
struct Send<T,*@\label{line:overview_types:send:parameter:payload}@*
 const CLOCK: char,*@\label{line:overview_types:send:parameter:clock}@*
 const START: i128,*@\label{line:overview_types:send:parameter:start}@*
 const INCLUDE_START: bool,*@\label{line:overview_types:send:parameter:include_start}@*
 const END: i128,*@\label{line:overview_types:send:parameter:end}@*
 const INCLUDE_END: bool,*@\label{line:overview_types:send:parameter:include_end}@*
 const RESET: char,*@\label{line:overview_types:send:parameter:reset}@*
 S>*@\label{line:overview_types:send:parameter:continuation}@*
\end{rustlisting}
\end{tabular}
\caption{\CODE{Send} type}%
\label{subfig:overview_types:send}
\end{subfigure}\hfill%
\begin{subfigure}[t]{0.3\textwidth}
\centering
\begin{tabular}{c}
\begin{rustlisting}
struct Recv<T,
 const CLOCK: char,
 const START: i128,
 const INCLUDE_START: bool,
 const END: i128,
 const INCLUDE_END: bool,
 const RESET: char,
 S>
\end{rustlisting}
\end{tabular}
\caption{\CODE{Recv} type}%
\label{subfig:overview_types:recv}
\end{subfigure}\hfill%
\begin{subfigure}[t]{0.4\textwidth}
\centering
\begin{tabular}{c}
\begin{rustlisting}
MeshedChannels<
 Recv<Data,*@\label{line:overview_types:meshedchannels:parameter:sensor}@*
  'a',6,true,7,true,'a',End>,
 Send<Data,*@\label{line:overview_types:meshedchannels:parameter:server}@*
  'b',6,true,7,true,'b',End>,
 RoleSen<RoleSer<End>>,*@\label{line:overview_types:meshedchannels:parameter:stack}@*
 NameSat,*@\label{line:overview_types:meshedchannels:parameter:name}@*
>
\end{rustlisting}
\end{tabular}
\caption{\CODE{MeshedChannels} type for~$\roleFmt{Sat}$}
\label{subfig:overview_types:meshedchannels}
\end{subfigure}
\vspace{-.5em}
\caption{Main types of~\timedmulticrusty}
\label{fig:overview_types}
\vspace{-1em}
\end{figure}

To augment the theory,
we introduce the~\timedmulticrusty library,
a toolchain for implementing communication protocols in~\Rust.
\timedmulticrusty specifies protocols where communication operations
must adhere to specific time limits (\emph{timed}),
allowing for \emph{asynchronous} message reception and
runtime handling of certain failures (\emph{affine}).
This library relies on two fundamental types:~{\small \CODE{Send}} and~{\small \CODE{Recv}},
representing message sending and receiving, respectively.
Additionally, it incorporates the~{\small \CODE{End}} type, signifying termination to close the connection.
\Cref{subfig:overview_types:send,subfig:overview_types:recv}
illustrate the~{\small \CODE{Send}} and~{\small \CODE{Recv}} types respectively,
used for sending and receiving messages of any thread-safe type
(represented as~{\small \CODE{T}} in~\Cref{line:overview_types:send:parameter:payload}).
After sending or receiving a message,
the next operation or continuation
({\small \CODE{S}} in~\Cref{line:overview_types:send:parameter:continuation}) is determined,
which may entail sending another message,
receiving another message,
or terminating the connection.

Similar to %
\ATMP,
each communication operation in~\timedmulticrusty is constrained by specific time boundaries to avoid infinite waiting.
These time bounds are represented by the parameters in
\Crefrange{line:overview_types:send:parameter:clock}{line:overview_types:send:parameter:reset}
of~\Cref{subfig:overview_types:send},
addressing scenarios where a role may be required to send
after a certain time unit or receive between two specific time units.
Consider the final communication operation in the first branch
of~\Cref{fig:implementation:remote_data} from {\small $\roleFmt{Sat}$}'s perspective.
To remain consistent with~\Cref{sec:overview:theory},
the communication is terminated here instead of looping
back to the beginning of the protocol.
In this operation,
{\small $\roleFmt{Sat}$} sends a message labelled {\small \CODE{Data}}
to {\small $\roleFmt{Ser}$} between time units {\small $6$} and {\small $7$},
with respect to its inner clock {\small \CODE{'b'}},
and then terminates after resetting its clock.
This can be implemented as: %
{\small \CODE{Send<Data, 'b', 6, true, 7, true, 'b', End>}}.

To enable multiparty communication in~\timedmulticrusty, we
use the~{\small \CODE{MeshedChannels}} type, inspired by~\cite{lagaillardie2022Affine}.
This choice is necessary as~{\small \CODE{Send}} and~{\small \CODE{Recv}} types are primarily
designed for \emph{binary}~(peer-to-peer) communication.
Within~{\small \CODE{MeshedChannels}}, each binary channel pairs the owner role with another,
establishing a mesh of communication channels that encompasses all participants.
\Cref{subfig:overview_types:meshedchannels} demonstrates
an example of using~{\small \CODE{MeshedChannels}}
for~{\small $\roleFmt{Sat}$} in our running example:
{\small $\roleFmt{Sat}$} receives a~{\small \CODE{Data}} message from {\small $\roleFmt{Sen}$}
(\Cref{line:overview_types:meshedchannels:parameter:sensor}) and forwards it %
to {\small $\roleFmt{Ser}$} (\Cref{line:overview_types:meshedchannels:parameter:server})
before ending all communications,
following the order specified by the stack in~\Cref{line:overview_types:meshedchannels:parameter:stack}.

Creating these types manually in~\Rust can be challenging and error-prone,
especially because they represent the local perspective of each role in the protocol.
Therefore, as depicted in~\Cref{fig:theory_methodology}~(right),
\timedmulticrusty employs a top-down methodology similar to~\ATMP to
generate local viewpoints from a global protocol,
while ensuring the \emph{correctness} of the generated types \emph{by construction}.
To achieve this, we extend the syntax of~\nuscr~\cite{zhouCFSM2021},
a language for describing multiparty communication protocols,
to include time constraints, resulting in~\timednuscr.
A timed global protocol represented in~\timednuscr
is then projected onto local types,
which are used for generating~\Rust types in~\timedmulticrusty.

\section{Affine Timed Multiparty Session Calculus}
\label{sec:session-calculus}
\label{SEC:SESSION-CACULUS}
In this section, %
we formalise an affine timed multiparty session $\pi$-calculus,
where processes are capable of performing time actions, raising timeouts, and handling failures. 
We start with the formal definitions of time constraints used in the paper.

\vspace{-1em}
\subparagraph{Clock Constraint, Valuation, and Reset}
\label{sec:all_about_clocks}
Our time model is based on the timed automata
formalism~\cite{DBLP:journals/tcs/AlurD94,DBLP:conf/cav/KrcalY06}.
Let $\mathbf{C}$ %
denote a finite set of \emph{clocks}, ranging over $C, C', C_1,\ldots$, 
that take non-negative real values in $\rn_{\ge 0}$. %
Additionally, let $\cUnit, \cUniti, \cUnit[1], \ldots$ be \emph{time constants} ranging over $\rn_{\ge 0}$. 
A \emph{clock constraint} $\ccst$ over $\cs$ is defined as: %

\smallskip
\centerline{\(
\small
\ccst%
    \,\coloncolonequals\,%
    \texttt{true}
    \bnfsep%
     C > \mathfrak{b}%
    \bnfsep%
    C = \mathfrak{b}%
    \bnfsep%
    \neg \ccst%
    \bnfsep%
    \ccst_1 \land \ccst_2    
    \)} 
 
\smallskip
 \noindent
  where $C \in \cs$ and $\mathfrak{b}$ is a \emph{constant time bound} ranging over non-negative rationals  
    $\mathbb{Q}_{\ge 0}$. We define $\texttt{false}, <, \ge, \le$ in the standard way.
For simplicity and consistency with our implementation~(\Cref{sec:implementation:implementation}), 
we assume %
each clock constraint contains a \emph{single} clock.
Extending  a clock constraint with \emph{multiple} clocks is straightforward.

A \emph{clock valuation} $\cVal : \cs \rightarrow \rn_{\ge 0}$
assigns time to each clock in $\cs$.
We define $\cVal + t$ as the valuation that assigns to each $C \in \cs$ the value $\cVal(C) + t$.
The \emph{initial} valuation that maps all clocks to $0$ is denoted as $\cVal^{0}$,
 and the valuation that assigns a value of $t$ to all clocks is denoted as $\cVal^{t}$.
  $\cVal \models \ccst$ indicates that the constraint
  $\ccst$ is satisfied by the valuation
  $\cVal$.
   Additionally,
   we use $\sqcup_{i \in I} \cVal[i]$ to represent the overriding union of the valuations $\cVal[i]$ for $i \in I$. %

A~\emph{reset predicate} $\crst$ over $\cs$
is a subset of $\cs$ that defines the clocks to be reset.
If $\crst = \emptyset$, no reset is performed.
Otherwise,
the valuation for each clock $C \in \crst$ is set to $0$.
For %
clarity,
we represent a reset predicate as  $C := 0$
when a single clock $C$ needs to be reset.
To denote the clock valuation identical to $\cVal$
but with the values of clocks in $\crst$ to $0$,
we use $\cValUpd{\cVal}{\crst}{0}$.

\vspace{-1em}
\subparagraph{Syntax of Processes} %
\label{sec:session-calculus:syntax}
\label{SEC:SESSION-CALCULUS:SYNTAX}
Our session $\pi$-calculus for affine timed multiparty session types~(\ATMP) 
models timed processes interacting 
via affine meshed multiparty channels.
It extends the calculus for affine multiparty session types~(\AMPST)~\cite{lagaillardie2022Affine}
by incorporating asynchronous communication,
time features, timeouts, and failure handling.\footnote{
To simplify,
our calculus exclusively emphasises send/receive communication channels.
Standard extensions, \eg integers, booleans, and conditionals,
can easily integrate into our framework since they are independent of our main formulation;
they will be used in the examples.}

\begin{definition}[Syntax] %
\label{def:aat-mpst-syntax-terms}
Let $\roleP, \roleQ, \roleR, \ldots$ denote \emph{roles} belonging to a (fixed) set $\roleSet$; 
$\mpS, \mpSi, \ldots$ for \emph{sessions}; 
$x, y, \ldots$ for \emph{variables}; 
$\mpLab, \mpLabi, \ldots$ for \emph{message labels};  and 
$\mpX, \mpY, \ldots$ for \emph{process variables}.
The \emph{affine timed multiparty session} \emph{$\pi$-calculus} syntax is defined as follows:

\smallskip%
\centerline{
\(
\small
\begin{array}{r@{\hskip 2mm}c@{\hskip 2mm}l@{\hskip 2mm}l}
  \textstyle%
  \mpC, \mpD
  &\coloncolonequals&%
  \mpFmt{x} \bnfsep \mpChanRole{\mpS}{\roleP}%
  & \mbox{\small(variable, channel with role $\roleP$)}
  \\[1mm]
  \mpP, \mpQ 
  &\coloncolonequals&%
  \mpNil
    \bnfsep \mpP \mpPar \mpQ%
    \bnfsep \mpRes{\mpS}{\mpP}%
  &
  \mbox{\small (inaction, parallel composition, restriction)}
  \\
  &&
  \highlight{\mpTSel{\mpC}{\roleQ}{\mpLab}{\mpD}{\mpP}{\mathfrak{n}}}
  &
  \mbox{\small(timed selection towards role $\roleQ$)}
  \\
  &&
  \highlight{\mpTBranch{\mpC}{\roleQ}{i \in I}{\mpLab[i]}{x_i}{\mpP[i]}{\mathfrak{n}}{}}%
  &
  \mbox{\small(timed branching from role $\roleQ$ with $I \neq \emptyset$)}
  \\
  &&
  \mpDefAbbrev{\mpDefD}{\mpP}%
  \bnfsep \mpCall{\mpX}{\widetilde{\mpC}}
  &
  \mbox{\small(process definition, process call)} %
  \\
  &&
 \emph{\mbox{$\highlight{\delay{\ccst}{\mpP}}$}}
  \bnfsep
\emph{\mbox{\ul{$\highlight{\delay{\cUnit}{\mpP}}$}}}
&
   \mbox{\small(time-consuming delay, \ul{deterministic delay})}
  \\
  &&
  \mbox{\ul{$\highlight{\mpFailedP{\mpP}}$}}
    \bnfsep
 \emph{\mbox{$\trycatch{\mpP}{\mpQ}$}}%
         &
        \mbox{\small(\ul{timeout failure}, try-catch)}
        \\
        &&
   \emph{ \mbox{$\mpCancel{\mpC}{\mpP}$}}
       \bnfsep 
       \mpCErr
        \bnfsep%
\mbox{\ul{$\kills{\mpS}$}}
&
        \mbox{\small(cancel, communication error, \ul{kill})} %
        \\
   &&
  \highlight{ \mpSessionQueueO{\mpS}{\roleP}{\mpQueue}}
   &
    \mbox{\small(output message queue of role $\roleP$ in session $\mpS$)}
   \\[1mm]
  \mpDefD%
  &\coloncolonequals&%
  \mpJustDef{\mpX}{\widetilde{x}}{P}
  &
  \mbox{\small(declaration of process variable $\mpX$)}
  \\[1mm]
  \highlight{ \mpQueue}%
  &\highlight{\coloncolonequals}&%
  \highlight{\mpQueueCons{%
    \mpQueueOElem{\roleQ}{\mpLab}{\mpChanRole{\mpS}{\roleR}}
  }{%
    \mpQueue%
  }}  \bnfsep \highlight{\mpQueueEmpty}
  &
  \mbox{\small(message queue, non-empty or empty)}
  \end{array}
\)}

\smallskip
\noindent%
Restriction, branching, and process definitions and declarations act as binders;
$\fc{\mpP}$ is the set of \emph{free channels with roles} in $\mpP$, $\fv{\mpP}$ is
the set of \emph{free variables} in $\mpP$, and $\Pi_{i \in I} \mpP[i]$
is the parallel composition of processes $\mpP[i]$.
Extensions w.r.t. \emph{\AMPST} calculus are $\highlight{\text{highlighted}}$. 
Runtime processes, generated dynamically during program execution rather than explicitly written by users, 
are \ul{underlined}. %
\end{definition}

Our calculus comprises: 

\noindent
{\bf \emph{Channels}} $\mpC, \mpD$, %
 being either variables $\mpFmt{x}$
or channels with roles (\aka session \emph{endpoints}) $\mpChanRole{\mpS}{\roleP}$. 

\vspace{.2ex}
\noindent
{\bf \emph{Standard}} processes as in~\cite{DBLP:journals/pacmpl/ScalasY19,lagaillardie2022Affine}, including 
inaction $\mpNil$, parallel composition $\mpP \mpPar \mpQ$, session scope restriction
$\mpRes{\mpS}{\mpP}$, process definition $\mpDefAbbrev{\mpDefD}{\mpP}$, %
process
call $\mpCall{\mpX}{\widetilde{\mpC}}$, and communication error $\mpCErr$. 

\vspace{.2ex}
\noindent
{\bf \emph{Time}} processes that follow the program time behaviour of~\Cref{subfig:overview_types:meshedchannels}: 
\begin{itemize}[left=0pt,topsep=0pt]
\item {\bf \emph{Timed selection}} (or {\bf \emph{timed internal choice}}) $\mpTSel{\mpC}{\roleQ}{\mpLab}{\mpD}{\mpP}{\mathfrak{n}}$
indicates that a message
$\mpLab$ with payload $\mpD$ is sent to role $\roleQ$ via endpoint $\mpC$, 
whereas {\bf \emph{timed branching}} (or {\bf \emph{timed external choice}}) 
$\mpTBranch{\mpC}{\roleQ}{i \in I}{\mpLab[i]}{x_i}{\mpP[i]}{\mathfrak{n}}{}$ waits to
receive a message $\mpLab[i]$ from role $\roleQ$ via endpoint $\mpC$ and then proceeds as $\mpP[i]$.

The parameter $\mathfrak{n}$ in both timed selection and branching is a \emph{timeout} that allows modelling
different types of communication primitives: \emph{blocking with a timeout}
 ($\mathfrak{n} \in \rn_{> 0}$), \emph{blocking} ($\mathfrak{n} = \infty$), or \emph{non-blocking} ($\mathfrak{n} = 0$).
When $\mathfrak{n} \in \mathbb{R}_{\ge 0}$, the timed selection (or timed branching) process waits for up to $\mathfrak{n}$ time units to send (or receive) a message.
If the message cannot be sent (or received) within this time, the process moves into a \emph{timeout state}, raising a \emph{time failure}.
If $\mathfrak{n}$ is set to $\infty$, the timed selection (or timed branching) process blocks until a message is successfully sent (or received).

In our system, we allow~\emph{send} processes to be time-consuming, enabling processes to wait before sending messages. 
Consider the remote data example shown in~\Cref{fig:implementation:remote_data}. 
This practical scenario illustrates how a process might wait before sending a message, resulting in the possibility of send actions failing due to timeouts. It highlights the importance of timed selection, contrasting with systems like in~\cite{bocchi2019Asynchronous} where send actions are instantaneous.

\item $\delay{\ccst}{\mpP}$ represents a {\bf \emph{time-consuming delay}} action, such as  method invocation or sleep.
Here, $\ccst$ is a clock constraint involving a \emph{single} clock variable $C$, used to specify the interval for the delay. 
When executing $\delay{\ccst}{\mpP}$, any time value $\cUnit$ that satisfies the constraint $\ccst$ can be consumed.
Consequently, the \emph{runtime} {\bf \emph{deterministic delay}} process $\delay{\cUnit}{\mpP}$, arising during the execution of 
$\delay{\ccst}{\mpP}$, is introduced. In $\delay{\cUnit}{\mpP}$, $\cUnit$ is a constant and a solution to $\ccst$, 
and $\mpP$ is %
executed after a precise delay of $\cUnit$ time units.  

\item $\mpFailedP{\mpP}$ signifies that the process $\mpP$ has violated a time constraint, resulting in a 
{\bf \emph{timeout failure}}.
\end{itemize}

\vspace{.2ex}
\noindent
{\bf \emph{Failure-handling}} processes that adopt the \AMPST approach~\cite{lagaillardie2022Affine}:  
\begin{itemize}[left=0pt,topsep=0pt]
\item 
$\trycatch{\mpP}{\mpQ}$
consists of a $\mpFmt{\ensuremath{\proclit{try}}}$ process $\mpP$ that is prepared to communicate
with a parallel composed process,
and a $\mpFmt{\ensuremath{\proclit{catch}}}$ process $\mpQ$, which becomes active %
in the event of a cancellation or timeout. 
For clarity,  $\trycatch{\mpNil}{\mpQ}$ is not allowed within our calculus.  %
\item 
$\mpCancel{\mpC}{\mpP}$ performs the {\bf \emph{cancellation}} of other processes with channel $\mpC$.
\item 
$\kills{\mpS}$ {\bf \emph{kills}} (terminates) all processes with session $\mpS$, 
and is dynamically generated only at \emph{runtime} from timeout failure or cancel processes.

\end{itemize}

\vspace{.2ex}
\noindent
{\bf \emph{Message queues}}: 
$\mpSessionQueueO{\mpS}{\roleP}{\mpQueue}$
represents the {\bf \emph{output message queue}} of role $\roleP$ in session $\mpS$.
It contains all the messages previously sent by $\roleP$. The queue $\mpQueue$ can be a sequence of messages of the form
$\mpQueueOElem{\roleQ}{\mpLab}{\mpChanRole{\mpS}{\roleR}}$, where $\roleQ$ is the receiver,
or $\mpQueueEmpty$, indicating an \emph{empty} message queue.
The set of \emph{receivers in} $\mpQueue$, denoted as $\operatorname{receivers}(\mpQueue)$, 
is defined in a standard way as:  

\centerline{\(%
\small
 \operatorname{receivers}(\mpQueueCons{
    \mpQueueOElem{\roleQ}{\mpLab}{\mpChanRole{\mpS}{\roleR}}
  }{%
    \mpQueuei%
  }) = \setenum{\roleQ} \cup \operatorname{receivers}(\mpQueuei)
  \quad \quad 
  \operatorname{receivers}(\mpQueueEmpty) = \emptyset
 \)}

\begin{figure}[t!]
\small
  \centerline{\(%
    \begin{array}{c}
    \begin{array}{rl}
        \inferrule{\iruleMPRedOut}&
    \mpEtxApp{\mpEtx}{\mpTSel{\mpChanRole{\mpS}{\roleQ}}{\roleP}{\mpLab}{%
      \mpChanRole{\mpSi}{\roleR}%
    }{\mpQ}{\mathfrak{n}}}%
    \mpPar
    \mpSessionQueueO{\mpS}{\roleQ}{\mpQueue}%
    \;\mpnonTMove\;%
    \mpQ%
    \mpPar
    \mpSessionQueueO{\mpS}{\roleQ}{%
      \mpQueueCons{\mpQueue}{%
        \mpQueueCons{%
          \mpQueueOElem{\roleP}{\mpLab}{\mpChanRole{\mpSi}{\roleR}}%
        }{%
          \mpQueueEmpty%
        }%
      }%
    }%
\\[.5mm]
   \inferrule{\iruleMPRedIn}&%
      \mpEtxApp{\mpEtx}{\mpTBranch{\mpChanRole{\mpS}{\roleP}}{\roleQ}{i \in I}{%
      \mpLab[i]}{x_i}{\mpP[i]}{\mathfrak{n}}{}}%
    \mpPar
    \mpSessionQueueO{\mpS}{\roleQ}{%
      \mpQueueCons{%
        \mpQueueOElem{\roleP}{\mpLab[k]}{\mpChanRole{\mpSi}{\roleR}}%
      }{\mpQueue}%
    }%
    \;\mpnonTMove\;%
    \mpP[k]\subst{\mpFmt{x_k}}{\mpChanRole{\mpSi}{\roleR}}%
    \mpPar
    \mpSessionQueueO{\mpS}{\roleQ}{\mpQueue}%
   \hspace{16mm}
    \text{\footnotesize%
      ($k \!\in\! I$)
    }%
   \\[.5mm]%
   \inferrule{\iruleMPRedErr}& 
      \mpEtxApp{\mpEtx}{\mpTBranch{\mpChanRole{\mpS}{\roleP}}{\roleQ}{i \in I}{%
      \mpLab[i]}{x_i}{\mpP[i]}{\mathfrak{n}}{}}%
    \mpPar
    \mpSessionQueueO{\mpS}{\roleQ}{%
      \mpQueueCons{%
        \mpQueueOElem{\roleP}{\mpLab}{\mpChanRole{\mpSi}{\roleR}}%
      }{\mpQueue}%
    }%
    \;\mpnonTMove\;%
   \mpCErr
   \hspace{29.3mm}
    \text{\footnotesize%
      ($\forall i \!\in\! I: \mpLab[i] \neq \mpLab$)
    }%
   \\[.5mm]
     \inferrule{\iruleMPRedDet}&
   \highlight{ \models \ccstSubt{\ccst}{\cUnit}{C}
    \;\;\text{implies}\;\;%
    \mpEtxApp{\mpEtx}{\delay{\ccst}{\mpP}}%
     \;\mpnonTMove\;%
      \delay{\cUnit}{\mpP}}
    \\[.5mm]
    \inferrule{\iruleMPRedDelay}&
    \highlight{ \mpP
    \;\mpMoveTime\;%
    \timePass{\cUnit}{\mpP}}
     \\[.5mm]
   \inferrule{\iruleMPRedTryFail}& %
   \highlight{
   \mpFailedP{\mpP} 
    \;\mpnonTMove\;%
    \kills{\mpS}}
    \hspace{67.2mm}
   \highlight{\text{\footnotesize%
       ($\exists \roleR. \,
      \procSubject{\mpP} = \setenum{\mpChanRole{\mpS}{\roleR}}$)}
    }
    \\[.5mm]
     \inferrule{\iruleMPRedCan}&
    \mpEtxApp{\mpEtx}{\mpCancel{\mpChanRole{\mpS}{\roleP}}{\mpQ}}
     \;\mpnonTMove\;%
     \kills{\mpS}%
      \mpPar
       \mpQ
       \\[.5mm]
   \inferrule{\iruleMPRedFailCatch}&  
     \highlight{
   \trycatch{\mpFailedP{\mpP}}{\mpQ}
    \;\mpnonTMove\;%
    \kills{\mpS}
     \mpPar \mpQ}
    \hspace{43mm}
   \highlight{\text{\footnotesize%
       ($\exists \roleR. \,
      \procSubject{\mpP} = \setenum{\mpChanRole{\mpS}{\roleR}}$)}
    }  
    \\[.5mm]
    \inferrule{\iruleMPCCat} &  %
     \trycatch{\mpP}{\mpQ}%
     \mpPar
     \kills{\mpS}
       \;\mpnonTMove\;%
       \mpQ%
       \mpPar
     \kills{\mpS}
      \hspace{50.7mm}
    \text{\footnotesize%
      ($\exists \roleR. \, 
      \procSubject{\mpP} = \setenum{\mpChanRole{\mpS}{\roleR}}$)%
    }%
    \\[.5mm]
     \inferrule{\iruleMPRedCanIn}&
   \highlight{\mpTBranch{\mpChanRole{\mpS}{\roleP}}{\roleQ}{i \in I}{%
      \mpLab[i]}{x_i}{\mpP[i]}{\mathfrak{n}}{}%
       \mpPar
       \mpSessionQueueO{\mpS}{\roleQ}{\mpQueue}
     \mpPar
      \kills{\mpS}}
      \\
      &
     \hspace{2mm}  \highlight{\;\mpnonTMove\;%
       \mpRes{\mpSi}{%
        (\mpP[k]\subst{\mpFmt{x_k}}{\mpChanRole{\mpSi}{\roleR}}%
    \mpPar %
    \kills{\mpSi})}
    \mpPar
     \mpSessionQueueO{\mpS}{\roleQ}{\mpQueue}
       \mpPar %
       \kills{\mpS}}%
      \hspace{17.6mm}
   \highlight{ \text{\footnotesize%
       ($\roleP \notin \operatorname{receivers}(\mpQueue)$,  $k \!\in\! I$, %
    $\mpSi \notin \fc{\mpP[k]}$)
    }}%
\\[.8mm]
\inferrule{\iruleMPRedCanQ}&
    \highlight{
    \mpSessionQueueO{\mpS}{\roleP}{%
      \mpQueueCons{%
        \mpQueueOElem{\roleQ}{\mpLab}{\mpChanRole{\mpSi}{\roleR}}%
      }{\mpQueue}}
    \mpPar
     \kills{\mpS}
       \;\mpnonTMove\;%
      \mpSessionQueueO{\mpS}{\roleP}{\mpQueue} \mpPar \kills{\mpS} \mpPar \kills{\mpSi}}
      \\[.5mm]
   \inferrule{\iruleMPRedCall}& %
      \mpDef{\mpX}{x_1,\ldots,x_n}{\mpP}{(%
        \mpCall{\mpX}{%
          \mpChanRole{\mpS[1]}{\roleP[1]}, \ldots,%
          \mpChanRole{\mpS[n]}{\roleP[n]}%
        }%
       \mpPar
        \mpQ%
        )%
      }%
      \\
      &\hspace{4mm}%
      \;\mpnonTMove\;%
      \mpDef{\mpX}{x_1,\ldots,x_n}{\mpP}{(%
        \mpP\subst{\mpFmt{x_1}}{\mpChanRole{\mpS[1]}{\roleP[1]}}%
        \cdots%
        \subst{\mpFmt{x_n}}{\mpChanRole{\mpS[n]}{\roleP[n]}}%
       \mpPar
          \mpQ%
          )%
      }%
      \\[.5mm]%
        \inferrule{\iruleMPRedCtx}&%
      \mpP \;\mpnonTMove\; \mpPi%
      \;\;\text{implies}\;\;%
      \mpCtxApp{\mpCtx}{\mpP} \;\mpnonTMove\; \mpCtxApp{\mpCtx}{\mpPi}%
      \\[1mm]%
    \inferrule{\iruleMPRedCongr} &
    \mpPi \equiv \mpP \;\mpnonTMove\; \mpQ \equiv \mpQi
    \;\;\text{implies}\;\;
    \mpPi \;\mpnonTMove\; \mpQi
   \quad 
    \inferrule{\iruleMPRedCongrTime} \quad 
   \highlight{ \mpPi \equiv \mpP \;\mpMoveTime\; \mpQ \equiv \mpQi
    \;\;\text{implies}\;\;
    \mpPi \;\mpMoveTime\; \mpQi}
    \\[.5mm]
     \inferrule{\iruleMPRedInstant} &
   \highlight{ \mpP \;\mpnonTMove\; \mpPi
    \;\;\text{implies}\;\;
    \mpP \;\mpMove\; \mpPi}
\qquad \qquad \quad \,\,\;\,\,
     \inferrule{\iruleMPRedTimeConsume} \quad %
   \highlight{ \mpP \;\mpMoveTime\; \mpPi
    \;\;\text{implies}\;\;
    \mpP \;\mpMove\; \mpPi }
    \end{array}
    \\[0.5mm]
    \\[-2mm]%
    \hdashline%
    \\[-2mm]
      \begin{array}{@{\hskip 0mm}c@{\hskip 0mm}}
      \mpP \mpPar \mpQ%
      \equiv%
      \mpQ \mpPar \mpP%
      \quad\;\;%
      \mpFmt{(\mpP \mpPar \mpQ) \mpPar \mpR}%
      \equiv%
      \mpP \mpPar \mpFmt{(\mpQ \mpPar \mpR)}%
      \quad\;\;%
      \mpP \mpPar \mpNil%
      \equiv%
      \mpP%
      \quad\;\;%
      \mpRes{\mpS}{\mpNil}%
      \equiv%
      \mpNil%
      \\[.2mm]%
      \mpRes{\mpS}{%
        \mpRes{\mpSi}{%
          \mpP%
        }%
      }%
      \equiv%
      \mpRes{\mpSi}{%
        \mpRes{\mpS}{%
          \mpP%
        }%
      }%
      \quad\;\;%
      \mpRes{\mpS}{(\mpP \mpPar \mpQ)}%
      \equiv%
      \mpP \mpPar \mpRes{\mpS}{\mpQ}%
      \;\;%
      \text{\footnotesize{}if $\mpS \!\not\in\! \fc{\mpP}$}%
      \\[.2mm]%
      \mpDefAbbrev{\mpDefD}{\mpNil}%
      \equiv%
      \mpNil%
      \qquad%
      \mpDefAbbrev{\mpDefD}{\mpRes{\mpS}{\mpP}}%
      \,\equiv\,%
      \mpRes{\mpS}{(
        \mpDefAbbrev{\mpDefD}{\mpP}%
        )}%
      \quad%
      \text{\footnotesize{}if\, $\mpS \!\not\in\! \fc{\mpDefD}$}%
      \\[.2mm]%
      \mpDefAbbrev{\mpDefD}{(\mpP \mpPar \mpQ)}%
      \,\equiv\,%
      \mpFmt{(\mpDefAbbrev{\mpDefD}{\mpP})} \mpPar \mpQ%
      \quad%
      \text{\footnotesize{}if\, $\dpv{\mpDefD} \cap \fpv{\mpQ} = \emptyset$}%
      \\[.2mm]%
      \mpDefAbbrev{\mpDefD}{%
        (\mpDefAbbrev{\mpDefDi}{\mpP})%
      }%
      \;\equiv\;%
      \mpDefAbbrev{\mpDefDi}{%
        (\mpDefAbbrev{\mpDefD}{\mpP})%
      }%
      \\%
      \text{\footnotesize%
        if\, %
        \(%
        (\dpv{\mpDefD} \cup \fpv{\mpDefD}) \cap \dpv{\mpDefDi}%
        \,=\,%
        (\dpv{\mpDefDi} \cup \fpv{\mpDefDi}) \cap \dpv{\mpDefD}%
        \,=\,%
        \emptyset%
        \)%
      }%
      \\
    \highlight{\delay{0}{\mpP} %
    \equiv%
   \mpP}
     \quad 
        \kills{\mpS}
  \mpPar
  \kills{\mpS}
  \equiv
  \kills{\mpS}
  \quad 
   \highlight{
       \mpRes{\mpS}{%
        \left(%
        \mpSessionQueueO{\mpS}{\roleP[1]}{\mpQueueEmpty} \mpPar \cdots \mpPar
        \mpSessionQueueO{\mpS}{\roleP[n]}{\mpQueueEmpty}
        \right)
        \equiv%
        \mpNil%
      }}
 \\[.5mm]
 \highlight{
 \mpSessionQueueO{\mpS}{\roleP}{%
          \mpQueueCons{%
            \mpQueue%
          }{%
            \mpQueueCons{%
              \mpQueueOElem{\roleQ[1]}{\mpLab[1]}{\mpChanRole{\mpS[1]}{\roleR[1]}}
            }{%
              \mpQueueCons{%
                \mpQueueOElem{\roleQ[2]}{\mpLab[2]}{\mpChanRole{\mpS[2]}{\roleR[2]}}%
              }{%
                \mpQueuei%
              }%
            }%
          }%
        } \equiv 
         \mpSessionQueueO{\mpS}{\roleP}{%
          \mpQueueCons{%
            \mpQueue%
          }{%
            \mpQueueCons{%
              \mpQueueOElem{\roleQ[2]}{\mpLab[2]}{\mpChanRole{\mpS[2]}{\roleR[2]}}
            }{%
              \mpQueueCons{%
                \mpQueueOElem{\roleQ[1]}{\mpLab[1]}{\mpChanRole{\mpS[1]}{\roleR[1]}}%
              }{%
                \mpQueuei%
              }%
            }%
          }%
        }}
             \quad  \highlight{\text{\footnotesize%
        if\, $\roleQ[1] \neq \roleQ[2]$%
      }}
    \end{array}
\end{array}
    \)}
    \vspace{-.5em}
  \caption{Top: reduction rules for \ATMP session $\pi$-calculus. 
  Bottom: structural congruence rules for the \ATMP $\pi$-calculus, where 
    $\fpv{\mpDefD}$ %
    is the set of \emph{free process variables} in $\mpDefD$, %
    and
    $\dpv{\mpDefD}$ %
    is the set of \emph{declared process variables} in $\mpDefD$. %
    New rules are $\highlight{\text{highlighted}}$. }%
  \label{fig:aat-mpst-pi-semantics}%
\vspace{-1em}
\end{figure}

 \begin{figure}[t]
 \centerline{\(
 \small
   \begin{array}{c}
\begin{array}{lll}
\timePass{\cUnit}{\mpP[1] \mpPar \mpP[2]} =
\timePass{\cUnit}{\mpP[1]}
\mpPar
\timePass{\cUnit}{\mpP[2]}
&
 \timePass{\cUnit}{\mpRes{\mpS}{\mpP}}
=
\mpRes{\mpS}{\timePass{\cUnit}{\mpP}}
&
\timePass{\cUnit}{\mpFailedP{\mpP}} =
  \mpFailedP{\mpP}
\end{array}
\\[.5mm]
\begin{array}{lll}
 \timePass{\cUnit}{\mpNil} = \mpNil
& 
 \timePass{\cUnit}{\mpDefAbbrev{\mpDefD}{\mpP}} =
  \mpDefAbbrev{\mpDefD}{\timePass{\cUnit}{\mpP}}
  &
\timePass{\cUnit}{\trycatch{\mpP}{\mpQ}}
 =
 \trycatch{\timePass{\cUnit}{\mpP}}{\timePass{\cUnit}{\mpQ}}
 \end{array}
 \\[.5mm]
 \begin{array}{llll}
\timePass{\cUnit}{\mpCErr} = \mpCErr
&
 \timePass{\cUnit}{\mpCancel{\mpC}{\mpQ}} = \mpCancel{\mpC}{\timePass{\cUnit}{\mpQ}}
 &
 \timePass{\cUnit}{\kills{\mpS}} = \kills{\mpS}
 &
  \timePass{\cUnit}{\mpSessionQueueO{\mpS}{\roleP}{\mpQueue}} =
  \mpSessionQueueO{\mpS}{\roleP}{\mpQueue}
\end{array}
\\[.5mm]
\begin{array}{ll}
\timePass{\cUnit}{\delay{\ccst}{\mpP}} = \texttt{undefined}
&
\begin{array}{@{}r@{~}c@{~}l@{}}
    \timePass{\cUnit}{\delay{\cUniti}{\mpP}} 
    &=&
    \left\{
      \begin{array}{@{}ll@{}}
        \delay{\cUniti - \cUnit}{\mpP}
          &
        \text{if~} \cUniti \geq \cUnit
        \\
      \texttt{undefined} 
          &
        \text{otherwise}
     \end{array}
    \right.
 \end{array}
\end{array}
\\[.5mm]
\begin{array}{l}
 \begin{array}{@{}r@{~}c@{~}l@{}}
    \timePass{\cUnit}{
     \mpTSel{\mpC}{\roleQ}{\mpLab}{\mpD}{\mpP}{\cUniti}
     }
    &=&
    \left\{
      \begin{array}{@{}ll@{}}
        \mpTSel{\mpC}{\roleQ}{\mpLab}{\mpD}{\mpP}{\cUniti - \cUnit}
          &
        \text{if~} \cUniti \geq \cUnit
        \\[1mm]
       \mpFailedP{\mpTSel{\mpC}{\roleQ}{\mpLab}{\mpD}{\mpP}{\cUniti}}
          &
        \text{otherwise}
     \end{array}
    \right.
 \end{array}
 \\[.5mm]
\begin{array}{@{}r@{~}c@{~}l@{}}
    \timePass{\cUnit}{
     \mpTBranch{\mpC}{\roleQ}{i \in I}{\mpLab[i]}{x_i}{\mpP[i]}{\cUniti}{}
    }
    &=&
    \left\{
      \begin{array}{@{}ll@{}}
        \mpTBranch{\mpC}{\roleQ}{i \in I}{\mpLab[i]}{x_i}{\mpP[i]}{\cUniti - \cUnit}{}
          &
        \text{if~} \cUniti \geq \cUnit
        \\[1mm]
         \mpFailedP{\mpTBranch{\mpC}{\roleQ}{i \in I}{\mpLab[i]}{x_i}{\mpP[i]}{\cUniti}{}}
          &
        \text{otherwise}
     \end{array}
    \right.
 \end{array}
   \end{array}
\\[.5mm]
\begin{array}{ll}
 \timePass{\cUnit}{\mpTSel{\mpC}{\roleQ}{\mpLab}{\mpD}{\mpP}{\infty}}
  =
  \mpTSel{\mpC}{\roleQ}{\mpLab}{\mpD}{\mpP}{\infty}
&
  \timePass{\cUnit}{\mpTBranch{\mpC}{\roleQ}{i \in I}{\mpLab[i]}{x_i}{\mpP[i]}{\infty}{}}
  =
  \mpTBranch{\mpC}{\roleQ}{i \in I}{\mpLab[i]}{x_i}{\mpP[i]}{\infty}{}
\end{array}
    \end{array}
    \)}%
 \vspace{-.5em}
  \caption{
     Time-passing function $\timePass{\cUnit}{\mpP}$. %
  }%
  \label{fig:aat-mpst-time-passing}%
  \vspace{-1em}
\end{figure}

\vspace{-1em}
\subparagraph{Operational Semantics} %
\label{sec:session-calculus:semantics}
\label{SEC:SESSION-CALCULUS:SEMANTICS}
We present the operational semantics of our session $\pi$-calculus for modelling the behaviour of affine timed processes,
including asynchronous communication, time progression, timeout activation, and failure handling.

\begin{definition}[Semantics]%
  \label{def:mpst-proc-context}%
  \label{def:mpst-pi-reduction-ctx}%
  \label{def:mpst-pi-semantics}%
  \label{def:mpst-pi-error}%
  A \emph{$\trycatchB$ context} $\mpEtx$ is defined as %
\emph{$\mpEtx \coloncolonequals%
  \trycatch{\mpEtx}{\mpP}%
  \bnfsep%
  \mpCtxHole%
  $}, and
  a \emph{reduction context} $\mpCtx$ is defined as %
$\mpCtx \coloncolonequals%
  \mpCtx \mpPar \mpP%
  \bnfsep%
  \mpRes{\mpS}{\mpCtx}%
  \bnfsep%
  \mpDefAbbrev{\mpDefD}{\mpCtx}%
  \bnfsep%
  \mpCtxHole%
  $.
 \emph{The reductions $\mpMove$, $\mpnonTMove$, and $\mpMoveTime$}
  are inductively defined %
  in~\Cref{fig:aat-mpst-pi-semantics}~(top), %
  with respect  to a \emph{structural congruence\;$\equiv$} depicted in~\Cref{fig:aat-mpst-pi-semantics}~(bottom).
  We write $\mpMoveStar$, $\mpnonTMoveStar$, and $\mpMoveTimeStar$ 
  for their  reflexive and transitive closures,  respectively. 
 $\mpNotMoveP{\mpP}$ (or  
$\mpnonTNotMoveP{\mpP}$, $\mpNotMoveTimeP{\mpP}$)
means  $\not\exists \mpPi$
such that $\mpP \!\mpMove\! \mpPi$ (or $\mpP \!\mpnonTMove\! \mpPi$,
$\mpP \!\mpMoveTime\! \mpPi$)
is derivable. We say $\mpP$ \emph{has a communication error} iff $\exists \mpCtx$ with 
$\mpP = \mpCtxApp{\mpCtx}{\mpCErr}$.

\end{definition}

\begin{remark}
The $\trycatchB$ contexts $\mpEtx$ in~\Cref{def:mpst-proc-context}
are exclusively 
used for defining reductions at the top level of processes, 
without applying them for nested exception handling as in~\cite{fowler2019Exceptional}. %
\end{remark}

We decompose the reduction rules in~\Cref{fig:aat-mpst-pi-semantics}
into three relations:
 $\mpnonTMove$ represents \emph{instantaneous} reductions
without time consumption,  $\mpMoveTime$  handles time-consuming steps,
and $\mpMove$ is a general relation that can arise either from
$\mpnonTMove$ by \inferrule{\iruleMPRedInstant}
or $\mpMoveTime$ by \inferrule{\iruleMPRedTimeConsume}. %
Now let us explain the operational semantics rules for our session $\pi$-calculus. %

\vspace{.2ex}
\noindent
{\bf \emph{Communication}}:  Rules 
\inferrule{\iruleMPRedOut} and \inferrule{\iruleMPRedIn} model asynchronous communication by queuing and dequeuing  \emph{pending} messages, respectively. Rule \inferrule{\iruleMPRedErr} is triggered by a message label mismatch, resulting in a fatal 
$\mpFmt{\boldsymbol{\mathsf{c}}}$ommunication  $\mpFmt{\boldsymbol{\mathsf{err}}}$or. 

\vspace{.2ex}
\noindent
{\bf \emph{Time}}:   
Rule \inferrule{\iruleMPRedDet} specifies a deterministic delay of a specific duration  
$\cUnit$, where $\cUnit$ is a solution to the clock constraint $\ccst$. 
Rule \inferrule{\iruleMPRedDelay} incorporates a time-passing function $\timePass{\cUnit}{\mpP}$, 
depicted in~\Cref{fig:aat-mpst-time-passing}, 
to represent time delays within a process.
This \emph{partial} function simulates a delay of time $\cUnit$ that may occur at different parts of the process. 
It is \emph{undefined} only if $\mpP$ is a time-consuming delay, \ie $\mpP = \delay{\ccst}{\mpPi}$, or if the specified delay time $\cUnit$ exceeds the duration of a runtime deterministic delay, \ie $\mpP = \delay{\cUniti}{\mpPi}$ with $\cUnit > \cUniti$. The latter case arises because deterministic delays must always adhere to their specified durations, \eg if a program is instructed to sleep for 5 time units, it must strictly follow this duration.

Notably, $\timePass{\cUnit}{\mpP}$ acts as the \emph{only mechanism} for triggering 
a timeout failure $\mpFailedP{\mpP}$, resulting from a timed selection or branching. 
Such a timeout failure occurs when $\timePass{\cUnit}{\mpP}$ is defined, and the specified delay $\cUnit$ exceeds a \emph{deadline} set within $\mpP$.

\vspace{.2ex}
\noindent
{\bf \emph{Cancellation}}: 
Rules \inferrule{\iruleMPRedCanIn} and \inferrule{\iruleMPRedCanQ} model the process cancellations. 
\inferrule{\iruleMPRedCanIn} is triggered only
when there are no messages in the queue that
can be received from $\roleQ$ via the endpoint
$\mpChanRole{\mpS}{\roleP}$. 
Cancellation of a timed selection is expected to 
eventually occur via \inferrule{\iruleMPRedCanQ}; 
therefore, 
there is no specific rule dedicated to it.
Similarly,
in our implementation,
the timed selection is not directly cancelled either.

Rules \inferrule{\iruleMPRedCan} and \inferrule{\iruleMPCCat},
adapted from~\cite{lagaillardie2022Affine},
state cancellations from other
parties.   
\inferrule{\iruleMPRedCan} 
facilitates cancellation and generates a kill process, while \inferrule{\iruleMPCCat} transitions 
to the $\mpFmt{\ensuremath{\proclit{catch}}}$ process $\mpQ$ 
due to the termination of session $\mpS$, 
where  the $\mpFmt{\ensuremath{\proclit{try}}}$ process $\mpP$ is communicating on $\mpS$. 
Therefore, the set of \emph{subjects} of process $\mpP$, denoted as $\procSubject{\mpP}$, is included 
 in the side condition of \inferrule{\iruleMPCCat} to ensure that $\mpP$ has a prefix at $\mpS$, 
 as defined below:

\smallskip
\centerline{\(
\footnotesize
\begin{array}{c}
\begin{array}{l}
\procSubject{\mpNil}
   =
    \procSubject{\mpCErr}
   =
   \emptyset
   \quad
      \procSubject{\mpP \mpPar \mpQ}
   =
    \procSubject{\mpP}
    \cup
     \procSubject{\mpQ}
     \quad 
     \procSubject{\mpSessionQueueO{\mpS}{\roleP}{\mpQueue}}
       =
   \setenum{{\mpChanRole{\mpS}{\roleP}}^{\tiny{\mathscr{Q}}}}
     \\
     \procSubject{\mpRes{\mpS}{\mpP}}
    =
     \procSubject{\mpP} \setminus (\setenum{\mpChanRole{\mpS}{\roleP[i]}}_{i \in I} \cup \setenum{{\mpChanRole{\mpS}{\roleP[i]}}^{\mathscr{Q}}}_{i \in I}) 
  \\
    \procSubject{\mpDefAbbrev{ \mpJustDef{\mpX}{\widetilde{x}}{\mpP}}{\mpQ}}
    =
     \procSubject{\mpQ} \cup \procSubject{\mpP} \setminus \setenum{\widetilde{x}} \text{ \,with\, }
     \procSubject{\mpCall{\mpX}{\widetilde{\mpC}}}
   =
      \procSubject{\mpP\subst{\widetilde{x}}{\widetilde{\mpC}}}
      \\
    \procSubject{\mpTSel{\mpC}{\roleQ}{\mpLab}{\mpD}{\mpP}{\mathfrak{n}}}
    =
      \procSubject{\mpTBranch{\mpC}{\roleQ}{i \in I}{\mpLab[i]}{x_i}{\mpP[i]}{\mathfrak{n}}{}}
      =
        \procSubject{\mpCancel{\mpC}{\mpP}}
      = \setenum{\mpC}
      \\[1.3mm]
         \procSubject{\delay{\ccst}{\mpP}}
   =%
      \procSubject{\delay{\cUnit}{\mpP}}
      =  
      \procSubject{\trycatch{\mpP}{\mpQ}}
      =
       \procSubject{\mpFailedP{\mpP}}
      = \procSubject{\mpP}
      \end{array}
    \end{array}
 \)}
 
 \smallskip
Subjects of processes determine sessions that may need cancellation, a crucial aspect for handling failed or 
cancelled processes properly. In our definition, subjects not only denote the endpoints via which processes start interacting but also indicate whether they are used for message queue processes.
Specifically, an endpoint $\mpChanRole{\mpS}{\roleP}$ annotated with $\mathscr{Q}$ signifies its use in a queue process.
This additional annotation, and thus the distinction it implies, is pivotal in formulating the typing rule for the $\trycatchB$ process, as discussed later in~\Cref{sec:aat-mpst-typing-system}, where we rely on subjects to exclude queue processes within any $\mpFmt{\ensuremath{\proclit{try}}}$ construct.

\vspace{.3ex}
\noindent
{\bf \emph{Timeout Handling}}: %
Rules \inferrule{\iruleMPRedTryFail} and \inferrule{\iruleMPRedFailCatch} address time failures. 
In the event of a timeout,  a killing process is generated.  
Moreover, in \inferrule{\iruleMPRedFailCatch}, 
the $\mpFmt{\ensuremath{\proclit{catch}}}$ process $\mpQ$ is triggered. 
To identify the session requiring termination, 
the set of subjects of the failure process $\mpFailedP{\mpP}$ is considered in both rules as a side condition.
Note that a timeout arises exclusively from timed selection or branching. Therefore, the subject set of 
$\mpFailedP{\mpP}$ must contain a \emph{single} endpoint devoid of $\mathscr{Q}$, indicating the generation of only one killing process.

\vspace{.3ex}
\noindent
{\bf \emph{Standard}}: Rules \inferrule{\iruleMPRedCall}, \inferrule{\iruleMPRedCtx}, and \inferrule{\iruleMPRedCongr}
are standard~\cite{DBLP:journals/pacmpl/ScalasY19,lagaillardie2022Affine}. %
\inferrule{\iruleMPRedCall} expands process definitions when invoked;
\inferrule{\iruleMPRedCtx} and \inferrule{\iruleMPRedCongr}
allow processes to reduce under reduction contexts
and through structural congruence,
respectively. 
Rule~\inferrule{\iruleMPRedCongrTime}
introduces a timed variant of~\inferrule{\iruleMPRedCongr}, enabling  
time-consuming reductions via structural congruence.

\vspace{.3ex}
\noindent
{\bf \emph{Congruence}}: %
As shown in~\Cref{fig:aat-mpst-pi-semantics} (bottom),  we introduce additional congruence rules related to  queues, delays, and process killings, alongside standard rules from \cite{DBLP:journals/pacmpl/ScalasY19}. %
Specifically, two %
rules are proposed for queues: the first addresses the \emph{garbage} collection of queues that are not referenced by any process, while the second rearranges messages with different receivers. 
The %
rule for delays states that adding a delay of zero time units has no effect on the process execution. 
 The rule regarding process killings eliminates duplicate kills.

\begin{example}
\label{ex:subjects}
Consider the processes: {\small $\mpP[1] =
\mpSTSel{\mpChanRole{\mpS}{\roleFmt{Sat}}}{\roleFmt{Ser}}{\mpLabFmt{Data}}{\mpNil}{0.4}$},
 {\small $\mpP[2] =
   \mpSTBranch{\mpChanRole{\mpS}{\roleFmt{Ser}}}{\roleFmt{Sat}}{\mpLabFmt{Data}}{\mpNil}{0.3}$}, and 
 {\small $\mpP[3] = \mpSessionQueueO{\mpS}{\roleFmt{Sat}}{\mpQueueEmpty}$}. 
 Rule \inferrule{\iruleMPCCat} can be applied to {\small $\trycatch{\mpP[1]}{\mpQ} \mpPar \kills{\mpS}$},  
 as {\small $\procSubject{\mpP[1]} = \setenum{\mpChanRole{\mpS}{\roleFmt{Sat}}}$} satisfies its side condition.  
However, neither {\small $\mpFailedP{\mpP[1] \mpPar \mpP[2]}$} nor {\small $\mpFailedP{\mpP[3]}$} can generate the killing process 
$\kills{\mpS}$,  as 
{\small $\procSubject{\mpP[1] \mpPar \mpP[2]} = \setenum{\mpChanRole{\mpS}{\roleFmt{Sat}}, \mpChanRole{\mpS}{\roleFmt{Ser}}}$}, 
whereas  
{\small $\procSubject{\mpP[3]} = \setenum{ {\mpChanRole{\mpS}{\roleFmt{Sat}}}^{\mathscr{Q}}}$}.

\end{example}

\vspace{-2ex}
\begin{example}
\label{ex:calculus_reductions}
Processes $\mpQ[\roleFmt{Sen}]$, $\mpQ[\roleFmt{Sat}]$, and $\mpQ[\roleFmt{Ser}]$ 
interact on a session $\mpS$:

\smallskip 
\centerline{\( 
\footnotesize
 \begin{array}{l}
\mpQ[\roleFmt{Sen}] = 
\delay{C_{\roleFmt{Sen}} = 6.5}{\mpQi[\roleFmt{Sen}]} \mpPar  \mpSessionQueueO{\mpS}{\roleFmt{Sen}}{\mpQueueEmpty} \,\,{\footnotesize \text{ where }}\,\,  \mpQi[\roleFmt{Sen}] = \trycatch{\mpTSel{\mpChanRole{\mpS}{\roleFmt{Sen}}}{\roleFmt{Sat}}{\mpLabFmt{Data}}{}{}{0.3}}{\mpCancel{\mpChanRole{\mpS}{\roleFmt{Sen}}}{}}
\\[.5mm]
 \mpQ[\roleFmt{Sat}]
= 
\delay{C_{\roleFmt{Sat}} = 6}{\mpQi[\roleFmt{Sat}]} \mpPar  \mpSessionQueueO{\mpS}{\roleFmt{Sat}}{\mpQueueEmpty} \,\,{\footnotesize \text{ where }}\,\, \mpQi[\roleFmt{Sat}] =  \mpTBranch{\mpChanRole{\mpS}{\roleFmt{Sat}}}{\roleFmt{Sen}}{}
{\left\{
\begin{array}{@{\hskip 0mm}l@{\hskip 0mm}}
\mpChoice{Data}{}{\mpTSel{\mpChanRole{\mpS}{\roleFmt{Sat}}}{\roleFmt{Ser}}{\mpLabFmt{Data}}{}{}{0.3}}
\\
\mpChoice{fail}{}{\mpTSel{\mpChanRole{\mpS}{\roleFmt{Sat}}}{\roleFmt{Ser}}{\mpLabFmt{fatal}}{}{}{0.4}}
\end{array}
\right\}}{}{}{0.2}{}
\\[.5mm]
\mpQ[\roleFmt{Ser}]
= 
\delay{C_{\roleFmt{Ser}} = 6}{\mpQi[\roleFmt{Ser}]} \mpPar  \mpSessionQueueO{\mpS}{\roleFmt{Ser}}{\mpQueueEmpty}  \,\,{\footnotesize \text{ where }}\,\,  \mpQi[\roleFmt{Ser}] = \mpTBranch{\mpChanRole{\mpS}{\roleFmt{Ser}}}{\roleFmt{Sat}}{}
{\left\{
\mpChoice{Data}{}{}, 
\mpChoice{fatal}{}{}
\right\}}{}{}{0.8}{}
\end{array}
\)}

\smallskip
\noindent
Process $\mpQ[\roleFmt{Sen}]$ delays for exactly 6.5 time units before executing process 
$\mpQi[\roleFmt{Sen}]$. Here, $\mpQi[\roleFmt{Sen}]$ attempts to use $\mpChanRole{\mpS}{\roleFmt{Sen}}$ 
to send $\mpLabFmt{Data}$ to  $\roleFmt{Sat}$ within 0.3 time units.  If the attempt fails, the cancellation of $\mpChanRole{\mpS}{\roleFmt{Sen}}$ is triggered. 
Process $\mpQ[\roleFmt{Sat}]$ waits for precisely 6 time units before using 
$\mpChanRole{\mpS}{\roleFmt{Sat}}$ to receive either $\mpLabFmt{Data}$ or $\mpLabFmt{fail}$ from 
$\roleFmt{Sen}$ within 0.2 time units; subsequently, in the first case, it uses 
$\mpChanRole{\mpS}{\roleFmt{Sat}}$ to send $\mpLabFmt{Data}$ to $\roleFmt{Ser}$ within 0.3 time units, while in the latter, it uses $\mpChanRole{\mpS}{\roleFmt{Sat}}$ to send $\mpLabFmt{fail}$ to $\roleFmt{Ser}$ within 0.4 time units. Similarly, process $\mpQ[\roleFmt{Ser}]$ waits 6 time units before using $\mpChanRole{\mpS}{\roleFmt{Ser}}$ to receive either 
$\mpLabFmt{Data}$ or $\mpLabFmt{fatal}$ from $\roleFmt{Sat}$ within 0.8 time units.

In $\mpQ[\roleFmt{Sen}]$, $\mpChanRole{\mpS}{\roleFmt{Sen}}$ can only start sending $\mpLabFmt{Data}$ to 
$\roleFmt{Sat}$ after 6.5 time units, whereas in $\mpQ[\roleFmt{Sat}]$, $\mpChanRole{\mpS}{\roleFmt{Sat}}$ must receive the message from $\roleFmt{Sen}$ within 0.2 time units after a 6-time unit delay. Consequently, 
$\mpChanRole{\mpS}{\roleFmt{Sat}}$ fails to receive the message from $\roleFmt{Sen}$ within the specified interval, resulting in a timeout failure, \ie 

\smallskip
\centerline{\(
\footnotesize
\begin{array}{l@{\hskip .5mm}l@{\hskip .5mm}l}
\mpQ[\roleFmt{Sen}] \mpPar \mpQ[\roleFmt{Sat}] \mpPar \mpQ[\roleFmt{Ser}] 
&
\mpnonTMove
&
\delay{6.5}{\mpQi[\roleFmt{Sen}]} \mpPar   \mpSessionQueueO{\mpS}{\roleFmt{Sen}}{\mpQueueEmpty} \mpPar \delay{6}{\mpQi[\roleFmt{Sat}]} \mpPar   \mpSessionQueueO{\mpS}{\roleFmt{Sat}}{\mpQueueEmpty} \mpPar \delay{6}{\mpQi[\roleFmt{Ser}]} \mpPar   \mpSessionQueueO{\mpS}{\roleFmt{Ser}}{\mpQueueEmpty} 
\\[0.5mm] 
&
\mpMoveTime 
&
\timePass{6.5}{\delay{6.5}{\mpQi[\roleFmt{Sen}]} \mpPar   \mpSessionQueueO{\mpS}{\roleFmt{Sen}}{\mpQueueEmpty} \mpPar \delay{6}{\mpQi[\roleFmt{Sat}]} \mpPar   \mpSessionQueueO{\mpS}{\roleFmt{Sat}}{\mpQueueEmpty} \mpPar \delay{6}{\mpQi[\roleFmt{Ser}]} \mpPar   \mpSessionQueueO{\mpS}{\roleFmt{Ser}}{\mpQueueEmpty}}  
\\[0.5mm]
&
\equiv
&
\mpQi[\roleFmt{Sen}] \mpPar \mpSessionQueueO{\mpS}{\roleFmt{Sen}}{\mpQueueEmpty} \mpPar \mpFailedP{\mpQi[\roleFmt{Sat}]} \mpPar   \mpSessionQueueO{\mpS}{\roleFmt{Sat}}{\mpQueueEmpty} \mpPar \timePass{0.5}{\mpQi[\roleFmt{Ser}]} \mpPar   \mpSessionQueueO{\mpS}{\roleFmt{Ser}}{\mpQueueEmpty}
\end{array}
\)}

\smallskip
\noindent
Therefore, the kill process $\kills{\mpS}$ is generated from $\mpFailedP{\mpQi[\roleFmt{Sat}]}$, successfully terminating  the process $\mpQ[\roleFmt{Sen}] \mpPar \mpQ[\roleFmt{Sat}] \mpPar \mpQ[\roleFmt{Ser}]$  by the following reductions: 

\smallskip
\centerline{\(
\footnotesize
\begin{array}{r@{\hskip 2mm}l@{\hskip 2mm}}
&
\mpQi[\roleFmt{Sen}] \mpPar \mpSessionQueueO{\mpS}{\roleFmt{Sen}}{\mpQueueEmpty} \mpPar \mpFailedP{\mpQi[\roleFmt{Sat}]} \mpPar   \mpSessionQueueO{\mpS}{\roleFmt{Sat}}{\mpQueueEmpty} \mpPar \timePass{0.5}{\mpQi[\roleFmt{Ser}]} \mpPar   \mpSessionQueueO{\mpS}{\roleFmt{Ser}}{\mpQueueEmpty}
\\[.5mm]
\mpnonTMove
&
\mpQi[\roleFmt{Sen}] \mpPar \mpSessionQueueO{\mpS}{\roleFmt{Sen}}{\mpQueueEmpty} \mpPar \kills{\mpS} \mpPar   \mpSessionQueueO{\mpS}{\roleFmt{Sat}}{\mpQueueEmpty} \mpPar \timePass{0.5}{\mpQi[\roleFmt{Ser}]} \mpPar   \mpSessionQueueO{\mpS}{\roleFmt{Ser}}{\mpQueueEmpty}
\\[.5mm]
\mpnonTMove
&
\mpCancel{\mpChanRole{\mpS}{\roleFmt{Sen}}}{} \mpPar   
\mpSessionQueueO{\mpS}{\roleFmt{Sen}}{\mpQueueEmpty} \mpPar 
\kills{\mpS} \mpPar   \mpSessionQueueO{\mpS}{\roleFmt{Sat}}{\mpQueueEmpty} \mpPar 
\mpNil  \mpPar   \mpSessionQueueO{\mpS}{\roleFmt{Ser}}{\mpQueueEmpty} 
\\[.5mm]
\mpnonTMove
&
\kills{\mpS} \mpPar \mpNil \mpPar   
\mpSessionQueueO{\mpS}{\roleFmt{Sen}}{\mpQueueEmpty} \mpPar 
\kills{\mpS} \mpPar   \mpSessionQueueO{\mpS}{\roleFmt{Sat}}{\mpQueueEmpty} \mpPar 
\mpNil  \mpPar   \mpSessionQueueO{\mpS}{\roleFmt{Ser}}{\mpQueueEmpty}  
\;\equiv\; 
\mpNil \mpPar \kills{\mpS}
\end{array}
\)}

\end{example}

\section{Affine Timed Multiparty Session Type System}
\label{sec:aat-mpst-type-system-popl}
\label{SEC:AAT-MPST-TYPE-SYSTEM-POPL}

In this section, we introduce our affine timed multiparty session type system.
We begin by exploring the types used in \ATMP, as well as subtyping and projection, in~\Cref{sec:aat-mpst-all-types}.  %
We furnish a Labelled Transition System (LTS) semantics for typing environments~(collections of timed local types and queue types) in~\Cref{sec:aat-mpst-typing-env}, and %
timed global types in~\Cref{sec:global_local_relation}, 
illustrating their relationship with~\Cref{lem:comp_proj,lem:sound_proj}. 
Furthermore, we present a type system for our~\ATMP session $\pi$-calculus in~\Cref{sec:aat-mpst-typing-system}. 
Finally,  we show the main properties of the type system:
\emph{subject reduction}~(\Cref{lem:sr_global}), \emph{session fidelity}~(\Cref{lem:aat-mpst-session-fidelity-global}), and \emph{deadlock-freedom}~(\Cref{lem:aat-mpst-process-df-proj}),
in~\Cref{subsec:pptes_atmp_session_types}.

\subsection{Timed Multiparty Session Types}
\label{sec:aat-mpst-all-types}
\label{SEC:AAT-MPST-ALL-TYPES}

\begin{figure}[t]%
  \centerline{\(
  \small
    \begin{array}{r@{\quad}c@{\quad}l@{\quad}l}
       \stS& \bnfdef & \stDelegate{\ccst}{\stT}
      & \text{\footnotesize{Session Type}}
      \\
      \gtG & \bnfdef &
      \gtCommT{\roleP}{\roleQ}{i \in
        I}{\gtLab[i]}{\stS[i]}{\ccstO[i], \crstO[i], \ccstI[i], \crstI[i]}{\gtG[i]}
        &
        {\footnotesize\text{Transmission}} \\
        & \bnfsep &
          \gtCommTTransit{\roleP}{\roleQ}{i \in
          I}{\gtLab[i]}{\stS[i]}{\ccstO[i], \crstO[i], \ccstI[i], \crstI[i]}{\gtG[i]}{j}~(j \in I)
        &
      \footnotesize\text{Transmission en route}%
      \\
        & \bnfsep & \gtRec{\gtRecVar}{\gtG} \quad \bnfsep \quad \gtRecVar \quad
        \bnfsep \quad \gtEnd &
        \text{\footnotesize Recursion, Type variable, Termination}
       \\[0.5ex]
      \stT
        & \bnfdef & \stExtSum{\roleP}{i \in I}{\stTChoice{\stLab[i]}{\stS[i]}{\ccst[i], \crst[i]} \stSeq \stT[i]}
          & \text{\footnotesize External choice}\\
        & \bnfsep & \stIntSum{\roleP}{i \in I}{\stTChoice{\stLab[i]}{\stS[i]}{\ccst[i], \crst[i]} \stSeq \stT[i]}
          & \text{\footnotesize Internal choice}\\
        & \bnfsep & \stRec{\stRecVar}{\stT} \ \bnfsep \ \stRecVar \
          \bnfsep \ \stEnd &
        \text{\footnotesize Recursion, Type variable, Termination}
        \\
          \stQType
 &\coloncolonequals&%
  \stQCons{\stQMsg{\roleP}{\stLab}{\stS}}{\stQType}%
  \bnfsep \stQEmptyType
  &
   \text{\footnotesize{Queue types}}%
    \end{array}
  \)}
  \vspace{-.5em}
  \caption{Syntax of timed global types, timed local types, and queue types.}
  \label{fig:syntax-global-type}%
  \label{fig:syntax-local-type}%
  \label{fig:syntax-aatmpst}
  \label{fig:aat-mpst-syntax}
  \vspace{-1em}
\end{figure}

Affine session frameworks keep the original system's type-level syntax intact, requiring no changes.
To introduce affine timed asynchronous multiparty session types,
we simply need to augment global and local types with clock constraints and resets 
introduced in~\Cref{sec:session-calculus} 
to derive \emph{timed global and local types}.
The syntax of types used in this paper is presented in~\Cref{fig:aat-mpst-syntax}. %
As usual, all types are required to be closed and have guarded recursion variables.

\vspace{-1em}
\subparagraph*{Sorts} 
are ranged over $\stS, \stSi, \stS[i], \ldots$, and
facilitate the delegation of the
remaining behaviour $\stT$ to the receiver, who can execute it under any
clock assignment satisfying $\ccst$.

\vspace{-1em}
\subparagraph*{Timed Global Types} are ranged over $\gtG, \gtGi, \gtG[i], \ldots$, and %
describe an \emph{overview} of the behaviour for all roles~($\roleP, \roleQ, \roleS, \roleT, \ldots$) belonging to a (fixed) set $\roleSet$.  The set of roles in a timed global type 
$\gtG$ is denoted as $\gtRoles{\gtG}$, while  the set of its free variables as $\fv{\gtG}$. %
We explain each syntactic construct of global types as follows.  

\vspace{.2ex}
\noindent
{\bf  \emph{Transmission}} 
$\gtCommT{\roleP}{\roleQ}{i \in
        I}{\gtLab[i]}{\stS[i]}{\ccstO[i], \crstO[i], \ccstI[i], \crstI[i]}{\gtG[i]}$ 
        represents a message sent from role $\roleP$ to role $\roleQ$, with labels $\gtLab[i]$,
payload types $\stS[i]$ (which are sorts),
and continuations $\gtG[i]$, where $i$ is taken from an index set $I$, and
$\gtLab[i]$ taken from a fixed set of all labels $\labSet$.
Each branch
is associated with a \emph{time assertion}
 consisting of four components: $\ccstO[i]$ and $\crstO[i]$ for the output (sending) action,
and $\ccstI[i]$ and $\crstI[i]$ for the input (receiving) action.
These components specify the clock constraint and reset predicate for the respective actions.
A message %
can be sent (or received) at any time satisfying the guard $\ccstO[i]$ (or $\ccstI[i]$),
and the clocks in $\crstO[i]$ (or $\crstI[i]$) are reset upon sending (or receiving).

In addition to the standard requirements for global types, including the index set being non-empty ($I \neq \emptyset$), labels $\gtLab[i]$ being pairwise distinct, and self-receptions being excluded ($\roleP \neq \roleQ$),
we impose a condition from~\cite{DBLP:conf/concur/BocchiYY14}, stating that %
sets of clocks ``owned'' by
different roles, \ie those that can be read and reset, must be pairwise disjoint. Furthermore, 
the clock constraint and %
reset predicate of an output or %
input action performed by a role
 are defined only over the clocks owned by that role. %
 For example,  
  $
   \gtCommT{\roleFmt{Sat}}{\roleFmt{Ser}}{}{
      Data
   }{}{
      \mpFmt{6 \leq C_{\roleFmt{Sat}} \leq 7, C_{\roleFmt{Sat}} := 0, 7 \leq C_{\roleFmt{Sat}} \leq 8, C_{\roleFmt{Sat}} := 0}
   }{
      \gtEnd
   }$ violates this assumption. 
   Both $\roleFmt{Sat}$'s sending and $\roleFmt{Ser}$'s receiving actions have time constraints defined over $C_{\roleFmt{Sat}}$,
  which can be owned by either of them. The same issue exists for the reset %
  $\setenum{C_{\roleFmt{Sat}}}$.
   
\vspace{.2ex}
\noindent
{\bf \emph{Transmission en route}} %
$\gtCommTTransit{\roleP}{\roleQ}{i \in
          I}{\gtLab[i]}{\stS[i]}{\ccstO[i], \crstO[i], \ccstI[i], \crstI[i]}{\gtG[i]}{j}~(j \in I)$  is a \emph{runtime} construct
to represent a message $\gtLab[j]$ sent by $\roleP$, and yet to be received by $\roleQ$. 
The distinction between $\gtFmt{\rightarrow}$ and $\gtFmt{\rightsquigarrow}$ is crucial for simulating asynchronous communication. %

\vspace{.2ex} 
\noindent
{\bf \emph{Recursion}} $\gtRec{\gtRecVar}{\gtG}$ and {\bf \emph{termination}} $\gtEnd$ (omitted where unambiguous)
are standard~\cite{ICALP13CFSM}. Note that contractive requirements~\cite[\S 21.8]{DBLP:books/daglib/0005958}, \ie ensuring that each recursion variable $\gtRecVar$ is bound within a $\gtRec{\gtRecVar}{\ldots}$ and is guarded, are applied in recursive types.

\vspace{-1em}
\subparagraph{Timed Local Types (timed session types)}
are ranged over $\stT, \stU, \stTi, \stUi, \stT[i], \stU[i], 
\ldots$, and describe the behaviour of a \emph{single}  role.
An \emph{internal choice} (\emph{selection})
$\stIntSum{\roleP}{i \in I}{\stTChoice{\stLab[i]}{\stS[i]}{\ccst[i], \crst[i]} \stSeq \stT[i]}$
(or \emph{external choice} (\emph{branching})
$\stExtSum{\roleP}{i \in I}{\stTChoice{\stLab[i]}{\stS[i]}{\ccst[i], \crst[i]} \stSeq \stT[i]}$%
)
states that the \emph{current} role is to \emph{send} to (or
\emph{receive} from) the role $\roleP$ when $\ccst[i]$ is satisfied,
followed by resetting the clocks in $\crst[i]$.
\emph{Recursive} and \emph{termination} types are defined
similarly to timed global types.
The requirements for the index set, labels,  
clock constraints, and reset predicates
in timed local types mirror those in timed global types.

\vspace{-1em}
\subparagraph{Queue Types} are ranged over $\stQType, \stQTypei, \stQType[i], \ldots$, and  %
represent (possibly empty) sequences of %
\emph{message types} $\stQMsg{\roleP}{\stLab}{\stS}$ %
having receiver $\roleP$, %
label $\stLab$, and payload type $\stS$~(omitted when $\stS \!=\! \stDelegate{\ccst}{\stEnd}$). %
As interactions in our formalisation are asynchronous,
queue types are used to capture the states in which messages are in transit. 
We adopt the notation $\operatorname{receivers}(\cdot)$ from \Cref{sec:session-calculus} to denote 
the set of \emph{receivers} in 
$\stQType$ as $\operatorname{receivers}(\stQType)$ as well, with a similar definition.

\vspace{-1em}
\subparagraph{Subtyping}   %
 We introduce a subtyping relation $\stSub$ on timed local types 
 in~\Cref{def:main_subtyping},    based on 
the standard behaviour-preserving subtyping~\cite{DBLP:journals/pacmpl/ScalasY19}. 
This %
relation indicates that a smaller type entails fewer external choices but more internal choices.

\begin{definition}[Subtyping]
\label{def:main_subtyping}
The %
\emph{subtyping relation} $\stSub$ is coinductively defined:

\smallskip
\centerline{\(
\small 
\begin{array}{c}
\cinference[\iruleStSubOut]{
  \forall i \in I
  &
  \stSi[i] \stSub \stS[i]
  &
  \ccst[i] = \ccsti[i]
  &
  \crst[i] = \crsti[i]
  &
  \stT[i] \stSub \stTi[i] 
}{
  \stIntSum{\roleP}{i \in I \cup J}{\stTChoice{\stLab[i]}{\stS[i]}{\ccst[i], \crst[i]} \stSeq \stT[i]}%
  \stSub
  \stIntSum{\roleP}{i \in I}{\stTChoice{\stLab[i]}{\stSi[i]}{\ccsti[i], \crsti[i]} \stSeq \stTi[i]}%
}
\\[.5mm]
\cinference[\iruleStSubIn]{
  \forall i \in I
  &
  \stS[i] \stSub \stSi[i]
  &
  \ccst[i] = \ccsti[i]
  &
  \crst[i] = \crsti[i]
  &
  \stT[i] \stSub \stTi[i]
}{
  \stExtSum{\roleP}{i \in I}{\stTChoice{\stLab[i]}{\stS[i]}{\ccst[i], \crst[i]} \stSeq \stT[i]}
  \stSub
  \stExtSum{\roleP}{i \in I \cup J}{\stTChoice{\stLab[i]}{\stSi[i]}{\ccsti[i], \crsti[i]} \stSeq \stTi[i]}
}
\quad 
\cinference[\iruleStSubEnd]{}{
  \stEnd \stSub \stEnd}
\\[.5mm]%
\cinference[\iruleStSubSort]{\stT \stSub \stTi}{
 \stDelegate{\ccst}{\stT} \stSub \stDelegate{\ccst}{\stTi}
}
\quad
\cinference[\iruleStSubRecL]{
  \stT{}\subst{\stRecVar}{\stRec{\stRecVar}{\stT}} \stSub \stTi
}
{
  \stRec{\stRecVar}{\stT} \stSub \stTi
}
\quad
\cinference[\iruleStSubRecR]{
  \stT \stSub \stTi{}\subst{\stRecVar}{\stRec{\stRecVar}{\stTi}}
}{
  \stT \stSub \stRec{\stRecVar}{\stTi}
}
\end{array}
\)}
\end{definition}

\vspace{-1em}  
\subparagraph{Projection} of a timed global type $\gtG$ onto a role $\roleP$ %
yields a timed local type. 
Our definition of projection %
in~\Cref{def:normal_proj} is mostly standard~\cite{DBLP:journals/pacmpl/ScalasY19}, 
with the addition of projecting time assertions onto the sender and receiver, respectively.
For example, when projecting a transmission from $\roleP$ to $\roleQ$ onto $\roleP$~(or $\roleQ$), 
an internal (or external) choice with time assertions for output (or input) is obtained, provided that 
each branching's continuation is projectable. 
For the projection of the transmission onto other participants $\roleR$,  a \emph{merge operator} is used to %
ensure the ``compatibility'' of the projections of all continuations. 

\begin{definition}[Projection]
\label{def:main_proj}
\label{def:normal_proj}
  \label{def:global-proj}%
  \label{def:local-type-merge}%
  The \emph{projection of a timed global type $\gtG$ onto a role $\roleP$}, 
  written as $\gtProj{\gtG}{\roleP}$, %
  is:

  \smallskip%
  \centerline{\(%
  \small%
  \begin{array}{c}
    \gtProj{\left(%
      \gtCommTSmall{\roleQ}{\roleR}
      {i \in I}{\gtLab[i]}{\stS[i]}{\ccstO[i], \crstO[i], \ccstI[i], \crstI[i]}{\gtG[i]}%
      \right)}{\roleP}%
    =\!%
    \left\{%
    \begin{array}{@{}l@{\hskip 2mm}l@{}}
      \stIntSum{\roleR}{i \in I}{ %
        \stTChoice{\stLab[i]}{\stS[i]}{\ccstO[i], \crstO[i]} \stSeq (\gtProj{\gtG[i]}{\roleP})%
      }%
      &\text{\footnotesize%
        if\, $\roleP = \roleQ$%
      }%
      \\[1mm]%
      \stExtSum{\roleQ}{i \in I}{%
        \stTChoice{\stLab[i]}{\stS[i]}{\ccstI[i], \crstI[i]} \stSeq (\gtProj{\gtG[i]}{\roleP})%
      }%
      &\text{\footnotesize%
        if\, $\roleP = \roleR$%
      }%
          \\[1mm]%
      \stMerge{i \in I}{\gtProj{\gtG[i]}{\roleP}}%
      &
      \text{\footnotesize%
      otherwise
      }%
    \end{array}
    \right.
    \\[8mm]%
      \gtProj{\left(%
      \gtCommTTransit{\roleQ}{\roleR}
      {i \in I}{\gtLab[i]}{\stS[i]}{\ccstO[i], \crstO[i], \ccstI[i], \crstI[i]}{\gtG[i]}{j}%
      \right)}{\roleP}%
    =\!%
    \left\{%
    \begin{array}{@{}l@{\hskip 2mm}l@{}}
     \gtProj{\gtG[j]}{\roleP}
      &\text{\footnotesize%
        if\, $\roleP = \roleQ$}
        \\[1mm]       
     \stExtSum{\roleQ}{i \in I}{%
        \stTChoice{\stLab[i]}{\stS[i]}{\ccstI[i], \crstI[i]} \stSeq (\gtProj{\gtG[i]}{\roleP})%
      }%
      &\text{\footnotesize%
        if\, $\roleP = \roleR$%
      }%
   \\[1mm]
 \stMerge{i \in I}{\gtProj{\gtG[i]}{\roleP}}%
      &
      \text{\footnotesize%
       otherwise
      }
    \end{array}
    \right.
    \\[4mm]%
    \gtProj{(\gtRec{\gtRecVar}{\gtG})}{\roleP}%
    \;=\;%
    \left\{%
    \begin{array}{@{\hskip 0.5mm}l@{\hskip 5mm}l@{}}
      \stRec{\stRecVar}{(\gtProj{\gtG}{\roleP})}%
      &%
      \text{\footnotesize%
        if\,
        $\roleP \in \gtRoles{\gtG}$ \,or\,
        $\fv{\gtRec{\gtRecVar}{\gtG}} \neq \emptyset$%
      }%
      \\%
      \stEnd%
      &%
      \text{\footnotesize%
        otherwise}
    \end{array}
    \right.%
    \quad\qquad%
    \begin{array}{@{}r@{\hskip 1mm}c@{\hskip 1mm}l@{}}
      \gtProj{\gtRecVar}{\roleP}%
      &=&%
      \stRecVar%
      \\%
      \gtProj{\gtEnd}{\roleP}%
      &=&%
      \stEnd%
    \end{array}
  \end{array}
  \)}%
  \smallskip%

  \noindent%
  where
  $\stMerge{}{}$ is %
  the \emph{merge operator for timed session types}: %

  \smallskip%
  \centerline{\(%
  \small%
  \begin{array}{c}%
    \textstyle%
    \stExtSum{\roleP}{i \in I}{\stTChoice{\stLab[i]}{\stS[i]}{\ccst[i], \crst[i]} \stSeq \stT[i]}%
    \;\stBinMerge\;%
    \stExtSum{\roleP}{\!j \in J}{\stTChoice{\stLab[j]}{\stSi[j]}{\ccsti[j], \crsti[j]} \stSeq \stTi[j]} = %
    \\[1mm]
    \stExtSum{\roleP}{k \in I \cap J}{\stTChoice{\stLab[k]}{\stS[k]}{\ccst[k], \crst[k]} \stSeq%
      (\stT[k] \!\stBinMerge\! \stTi[k])%
    }%
    \stExtC%
    \stExtSum{\roleP}{i \in I \setminus J}{\stTChoice{\stLab[i]}{\stS[i]}{\ccst[i], \crst[i]} \stSeq \stT[i]}%
    \stExtC%
    \stExtSum{\roleP}{\!j \in J \setminus I}{\stTChoice{\stLab[j]}{\stSi[j]}{\ccsti[j], \crsti[j]} \stSeq \stTi[j]}%
    \\[1mm]%
    \stIntSum{\roleP}{i \in I}{\stTChoice{\stLab[i]}{\stS[i]}{\ccst[i], \crst[i]} \stSeq \stT[i]}%
    \,\stBinMerge\,%
    \stIntSum{\roleP}{i \in I}{\stTChoice{\stLab[i]}{\stS[i]}{\ccst[i], \crst[i]} \stSeq \stTi[i]}%
    \;=\;%
    \stIntSum{\roleP}{i \in I}{\stTChoice{\stLab[i]}{\stS[i]}{\ccst[i], \crst[i]} \stSeq (\stT[i]
    \stBinMerge \stTi[i])}%
    \\[1mm]%
    \stRec{\stRecVar}{\stT} \,\stBinMerge\, \stRec{\stRecVar}{\stU}%
    \,=\,%
    \stRec{\stRecVar}{(\stT \stBinMerge \stU)}%
    \qquad%
    \stRecVar \,\stBinMerge\, \stRecVar%
    \,=\,%
    \stRecVar%
    \qquad%
    \stEnd \,\stBinMerge\, \stEnd%
    \,=\,%
    \stEnd%
  \end{array}
  \)}%

\end{definition}

\begin{example}
\label{ex:proj_example}
Take the timed global type $\gtG$, and timed local types  $\stT[\roleFmt{Sat}]$ and 
$\stT[\roleFmt{Ser}]$ from~\Cref{sec:overview:theory}. 
Consider a timed global type $\gtG[\text{data}]$, 
derived from remote data~(\Cref{fig:implementation:remote_data}) as well, 
representing  data transmission from  
$\roleFmt{Sen}$ to $\roleFmt{Ser}$ via $\roleFmt{Sat}$:  

\centerline{\(
\footnotesize
\gtG[\text{data}] = \gtCommT{\roleFmt{Sen}}{\roleFmt{Sat}}{}{
      \gtMsgFmt{Data}
   }{}{
  \mpFmt{6 \leq C_{\roleFmt{Sen}} \leq 7, C_{\roleFmt{Sen}} := 0, 
  6 \leq C_{\roleFmt{Sat}} \leq 7, \emptyset}
   }{\gtG} 
   \)}
   
   \noindent
which can be projected onto roles $\roleFmt{Sen}$, 
   $\roleFmt{Sat}$, and $\roleFmt{Ser}$, respectively, as: 
 
 \smallskip
 \centerline{\(
 \footnotesize
 \begin{array}{c}
 \gtProj{\gtG[\text{data}]}{\roleFmt{Sen}} = \roleFmt{Sat}\stFmt{\oplus}
   \stLabFmt{Data}\stFmt{\{ \mpFmt{6 \leq C_{\roleFmt{Sen}} \leq 7, C_{\roleFmt{Sen}} := 0}\}}
   \stSeq
   \stEnd
   \qquad 
    \gtProj{\gtG[\text{data}]}{\roleFmt{Ser}} =  \gtProj{\gtG}{\roleFmt{Ser}} = \stT[\roleFmt{Ser}]
 \\[0.5mm]
  \gtProj{\gtG[\text{data}]}{\roleFmt{Sat}} =  
  \roleFmt{Sen}\stFmt{\&}
   \stLabFmt{Data}\stFmt{\{\mpFmt{6 \leq C_{\roleFmt{Sat}} \leq 7, \emptyset}\}}
   \stSeq
   \gtProj{\gtG}{\roleFmt{Sat}}
= 
\roleFmt{Sen}\stFmt{\&}
   \stLabFmt{Data}\stFmt{\{\mpFmt{6 \leq C_{\roleFmt{Sat}} \leq 7, \emptyset}\}}
   \stSeq
   \stT[\roleFmt{Sat}] 
 \end{array}
 \)}
\end{example}

\subsection{Typing Environments}
\label{sec:aat-mpst-typing-env}
\label{SEC:AAT-MPST-TYPING-ENV}
To reflect the behaviour of timed global types~(\Cref{sec:global_local_relation}),
present a typing system for our session $\pi$-calculus (\Cref{sec:aat-mpst-typing-system}),
and introduce type-level properties (\Cref{subsec:pptes_atmp_session_types}),
we %
formalise \emph{typing environments} in~\Cref{def:aat-mpst-typing-env-syntax}, followed by their Labelled Transition System (LTS) semantics in~\Cref{def:aat-mpst-typing-env-reduction}.

\begin{figure}[!t]
\centerline{\(
\small
\begin{array}{c}
 \inference{%
      \vphantom{\stT \stQEquiv \stTi}%
    }{%
      \stCPair{\cVal}{\stT} \stQEquiv \stCPair{\cVal}{\stT}%
    }%
    \qquad%
    \inference{%
     \highlight{ \roleP \neq \roleQ}%
    }{%
      \stQCons{\stQMsg{\roleP}{\stLab[1]}{\stS[1]}}{%
        \stQCons{\stQMsg{\roleQ}{\stLab[2]}{\stS[2]}}{%
          \stQType%
        }%
      }%
      \;\stQEquiv\;%
      \stQCons{\stQMsg{\roleQ}{\stLab[2]}{\stS[2]}}{%
        \stQCons{\stQMsg{\roleP}{\stLab[1]}{\stS[1]}}{%
          \stQType%
        }%
      }%
    }%
   \\[0.5mm]
    \inference{}
    {\stQCons{\stQEmptyType}{\stQEmptyType} \stQEquiv \stQEmptyType}
    \qquad 
    \inference{}
    {%
      \stQCons{\stQMsg{\roleP}{\stLab}{\stS}}{%
        \stQCons{\stQEmptyType}{%
          \stQType%
        }%
      }%
      \;\stQEquiv\;%
      \stQCons{\stQEmptyType}{%
        \stQCons{\stQMsg{\roleP}{\stLab}{\stS}}{%
          \stQType%
        }%
      }%
    }
    \qquad 
    \inference{%
      \stQType \stQEquiv \stQTypei%
      \quad%
      \stCPair{\cVal}{\stT} \stQEquiv \stCPair{\cVal}{\stTi}%
    }{%
      \stMPair{\stCPair{\cVal}{\stT}}{\stQType} \stQEquiv \stMPair{\stCPair{\cVal}{\stTi}}{\stQTypei}%
    }%
\\[1mm]
 \hdashline
\\[-2.5mm]
\cinference[\iruleSubEQueueTypeEmpty]{%
 \vphantom{X}%
  }{%
    \stQEmptyType \stSub \stQEmptyType
   }%
\quad
    \cinference[\iruleSubEQueueType]{%
     \highlight{\stSi \stSub \stS} & \stQType \stSub \stQTypei%
    }{%
       \stQCons{\stQMsg{\roleQ}{\gtLab}{\stS}}{\stQType}
       \stSub   \stQCons{\stQMsg{\roleQ}{\gtLab}{\stSi}}{\stQTypei}}
       \\[0.5mm]
        \cinference[\iruleSubESessionType]{%
      \stT \stSub \stTi%
    }{%
      \stCPair{\cVal}{\stT} \stSub \stCPair{\cVal}{\stTi}%
    }%
\qquad
 \cinference[\iruleSubQueue]{%
      \stQType \stSub \stQTypei%
      &
      \stCPair{\cVal}{\stT} \stSub \stCPair{\cVal}{\stTi}%
    }{%
      \stMPair{\stCPair{\cVal}{\stT}}{\stQType} \stSub \stMPair{\stCPair{\cVal}{\stTi}}{\stQTypei}%
    }%
\end{array}
\)}
\vspace{-.5em}
\caption{Congruence (top) and subtyping (bottom) rules for timed-session/queue types.}
\label{fig:sub_cong_com_types}
\vspace{-1em}
\end{figure}
\begin{definition}[Typing Environments]%
\label{def:aat-mpst-typing-env-syntax}%
The typing environments $\mpEnv$ and $\stEnv$ are defined as:%

 \smallskip%
  \centerline{\(%
  \small
  \mpEnv%
  \;\;\coloncolonequals\;\;%
  \mpEnvEmpty%
  \bnfsep%
 \mpEnv \mpEnvComp\, \mpEnvMap{\mpX}{\stCPair{\cVal[1]}{\stT[1]}, \ldots, \stCPair{\cVal[n]}{\stT[n]}}%
  \quad\quad
  \stEnv%
  \,\coloncolonequals\,%
  \stEnvEmpty%
  \bnfsep%
  \stEnv \stEnvComp \stEnvMap{x}{\stCPair{\cVal}{\stT}}%
  \bnfsep%
  \stEnvQ \stEnvQComp%
  \stEnvMap{%
    \mpChanRole{\mpS}{\roleP}%
  }{
    \stSpecf
  }%
  \)}%

\smallskip%
  \noindent%
  where $\stSpecf$ is a \emph{timed-session/queue} type:  $\stSpecf \coloncolonequals  \stCPair{\cVal}{\stT} \bnfsep \stQType \bnfsep \stMPair{\stCPair{\cVal}{%
         \stT}}{\stQType}$, \ie either a timed session type, a queue type, or a combination. 
  
  The environment \emph{composition} $\stEnv[1] \stEnvQComp \stEnv[2]$ 
 is defined iff 
 \(%
  \forall \mpC \in \dom{\stEnv[1]} \cap \dom{\stEnv[2]}:%
  \stEnvApp{\stEnv[i]}{\mpC} \!=\! \stQType%
  \text{ and }%
  \stEnvApp{\stEnv[\!j]}{\mpC} \!=\! \stCPair{\cVal}{\stT}%
  \text{ with $i, j \in \setenum{1, 2}$}
  \), and for all such $\mpC$, we posit  %
  $\stEnvApp{(\stEnv[1] \stEnvQComp \stEnv[2])}{\mpC}%
  = \stMPair{\stCPair{\cVal}{\stT}}{\stQType}$.  
 
   We write $\dom{\stEnvQ} = \setenum{\mpS}$ iff for any $\mpC \in \dom{\stEnvQ}$, there
  is $\roleP$ such that $\mpC =  \mpChanRole{\mpS}{\roleP}$ (\ie $\stEnvQ$ only contains session $\mpS$).
 We write $\mpS \not \in \stEnvQ$ iff
 $\forall \roleP: \mpChanRole{\mpS}{\roleP} \not \in \dom{\stEnvQ}$ (\ie session $\mpS$ does not occur in $\stEnvQ$).  
 We write $\stEnv[\mpS]$ iff $\dom{\stEnv[\mpS]} = \setenum{\mpS}$, $\dom{\stEnv[\mpS]} \!\subseteq\! \dom{\stEnv}$, 
  and $\forall \mpChanRole{\mpS}{\roleP} \!\in\! \dom{\stEnv}: \stEnvApp{\stEnv}{\mpChanRole{\mpS}{\roleP}} = \stEnvApp{\stEnv[\mpS]}{\mpChanRole{\mpS}{\roleP}}$~(\ie restriction of  $\stEnv$ to session $\mpS$). 
 We denote updates as $\stEnvUpd{\stEnv}{\mpC}{\stSpecf}$:
  $\stEnvApp{\stEnvUpd{\stEnv}{\mpC}{\stSpecf}}{ \mpC} = \stSpecf$ and %
  $\stEnvApp{\stEnvUpd{\stEnv}{ \mpC}{\stSpecf}}{\mpCi} =
  \stEnvApp{\stEnv}{\mpCi}$ (where $\mpC \neq \mpCi$).

   Congruence and subtyping are imposed on typing environments: 
  $\stEnvQ \stQEquiv \stEnvQi$ (resp. $\stEnvQ \stSub \stEnvQi$) 
  iff %
  $\dom{\stEnvQ} = \dom{\stEnvQi}$ %
   and  %
  $\forall \mpC \in \dom{\stEnvQ}:
  \stEnvApp{\stEnvQ}{\mpC} \stQEquiv \stEnvApp{\stEnvQi}{\mpC}$~(resp. %
  $\stEnvApp{\stEnvQ}{\mpC} \stSub \stEnvApp{\stEnvQi}{\mpC}$), incorporating  additional 
  congruence and subtyping rules for time-session/queue types, as depicted in~\Cref{fig:sub_cong_com_types}.

\end{definition}

In~\Cref{def:aat-mpst-typing-env-syntax}, the typing environment  $\mpEnv$ maps process variables 
to $n$-tuples of timed session types, while
 $\stEnv$ maps variables to timed session types, and 
 channels with roles  to timed-session/queue types. 
Note that in our typing environments, timed session types are annotated with clock valuations, 
denoted as $\stCPair{\cVal}{\stT}$. 
This enables us to capture timing information within the type system, facilitating the tracking of the (virtual) time at which the next action can be validated during the execution of a process. 

The congruence relation $\stQEquiv$ for timed-session/queue types is inductively defined as in~\Cref{fig:sub_cong_com_types} (top),  reordering queued messages with different receivers. Subtyping for timed-session/queue types 
extends~\Cref{def:main_subtyping} with rules in~\Cref{fig:sub_cong_com_types} (bottom): particularly,  
rule \inferrule{\iruleSubEQueueType} states that a sequence of queued message types is a subtype of another if messages in the same position have identical receivers and labels, and their payload sorts are related by subtyping.

\begin{figure}[t]%
  \centerline{\(
  \small
  \begin{array}{c}
     \inference[\iruleTCtxRec]{%
      \stEnv \!\stEnvComp \stEnvMap{%
        \mpChanRole{\mpS}{\roleP}%
      }{
      \stCPair{\cVal}{%
        \stT\subst{\stRecVar}{\stRec{\stRecVar}{\stT}}}%
      }%
      \stEnvQTMoveGenAnnotT \stEnvi%
    }{%
      \stEnv \!\stEnvComp \stEnvMap{%
       \mpChanRole{\mpS}{\roleP}
      }{
      \stCPair{\cVal}{%
        \stRec{\stRecVar}{\stT}}%
      }%
      \stEnvQTMoveGenAnnotT \stEnvi%
    }%
    \quad
    \inference[\iruleTCtxCongX]{%
      \stEnv \stEnvQTMoveGenAnnotT \stEnvi
   &
  \highlight{\stTEnvAnnotGenericSymT \neq \timeLab}
    }{%
      \stEnv \!\stEnvComp \stEnvMap{\mpFmt{x}}{\stCPair{\cVal}{\stT}}
      \stEnvQTMoveGenAnnotT%
      \stEnvi \!\stEnvComp \stEnvMap{\mpFmt{x}}{\stCPair{\cVal}{\stT}}
    }%
    \quad
     \inference[\iruleTCtxCongCombined]{%
      \stEnv \stEnvQTMoveGenAnnotT \stEnvi
      &
   \highlight{ \stTEnvAnnotGenericSymT \neq \timeLab}
    }{%
      \stEnv \!\stEnvComp \stEnvMap{\mpChanRole{\mpS}{\roleP}}{\stSpecf}
      \stEnvQTMoveGenAnnotT%
      \stEnvi \!\stEnvComp \stEnvMap{\mpChanRole{\mpS}{\roleP}}{\stSpecf}}%
\\[1mm]
       \inference[\iruleTCtxSend]
    {k \!\in\! I
     &
    \highlight{\cVal \models \ccst[k]}}
     { \stEnvMap{%
        \mpChanRole{\mpS}{\roleP}%
      }{
        \stMPair{
        \stCPair{\cVal}{%
          \stIntSum{\roleQ}{i \in I}{\stTChoice{\stLab[i]}{\stS[i]}{\ccst[i], \crst[i]} \stSeq \stT[i]}}%
        }{%
          \stQType%
        }%
      \;\stEnvQTMoveQueueAnnot{\roleP}{\roleQ}{\stLab[k]}\;%
      \stEnvMap{%
        \mpChanRole{\mpS}{\roleP}%
      }{%
        \stMPair{
        \stCPair{\mpFmt{\highlight{\cValUpd{\cVal}{\crst[k]}{0}}}}
          {\stT[k]}%
        }{%
          \stQCons{\stQType}{%
            \stQCons{%
              \stQMsg{\roleQ}{\stLab[k]}{\stS[k]}%
            }{%
              \stQEmptyType%
            }%
          }%
      }}
      }}
      \\[1mm]
      \inference[\iruleTCtxRcv]
       {k \!\in\! I &
      \highlight{\cVal \models \ccst[k]} &
      \stS[k]\!\stSub\!\stSi[k]
}
      {
      \stEnvMap{%
        \mpChanRole{\mpS}{\roleP}%
      }{%
       \stQCons{\stQMsg{\roleQ}{\stLab[k]}{\stS[k]}}{\stQType}%
      } %
      \stEnvQComp %
      \stEnvMap{%
        \mpChanRole{\mpS}{\roleQ}%
      }{
      \stCPair{\cVal}{%
          \stExtSum{\roleP}{i \in I}{\stTChoice{\stLab[i]}{\stSi[i]}{\ccst[i], \crst[i]} \stSeq \stT[i]}}%
      }%
      \;\stEnvQTMoveRecvAnnot{\roleQ}{\roleP}{\stLab[k]}\;
      \stEnvMap{%
        \mpChanRole{\mpS}{\roleP}%
      }{%
               \stQType
      }%
      \stEnvQComp%
      \stEnvMap{%
        \mpChanRole{\mpS}{\roleQ}%
      }{
        \stCPair{\highlight{\mpFmt{\cValUpd{\cVal}{\crst[k]}{0}}}}%
          {\stT[k]}
          }
          }
 \\[1mm]
   \inferenceSingle[\iruleTCtxTimeSession]
      {\highlight{\stEnvMap{%
       \mpC
      }{\stCPair{\cVal}{%
         \stT}}
\;\stEnvQTMoveTimeAnnot\;
       \stEnvMap{%
        \mpC
      }{%
          \stCPair{\cVal + \timeLab}{\stT}%
        }}}
      \qquad 
        \inferenceSingle[\iruleTCtxTimeQ]
        { \highlight{\stEnvMap{%
        \mpChanRole{\mpS}{\roleP}%
      }{\stQType}
       \;\stEnvQTMoveTimeAnnot\;
       \stEnvMap{%
        \mpChanRole{\mpS}{\roleP}%
      }{%
          \stQType}}}
\\[1mm] 
            \inferenceSingle[\iruleTCtxTimeCombined]
      { \highlight{\stEnvMap{%
        \mpChanRole{\mpS}{\roleP}%
      }{\stMPair{\stCPair{\cVal}{%
         \stT}}{\stQType}}
       \;\stEnvQTMoveTimeAnnot\;
       \stEnvMap{%
        \mpChanRole{\mpS}{\roleP}%
      }{%
          \stMPair{\stCPair{\cVal + \timeLab}{%
         \stT}}{\stQType}%
        }}}
          \\[1mm]
 \begin{array}{c}
 \inference[\iruleTCtxTime]
       {\highlight{\stEnv[1] \;\stEnvQTMoveTimeAnnot\; \stEnvi[1]}
       &  \; { \highlight{\stEnv[2] \;\stEnvQTMoveTimeAnnot\; \stEnvi[2]}}
    }
    {
     \highlight{ \stEnv[1] \stEnvComp \stEnv[2]
          \;\stEnvQTMoveTimeAnnot\;
      \stEnvi[1] \stEnvComp \stEnvi[2]}}
   \qquad 
    \inference[\iruleTCtxStruct]
       {\stEnv \;\stQEquiv\; \stEnv[1]%
       &
       \stEnv[1] \;\stEnvQTMoveGenAnnotT\; \stEnvi[1]
       &
       \stEnvi[1] \;\stQEquiv\; \stEnvi
       }{\stEnv \;\stEnvQTMoveGenAnnotT\; \stEnvi
     }
     \end{array}
 \end{array}
\)}
\vspace{-.5em}
  \caption{Typing environment semantics.}
  \label{fig:aat-mpst-semantics-typing-env}%
  \vspace{-1em}
\end{figure}

\begin{definition}[Typing Environment Reduction]%
  \label{def:aat-mpst-typing-env-reduction}%
 Let $\stTEnvAnnotGenericSymT$ be a transition label of the form %
  $\stEnvQQueueAnnotSmall{\roleP}{\roleQ}{\stLab}$,  %
  $\stEnvQRecvAnnotSmall{\roleP}{\roleQ}{\stLab}$,  %
or $\timeLab$. %
  The \emph{typing environment transition\;$\stEnvQTMoveGenAnnotT$}\;%
  is inductively defined by the rules in
  \Cref{fig:aat-mpst-semantics-typing-env}. %
  We write $\stEnvQ \!\stEnvQTMoveGenAnnotT$ %
 iff
  $\stEnvQ \!\stEnvQTMoveGenAnnotT\! \stEnvQi$ %
  for some $\stEnvQi$.
 We define two \emph{reductions} $\stEnv \!\stEnvMoveTWithSession[\mpS]\! \stEnvi$ (where $\mpS$ is a session) and 
 $\stEnv \!\stEnvMove\! \stEnvi$ as follows:
 \begin{itemize}[left=0pt,topsep=0pt]
 \item $\stEnv \!\stEnvMoveTWithSession[\mpS]\! \stEnvi$ holds iff 
    $\stEnv \!\stEnvQTMoveGenAnnotT\! \stEnvi$
    with $\stEnvAnnotGenericSym \in \setcomp{\stEnvQQueueAnnotSmall{\roleP}{\roleQ}{\stLab},   %
  \stEnvQRecvAnnotSmall{\roleP}{\roleQ}{\stLab}, \timeLab}{\roleP, \roleQ \in \roleSet}$ (where 
  $\roleSet$ is the set of all roles).     
  We write $\stEnv \!\stEnvMoveWithSession[\mpS]$ %
    iff $\stEnv \!\stEnvMoveWithSession[\mpS]\! \stEnvi$ for some $\stEnvi$, %
    and $\stEnvMoveWithSessionStar[\mpS]$ as the reflexive and transitive closure of $\stEnvMoveWithSession[\mpS]$; %
    \item $\stEnv \!\stEnvMove\! \stEnvi$ holds iff %
    $\stEnv \stEnvMoveWithSession[\mpS] \stEnvi$ for some $\mpS$. 
    We write $\stEnvMoveP{\stEnv}$ iff $\stEnv \!\stEnvMove\! \stEnvi$ for some $\stEnvi$, %
  and $\stEnvMoveStar$ %
    as the reflexive and transitive closure of $\stEnvMove$.  
  \end{itemize}
\end{definition}

The label $\stEnvQQueueAnnotSmall{\roleP}{\roleQ}{\stLab}$ indicates that $\roleP$ 
sends the message $\stLab$ to $\roleQ$ on session $\mpS$, 
while $\stEnvQRecvAnnotSmall{\roleP}{\roleQ}{\stLab}$ denotes 
the reception of $\stLab$ from $\roleQ$ by $\roleP$ on $\mpS$. 
Additionally, the label $\timeLab~(\in \rn_{\ge 0})$ represents a time action modelling the passage of time.
  
The ($\highlight{\text{highlighted}}$) main modifications in the reduction rules for typing environments, 
compared to standard rules, concern time. 
Rule \inferrule{\iruleTCtxSend} states that an entry can perform an
output transition by appending a message at the respective queue within the time specified by the output clock constraint. 
Dually, 
rule  \inferrule{\iruleTCtxRcv} allows an entry to execute an input transition, consuming a message from the corresponding queue  within the specified input clock constraint, provided that the payloads are compatible through subtyping.  
Note that in both rules, the associated clock valuation of the reduced 
entry must be updated according to the reset. %

Rules \inferrule{\iruleTCtxCongX} and \inferrule{\iruleTCtxCongCombined} pertain to \emph{untimed} reductions, 
\ie $\stTEnvAnnotGenericSymT \neq \timeLab$,  within a larger environment. 
Rule \inferrule{\iruleTCtxTimeSession} models time passing on an entry of timed session type by incrementing the associated clock valuation, while rule \inferrule{\iruleTCtxTimeQ} specifies that an entry of queue type is not affected with respect to time progression. 
Thus,  
rule \inferrule{\iruleTCtxTimeCombined} captures the corresponding time behaviour for a timed-session/queue type entry. 
Additionally, rule \inferrule{\iruleTCtxTime} ensures that time elapses
uniformly across compatibly composed %
environments. 
Other rules are standard: \inferrule{\iruleTCtxRec} is for recursion, 
and \inferrule{\iruleTCtxStruct} 
 ensures that reductions are closed under %
 congruence.
 
The reduction $\stEnv \stEnvMoveWithSession[\mpS] \stEnvi$ indicates 
that the typing environment $\stEnv$ can advance on session $\mpS$, involving any roles, while 
$\stEnv \stEnvMove \stEnvi$ signifies $\stEnv$ progressing on any session. 
This distinction helps in illustrating properties of typed processes discussed in~\Cref{subsec:pptes_atmp_session_types}.

\subsection{Relating Timed Global Types and Typing Environments}
\label{sec:global_local_relation}
One of our main results is establishing an operational relationship between the semantics of timed global types and typing environments, ensuring the \emph{correctness} of processes typed by environments that reflect timed global types. 
To accomplish this, we begin by assigning LTS semantics to timed global types.

\vspace{-1em}
\subparagraph*{Semantics of Timed Global Types} %
Similar to that of typing environments,
we define the LTS semantics for timed global types $\gtG$
over tuples of the form
 $\gtWithTime{\cVal}{\gtG}$,
where $\cVal$ is a clock valuation. 
Additionally, we specify the subject of an action $\stEnvAnnotGenericSym$ 
as its responsible principal: 
$\ltsSubject{\stEnvQQueueAnnotSmall{\roleP}{\roleQ}{\stLab}}
=
\ltsSubject{\stEnvQRecvAnnotSmall{\roleP}{\roleQ}{\stLab}} =
\setenum{\roleP}$, and 
$ \ltsSubject{\timeLab} = \emptyset$.
      
\begin{figure}[t]
\centerline{\(
\small
\begin{array}{c}
\begin{array}{c}
 \highlight{ \inferenceSingle[\iruleGtMoveTime]{
    \gtWithTime{\cVal}{\gtG}
    \gtMove[\timeLab]
    \gtWithTime{\cVal + \timeLab}{\gtG}
  }}
\end{array}
\qquad
\begin{array}{c}
  \inference[\iruleGtMoveRec]{
    \gtWithTime{\cVal}{\gtG{}\subst{\gtRecVar}{\gtRec{\gtRecVar}{\gtG}}}
    \gtMove[\stEnvAnnotGenericSym]
    \gtWithTime{\cVali}{\gtGi}}
  {
    \gtWithTime{\cVal}{\gtRec{\gtRecVar}{\gtG}}
    \gtMove[\stEnvAnnotGenericSym]
    \gtWithTime{\cVali}{\gtGi}
  }
  \end{array}
  \\[2ex]
  \inference[\iruleGtMoveOut]{
     j \in I
    &
   \highlight{ \cVal \models \ccstO[j]}
    &
   \highlight{\cVali = \cValUpd{\cVal}{\crstO[j]}{0}}
  }{
    \gtWithTime{\cVal}{
      \gtCommTSmall{\roleP}{\roleQ}{i \in I}{\gtLab[i]}{\stS[i]}{\tcFmt{\mathcal{A}_i}}{\gtGi[i]}
    }
    \gtMove[
     \stEnvQTQueueAnnotSmall{\roleP}{\roleQ}{\gtLab[j]}
    ]
    \gtWithTime{\cVali}{
      \gtCommTTransit{\roleP}{\roleQ}{i \in I}{\gtLab[i]}{\stS[i]}{\tcFmt{\mathcal{A}_i}}{\gtGi[i]}{j}
    }
  }
  \\[1ex]
  \inference[\iruleGtMoveIn]{
    j \in I
    &
   \highlight{\cVal \models \ccstI[j]}
    &
   \highlight{ \cVali = \cValUpd{\cVal}{\crstI[j]}{0}}
  }{
    \gtWithTime{\cVal}{
      \gtCommTTransit{\roleP}{\roleQ}{i \in I}{\gtLab[i]}{\stS[i]}{\tcFmt{\mathcal{A}_i}}{\gtGi[i]}{j}
    }
    \gtMove[
     \stEnvQTRecvAnnotSmall{\roleQ}{\roleP}{\gtLab[j]}
    ]
    \gtWithTime{\cVali}{\gtGi[j]}
  }\\[1ex]
  \inference[\iruleGtMoveCtx]{
    \forall i \in I :
    \gtWithTime{\cVal}{\gtGi[i]}
    \gtMove[\stEnvAnnotGenericSym]
    \gtWithTime{\cVali}{\gtGii[i]}
    &
    \roleP, \roleQ \notin \ltsSubject{\stEnvAnnotGenericSym}
    &
    \highlight{\stEnvAnnotGenericSym \neq \timeLab}
  }{
    \gtWithTime{\cVal}{
      \gtCommTSmall{\roleP}{\roleQ}{i \in
      I}{\gtLab[i]}{\stS[i]}{\tcFmt{\mathcal{A}_i}}{\gtGi[i]}
    }
    \gtMove[\stEnvAnnotGenericSym]
    \gtWithTime{\cVali}{
      \gtCommTSmall{\roleP}{\roleQ}{i \in
      I}{\gtLab[i]}{\stS[i]}{\tcFmt{\mathcal{A}_i}}{\gtGii[i]}
    }
  }
  \\[1ex]
  \inference[\iruleGtMoveCtxi]{
    \forall i \in I :
    \gtWithTime{\cVal}{\gtGi[i]}
    \gtMove[\stEnvAnnotGenericSym]
    \gtWithTime{\cVali}{\gtGii[i]}
    &
    \roleQ \notin \ltsSubject{\stEnvAnnotGenericSym} 
    &
     \highlight{\stEnvAnnotGenericSym \neq \timeLab}
  }{
    \gtWithTime{\cVal}{
      \gtCommTTransit{\roleP}{\roleQ}{i \in
      I}{\gtLab[i]}{\stS[i]}{\tcFmt{\mathcal{A}_i}}{\gtGi[i]}{j}
     }
    \gtMove[\stEnvAnnotGenericSym]
    \gtWithTime{\cVali}{
      \gtCommTTransit{\roleP}{\roleQ}{i \in
      I}{\gtLab[i]}{\stS[i]}{\tcFmt{\mathcal{A}_i}}{\gtGii[i]}{j}
    }
  }%

\end{array}
\)}
\vspace{-.5em}
\caption{Timed global type reduction rules, where $\tcFmt{\mathcal{A}_i} = \ccstO[i], \crstO[i], \ccstI[i], \crstI[i]$. }
\label{fig:gtype_rules}
\vspace{-1em}
\end{figure}

\begin{definition}[Timed Global Type Reduction]
\label{def:tgtype-semantics}
 \label{def:tgtype:lts-gt}
  The \emph{timed global type
  transition}
  $\gtMove[\stEnvAnnotGenericSym]$
  is inductively
  defined by the rules in~\Cref{fig:gtype_rules}. 
    We denote 
  $\gtWithTime{\cVal}{\gtG} \gtMove
  \gtWithTime{\cVali}{\gtGi}$
  if there
  exists $\stEnvAnnotGenericSym$ such that
  $\gtWithTime{\cVal}{\gtG}
  \gtMove[\stEnvAnnotGenericSym]
  \gtWithTime{\cVali}{\gtGi}$, 
  $\gtWithTime{\cVal}{\gtG} \gtMove$ if there exists 
  $\gtWithTime{\cVali}{\gtGi}$ such that 
   $\gtWithTime{\cVal}{\gtG}
  \gtMove 
  \gtWithTime{\cVali}{\gtGi}$, and $\gtMoveStar$  as the transitive and reflexive closure of $\gtMove$.
  \end{definition}

In~\Cref{fig:gtype_rules}, the ($\highlight{\text{highlighted}}$) changes from the standard 
global type reduction rules~\cite{ICALP13CFSM} focus on time. 
Rule \inferrule{\iruleGtMoveTime} accounts for the passage of time by incrementing the clock valuation.
Rules \inferrule{\iruleGtMoveOut} and \inferrule{\iruleGtMoveIn}
model the sending and receiving of messages within specified clock constraints, respectively. 
Both rules also require the adjustment of the clock valuation using the reset predicate.
Rule \inferrule{\iruleGtMoveRec} handles recursion. 
Finally, rules \inferrule{\iruleGtMoveCtx} and \inferrule{\iruleGtMoveCtxi} 
allow reductions of (intermediate) global types causally independent of their prefixes. 
Note that the execution of any timed global type transition always starts with an initial clock valuation $\cVal^{0}$, 
\ie all clocks in $\cVal$ are set to $0$. 

\vspace{-1em}
\subparagraph{Associating Timed Global Types with Typing Environments} %
We are now ready to establish a \emph{new} relationship, \emph{association}, between timed global types and typing environments. This association, which is more general than projection~(\Cref{def:main_proj}) by incorporating 
subtyping $\stSub$~(\Cref{def:main_subtyping}), plays a crucial role in formulating the typing rules (\Cref{sec:aat-mpst-typing-system}) and demonstrating the properties of typed processes (\Cref{subsec:pptes_atmp_session_types}).

\begin{definition}[Association]
\label{def:assoc}
  A typing environment $\stEnv$ is \emph{associated with} %
  a timed global type $\gtWithTime{\cVal}{\gtG}$ for 
  a multiparty session $\mpS$, 
  written 
  $\stEnvAssoc{\gtWithTime{\cVal}{\gtG}}{\stEnv}{\mpS}$, 
  iff  
  $\stEnv$ 
  can be split into three %
  (possibly empty) sub-environments
 $\stEnv = \stEnv[\gtG] \stEnvComp \stEnv[\Delta] \stEnvComp \stEnv[\stEnd]$ where: 
  \begin{enumerate}[left=0pt, topsep=0pt,resume, nosep]
    \item %
    $\stEnv[\gtG]$ is associated with $\gtWithTime{\cVal}{\gtG}$ for $\mpS$, provided as: %
         \begin{enumerate*}[label=\emph{(\roman*)},left=0pt, topsep=0pt]
         \item  $\dom{\stEnv[\gtG]}
        =
        \setcomp{\mpChanRole{\mpS}{\roleP}}{\roleP \in \gtRoles{\gtG}}$; 
       
        \item $\forall \mpChanRole{\mpS}{\roleP} \in \dom{\stEnv[\gtG]}:   
        \stEnvApp{\stEnv[\gtG]}{\mpChanRole{\mpS}{\roleP}} = \stCPair{\cVal[\roleP]}{\stT[\roleP]}$;

        \item  $\forall \roleP \in \gtRoles{\gtG}: \gtProj{\gtG}{\roleP} \stSub \stT[\roleP]$; and %
        
        \item $\cVal = \sqcup_{\roleP \in \gtRoles{\gtG}}\cVal[\roleP]$ (recall that  $\sqcup$ is an overriding union). 

         \end{enumerate*}
   \item
        \label{item:assoc:qenv}
         $\stEnv[\Delta]$ is associated with $\gtG$ for $\mpS$, given as follows: 
         \begin{enumerate}[label=\emph{(\roman*)}, left=0pt, topsep=0pt]
         \item $\dom{\stEnv[\Delta]}
        =
        \setcomp{\mpChanRole{\mpS}{\roleP}}{\roleP \in \gtRoles{\gtG}}$; 
        
        \item  $\forall \mpChanRole{\mpS}{\roleP} \in \dom{\stEnv[\Delta]}:   
        \stEnvApp{\stEnv[\Delta]}{\mpChanRole{\mpS}{\roleP}} = \stQType[\roleP]$; 

         \item if $\gtG = \gtEnd$ or $\gtG = \gtRec{\gtRecVar}{\gtGi}$, then 
         $\forall \mpChanRole{\mpS}{\roleP} \in \dom{\stEnv[\Delta]}: \stEnvApp{\stEnv[\Delta]}{\mpChanRole{\mpS}{\roleP}} = \stQEmptyType$;

         \item if $\gtG =  \gtCommT{\roleP}{\roleQ}{i \in
        I}{\gtLab[i]}{\stS[i]}{\ccstO[i], \crstO[i], \ccstI[i], \crstI[i]}{\gtG[i]}$, then
        \begin{enumerate*}[label={\emph{(a\arabic*)}}]
        \item $\roleQ \notin \operatorname{receivers}(\stEnvApp{\stEnv[\Delta]}{\mpChanRole{\mpS}{\roleP}})$, and 
        \item $\forall i \in I$: $\stEnv[\Delta]$ is associated with $\gtG[i]$ for $\mpS$; 
         \end{enumerate*}
         
         \item if $\gtG = \gtCommTTransit{\roleP}{\roleQ}{i \in
          I}{\gtLab[i]}{\stS[i]}{\ccstO[i], \crstO[i], \ccstI[i], \crstI[i]}{\gtG[i]}{j}$, then 
           \begin{enumerate*}[label={\emph{(b\arabic*)}}]
         \item  $\stEnvApp{\stEnv[\Delta]}{\mpChanRole{\mpS}{\roleP}} = 
          \stQCons{\stQMsg{\roleQ}{\stLab[j]}{\stSi[j]}}{\stQType}$ with $\stSi[j] \stSub \stS[j]$, and
          
          \item  $\stEnvUpd{\stEnv[\Delta]}{\mpChanRole{\mpS}{\roleP}}{\stQType}$ is associated with $\gtG[j]$ for $\mpS$. 
          \end{enumerate*}
        \end{enumerate}
  \item
 \label{item:assoc:end} 
 $\forall \mpChanRole{\mpS}{\roleP} \in \dom{\stEnv[\stEnd]}: 
 \stEnvApp{\stEnv[\stEnd]}{\mpChanRole{\mpS}{\roleP}} = \stMPair{\stCPair{\cVal[\roleP]}{\stEnd}}{\stQEmptyType}$.  
       \end{enumerate}
\end{definition}

The association $\stEnvAssoc{\gtFmt{\cdot}}{\stFmt{\cdot}}{\mpFmt{\cdot}}$ is a
binary relation over timed global types $\gtWithTime{\cVal}{\gtG}$ and typing environments $\stEnv$,
parameterised by  multiparty sessions $\mpS$. 
There are three requirements for the association: 
\begin{enumerate*}[label={(\arabic*)}]
\item the typing environment $\stEnv$ must include two entries for each role of the global type $\gtG$ 
in $\mpS$: one of timed session type and another of queue type; 
\item the timed session type entries in $\stEnv$ %
reflect $\gtWithTime{\cVal}{\gtG}$ by ensuring 
that:  
\begin{enumerate*}
\item they align with the projections of $\gtG$ via subtyping, %
and 
\item their clock valuations match $\cVal$; %
\end{enumerate*}
\item the queue type entries in $\stEnv$ correspond to the transmissions en route in $\gtG$.  
\end{enumerate*}
Note that $\stEnv[\stEnd]$ %
is specifically used to associate
typing environments and $\gtEnd$-types $\gtWithTime{\cVal}{\gtEnd}$, as in this case, both $\stEnv[\gtG]$ and  
$\stEnv[\Delta]$ are empty. 

\begin{example}
\label{ex:assoc}
Consider 
 the timed global type $\gtWithTime{\setenum{C_{\roleFmt{Sen}} = 0, C_{\roleFmt{Sat}} = 0, C_{\roleFmt{Ser}} = 0}}{\gtG[\text{data}]}$, where $\gtG[\text{data}]$ is from~\Cref{ex:proj_example}, and a typing environment 
 $\stEnv[\text{data}] = \stEnv[\gtG_{\gtFmt{\text{data}}}] \stEnvComp \stEnv[\Delta_{\gtFmt{\text{data}}}]$, where:  
 
 \smallskip
 \centerline{\(
 \footnotesize
 \begin{array}{l@{\hskip 0mm}l@{\hskip .6mm}l}
 \stEnv[\gtG_{\gtFmt{\text{data}}}] 
 &=&   \stEnvMap{%
    \mpChanRole{\mpS}{\roleFmt{Sen}}%
  }{
    \stCPair{\setenum{C_{\roleFmt{Sen}} = 0}}{\roleFmt{Sat}\stFmt{\oplus}
   \stLabFmt{Data}\stFmt{\{ \mpFmt{6 \leq C_{\roleFmt{Sen}} \leq 7, C_{\roleFmt{Sen}} := 0}\}}
  }
  } 
  \stEnvComp 
  \\[.5mm]
  &
  &
  \stEnvMap{\mpChanRole{\mpS}{\roleFmt{Sat}}}{
  \stCPair{\setenum{C_{\roleFmt{Sat}} = 0}}
  {\stExtSum{\roleFmt{Sen}}{}{%
        \begin{array}{@{\hskip 0mm}l@{\hskip 0mm}}%
          \stTChoice{\stLabFmt{Data}}{} {\mpFmt{6 \leq C_{\roleFmt{Sat}} \leq 7, \emptyset}} \stSeq%
          \stOut{\roleFmt{Ser}}{%
            \stTChoice{\stLabFmt{Data}}{}{\mpFmt{6 \leq C_{\roleFmt{Sat}} \leq 7, C_{\roleFmt{Sat}} := 0}}
          }{}%
         \\
         \stTChoice{\stLabFmt{fail}}{}{\mpFmt{6 \leq C_{\roleFmt{Sat}} \leq 7, \emptyset}} \stSeq%
          \stOut{\roleFmt{Ser}}{%
            \stTChoice{\stLabFmt{fatal}}{}{\mpFmt{6 \leq C_{\roleFmt{Sat}} \leq 7, C_{\roleFmt{Sat}} := 0}}}{}
        \end{array}
      }} } \stEnvComp 
      \\[.5mm]
      &&
        \stEnvMap{\mpChanRole{\mpS}{\roleFmt{Ser}}}{
  \stCPair{\setenum{C_{\roleFmt{Ser}} = 0}}
  {\stExtSum{\roleFmt{Sat}}{}{%
        \begin{array}{@{\hskip 0mm}l@{\hskip 0mm}}%
          \stTChoice{\stLabFmt{Data}}{} {\mpFmt{6 \leq C_{\roleFmt{Ser}} \leq 7, C_{\roleFmt{Ser}} := 0}}     
          \\
         \stTChoice{\stLabFmt{fatal}}{}{\mpFmt{6 \leq C_{\roleFmt{Ser}} \leq 7, C_{\roleFmt{Ser}} := 0}} 
        \end{array}
      }}}
      \\[.5mm]
    \stEnv[\Delta_{\gtFmt{\text{data}}}] &=& 
     \stEnvMap{\mpChanRole{\mpS}{\roleFmt{Sen}}}{\stQEmptyType} \stEnvComp  
     \stEnvMap{\mpChanRole{\mpS}{\roleFmt{Sat}}}{\stQEmptyType} \stEnvComp 
     \stEnvMap{\mpChanRole{\mpS}{\roleFmt{Ser}}}{\stQEmptyType}
      \end{array}
 \)}
 
 \smallskip
 \noindent
 $\stEnv[\text{data}]$ is associated with $\gtWithTime{\setenum{C_{\roleFmt{Sen}} = 0, C_{\roleFmt{Sat}} = 0, C_{\roleFmt{Ser}} = 0}}{\gtG[\text{data}]}$ for $\mpS$, which can be formally verified by ensuring that 
 $\stEnv[\text{data}]$ satisfies all conditions outlined in~\Cref{def:assoc}.  
 Intuitively, $\stEnv[\gtG_{\gtFmt{\text{data}}}]$ is associated with $\gtWithTime{\setenum{C_{\roleFmt{Sen}} = 0, C_{\roleFmt{Sat}} = 0, C_{\roleFmt{Ser}} = 0}}{\gtG[\text{data}]}$ because:  
\begin{enumerate*}[label={(\arabic*)}]
\item  
 $\roleFmt{Sen}$ sends $\stLabFmt{Data}$, while $\roleFmt{Sat}$ and 
 $\roleFmt{Ser}$ expect to receive additional messages $\stLabFmt{fail}$ and $\stLabFmt{fatal}$, respectively.  Consequently, the entries in $\stEnv[\gtG_{\gtFmt{\text{data}}}]$ adhere to the expected communication behaviour of each role in 
of $\gtG[\text{data}]$,  
though with more 
input messages; 
\item %
the time tracking for each entry in $\stEnv[\gtG_{\gtFmt{\text{data}}}]$ 
matches the clock valuation of $\gtG[\text{data}]$.
\end{enumerate*}
Additionally, $\stEnv[\Delta_{\gtFmt{\text{data}}}]$ only contains  empty queues, aligning with that  
$\gtG[\text{data}]$ is not currently in transit.

\end{example}

We establish the operational correspondence between a timed global type and 
its associated typing environment, our main result for timed multiparty session types, through two theorems: \Cref{lem:comp_proj} 
demonstrates that every possible reduction of a typing environment is mirrored by a corresponding action in reductions of the associated timed global type, while~\Cref{lem:sound_proj} indicates that the reducibility of a timed global type is equivalent to its associated  
environment. 
The proofs of~\Cref{lem:comp_proj,lem:sound_proj} are by inductions on reductions of typing environments and timed global types, respectively, with further details in~\Cref{sec:proof:relating}\iftoggle{full}{}{ of the full version}.

\begin{restatable}[Completeness of Association]{theorem}{lemCompProj}
\label{lem:comp_proj}
Given associated timed global type $\gtWithTime{\cVal}{\gtG}$ and typing environment $\stEnvQ$:
$\stEnvAssoc{\gtWithTime{\cVal}{\gtG}}{\stEnv}{\mpS}$.
If 
 $\stEnvQ \stEnvQTMoveGenAnnotT \stEnvQi$, then there exists $ \gtWithTime{\cVali}{\gtGi}$ 
 such that
 $\gtWithTime{\cVal}{\gtG}
  \gtMove[\stEnvAnnotGenericSym]
  \gtWithTime{\cVali}{\gtGi}$ and $\stEnvAssoc{\gtWithTime{\cVali}{\gtGi}}{\stEnvi}{\mpS}$.
\end{restatable}

\vspace{-2ex}
\begin{restatable}[Soundness of Association]{theorem}{lemSoundProj}
\label{lem:sound_proj}
Given associated timed global type $\gtWithTime{\cVal}{\gtG}$ 
and typing environment $\stEnvQ$:
$\stEnvAssoc{\gtWithTime{\cVal}{\gtG}}{\stEnv}{\mpS}$.
If 
 $\gtWithTime{\cVal}{\gtG} \gtMove$, 
 then there exists $\stEnvAnnotGenericSymi$, $\cVali$,  
  $\gtWithTime{\cValii}{\gtGii}$, 
  $\stEnvi$, and $\stEnvii$, such that 
  $\stEnvAssoc{\gtWithTime{\cVali}{\gtG}}{\stEnvi}{\mpS}$, 
  $\gtWithTime{\cVali}{\gtG}
  \gtMove[\stEnvAnnotGenericSymi]
  \gtWithTime{\cValii}{\gtGii}$, 
  $\stEnvQi 
 \stEnvQTMoveAnnot{\stEnvAnnotGenericSymi} \stEnvQii$, and 
  $\stEnvAssoc{\gtWithTime{\cValii}{\gtGii}}{\stEnvii}{\mpS}$.
\end{restatable}

\vspace{-1ex}
\begin{remark}
\label{rem:weak_soundness}
We formulate a soundness theorem that does not mirror the completeness theorem, differing from  
prior work such as~\cite{ICALP13CFSM}. 
This choice stems from our reliance on subtyping~(\Cref{def:main_subtyping}), 
notably \inferrule{\iruleStSubOut}.  In our framework, a timed local type in the typing environment might offer fewer selection branches 
compared to the corresponding projected timed local type. Consequently, certain sending actions with their associated clock valuations may remain uninhabited within the timed global type. Consider, \eg a timed global type: 
 
 \smallskip 
 \centerline{\(
 \footnotesize 
\gtWithTime{\cVal[\text{r}]}{\gtG[\text{r}]} = \gtWithTime{\setenum{C_{\roleP} = 3, C_{\roleQ} = 3}}{\gtCommRaw{\roleP}{\roleQ}{
 \begin{array}{@{\hskip 0mm}l@{\hskip 0mm}}
    \gtCommTChoice{\gtLab[1]}{}{\mpFmt{0 \leq C_{\roleP} \leq 1, \emptyset, 1 \leq C_{\roleQ} \leq 2, \emptyset}}{\gtEnd}
    \\
    \gtCommTChoice{\gtLab[2]}{}{\mpFmt{2 \leq C_{\roleP} \leq 4, \emptyset, 5 \leq C_{\roleQ} \leq 6, \emptyset}}{\gtEnd}
    \end{array}
    }}
    \)}

\smallskip
\noindent
An associated typing environment  $\stEnv[\text{r}]$ may have: 

\centerline{\(
\footnotesize
\stEnvApp{\stEnv[\text{r}]}{\mpChanRole{\mpS}{\roleP}} = 
\stMPair{\stQEmptyType}{\stCPair{\setenum{C_{\roleP} = 3}}{\roleQ\stFmt{\oplus}
   \stLab[1] \stFmt{\{ \mpFmt{0 \leq C_{\roleP} \leq 1, \emptyset}\}} \stSeq \stEnd}} 
   \stSup
 \stMPair{\stQEmptyType}{\stCPair{\setenum{C_{\roleP} = 3}}{\roleQ\stFmt{\oplus}
 \left\{
 \begin{array}{@{\hskip 0mm}l@{\hskip 0mm}}
   \stLab[1] \stFmt{\{ \mpFmt{0 \leq C_{\roleP} \leq 1, \emptyset}\}} \stSeq \stEnd
   \\
   \stLab[2] \stFmt{\{ \mpFmt{2 \leq C_{\roleP} \leq 4, \emptyset}\}} \stSeq \stEnd
   \end{array}
   \right\}
   }}  
   \)}
   
   \smallskip
   \noindent
While the timed global type $\gtWithTime{\cVal[\text{r}]}{\gtG[\text{r}]}$ might transition through 
 $\mpS\stFmt{:}\roleP\stFmt{!}\roleQ\stFmt{:}\gtLab[2]$, 
 the associated environment $\stEnv[\text{r}]$ cannot. 
Nevertheless, our soundness theorem \emph{adequately} guarantees communication safety~(communication matches) via association.
\end{remark}

\subsection{Affine Timed Multiparty Session Typing System}
\label{sec:aat-mpst-typing-system}
\label{SEC:AAT-MPST-TYPING-SYSTEM}
We now present a typing system for \ATMP, which relies on 
 \emph{typing judgments} 
 of the form: 
 
 \smallskip
 \centerline{\(
 \stJudge{\mpEnv}{\stEnvQ}{\mpP} \quad \text{(with $\mpEnv$ omitted when empty)}
 \)}

 \noindent
This judgement indicates that the process $\mpP$ adheres to the usage of its variables and channels 
as specified in $\stEnv$~(\Cref{def:aat-mpst-typing-env-syntax}), guided by the process types in $\mpEnv$~(\Cref{def:aat-mpst-typing-env-syntax}). 
Our typing system is defined inductively by the typing rules shown in~\Cref{fig:aat-mpst-rules-main}, 
with channels annotated for convenience, especially those bound by process definitions and restrictions. 

\begin{figure}[!t]
\centerline{\(
\small
  \begin{array}{c}
    \inference[\iruleMPX]{%
      \mpEnvApp{\mpEnv}{X} = \stFmt{\stCPair{\cVal[1]}{\stT[1]},\ldots,\stCPair{\cVal[n]}{\stT[n]}}%
    }{%
      \mpEnvEntails{\mpEnv}{X}{\stCPair{\cVal[1]}{\stT[1]},\ldots,\stCPair{\cVal[n]}{\stT[n]}}%
    }%
    \qquad%
    \inference[\iruleMPSub]{
      \stCPair{\cVal}{\stT} \stSub \stCPair{\cVali}{\stTi}%
    }{%
      \stEnvEntails{\stEnvMap{\mpC}{ \stCPair{\cVal}{\stT}}}{\mpC}{ \stCPair{\cVali}{\stTi}}%
    }%
    \\[1mm]
     \inference[\iruleMPEnd]{
      \forall i \in 1..n%
      &%
      \stEnvEntails{\stEnvMap{\mpC[i]}{\stCPair{\cVal[i]}{\stT[i]}}}{%
        \mpC[i]%
      }{%
        \stCPair{\cVali[i]}{\stEnd}%
      }%
    }{%
      \stEnvEndP{%
        \stEnvMap{\mpC[1]}{\stCPair{\cVal[1]}{\stT[1]}}%
        \stEnvComp \ldots \stEnvComp%
        \stEnvMap{\mpC[n]}{\stCPair{\cVal[n]}{\stT[n]}}%
      }%
    }%
    \qquad 
    \inference[\iruleMPNil]{%
      \stEnvEndP{\stEnv}%
    }{%
      \stJudge{\mpEnv}{\stEnv}{\mpNil}%
    }%
    \\[1mm]%
    \inference[\iruleMPDef]{
    \begin{array}{c}
           \stJudge{%
          \mpEnv \mpEnvComp%
          \mpEnvMap{\mpX}{\stCPair{\cVal[1]}{\stT[1]}, \ldots, \stCPair{\cVal[n]}{\stT[n]}}%
        }{%
          \stEnvMap{x_1}{\stCPair{\cVal[1]}{\stT[1]}}%
          \stEnvComp \ldots \stEnvComp%
          \stEnvMap{x_n}{\stCPair{\cVal[n]}{\stT[n]}}%
        }{%
          \mpP%
        }%
        \\%
        \stJudge{%
          \mpEnv \mpEnvComp%
          \mpEnvMap{\mpX}{\stCPair{\cVal[1]}{\stT[1]}, \ldots, \stCPair{\cVal[n]}{\stT[n]}}%
        }{%
          \stEnv%
        }{}{%
          \mpQ%
        }%
     \end{array}
    }{%
      \stJudge{\mpEnv}{%
        \stEnv%
      }{%
        \mpDef{\mpX}{%
          \stEnvMap{x_1}{\stCPair{\cVal[1]}{\stT[1]}},%
          \ldots,%
          \stEnvMap{x_n}{\stCPair{\cVal[n]}{\stT[n]}}%
        }{\mpP}{\mpQ}%
      }%
    }%
    \\[1mm]%
    \inference[\iruleMPCall]{%
        \mpEnvEntails{\mpEnv}{X}{%
          \stCPair{\cVal[1]}{\stT[1]},\ldots, \stCPair{\cVal[n]}{\stT[n]}%
        }%
        &%
        \stEnvEndP{\stEnv[0]}%
        &%
        \forall i \in 1..n%
        &%
        \stEnvEntails{\stEnv[i]}{\mpC[i]}{\stCPair{\cVal[i]}{\stT[i]}}%
    }{%
      \stJudge{\mpEnv}{%
        \stEnv[0] \stEnvComp%
        \stEnv[1] \stEnvComp \ldots \stEnvComp \stEnv[n]%
      }{%
        \mpCall{\mpX}{\mpC[1],\ldots,\mpC[n]}%
      }%
    }%
    \\[1mm]%
        \inference[\iruleMPClock]{%
   \highlight{ \forall \cUnit \text{ s.t.} \models \ccstSubt{\ccst}{\cUnit}{C}:\stJudge{\mpEnv}{%
        \stEnvQ }{\delay{\cUnit}{\mpP}}}}
    {%
      \highlight{\stJudge{\mpEnv}{%
        \stEnvQ }{\delay{\ccst}{\mpP}} }
      }%
    \qquad
      \inference[\iruleMPTime]{%
     \highlight{ \stJudge{\mpEnv}{%
        \stEnvQ \,\tcFmt{+}\, \cUnit%
      }{%
        \mpP%
      }}%
      }{\highlight{%
     \stJudge{\mpEnv}{%
        \stEnvQ }{\delay{\cUnit}{\mpP}} }
    }%
    \\[1.5mm]%
    \inference[\iruleMPBranch]{
    \begin{array}{c}
   \highlight{ \forall i \!\in\! I \,\,\,
    \forall \cUnit:
    \cUnit \leq \mathfrak{n} \Longrightarrow \cVal + \cUnit \models \ccst[i] }
    \\[0.5mm]
        \stEnvEntails{\stEnv[1]}{\mpC}{%
         \stCPair{\cVal}{ \stExtSum{\roleQ}{i \in I}{\stTChoice{\stLab[i]}{\stS[i]}{\ccst[i], \crst[i]} \stSeq \stT[i]}}%
        }%
        \quad
      \highlight{\forall i \!\in\! I: \stS[i] = \stCPair{\ccsti[i]}{\stTi[i]}}
      \quad
       \highlight{ \cVali[i] \models \ccsti[i]}
        \\[0.5mm]
       \highlight{ \forall i \!\in\! I \,\,
       \forall \cUnit \leq \mathfrak{n}:
       \stJudge{\mpEnv}{%
          \stEnvQ \,\tcFmt{+}\, \cUnit \stEnvComp%
          \stEnvMap{y_i}{\stCPair{\cVali[i]}{\stTi[i]}} \stEnvComp%
          \stEnvMap{\mpC}{\stCPair{\cValUpd{\cVal + \cUnit}{\crst[i]}{0}}{\stT[i]}}%
        }{%
          \mpP[i]%
        }}
        \end{array}
    }{%
      \highlight{\stJudge{\mpEnv}{%
        \stEnv \stEnvComp \stEnv[1]%
      }{%
        \mpTBranch{\mpC}{\roleQ}{i \in I}{\mpLab[i]}{y_i}{\mpP[i]}{\mathfrak{n}}{}%
      }}%
    }%
    \\[1.5mm]%
    \inference[\iruleMPSel]{%
      \begin{array}{c}
    \highlight{\forall \cUnit:   \cUnit \leq \mathfrak{n} \Longrightarrow \cVal + \cUnit \models \ccst }
      \\[0.5mm]
    \stEnvEntails{\stEnv[1]}{\mpC}{%
       \stCPair{\cVal}{ \stIntSum{\roleQ}{}{\stTChoice{\stLab}{\stS}{\ccst, \crst} \stSeq \stT}}}%
      \quad
    \highlight{\stS = \stDelegate{\ccsti}{\stTi}}
      \quad
\highlight{\stEnvEntails{\stEnv[2]}{\mpD}{\stCPair{\cVali}{\stTi}}}%
     \quad
     \highlight{\cVali \models \ccsti}
    \\[0.5mm]
    \highlight{\forall \cUnit \leq \mathfrak{n}:  \stJudge{\mpEnv}{%
        \stEnvQ \,\tcFmt{+}\, \cUnit \stEnvComp \stEnvMap{\mpC}{\stCPair{\cValUpd{\cVal + \cUnit}{\crst}{0}}{\stT}}%
      }{%
        \mpP%
      }
      }%
      \end{array}
    }{%
      \highlight{\stJudge{\mpEnv}{%
        \stEnv \stEnvComp \stEnv[1] \stEnvComp \stEnv[2]%
      }{%
        \mpTSel{\mpC}{\roleQ}{\mpLab}{\mpD}{\mpP}{\mathfrak{n}}%
      }}%
    }%
    \\[1mm]%
      \inference[\iruleMPPar]{%
      \stJudge{\mpEnv}{%
        \stEnvQ[1]%
      }{%
        \mpP[1]%
      }%
      &%
      \stJudge{\mpEnv}{%
        \stEnvQ[2]%
      }{%
        \mpP[2]%
      }%
    }{%
      \stJudge{\mpEnv}{%
        \stEnvQ[1] \stEnvQComp \stEnvQ[2]%
      }{%
        \mpP[1] \mpPar \mpP[2]%
      }%
    }%
    \qquad
    \inference[\iruleMPKill]{%
     \stEnvEndP{\stEnv}
      &%
     n \geq 0
      }%
    {%
      \stJudge{\mpEnv}{%
        \stEnvQ \stEnvQComp \stEnvMap{%
          \mpChanRole{\mpS}{\roleP[1]}%
        }{%
           \stSpecf[1]
        }
        \stEnvQComp
        \ldots
         \stEnvQComp \stEnvMap{%
          \mpChanRole{\mpS}{\roleP[n]}%
        }{%
        \stSpecf[n]
        }%
      }{%
       \kills{\mpS}%
      }%
    }%
    \\[1mm]
      \inference[\iruleMPTry]{%
      \stJudge{\mpEnv}{%
        \stEnvQ%
      }{%
        \mpP%
      }%
      &%
     \highlight{ \procSubject{\mpP} = \setenum{\mpC}}
      &
      \stJudge{\mpEnv}{%
        \stEnvQ%
      }{%
        \mpQ%
      }%
    }{%
      \stJudge{\mpEnv}{%
        \stEnvQ%
      }{%
         \trycatch{\mpP}{\mpQ}
      }%
    }%
   \\[1mm]
    \inference[\iruleMPCancel]{%
      \stJudge{\mpEnv}{%
        \stEnvQ%
      }{%
        \mpQ%
      }%
      }{%
      \stJudge{\mpEnv}{%
        \stEnvQ%
         \stEnvQComp
          \stEnvMap{%
         \mpChanRole{\mpS}{\roleP}
        }{%
           \stSpecf
        }
      }{%
         \mpCancel{\mpChanRole{\mpS}{\roleP}}{\mpQ}
      }%
    }%
    \qquad
 \inference[\iruleMPFailed]
 {}%
    {\highlight{\stJudge{\mpEnv}{%
        \stEnvQ }{\mpFailedP{\mpP}}}}
\\[1mm]
     \inference[\iruleMPResPropG]{%
    \highlight{ \stEnvAssoc{\gtWithTime{\cVal}{\gtG}}{\stEnvi}{\mpS}}
       &%
      \mpS \not\in \stEnvQ%
      &%
      \stJudge{\mpEnv}{%
        \stEnvQ \stEnvQComp \stEnvi%
      }{%
        \mpP%
      }%
    }{%
      \stJudge{\mpEnv}{%
        \stEnv%
      }{%
        \mpRes{\stEnvMap{\mpS}{\stEnvi}}\mpP%
      }%
    }%
 \\[1mm]%
        \inference[\iruleMPQueueEmpty]
        {\highlight{\stEnvEndP{\stEnvQ}}}{%
     \highlight{ \stJudge{\mpEnv}{\stEnvQ \stEnvQComp %
        \stEnvMap{%
          \mpChanRole{\mpS}{\roleP}%
        }{%
          \stQEmptyType%
        }%
      }{%
        \mpSessionQueueO{\mpS}{\roleP}{\mpQueueEmpty}%
      }}%
    }%
    \qquad 
    \inference[\iruleMPQueue]{%
     \highlight{ \stJudge{\mpEnv}{\stEnvQ}{%
        \mpSessionQueueO{\mpS}{\roleP}{%
          \mpQueue%
        }%
      }}%
      &%
    \highlight{  \stS = \stCPair{\ccst}{\stT}}
      &
     \highlight{ \cVal \models \ccst}
      &
     \highlight{ \stEnvEntails{\stEnvi}{\mpChanRole{\mpSi}{\roleR}}{\stCPair{\cVal}{\stT}}}%
    }{%
     \highlight{ \stJudge{\mpEnv}{%
        \stEnvUpd{\stEnvQ}{\mpChanRole{\mpS}{\roleP}}{%
          \stQCons{\stQMsg{\roleQ}{\stLab}{\stS}}{\stEnvApp{\stEnvQ}{\mpChanRole{\mpS}{\roleP}}}%
        }
        \stEnvQComp%
        \stEnvi%
      }{%
        \mpSessionQueueO{\mpS}{\roleP}{%
          \mpQueueCons{%
            \mpQueueOElem{\roleQ}{\mpLab}{\mpChanRole{\mpSi}{\roleR}}%
          }{%
            \mpQueue%
          }%
        }
      }%
    }}%
\end{array}
\)}
\vspace{-.5em}
  \caption[\ATMP typing rules]{%
   \ATMP typing rules. %
  }
  \label{fig:aat-mpst-rules-main}%
\vspace{-1em}
\end{figure}

The innovations~($\highlight{\text{highlighted}}$)~in~\Cref{fig:aat-mpst-rules-main} primarily 
focus on typing processes with time, timeout failures, message queues, and using association~(\Cref{def:assoc}) to enforce session restrictions.

\vspace{.2ex}
\noindent
\textbf{\emph{Standard}} from~\cite{DBLP:journals/pacmpl/ScalasY19}:  %
Rule \inferrule{\iruleMPX} retrieves process variables.
Rule \inferrule{\iruleMPSub} applies subtyping within a singleton typing environment $\stEnvMap{\mpC}{\stCPair{\cVal}{\stT}}$. 
Rule \inferrule{\iruleMPEnd} introduces a predicate $\stEnvEndP{\cdot}$ for typing environments, 
signifying  the termination of all endpoints. This predicate is used in \inferrule{\iruleMPNil} to type an inactive process $\mpNil$. Rules \inferrule{\iruleMPDef} and
\inferrule{\iruleMPCall} deal with recursive processes
declarations and calls, respectively. 
Rule \inferrule{\iruleMPPar} partitions the typing environment into two, 
each dedicated to typing one sub-process.

\vspace{.2ex}
\noindent
{\bf \emph{Session Restriction}}: Rule \inferrule{\iruleMPResPropG} depends on a typing environment associated with a timed global type in a given session $\mpS$ to validate session restrictions. 

\vspace{.2ex}
\noindent
{\bf \emph{Delay}}: Rule~\inferrule{\iruleMPClock} ensures the typedness of time-consuming delay 
$\delay{\ccst}{\mpP}$ by checking every deterministic delay $\delay{t}{\mpP}$ with $t$ as a possible solution to 
$\ccst$.  Rule \inferrule{\iruleMPTime} types a deterministic delay $\delay{t}{\mpP}$ by adjusting the clock valuations in the environment used to type $\mpP$. Here, $\stEnv +t$ denotes the typing environment obtained from $\stEnv$ by increasing the associated clock valuation in each entry by $t$.

\vspace{.2ex}
\noindent
{\bf \emph{Timed Branching and Selection}}: Rules~\inferrule{\iruleMPBranch} and \inferrule{\iruleMPSel} are for timed branching and selection, respectively. We %
 elaborate on \inferrule{\iruleMPBranch}, as  \inferrule{\iruleMPSel} is its dual. 
The first premise in \inferrule{\iruleMPBranch}  
specifies a time interval $[\cVal, \cVal + \mathfrak{n}]$ within which the message must be received, 
in accordance with each $\ccst[i]$.
The last premise requires that each continuation process be well-typed against the continuation of the type in all possible typing environments
where the time falls between $[\cVal, \cVal + \mathfrak{n}]$.
Here, the clock valuation $\cVal$ is reset based on each $\crst[i]$.
The remaining premises stipulate that the clock valuation $\cVali[i]$ of each delegated receiving session must satisfy $\ccsti[i]$, and that $\mpC$ is typed. 

\vspace{.2ex}
\noindent
{\bf \emph{Try-Catch, Cancellation, and Kill}}: 
Rules \inferrule{\iruleMPTry}, \inferrule{\iruleMPCancel}, and \inferrule{\iruleMPKill} pertain to try-catch, cancellation, and kill processes, respectively, analogous to the corresponding rules in~\cite{lagaillardie2022Affine}.
\inferrule{\iruleMPCancel} is responsible for generating a kill process at its declared session.
\inferrule{\iruleMPKill} types a kill process arising during reductions: it involves broadcasting the cancellation of 
$\mpChanRole{\mpS}{\roleP}$ to all processes that belong to $\mpS$. 
\inferrule{\iruleMPTry} handles a $\mpFmt{\ensuremath{\proclit{try-catch}}}$ process $\trycatch{\mpP}{\mpQ}$ 
by ensuring that 
the $\mpFmt{\ensuremath{\proclit{try}}}$ process $\mpP$ and the $\mpFmt{\ensuremath{\proclit{catch}}}$ process $\mpQ$  
maintain consistent session typing. Additionally, $\mpP$ cannot be a queue  or parallel composition, as  indicated by
$\procSubject{\mpP} = \setenum{\mpC}$.

\vspace{.2ex}
\noindent
{\bf \emph{Timeout Failure}}: Rule \inferrule{\iruleMPFailed} indicates that a process raising timeout failure
can be typed by \emph{any} typing environment. 

\vspace{.2ex}
\noindent
{\bf \emph{Queue}}: 
Rules \inferrule{\iruleMPQueueEmpty} and \inferrule{\iruleMPQueue} concern the typing of queues. \inferrule{\iruleMPQueueEmpty} types an empty queue  
under an ended typing environment, while \inferrule{\iruleMPQueue} types a non-empty queue by inserting a message type into $\stEnvQ$. This insertion may either prepend the message to an existing queue type in $\stEnvQ$ or add a queue-typed entry to $\stEnvQ$ if not  present.

\vspace{-.3ex}

\begin{example}
\label{ex:typing_system}
Take the typing environment  $\stEnv[\text{data}]$ from~\Cref{ex:assoc}, along with the processes 
$\mpQ[\roleFmt{Sen}]$, $\mpQ[\roleFmt{Sat}]$, $\mpQ[\roleFmt{Ser}]$ from~\Cref{ex:calculus_reductions}. 
Verifying the typing of $\mpQ[\roleFmt{Sen}] \mpPar \mpQ[\roleFmt{Sat}] \mpPar \mpQ[\roleFmt{Ser}]$ by $\stEnv[\text{data}]$ is easy. 
 Moreover, 
 since $\stEnv[\text{data}]$ is associated with a timed global type 
 $\gtWithTime{\setenum{C_{\roleFmt{Sen}} = 0, C_{\roleFmt{Sat}} = 0, C_{\roleFmt{Ser}} = 0}}{\gtG[\text{data}]}$ for session $\mpS$~(as demonstrated in~\Cref{ex:assoc}), 
 \ie $\stEnvAssoc{\gtWithTime{\setenum{C_{\roleFmt{Sen}} = 0, C_{\roleFmt{Sat}} = 0, C_{\roleFmt{Ser}} = 0}}{\gtG[\text{data}]}}{\stEnv[\text{data}]}{\mpS}$, 
following \inferrule{\iruleMPResPropG}, 
$\mpQ[\roleFmt{Sen}] \mpPar \mpQ[\roleFmt{Sat}] \mpPar \mpQ[\roleFmt{Ser}]$ is closed under 
$\stEnv[\text{data}]$, \ie $\stJudge{}{}{\mpRes{\stEnvMap{\mpS}{\stEnv[\text{data}]}}{\mpQ[\roleFmt{Sen}] \mpPar \mpQ[\roleFmt{Sat}] \mpPar \mpQ[\roleFmt{Ser}]}}$.  

Note that the typing environment $\stEnvi[\text{data}] = \stEnvi[\gtG_{\gtFmt{\text{data}}}] \stEnvComp \stEnvi[\Delta_{\gtFmt{\text{data}}}]$, where $\stEnvi[\Delta_{\gtFmt{\text{data}}}] = \stEnv[\Delta_{\gtFmt{\text{data}}}]$~(\Cref{ex:assoc}), and $\stEnvi[\gtG_{\gtFmt{\text{data}}}]$ is obtained from  projecting 
$\gtG[\text{data}]$, \ie 
$\stEnvi[\Delta_{\gtFmt{\text{data}}}] = \stEnvMap{\mpChanRole{\mpS}{\roleFmt{Sen}}}{\stCPair{\setenum{C_{\roleFmt{Sen}} = 0}}{\gtProj{\gtG[\text{data}]}{\roleFmt{Sen}}}}
\stEnvComp$ 
\\
$\stEnvMap{\mpChanRole{\mpS}{\roleFmt{Sat}}}{\stCPair{\setenum{C_{\roleFmt{Sat}} = 0}}{\gtProj{\gtG[\text{data}]}{\roleFmt{Sat}}}}
\stEnvComp \stEnvMap{\mpChanRole{\mpS}{\roleFmt{Ser}}}{\stCPair{\setenum{C_{\roleFmt{Sen}} = 0}}{\gtProj{\gtG[\text{data}]}{\roleFmt{Sen}}}}$, 
is capable of typing the process $\mpQ[\roleFmt{Sen}] \mpPar \mpQ[\roleFmt{Sat}] \mpPar \mpQ[\roleFmt{Ser}]$ as well. In fact, if a process is typable by a typing environment associated with a timed global type, then it  can definitely be typed using the environment derived from the projection of the timed global type, highlighting the utility of association. 

\end{example}

\subsection{Typed Process Properties}
\label{subsec:pptes_atmp_session_types}
\label{SUBSEC:PPTES_ATMP_SESSION_TYPES}
We demonstrate that processes typed by the \ATMP typing system exhibit the desirable properties: 
\emph{subject reduction}~(\Cref{lem:sr_global}),
\emph{session fidelity}~(\Cref{lem:aat-mpst-session-fidelity-global}),
and \emph{deadlock-freedom}~(\Cref{lem:aat-mpst-process-df-proj}). 

\vspace{-1em}
\subparagraph*{Subject Reduction}
 ensures the preservation of well-typedness of processes during  reductions. 
Specifically, it states that if a well-typed process $\mpP$ reduces to $\mpPi$, this reduction is reflected in the typing environment $\stEnv$ used to type $\mpP$. Notably, in our subject
reduction theorem, $\mpP$ is constructed from a timed global type, \ie typed by an environment 
associated with a timed global type, and this construction approach persists as an invariant property throughout reductions.
Furthermore, the theorem does not require $\mpP$ to contain only a single session; instead,
it includes all restricted sessions in $\mpP$, ensuring that reductions on these
sessions uphold their respective restrictions. 
This enforcement is facilitated by rule \inferrule{\iruleMPResPropG} in~\Cref{fig:aat-mpst-rules-main}.

Association~(\Cref{def:assoc}), crucial for subject reduction, possesses the following properties that facilitate its proof: supertyping and reductions commute under association (\cref{lem:stenv-assoc-reduction-sub}), and  supertyping preserves association~(\Cref{lem:stenv-assoc-sub}). 
Their proofs are available in~\Cref{sec:proofs:subtyping-properties}\iftoggle{full}{}{ of the full version}.  
 
  \begin{restatable}{lemma}{lemStenvReductionSubAssoc}
    \label{lem:stenv-assoc-reduction-sub}%
    Assume that $\stEnvAssoc{\gtWithTime{\cVal}{\gtG}}{\stEnv}{\mpS}$ and
    $\stEnv \stSub \stEnvi \stEnvMoveGenAnnot \stEnvii$. 
    Then, there is  $\stEnviii$ such that 
    $\stEnv \stEnvMoveGenAnnot \stEnviii \stSub \stEnvii$.%
  \end{restatable}
  
  \vspace{-2ex}
  \begin{restatable}{lemma}{lemAssocSubP}
\label{lem:stenv-assoc-sub}
If $\stEnvAssoc{\gtWithTime{\cVal}{\gtG}}{\stEnv}{\mpS}$ and 
$\stEnv \stSub \stEnvi$, 
then $\stEnvAssoc{\gtWithTime{\cVal}{\gtG}}{\stEnvi}{\mpS}$. 
\end{restatable}

\noindent
Additionally, the typability of time passing, 
fundamental for demonstrating both subject reduction and session fidelity~(as discussed later),  is addressed below.  
Its proof is available in~\Cref{sec:app-aat-mpst-sj-proof}\iftoggle{full}{}{ of the full version}.

\begin{restatable}{lemma}{lemTimePass}
\label{lem:time_passing_typing}
Assume
  $\stJudge{\mpEnv}{\stEnvQ}{\mpP}$.
If 
$\mpP \mpMoveTime \timePass{\cUnit}{\mpP}$
 with 
   $\timePass{\cUnit}{\mpP}$
    defined, then 
$\stJudge{\mpEnv}{\stEnvQ \,\tcFmt{+}\, \cUnit}{\mpFmt{\timePass{\cUnit}{\mpP}}}$.
\end{restatable}

\vspace{-1.5ex}
\begin{restatable}[Subject Reduction]{theorem}{lemSubjectReductionGlobal}
\label{lem:sr_global}
  Assume $\stJudge{\mpEnv}{\stEnv}{\mpP}$ 
  where $\forall \mpS \in \stEnv: \exists \gtWithTime{\cVal}{\gtG}: 
  \stEnvAssoc{\gtWithTime{\cVal}{\gtG}}{\stEnv[\mpS]}{\mpS}$. 
  If $\mpP \!\mpMove\! \mpPi$, then 
  $\exists \stEnvi$
  such that 
  $\stEnv \!\stEnvMoveStar\! \stEnvi$, 
  $\stJudge{\mpEnv}{\stEnvi}{\mpPi}$, and $\forall \mpS \in \stEnvi: 
  \exists \gtWithTime{\cVali}{\gtGi}: 
  \stEnvAssoc{\gtWithTime{\cVali}{\gtGi}}{\stEnvi[\mpS]}{\mpS}$. 
\end{restatable}
\begin{proof}
By induction on the derivation of $\mpP \!\mpMove\! \mpPi$~(see~\Cref{sec:app-aat-mpst-sj-proof}\iftoggle{full}{}{ of the full version}).
\end{proof}

\vspace{-1.3ex}
\begin{restatable}[Type Safety]{corollary}{lemTypeSafety}
  \label{cor:type-safety}
Assume $\stJudge{\emptyset}{\emptyset}{\mpP}$. 
If $\mpP \mpMoveStar \mpPi$, then %
  $\mpPi$ has no communication error.%
\end{restatable}

\begin{example}
\label{ex:subject_reduction}
Take the typed process $\mpQ[\roleFmt{Sen}] \mpPar  \mpQ[\roleFmt{Sat}]  \mpPar \mpQ[\roleFmt{Ser}]$ 
and the typing environment $\stEnv[\text{data}]$ from~\Cref{ex:calculus_reductions,ex:assoc,ex:typing_system}. 
After a reduction using~\inferrule{\iruleMPRedDet},   $\mpQ[\roleFmt{Sen}] \mpPar  \mpQ[\roleFmt{Sat}]  \mpPar \mpQ[\roleFmt{Ser}]$  transitions to $\delay{6.5}{\mpQi[\roleFmt{Sen}]} \mpPar   \mpSessionQueueO{\mpS}{\roleFmt{Sen}}{\mpQueueEmpty} \mpPar \delay{6}{\mpQi[\roleFmt{Sat}]} \mpPar   \mpSessionQueueO{\mpS}{\roleFmt{Sat}}{\mpQueueEmpty} \mpPar \delay{6}{\mpQi[\roleFmt{Ser}]} \mpPar   \mpSessionQueueO{\mpS}{\roleFmt{Ser}}{\mpQueueEmpty} = \mpQ[2]$, which remains typable by 
$\stEnv[\text{data}]$~($\stEnv[\text{data}] \!\stEnvMoveStar\! \stEnv[\text{data}]$). 
Then, applying  \inferrule{\iruleMPRedDelay}, 
$\mpQ[2]$ evolves to 
$\timePass{6.5}{\mpQ[2]}$,  typed as $\stEnv[\text{data}] + 6.5$, derived from 
$\stEnv[\text{data}] \stEnvQTMoveAnnot{6.5} \stEnv[\text{data}] + 6.5$. 
Further reduction through~$\inferrule{\iruleMPRedTryFail}$ leads $\timePass{6.5}{\mpQ[2]}$ to 
$\mpQi[\roleFmt{Sen}] \mpPar \mpSessionQueueO{\mpS}{\roleFmt{Sen}}{\mpQueueEmpty} \mpPar \kills{\mpS} \mpPar \mpSessionQueueO{\mpS}{\roleFmt{Sat}}{\mpQueueEmpty} \mpPar \timePass{0.5}{\mpQi[\roleFmt{Ser}]} \mpPar \mpSessionQueueO{\mpS}{\roleFmt{Ser}}{\mpQueueEmpty} = \mpQ[3]$, typable by 
$\stEnv[\text{data}] + 6.5$. %
Later, via \inferrule{\iruleMPCCat}, $\mpQ[3]$ reduces to 
$\mpCancel{\mpChanRole{\mpS}{\roleFmt{Sen}}}{} \mpPar   
\mpSessionQueueO{\mpS}{\roleFmt{Sen}}{\mpQueueEmpty} \mpPar 
\kills{\mpS} \mpPar   \mpSessionQueueO{\mpS}{\roleFmt{Sat}}{\mpQueueEmpty} \mpPar 
\timePass{0.5}{\mpQi[\roleFmt{Ser}]}   \mpPar   \mpSessionQueueO{\mpS}{\roleFmt{Ser}}{\mpQueueEmpty}$, which can be typed by 
$\stEnvii[\text{data}]$, obtained from $\stEnv[\text{data}] + 6.5 \stEnvQTMoveQueueAnnot{\roleFmt{Sen}}{\roleFmt{Sat}}{\stLabFmt{Data}} \cdot \stEnvQTMoveRecvAnnot{\roleFmt{Sat}}{\roleFmt{Sen}}{\stLabFmt{Data}} \stEnvii[\text{data}]$.  
Due to~\Cref{lem:comp_proj}, $\stEnvii[\text{data}]$ keeps its association with some timed global type. 
\end{example}

\vspace{-1.3em}
\subparagraph*{Session Fidelity} 
states the converse implication of subject reduction: if a process $\mpP$ is typed by $\stEnvQ$ and $\stEnvQ$ can reduce, then $\mpP$ can simulate at least one of the reductions performed by $\stEnvQ$ -- although not necessarily all such reductions, as $\stEnvQ$ over-approximates the behavior of $\mpP$.
Consequently, we can infer
$\mpP$'s behaviour from  that of $\stEnvQ$.  
However, this result does not hold for certain well-typed processes,
such as those that get trapped in recursion loops like $\mpDef{\mpX}{...}{\mpX}{\mpX}$,
or deadlock due to interactions across multiple sessions~\cite{coppo2016Global}.
To address this, similarly to~\cite{DBLP:journals/pacmpl/ScalasY19} and most session type works, we
establish session fidelity specifically for processes featuring guarded recursion
and implementing a single multiparty session as a parallel composition of one
sub-process per role. The formalisation of session fidelity is provided in~\Cref{lem:session-fidelity-strength-global}, 
building upon the concepts introduced in~\Cref{lem:guarded-definitions}.

\begin{definition}[From~\cite{DBLP:journals/pacmpl/ScalasY19}]
  \label{lem:guarded-definitions}%
  \label{def:unique-role-proc}%
  Assume $\stJudge{\mpEnvEmpty}{\stEnv}{\mpP}$. %
  We say that $\mpP$:
  \begin{enumerate}[left=0pt, topsep=1pt]
  \item\label{item:guarded-definitions:stmt}%
    \emph{has guarded definitions} %
    if and only if 
   in each process definition in $\mpP$ of the form %
    $\mpDef{\mpX}{\stEnvMap{x_1}{\stCPair{\cVal[1]}{\stT[1]}},...,\stEnvMap{x_n}{\stCPair{\cVal[n]}{\stT[n]}}}{
      \mpQ}{\mpPi}$, %
    for all $i \in 1...n$, %
   $\stT[i] \!\stNotSub\! \stEnd$
    implies %
    that a call $\mpCall{\mpY}{...,x_i,...}$ %
    can only occur in $\mpQ$ %
    as a subterm of %
    $\mpTBranch{x_i}{\roleQ}{j \in J}{\mpLab[j]}{y_j}{\mpP[j]}{\mathfrak{n}}{}$ %
    or %
    $\mpTSel{x_i}{\roleQ}{\mpLab}{\mpD}{\mpPii}{\mathfrak{n}}$ %
    (\ie  after using $x_i$ for input or output);%
  \item\label{item:unique-role-proc:stmt}
    \emph{only plays role $\roleP$ in $\mpS$, by $\stEnv$\!,} %
    if and only if 
    \begin{enumerate*}[ref=\emph{(\roman*)},label=\emph{(\roman*)}]
    \item%
    $\mpP$ has guarded definitions; %
    \item%
    $\fv{\mpP} =  \emptyset$; %
    \item\label{item:unique-role-proc:stenv}%
      $\stEnv \!=\!%
      \stEnv[0] \stEnvComp \stEnvMap{\mpChanRole{\mpS}{\roleP}}{\stSpecf}$
      with %
      $\stSpecf \stNotSub \stCPair{\cVal}{\stEnd}$ %
      and %
      $\stEnvEndP{\stEnv[0]}$;
    \item\label{item:unique-role-proc:res-end}%
      in all subterms %
      $\mpRes{\stEnvMap{\mpSi}{\stEnvi}}{\mpPi}$ %
      of $\mpP$, %
      we have $\stEnvi \stSub \stEnvMap{\mpChanRole{\mpSi}{\rolePi}}{\stMPair{\stCPair{\cVali}{\stEnd}}{\stQEmptyType}}$
      or $\stEnvi \stSub \stEnvMap{\mpChanRole{\mpSi}{\rolePi}}{\stCPair{\cVali}{\stEnd}}$
      (for some $\rolePi, \cVali$).
    \end{enumerate*}
  \end{enumerate}
  We say ``\emph{$\mpP$ only plays role $\roleP$ in $\mpS$}'' %
  if and only if   %
  $\exists\stEnv:  \stJudge{\mpEnvEmpty}{\stEnv}{\mpP}$,  %
  and item~\ref{item:unique-role-proc:stmt} holds.
\end{definition}

In~\Cref{lem:guarded-definitions}, item~\ref{item:guarded-definitions:stmt} 
describes guarded recursion for processes, while 
item~\ref{item:unique-role-proc:stmt} %
specifies a process limited to playing exactly \emph{one} role within \emph{one} session,  
preventing an ensemble of such processes from 
deadlocking by waiting for each other on multiple sessions.

We proceed to present our session fidelity result, taking kill processes into account. 
We denote $\kills{\mpQ}$ to indicate that $\mpQ$ consists only of a parallel composition of kill processes. 
Similar to subject reduction~(\Cref{lem:sr_global}), our session fidelity relies on a typing environment associated with a timed global type for a specific session $\mpS$ to type the process, ensuring the fulfilment of 
single-session requirements~(\Cref{lem:guarded-definitions}) and maintaining invariance during reductions.
The proof is by induction on the derivation of $\stEnv \!\stEnvMoveWithSession[\mpS]$~(available in~\Cref{sec:app-aat-mpst-sf-proof}\iftoggle{full}{}{ of the full version}).

\begin{restatable}[Session Fidelity]{theorem}{lemSessionFidelityStrengthGlobal}%
  \label{lem:session-fidelity-strength-global}%
  \label{lem:session-fidelity-global}
   \label{lem:aat-mpst-session-fidelity-global}
  Assume 
 $\stJudge{\mpEnvEmpty}{\stEnvQ}{\mpP}$,
  with  $\stEnvAssoc{\gtWithTime{\cVal}{\gtG}}{\stEnvQ}{\mpS}$,
  $\mpP \equiv%
  \mpBigPar{\roleP \in I}{%
    \mpP[\roleP]%
  } \mpPar \kills{\mpQ}$,
  and 
  $\stEnvQ = \bigcup_{\roleP \in I}\stEnvQ[\roleP] \cup \stEnvQ[0]$, 
  such that 
  $\stJudge{\mpEnvEmpty}{\stEnvQ[0]}{\kills{\mpQ}}$,
  and
  for each $\mpP[\roleP]$: 
  \begin{enumerate*}[label=\emph{(\arabic*)}]
 \item 
   $\stJudge{\mpEnvEmpty}{\stEnvQ[\roleP]}{\mpP[\roleP]}$, and 
  \item 
  either $\mpP[\roleP] \equiv \mpNil$, %
  or $\mpP[\roleP]$ only plays role $\roleP$ in $\mpS$,
  by 
   $\stEnvQ[\roleP]$. %
   \end{enumerate*}
  Then, $\stEnvQ \!\stEnvMoveWithSession[\mpS]$ implies 
  $\exists \stEnvQi, \gtWithTime{\cVali}{\gtGi}, \mpPi$ %
  such that 
  $\stEnvQ \!\stEnvMoveWithSession[\mpS]\! \stEnvQi$, 
   $\mpP \!\mpMoveStar\! \mpPi$,  and 
  $\stJudge{\mpEnvEmpty}{\stEnvQi}{\mpPi}$, %
  with %
  $\stEnvAssoc{\gtWithTime{\cVali}{\gtGi}}{\stEnvQi}{\mpS}$,
  $\mpPi \equiv
  \mpBigPar{\roleP \in I}{%
    \mpPi[\roleP]%
  } \mpPar \kills{\mpQi}$,
  and $\stEnvQi = \bigcup_{\roleP \in I}\stEnvQi[\roleP] \cup \stEnvQi[0]$
  such that %
  $\stJudge{\mpEnvEmpty}{\stEnvQi[0]}{\kills{\mpQi}}$, and  
  for each 
  $\mpPi[\roleP]$:
  \begin{enumerate*}[label=\emph{(\arabic*)}] 
 \item 
  $\stJudge{\mpEnvEmpty}{\stEnvQi[\roleP]}{\mpPi[\roleP]}$,
  and %
 \item 
  either $\mpPi[\roleP] \equiv \mpNil$, %
  or $\mpPi[\roleP]$ only plays role $\roleP$ in $\mpS$,
  by 
  $\stEnvQi[\roleP]$.
  \end{enumerate*}
\end{restatable}

\vspace{-1.2ex}
\begin{example}
\label{ex:session_fidelity}
Consider the processes $\mpQ[\roleFmt{Sen}]$, $\mpQ[\roleFmt{Sat}]$, $\mpQ[\roleFmt{Ser}]$ from~\Cref{ex:calculus_reductions}, the process $\mpQ[3]$ from~\Cref{ex:subject_reduction}, 
and the typing environment $\stEnv[\text{data}]$ from~\Cref{ex:assoc}. 
 $\mpQ[\roleFmt{Sen}]$, $\mpQ[\roleFmt{Sat}]$, and $\mpQ[\roleFmt{Ser}]$ only play 
 roles $\roleFmt{Sen}$, $\roleFmt{Sat}$, and $\roleFmt{Ser}$, respectively, in $\mpS$, which can be easily verified. 
 
 As demonstrated in~\Cref{ex:subject_reduction}, $\mpQ[3]$ is typed by $\stEnv[\text{data}] + 6.5$, satisfying all prerequisites specified in~\Cref{lem:aat-mpst-session-fidelity-global}. 
Consequently, given $\stEnv[\text{data}] + 6.5 \stEnvQTMoveQueueAnnot{\roleFmt{Sen}}{\roleFmt{Sat}}{\stLabFmt{Data}} \stEnvi[\text{data}]$,  there exists $\mpQ[4] = \mpNil \mpPar \mpSessionQueueO{\mpS}{\roleFmt{Sen}}{\mpQueueCons{ \mpQueueEmpty}{\mpQueueCons{%
    \mpQueueOElem{\roleFmt{Sat}}{\mpLabFmt{Data}}{}
  }{%
    \mpQueueEmpty%
  }}} \mpPar \kills{\mpS} \mpPar \mpSessionQueueO{\mpS}{\roleFmt{Sat}}{\mpQueueEmpty} \mpPar \timePass{0.5}{\mpQi[\roleFmt{Ser}]} \mpPar \mpSessionQueueO{\mpS}{\roleFmt{Ser}}{\mpQueueEmpty} $, resulting from  $\mpQ[3] \mpMove \mpQ[4]$ via \inferrule{\iruleMPRedOut}, such that 
  $\stEnvi[\text{data}]$ can type
  $\mpQ[4]$, with $\stEnvi[\text{data}]$ and $\mpQ[4]$ fulfilling  the single session requirements of session fidelity.  
\end{example}

\vspace{-1.3em}
\subparagraph*{Deadlock-Freedom} ensures that a process can always either progress via reduction or  terminate properly. In our system, where time can be infinitely reduced and session killings may occur during reductions,  
deadlock-freedom implies that if a process cannot undergo any further instantaneous (communication) reductions,  and if any subsequent  time reduction leaves it unchanged, then it contains only inactive or 
kill  sub-processes.
This desirable runtime property is guaranteed by processes constructed from timed global types. 
We formalise the property in~\Cref{def:aat-mpst-process-properties}, and conclude, in~\Cref{lem:aat-mpst-process-df-proj}, that a typed ensemble of processes interacting on a single session, restricted by 
a typing environment  $\stEnv$ associated with a timed global type $\gtWithTime{\cVal^{0}}{\gtG}$, 
is deadlock-free. The proof is available in~\Cref{sec:app-aat-mpst-sf-proof}\iftoggle{full}{}{ of the full version}.

\begin{definition}[Deadlock-Free Process]%
  \label{def:aat-mpst-process-properties}%
  \label{def:async-process-df}%
  \label{def:async-process-liveness}%
$\mpP$ is \emph{deadlock-free} if and only if   $\mpP \!\mpMoveStar\! \mpnonTNotMoveP{\mpPi}$  and $\forall \cUnit \geq 0: \timePass{\cUnit}{\mpPi} = \mpPi$~(recall that $\timePass{\cUnit}{\cdot}$ is a time-passing function defined in~\Cref{fig:aat-mpst-time-passing})~implies
$\mpPi \equiv \mpNil \mpPar \Pi_{i \in I} \kills{\mpS[i]}$. 
\end{definition}

\vspace{-.7ex}
\begin{restatable}[Deadlock-Freedom]{theorem}{corATMPProcessDfProj}
  \label{lem:aat-mpst-process-df-proj}%
Assume 
  $\stJudge{\mpEnvEmpty\!}{\!\stEnvEmpty\!}{\!\mpP}$,
  where
$\mpP \equiv %
 \mpRes{\stEnvMap{\mpS}{\stEnvQ}}{\mpBigPar{\roleP \in \gtRoles{\gtG}}{%
    \mpP[\roleP]%
  }}$,  $\stEnvAssoc{\gtWithTime{\cVal^{0}}{\gtG}}{\stEnvQ}{\mpS}$,
 and each $\mpP[\roleP]$ is either $\mpNil$~(up to $\equiv$), 
 or only plays $\roleP$ in $\mpS$.  
 Then,  $\mpP$ is deadlock-free.
  \end{restatable}

\begin{example}
\label{ex:df}
Considering the processes $\mpQ[\roleFmt{Sen}]$, $\mpQ[\roleFmt{Sat}]$, and $\mpQ[\roleFmt{Ser}]$ from~\Cref{ex:calculus_reductions}, along with the typing environment $\stEnv[\text{data}]$ and timed global type $\gtWithTime{\cVal^{0}}{\gtG[\text{data}]}$ from~\Cref{ex:assoc}, we observe in~\Cref{ex:assoc,ex:session_fidelity,ex:typing_system} that $\stEnvAssoc{\gtWithTime{\cVal^{0}}{\gtG[\text{data}]}}{\stEnv[\text{data}]}{\mpS}$, $\mpQ[\roleFmt{Sen}]$, $\mpQ[\roleFmt{Sat}]$, and $\mpQ[\roleFmt{Ser}]$ only play 
$\roleFmt{Sen}$, $\roleFmt{Sat}$, and $\roleFmt{Ser}$ in $\mpS$, respectively, with $\stJudge{\mpEnvEmpty}{\stEnvEmpty}{\mpRes{\stEnvMap{\mpS}{\stEnvQ[\text{data}]}}{\mpQ[\roleFmt{Sen}] \mpPar \mpQ[\roleFmt{Sat}] \mpPar \mpQ[\roleFmt{Ser}]
}}$. Following the reductions described in~\Cref{ex:calculus_reductions}, we obtain $\mpRes{\stEnvMap{\mpS}{\stEnvQ[\text{data}]}}{\mpQ[\roleFmt{Sen}] \mpPar \mpQ[\roleFmt{Sat}] \mpPar \mpQ[\roleFmt{Ser}]}
\mpMoveStar \mpnonTNotMoveP{\mpPi}$, where $\mpPi = \mpRes{\stEnvMap{\mpS}{\stEnvQ[\text{data}]}}{\kills{\mpS} \mpPar \mpNil \mpPar
\mpSessionQueueO{\mpS}{\roleFmt{Sen}}{\mpQueueEmpty} \mpPar
\kills{\mpS} \mpPar \mpSessionQueueO{\mpS}{\roleFmt{Sat}}{\mpQueueEmpty}
\mpPar
\mpNil \mpPar \mpSessionQueueO{\mpS}{\roleFmt{Ser}}{\mpQueueEmpty}} \equiv \mpNil \mpPar \mpS$. Moreover, $\mpPi$ cannot undergo any further \emph{meaningful} time reduction, \ie $\timePass{\cUnit}{\mpNil \mpPar \kills{\mpS}} =  \mpNil \mpPar \kills{\mpS}$ for any $\cUnit$, indicating the deadlock-freedom of 
$\mpRes{\stEnvMap{\mpS}{\stEnvQ[\text{data}]}}{\mpQ[\roleFmt{Sen}] \mpPar \mpQ[\roleFmt{Sat}] \mpPar \mpQ[\roleFmt{Ser}]}$.
\end{example}

\section{Design and Implementation of \timedmulticrusty}
\label{sec:implementation:implementation}
\label{SEC:IMPLEMENTATION:IMPLEMENTATION}
In this section,
we present our toolchain,~\timedmulticrusty, an implementation of~\ATMP in~\Rust.
\timedmulticrusty is designed with two primary goals:
correctly cascading failure notifications,
and effectively representing and enforcing time constraints. 
To achieve the first goal,  
we use~\Rust's native \CODE{?}-operator along with 
optional types, inspired by~\cite{lagaillardie2022Affine}.
For the second objective, we begin by 
discussing the key challenges encountered during implementation. 

\smallskip
\noindent
\emph{Challenge 1:  Representation of Time Constraints.} %
To handle asynchronous timed communications using~\ATMP, 
we define a time window~($\ccst$ in \ATMP)
and a corresponding behaviour for each operation. 
Addressing this constraint involves two subtasks: 
creating and using clocks in \Rust, and representing
all clock constraints as shown in~\Cref{sec:session-calculus}.
 \Rust allows the creation of virtual clocks 
that rely on the CPU's clock and provide nanosecond-level accuracy.
Additionally, it is crucial to ensure that different behaviours can involve
blocking or non-blocking communications,
pre- or post-specific time tags,
or falling within specified time bounds.

\smallskip
\noindent
\emph{Challenge 2: Enforcement of Time Constraints.} %
To effectively enforce time windows,
implementing reliable and accurate clocks and using them correctly
is imperative.
This requires addressing all cases related to time
constraints properly: %
clocks may be considered unreliable if they skip ticks,
do not strictly increase,
or if the API for clock comparison does not yield
results quickly enough.
Enforcing time constraints in \timedmulticrusty
involves using two libraries: 
the \crosschan \Rust library~\cite{web:rust:crossbeamChannel} 
for \emph{asynchronous} message sending and receiving, 
and the \Rust standard library \timeRust~\cite{standardTimeLibraryWebsite} for
handling and comparing virtual clocks.

 \subsection{Time Bounds in \timedmulticrusty}
\label{subsec:implementing_time_bounds}
\label{subsec:implementation:enforcing_time_constraints}

\subparagraph{Implementing Time Bounds}
To demonstrate the integration of time bounds in \timedmulticrusty,
we consider the final interaction between $\roleFmt{Sen}$ and $\roleFmt{Sat}$ in~\Cref{fig:implementation:remote_data},
specifically from $\roleFmt{Sat}$'s perspective:
 $\roleFmt{Sat}$ sends a \CODE{Close} message between time units
$5$ and $6$ (both inclusive),
following clock $C_{\roleFmt{Sat2}}$,
which is not reset afterward. %

In \timedmulticrusty,
we define the \CODE{Send} type for message transmission,
incorporating various parameters to specify requirements as
\CODE{Send<[parameter1],[parameter2],...>}. 
Assuming the (payload) type~\CODE{Close} is defined,
sending it using the \CODE{Send} type involves initiating with
\CODE{Send<Close,...>}.
If we denote $C_{\roleFmt{Sat2}}$ as \CODE{'b'},
the clock \CODE{'b'} is employed for time constraints,
expressed as
\CODE{Send<Close,'b',...>}.
Time bounds parameters in the \CODE{Send} type follow the declaration of the clock.
In this case, %
both bounds are integers within the time window,
resulting in the \CODE{Send} type being parameterised as
\CODE{Send<Close,'b',0,true,1,false,...>}.
Notably,
bounds are integers due to the limitations of \Rust's generic type parameters.

To ensure that the clock \CODE{'b'} is not reset after triggering the \CODE{send} operation,
we represent this with a whitespace char value in the \CODE{Send} type:
\CODE{Send<Close,'b',0,true,1,false,' ',...>}.
The last parameter,
known as the \emph{continuation},
specifies the operation following the sending of the integer.
In this case,
closing the connection is achieved with an \CODE{End} type.
The complete sending operation is denoted as
\CODE{Send<Close, 'b', 0, true, 1, false, ' ', End>}.

Similarly,
the \CODE{Recv} type is instantiated as
\CODE{Recv<Close,'b',0,true,1,false,' ',End>}.
The inherent mirroring of \CODE{Send} and \CODE{Recv} reflects %
their dual compatibility.
\Cref{subfig:overview_types:send,subfig:overview_types:recv}
provide an analysis of the functioning of
\CODE{Send} and \CODE{Recv},
detailing their parameters and features. %

Generic type parameters preceded by \CODE{const}
within \CODE{Send} and \CODE{Recv} types also serve as values,
representing general type categories supported by \Rust.
This type-value duality facilitates easy verification during compilation,
ensuring compatibility between communicating parties.

\vspace{-1em}
\subparagraph{Enforcing Time Bounds}
It is crucial to rely on dependable clocks and APIs to enforce time constraints. 
\Rust's standard library provides the time module~\cite{standardTimeLibraryWebsite},
enabling developers to manage clocks and measure durations between events.
This library,
utilising OS API,
offers two clock types:
\CODE{Instant} (monotonic but non-steady) and
\CODE{SystemTime} (steady but non-monotonic). 
In \timedmulticrusty, the \CODE{Instant} type serves for 
both correctly prioritising event order and implementing virtual clocks. 
Virtual clocks are maintained through a dictionary (\CODE{HashMap} in \Rust).
Initialising an empty \CODE{HashMap} is achieved with:
\CODE{let clocks = HashMap::<char,Instant>::new()}, while
introducing a new clock is done using:
\CODE{clocks.insert('a',Instant::now())},
where \CODE{'a'} is the clock's name.

\Cref{tab:implementation:primitives} details the available primitives
provided by \timedmulticrusty for sending and receiving payloads,
implementing branching,
or closing a connection.
Except for the \CODE{close} method,
all primitives require the specific
\CODE{HashMap} of clocks to enforce time constraints.

\vspace{-.5em}
\begin{table}[htbp]
    \centering
    \caption{
        Primitives available in \timedmulticrusty.
        Let \CODE{s} be an affine meshed channel;
        \CODE{p},
        a payload of a given type;
        \CODE{clocks},
        a \CODE{HashMap} of all clocks
        linked to the specific process;
        $I$,
        a subset of all roles in the protocol
        but the current role;
        and $K$,
        a subset of all branches.
    }
    \label{tab:implementation:primitives}
    \vspace{-.3em}
    \begin{tabular}{p{0.39\textwidth}p{0.51\textwidth}}
        \hline
        \CODE{let s = s.send(p, clocks)?;}   &
       {\footnotesize  If allowed by the time constraint compared to the given clock in \CODE{clocks},
        sends a payload \CODE{p} on a channel \CODE{s}
        and assigns the continuation
        of the session (a new meshed channel)
        to a new variable \CODE{s}.}
        \\
        \hline
        \CODE{let (p, s) = s.recv(clocks)?;} &
        {\footnotesize If allowed by the time constraint compared to the given clock in \CODE{clocks},
        receives a payload \CODE{p}
        on channel \CODE{s} and assigns the continuation
        of the session to a new variable \CODE{s}.}
        \\
        \hline
        \CODE{s.close()}                     &
       {\footnotesize  Closes the channel \CODE{s}
        and returns a unit type.}
        \\
        \hline
        \makecell[l]{
        \CODE{offer!\(s, clocks, \{}
        \\
        \CODE{enum}$_\texttt{i}::$ \CODE{variant}$_\texttt{k}(\texttt{e}) => \{ ... \}_{k \in K}$ \CODE{\}\)}
        }                                       &
         {\vspace{-1.6em} {\footnotesize If allowed by the time constraint compared to the given clock in \CODE{clocks},
        role \norole{i} receives
        a choice as a message label on channel \CODE{s},
        and, depending on
        the label value which should
        match one of the variants
        \CODE{variant}$_k$ of \CODE{enum}$_\texttt{i}::$,
        runs the related block of code.}}
        \\
        \hline
        \makecell[l]{
        \CODE{choose_X!\(s, clocks, \{}
        \\
        \CODE{enum}$_\texttt{i}::$ \CODE{variant}$_\texttt{k}(\texttt{e}) \}_{i \in I}$ \CODE{\}\)}
        }                                       &
       {\vspace{-1.6em} {\footnotesize  For role ${\roleFmt{X}}$, if allowed by the time constraint compared to the given clock in \CODE{clocks},
        sends the chosen label,
        corresponding to \CODE{variant}$_k$
        to all other roles.}}
        \\
        \hline
    \end{tabular}
    \vspace{-.58em}
\end{table}

\vspace{-1em}
\subparagraph{Verifying Time Bounds}
Our \CODE{send} and \CODE{recv} primitives
use a series of conditions to ensure the integrity of a time window.
The verification process adopts a \emph{divide-and-conquer} strategy,
validating the left-hand side time constraint for each clock before assessing the right-hand side constraint.
The corresponding operation, whether sending or receiving a payload, is executed only after satisfying these conditions.
This approach guarantees the effective enforcement of time constraints without requiring complex solver mechanisms.

 \subsection{Remote Data Implementation}
\label{subsec:implementation:implenting_running_example}

\begin{figure}[t!]
\begin{minipage}[t]{0.49\textwidth}
\begin{rustlisting}
type EndpointSerData = MeshedChannels< *@ \label{line:short:name} @*
 Send<GetData, 'a', 5, true,5, true, ' ',*@ \label{line:short:Server:type:satellite} @*
  Recv<Data, 'a', 6, true, 7, true, 'a', End>>,
  End,*@ \label{line:short:Server:type:sensor} @*
 RoleSat<RoleSat<RoleBroadcast>>,*@ \label{line:short:Server:type:stack} @*
 NameSer>;*@ \label{line:short:Server:type:name} @*
\end{rustlisting}
\vspace{-4mm}
\vspace{-3mm}
\end{minipage}\hfill
\begin{minipage}[t]{0.49\textwidth}
{\lstset{firstnumber=7}\begin{rustlisting}
fn endpoint_data_ser(*@ \label{line:short:Server:function:endpoint_ser} @*
 s: EndpointSerData,
 clocks: &mut HashMap<char, Instant>,
) -> Result<(), Error> { [...]*@ \label{line:short:Server:function:endpoint_ser:start} @*
 let s = s.send(GetData {}, clocks)?;*@ \label{line:short:Server:function:recurs_ser:send:GetData} @*
 let (_data, s) = s.recv(clocks)?;[...]}*@ \label{line:short:Server:function:recurs_ser:recv:Data} @*
\end{rustlisting}}
\end{minipage}
\vspace{-.5em}
\caption{Types (left) and primitives (right) for ${\roleFmt{Ser}}$.}
\label{fig:implementation:types_primitives_server}
\vspace{-1em}
\end{figure}

\subparagraph{Implementation of Server} 
\Cref{fig:implementation:types_primitives_server} explores our \timedmulticrusty implementation of
${\roleFmt{Ser}}$ in the remote data protocol~(\Cref{fig:implementation:remote_data}).
Specifically, the left side of 
\Cref{fig:implementation:types_primitives_server} 
delves into  the \CODE{MeshedChannels} type,
representing the behaviour
of ${\roleFmt{Ser}}$ in the first branch and encapsulating various elements.
In~\timedmulticrusty, the \CODE{MeshedChannels} type incorporates $n + 1$ parameters,
where $n$ is the count of roles in the protocol.
These parameters include the role's name,
$n - 1$ binary channels for interacting with other roles,
and a stack dictating the sequence of binary channel usage.
All types relevant to ${\roleFmt{Ser}}$ are depicted in~\Cref{fig:implementation:types_primitives_server} (left).

The alias \CODE{EndpointSerData}, as indicated in~\Cref{line:short:name}, represents the \CODE{MeshedChannels} type. Binary types, defined in \Crefrange{line:short:Server:type:satellite}{line:short:Server:type:sensor}, facilitate communication between ${\roleFmt{Ser}}$, ${\roleFmt{Sat}}$, and ${\roleFmt{Sen}}$. When initiating communication with ${\roleFmt{Sat}}$, ${\roleFmt{Ser}}$ sends a \CODE{GetData} message in \Cref{line:short:Server:type:satellite}, receives a \CODE{Data} response, and ends  communication on this binary channel. These operations use the clock \CODE{'a'} and adhere to  time windows  between $5$ and $6$ seconds for the first operation and between $6$ and $7$ seconds for the second. Clock \CODE{'a'} is reset only within the second operation. The order of operations is outlined in \Cref{line:short:Server:type:stack}, where ${\roleFmt{Ser}}$ interacts twice with ${\roleFmt{Sat}}$ using \CODE{RoleSat} before initiating a choice with \CODE{RoleBroadcast}. \Cref{line:short:Server:type:name} designates ${\roleFmt{Ser}}$ as the owner of the \CODE{MeshedChannels} type. The behaviour of all roles in each branch can be specified similarly.

The right side of~\Cref{fig:implementation:types_primitives_server} illustrates the usage of \CODE{EndpointSerData} as an input type in the \Rust function \CODE{endpoint_data_ser}. The function's output type, \CODE{Result<(), Error>}, indicates the utilization of affinity in \Rust. In \Cref{line:short:Server:function:recurs_ser:send:GetData}, variable \CODE{'s'}, of type \CODE{EndpointSerData}, attempts to send a contentless message %
\CODE{GetData}. The \CODE{send} function can return either a value resembling \CODE{EndpointSerData} or an \CODE{Error}. If the clock's time does not adhere to the time constraint displayed in \Cref{line:short:Server:type:satellite} with respect to the clock \CODE{'a'} from the set of clocks \CODE{clocks}, an \CODE{Error} is raised. Similarly, in \Cref{line:short:Server:function:recurs_ser:recv:Data}, ${\roleFmt{Ser}}$ attempts to receive a message using the same set of clocks. Both \CODE{send} and \CODE{recv} functions verify compliance with time constraints by comparing the relevant clock provided in the type for the time window and resetting the clock if necessary.

\vspace{-1em}
\subparagraph{Error Handling} 

The error handling capabilities of \timedmulticrusty cover various potential errors that may arise during protocol implementation and execution. These errors include the misuse of generated types and timeouts, 
showcasing the flexibility of our implementation in verifying communication protocols. For instance, if Lines \ref{line:short:Server:function:recurs_ser:send:GetData} and \ref{line:short:Server:function:recurs_ser:recv:Data} in~\Cref{fig:implementation:types_primitives_server} are swapped, the program will fail to compile because it expects a \CODE{send} primitive in \Cref{line:short:Server:function:recurs_ser:send:GetData}, as indicated by the type of \CODE{'s'}. Another compile-time error occurs when a payload with the wrong type is sent. For example, attempting to send a \CODE{Data} message instead of a \CODE{GetData} in \Cref{line:short:Server:function:recurs_ser:send:GetData} will result in a compilation error. \timedmulticrusty can also identify errors at runtime. If the content of the function \CODE{endpoint_data_ser}, spanning in \Crefrange{line:short:Server:function:endpoint_ser:start}{line:short:Server:function:recurs_ser:recv:Data}, is replaced with a single \CODE{Ok(())}, the code will compile successfully. However, during runtime, the other roles will encounter failures as they consider ${\roleFmt{Ser}}$ to have failed.

Timeouts are handled dynamically within~\timedmulticrusty.
If a time-consuming task with a 10-second delay is introduced
between Lines \ref{line:short:Server:function:recurs_ser:send:GetData}
and \ref{line:short:Server:function:recurs_ser:recv:Data}, 
${\roleFmt{Ser}}$ will enter a sleep state for the same duration. Consequently, the \CODE{recv} operation
in \Cref{line:short:Server:function:recurs_ser:recv:Data} will encounter a time constraint violation,
resulting in the failure and termination of ${\roleFmt{Ser}}$.
Furthermore, the absence of clock \CODE{'a'} in the set of clocks, where it is required for a specific primitive,
will trigger a runtime error, as the evaluation of time constraints depends on %
the availability of the necessary clocks.

\begin{figure}[t!]
\centering
  \begin{SCRIBBLELISTING}
global protocol remote\_data(role Sen, role Sat, role Ser){ *@ \label{line:implementation:global_protocol:start} @*
  rec Loop { *@ \label{line:implementation:global_protocol:line:rec} @*
    choice at Ser { *@ \label{line:implementation:global_protocol:line:choice} @*
      GetData() from Ser to Sat within [5;6] using a and resetting (); *@ \label{line:implementation:global_protocol:line:GetData1} @*
      GetData() from Sat to Sen within [5;6] using b and resetting (); *@ \label{line:implementation:global_protocol:line:GetData2} @*
      Data() from Sen to Sat within [6;7] using b and resetting (b); *@ \label{line:implementation:global_protocol:line:Data1} @*
      Data() from Sat to Ser within [6;7] using a and resetting (a); *@ \label{line:implementation:global_protocol:line:Data2} @*
      continue Loop *@ \label{line:implementation:global_protocol:line:loop} @*
    } or {
      Close() from Ser to Sat within [5;6] using a and resetting (); *@ \label{line:implementation:global_protocol:line:Stop1} @*
      Close() from Sat to Sen within [5;6] using b and resetting (); *@ \label{line:implementation:global_protocol:line:Stop2} @* } } }
\end{SCRIBBLELISTING}
\vspace{-1em}
  \caption{Remote data protocol in \timednuscr.}
  \label{fig:implementation:global_protocol}
  \vspace{-1em}
\end{figure}

\vspace{-1em}
\subparagraph{Timed Protocol Specification}
To specify timed multiparty protocols, we extend \nuscr~\cite{zhouCFSM2021}, a multiparty protocol description language, with time features, resulting in \timednuscr. Additional keywords such as \lstscribble{within}, \lstscribble{using}, and \lstscribble{and resetting} are incorporated in \nuscr to support the specification of time windows, clocks, and resets, respectively. In \Cref{fig:implementation:global_protocol}, we illustrate the \timednuscr protocol for remote data, showcasing the application of these enhancements.  
\timednuscr ensures the accuracy of timed multiparty protocols by verifying interactions, validating time constraints, handling clock increments, and performing standard \MPST protocol checks.

\section{Evaluation: Expressiveness, Case Studies and Benchmarks}
\label{sec:evaluation}
\label{SEC:EVALUATION}
We evaluate our toolchain~\timedmulticrusty from two perspectives: \emph{expressivity} and \emph{feasibility}. 
In terms of expressivity, we implement protocols from the session type literature~\cite{hu2016Hybrid,neykova2013Spy,fielding2014Hypertext,jia2016Monitors,huSessionBased2008,postel1982Rfc0821}, 
as well as newly introduced protocols derived from real-world applications~\cite{DBLP:journals/sensors/ChenZLCJGYAN22,servoWebEngineBuggy,androidMotionSensors,Pine64,DBLP:journals/tches/WoutersMAGP19}. 
Regarding feasibility, we compare our system to~\multicrusty~\cite{lagaillardie2022Affine}, an untimed implementation of affine synchronous \MPST, demonstrating that our tool introduces negligible 
compile-time and runtime overhead in all cases, as expected.

\smallskip
\noindent
\emph{Setup.} 
For benchmarking purposes, we employ \criterionRust \cite{web:rust:criterion}, configuring it to utilise bootstrap sampling \cite{dixon2006bootstrap,efron1992bootstrap} for each protocol. We collect samples of size 100,000 and report the average execution time with a 95\% confidence interval. Additionally, we benchmark compilation time by compiling each protocol 100 times and reporting the mean. Compilation is performed using the commands \CODE{cargo clean} and \CODE{cargo build [name of the file]}. The \CODE{cargo clean} command removes all compiled files, including dependencies, while the second command, without additional parameters, compiles the given file and its dependencies without optimising the resulting files.
The benchmarking machine configuration is a workstation with an
Intel\textregistered~
Core\texttrademark~
i7-7700K CPU @ 4.20 GHz
with 8 threads and 32 GB of RAM,
Manjaro Linux 23.1.4,
Linux kernel 6.6.25-1,
along with the latest versions of
\rustup (1.27.0) and the \Rust cargo compiler (1.77.1).

\subsection{Performance: \timedmulticrusty vs. \multicrusty}
\label{subsec:evaluation:benchmarks}

When comparing \timedmulticrusty with  \multicrusty,
we evaluate their performance on two standard benchmark protocols: the \emph{ring} and \emph{mesh} protocols.
The ring protocol involves sequentially passing a message through roles,
while the mesh protocol requires each participant to send a message to every other.
Both protocols underwent 100 iterations within a time window of 0 to 10 seconds.
Benchmark results for roles ranging from 2 to 8 are presented in~\Cref{fig:benchmark_results_examples} (top).

In the \emph{ring} protocol, compile-time benchmarks (\Cref{fig:benchmark_results:compile:ring}) indicate that \timedmulticrusty experiences a marginal slowdown of less than 2\% with 2 roles, but achieves approximately  
5\% faster compilation time with 8 roles. Regarding runtime benchmarks (\Cref{fig:benchmark_results:running:ring}), \multicrusty demonstrates a 15\% speed advantage with 2 roles, which decreases to 5\% with 8 roles. The overhead remains consistent, with a difference of less than 0.5 ms at 6, 7, and 8 roles.

In the \emph{mesh} protocol, where all roles send and receive messages (compile-time benchmarks in \Cref{fig:benchmark_results:compile:mesh} and runtime benchmarks in \Cref{fig:benchmark_results:running:mesh}), \timedmulticrusty compiles slightly slower (less than 1\% at 2 roles, 
4\% at 8 roles) and runs slower as well (less than 1\% at 2 roles, 15\% at 8 roles). Compile times for \timedmulticrusty range from 18.9 s to 26 s, with running times ranging between 0.9 ms and 11.9 ms. The  gap widens exponentially due to the increasing number of enforced time constraints. In summary, as the number of roles increases, \timedmulticrusty demonstrates a growing overhead, mainly attributed to the incorporation of additional time constraint checking.

\subsection{Expressivity and Feasibility with Case Studies}
\label{subsec:evaluation:examples}
\begin{figure}[t!]
    \begin{subfigure}[t]{0.026\textwidth}
        \includegraphics[width=\textwidth, height=7em]{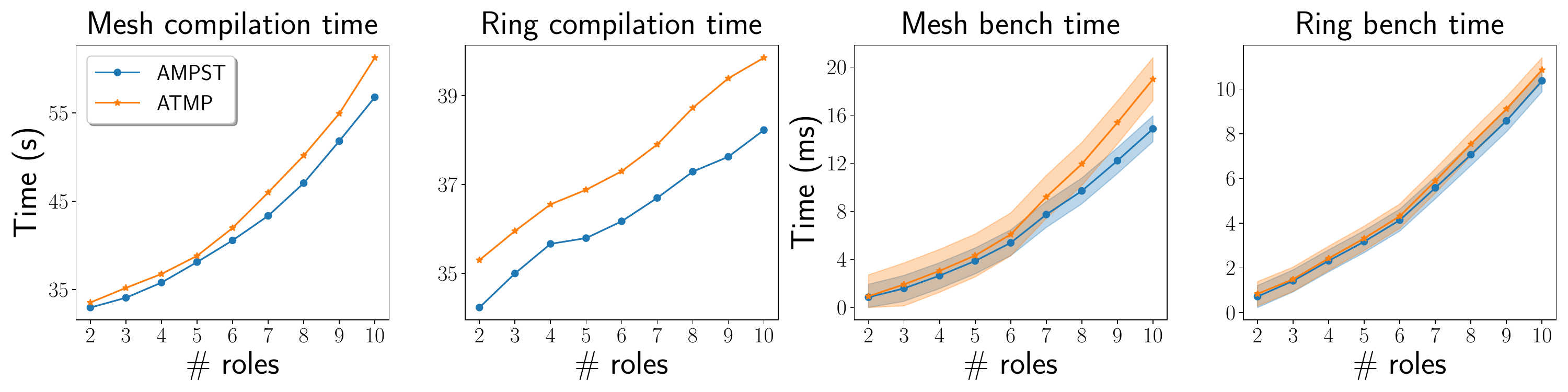}
    \end{subfigure}
    \begin{subfigure}[t]{0.23\textwidth}
        \includegraphics[width=\textwidth, height=7em]{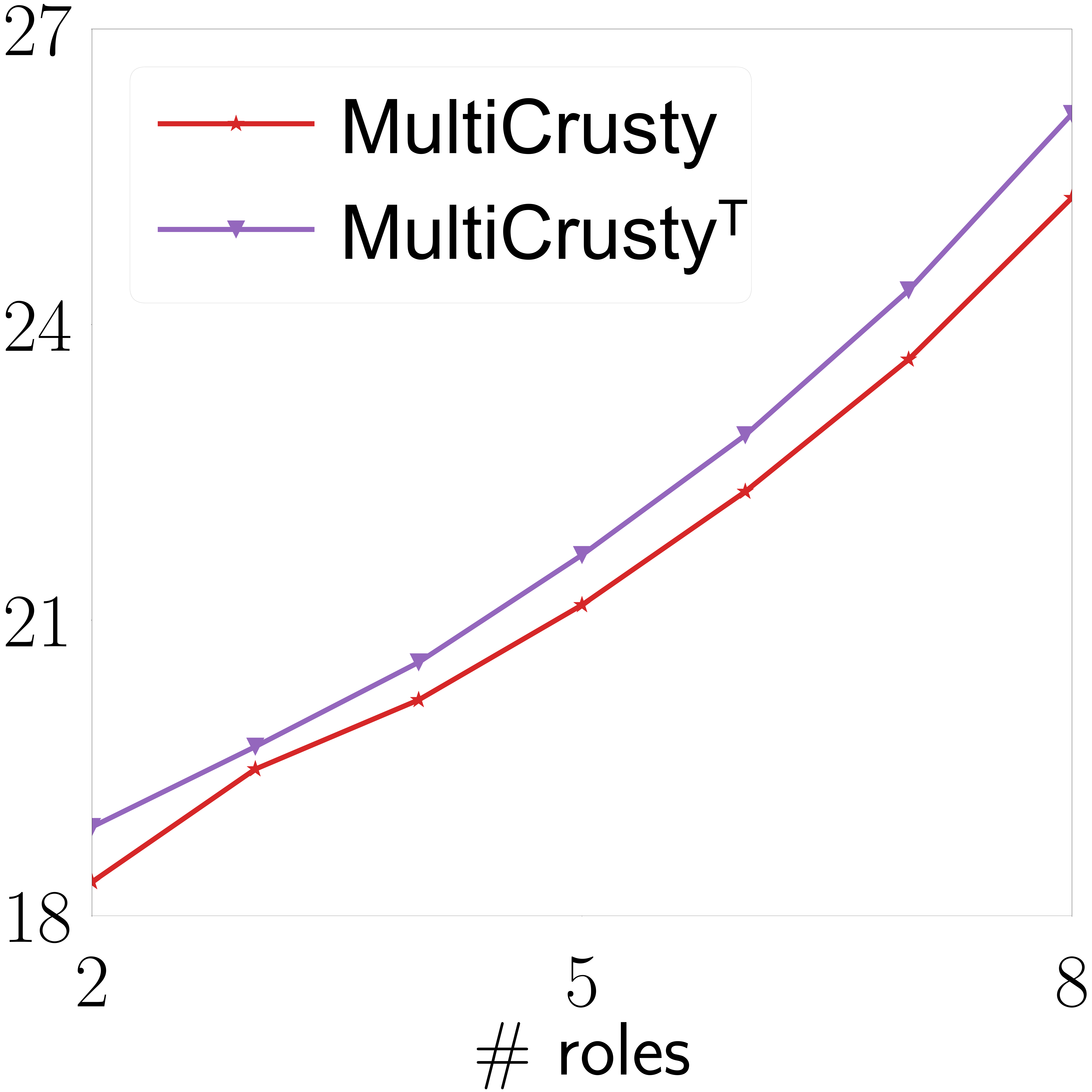}
        \caption{\emph{mesh} - compilation}
        \label{fig:benchmark_results:compile:mesh}
    \end{subfigure}
    \begin{subfigure}[t]{0.23\textwidth}
        \includegraphics[width=\textwidth, height=7em]{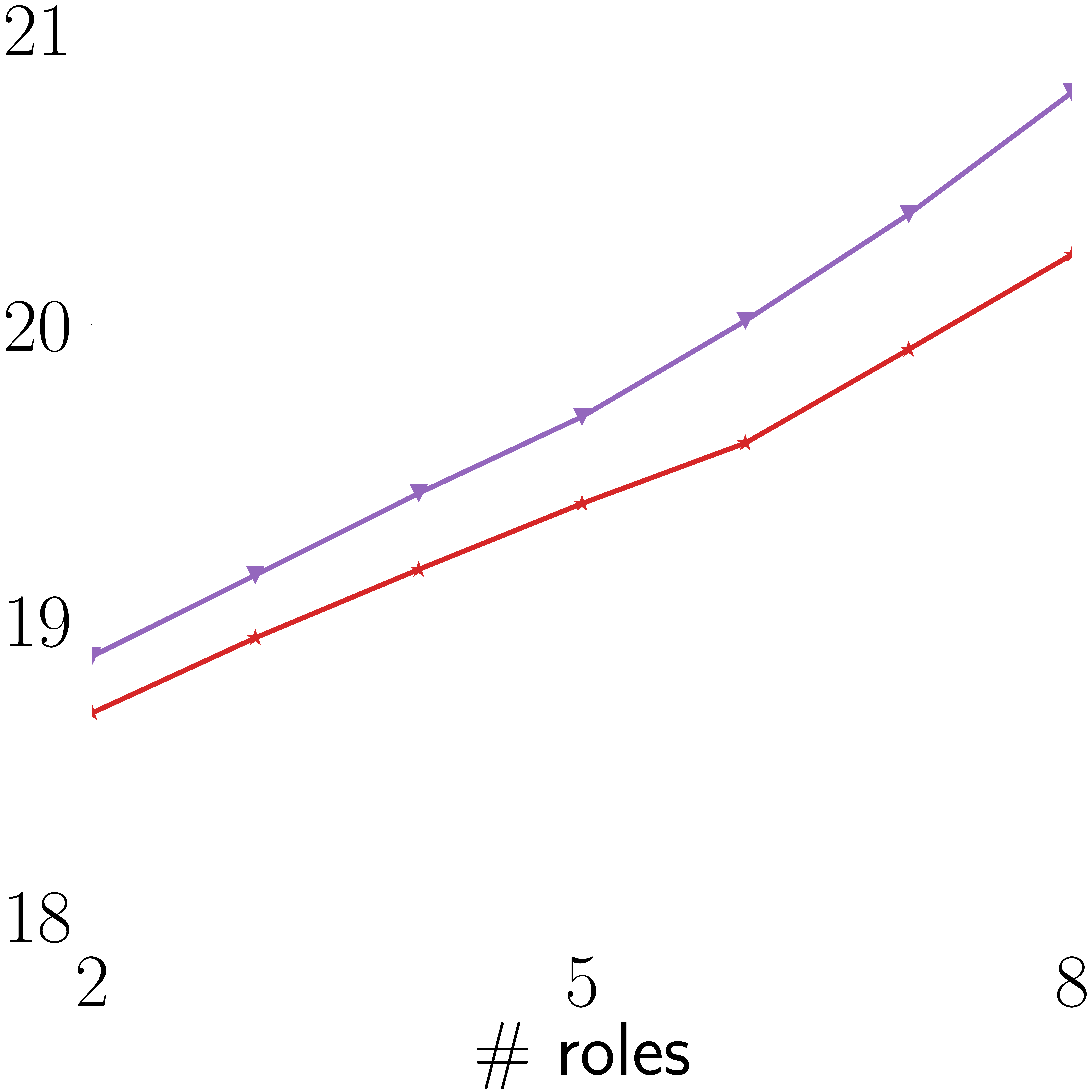}
        \caption{\emph{ring} - compilation}
        \label{fig:benchmark_results:compile:ring}
    \end{subfigure}\hfill%
    \begin{subfigure}[t]{0.026\textwidth}
        \includegraphics[width=\textwidth, height=7em]{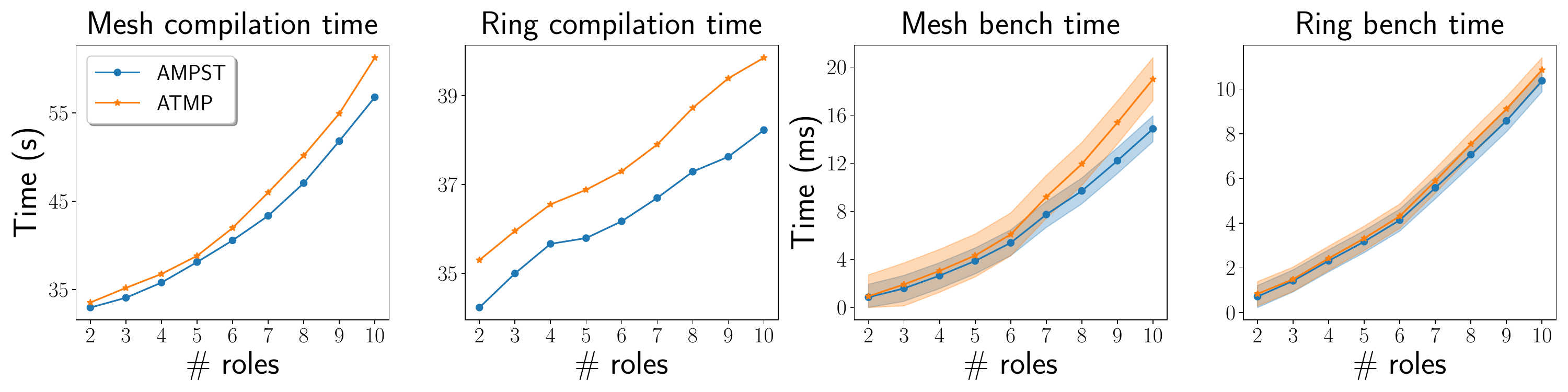}
    \end{subfigure}
    \begin{subfigure}[t]{0.23\textwidth}
        \includegraphics[width=\textwidth, height=7em]{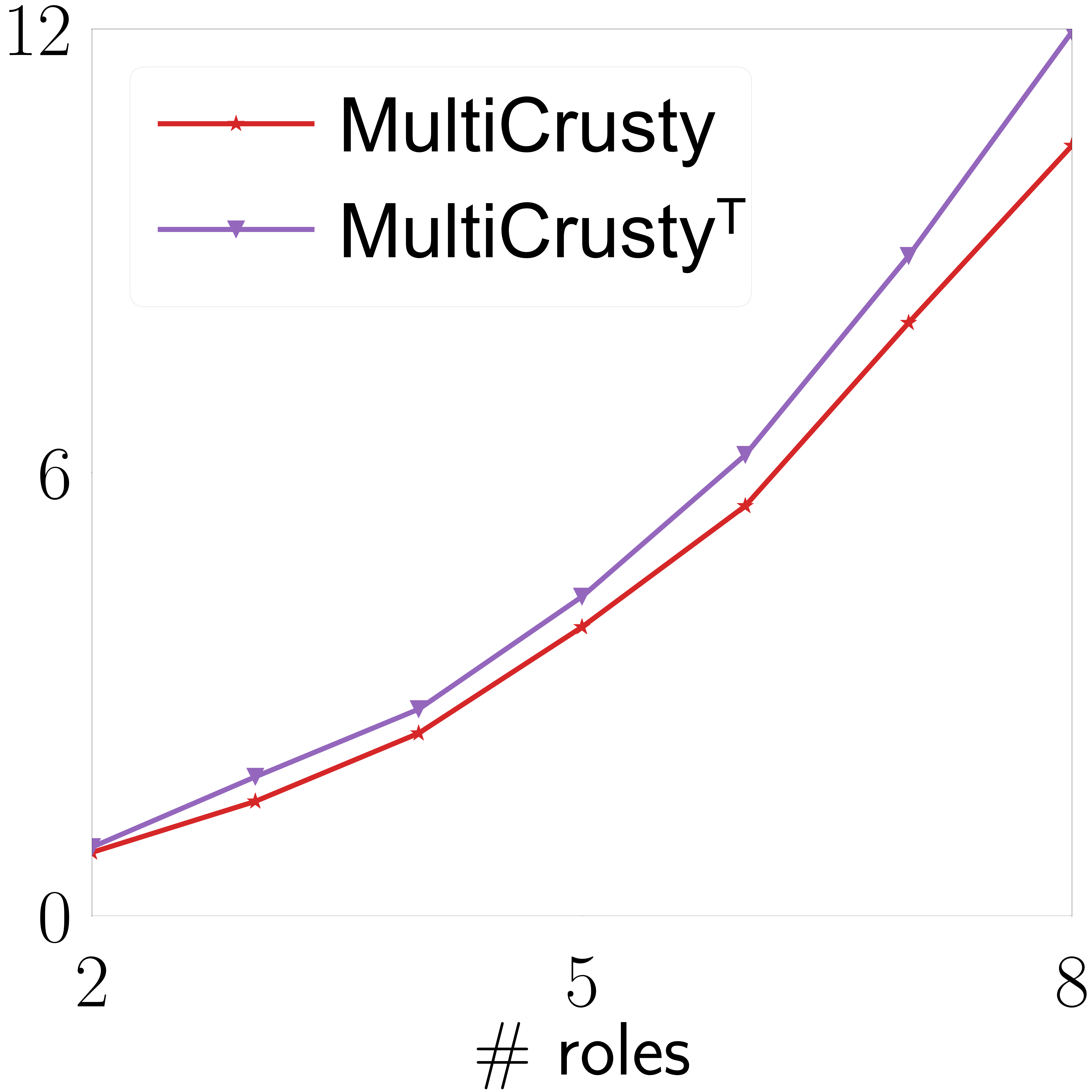}
        \caption{\emph{mesh} - runtime}
        \label{fig:benchmark_results:running:mesh}
    \end{subfigure}
    \begin{subfigure}[t]{0.23\textwidth}
        \includegraphics[width=\textwidth, height=7em]{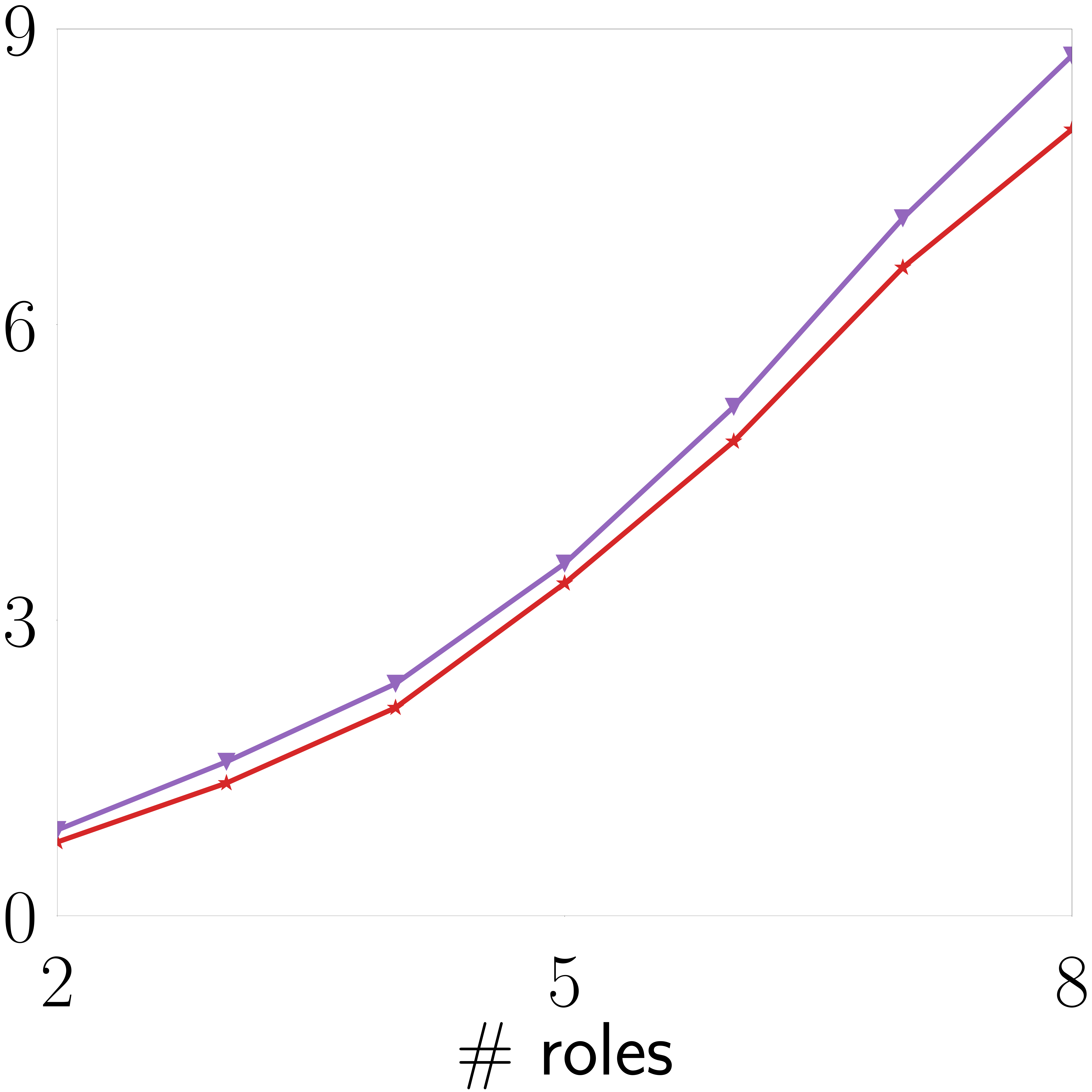}
        \caption{\emph{ring} - runtime}
        \label{fig:benchmark_results:running:ring}
    \end{subfigure}
    \begin{subfigure}[t]{\textwidth}
        \includegraphics[width=\textwidth, height=8.2em]{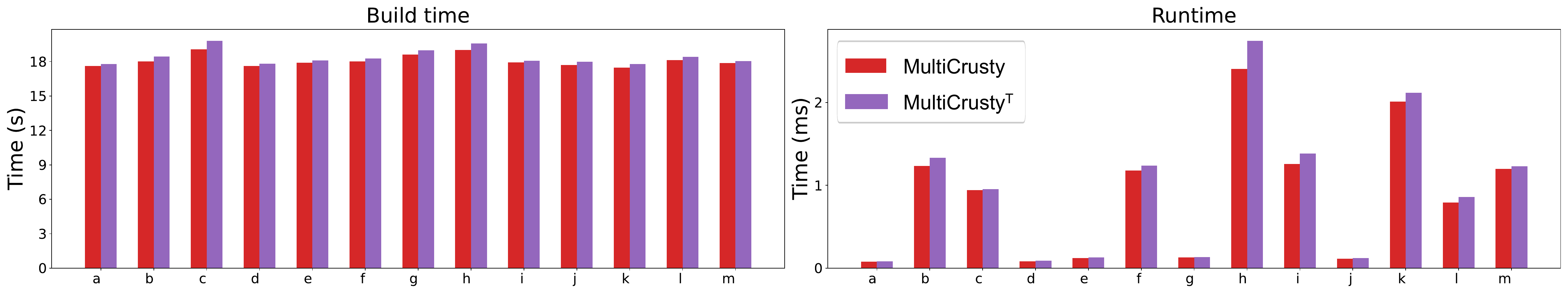}
    \end{subfigure}
    \caption{
        Top: microbenchmark results for mesh and ring protocols. 
         Bottom: benchmark results for
        Calculator~\cite{hu2016Hybrid} (a),
        Online wallet~\cite{neykova2013Spy} (b),
        SMTP~\cite{postel1982Rfc0821} (c),
        Simple voting~\cite{hu2016Hybrid} (d),
        Three buyers~\cite{jia2016Monitors} (e),
        Travel agency~\cite{huSessionBased2008} (f),
        OAuth~\cite{neykova2013Spy} (g),
        HTTP~\cite{fielding2014Hypertext} (h),
        Remote data \cite{DBLP:journals/sensors/ChenZLCJGYAN22} (i),
        Servo \cite{servoWebEngineBuggy} (j),
        Gravity sensor \cite{androidMotionSensors} (k),
        PineTime heart rate \cite{Pine64} (l),
        and Proximity based car key \cite{DBLP:journals/tches/WoutersMAGP19} (m).  %
    }
        \label{fig:benchmark_results_examples}
    \vspace{-1em}
\end{figure}

We implement a variety of protocols to showcase the expressivity, feasibility, and capabilities of \timedmulticrusty, conducting benchmarking using both \timedmulticrusty and \multicrusty. The \CODE{send} and \CODE{recv} operations in both libraries are ordered, directed, and involve the same set of participants. Additionally, when implemented with \timedmulticrusty, these operations are enriched with time constraints and reset predicates. The benchmark results for the selected case studies, including those from prior research and five additional protocols sourced from industrial use cases~\cite{DBLP:journals/sensors/ChenZLCJGYAN22,servoWebEngineBuggy,androidMotionSensors,Pine64,DBLP:journals/tches/WoutersMAGP19}, are presented in the bottom part of \Cref{fig:benchmark_results_examples}. To ensure a fair comparison between \timedmulticrusty (${\highlightbox{bordeaux}{\color{bordeaux}bars}}$) and \multicrusty (${\highlightbox{purple}{\color{purple}bars}}$), time constraints are enforced for all examples without introducing any additional sleep or timeouts.

It is noteworthy 
that rate-based protocols~((k), (l), (m) in~\Cref{fig:benchmark_results_examples} (bottom)) from real-time systems~\cite{androidMotionSensors,Pine64,DBLP:journals/tches/WoutersMAGP19}
are implemented in~\timedmulticrusty, showcasing its expressivity in real-time applications. 
These implementations feature the establishment of consistent time constraints and a shared clock for operations with identical rates.

Consider the Car Key protocol~\cite{DBLP:journals/tches/WoutersMAGP19} as an example: the car periodically sends a wake-up message to probe the presence of the key. The communication in this protocol can 
be easily captured by session types.  
The critical aspect is the rate requirement, indicating that interactions 
from one wake-up to the next must be completed within a period of (\eg) 100 ms. 
Consequently, when implementing this protocol with \timedmulticrusty, all time constraints use the same clock, ranging between 0 and 100 ms, and the clock resets at the end of each loop. For scenarios requiring a ``\emph{generalised expansion}'',  where processes and communication operations with periods of one are not multiples of the other, we extend the main loop to match the least common multiple of the periods, following a methodology detailed in~\cite{DBLP:journals/tches/WoutersMAGP19}.

The feasibility of our tool, \timedmulticrusty, is demonstrated in~\Cref{fig:benchmark_results_examples} (bottom). The results indicate that \timedmulticrusty incurs minimal compile-time overhead, averaging approximately 1.75\%. Moreover, the runtime for each protocol remains within milliseconds, ensuring negligible impact. Notably, in the case of the HTTP protocol, the runtime comparison percentage with \multicrusty is 87.60\%, primarily attributed to the integration of 126 time constraints within it. 
The relevant implementation metrics, including multiple participants~(MP), branching, recursion~(Rec), and time constraints, are  illustrated in~\Cref{tab:examples:metrics}.

\begin{table}[h]
    \centering
    \vspace{-.5em}
    \caption{Metrics for protocols implemented in~\timedmulticrusty.}
    \vspace{-.3em}
    \footnotesize
    \resizebox{\textwidth}{!}{%
        \begin{tabular}{|l|r|c|c|c|c|c|c|}
            \hline
            Protocol                                                 &
            \makecell[c]{Generated                                     \\Types}                                                                &
            \makecell[c]{Implemented                                   \\Lines of Code}                                                                 &
            MP                                                       &
            Branching                                                &
            Rec                                                      &
            \makecell[c]{Time                                          \\Constraints} \\ %
            \hline
            \hline
            Calculator~\cite{hu2016Hybrid}                           &
            52                                                       &
            51                                                       &
            \xmark                                                   &
            \cmark                                                   &
            \cmark                                                   &
            11                                                         \\ %
            \hline
            Online wallet~\cite{neykova2013Spy}                      &
            142                                                      &
            160                                                      &
            \cmark                                                   &
            \cmark                                                   &
            \cmark                                                   &
            24                                                         \\ %
            \hline
            SMTP~\cite{postel1982Rfc0821}                            &
            331                                                      &
            475                                                      &
            \xmark                                                   &
            \cmark                                                   &
            \cmark                                                   &
            98                                                         \\ %
            \hline
            Simple voting~\cite{hu2016Hybrid}                        &
            73                                                       &
            96                                                       &
            \xmark                                                   &
            \cmark                                                   &
            \xmark                                                   &
            16                                                         \\ %
            \hline
            Three buyers~\cite{jia2016Monitors}                      &
            108                                                      &
            78                                                       &
            \cmark                                                   &
            \cmark                                                   &
            \xmark                                                   &
            22                                                         \\ %
            \hline
            Travel agency~\cite{huSessionBased2008}                  &
            148                                                      &
            128                                                      &
            \cmark                                                   &
            \cmark                                                   &
            \cmark                                                   &
            30                                                         \\ %
            \hline
            OAuth~\cite{neykova2013Spy}                              &
            199                                                      &
            89                                                       &
            \cmark                                                   &
            \cmark                                                   &
            \xmark                                                   &
            30                                                         \\ %
            \hline
            HTTP~\cite{fielding2014Hypertext}                        &
            648                                                      &
            610                                                      &
            \cmark                                                   &
            \cmark                                                   &
            \cmark                                                   &
            126                                                        \\ %
            \hline
            Remote data \cite{DBLP:journals/sensors/ChenZLCJGYAN22}  &
            100                                                      &
            119                                                      &
            \cmark                                                   &
            \cmark                                                   &
            \cmark                                                   &
            16                                                         \\ %
            \hline
            Servo \cite{servoWebEngineBuggy} &
            74                                                       &
            48                                                       &
            \cmark                                                   &
            \xmark                                                   &
            \xmark                                                   &
            10                                                         \\ %
            \hline
            Gravity sensor \cite{androidMotionSensors}               &
            61                                                       &
            95                                                       &
            \xmark                                                   &
            \cmark                                                   &
            \cmark                                                   &
            9                                                          \\ %
            \hline
            PineTime heart rate \cite{Pine64}                        &
            101                                                      &
            111                                                      &
            \xmark                                                   &
            \cmark                                                   &
            \cmark                                                   &
            17                                                         \\ %
            \hline
            Proximity based car key \cite{DBLP:journals/tches/WoutersMAGP19}    &
            70                                                       &
            134                                                      &
            \xmark                                                   &
            \cmark                                                   &
            \cmark                                                   &
            22                                                         \\ %
            \hline
        \end{tabular}%
    }
    \label{tab:examples:metrics}
    \vspace{-1em}
\end{table}

\section{Related Work and Conclusion}
\label{sec:related_work}

\subparagraph{Time in Session Types} %
Bocchi \etal~\cite{DBLP:conf/concur/BocchiYY14} propose a timed extension of \MPST to model real-time choreographic interactions, while 
Bocchi \etal~\cite{bocchi2019Asynchronous} extend~\emph{binary} session types with time constraints, introducing
a subtyping relation and a blocking receive primitive with timeout in their calculus. In contrast to their strategies to avoid time-related failures, as %
discussed in~\Cref{sec:introduction,sec:overview}, \ATMP focuses on actively managing failures as they occur, offering a distinct approach to handling timed communication.

Iraci \etal~\cite{DBLP:journals/pacmpl/IraciandCHZ2023} 
extend \emph{synchronous binary} session types
with a periodic recursion primitive to model rate-based processes. 
To align their design with real-time systems, they encode time into a periodic construct, synchronised with a global clock.  With~\emph{rate compatibility},
a relation that facilitates communication between processes with different periods by 
synthesising and verifying a common superperiod type, their approach ensures that well-typed processes
remain free from rate errors during any specific period. 
On the contrary, \ATMP integrates time constraints directly into communication through local clocks, 
resulting in distinct time behaviour. Intriguingly, our method of time encoding can adapt to theirs, 
while the opposite is not feasible. 
Consequently, not all the timed protocols in our paper, \eg~\Cref{fig:implementation:remote_data}, can be accurately 
represented in their system. Moreover, due to its \emph{binary} and \emph{synchronous} features, 
their theory does not directly model and ensure the properties of real distributed systems.

Le Brun \etal~\cite{ESOP23MAGPi} develop a theory of multiparty
session types that accounts for different failure scenarios, including message losses, delays, reordering,
as well as link failures and network partitioning.
Unlike~\ATMP, their approach does not integrate time specifications
or address failures specifically related to time.
Instead,
they use \emph{timeout} as a generic message label for failure branches,
which triggers the failure detection mechanism.
For example, in~\cite{ESOP23MAGPi}, $\mpChanRole{\mpS}{\roleP}
 {\scriptsize \VarClock} \mpSeq \mpNil$ denotes a timeout branch, where  {\scriptsize $\VarClock$}
 is the timeout notation representing
 a non-deterministic time duration that a process waits before assuming a failure has arisen.
 Except for~\cite{DBLP:journals/pacmpl/IraciandCHZ2023}, all the mentioned works
 on session types with time are purely theoretical.

\vspace{-1em}
\subparagraph{Affinity, Exceptions and Error-Handling in Session Types} %
 Mostrous and Vasconcelos~\cite{mostrous2018Affine} propose affine binary session
types with explicit cancellation,
which Fowler \etal~\cite{fowler2019Exceptional}
extend
to define Exceptional GV for binary asynchronous communication.
Exceptions can be
nested and handled over multiple communication actions,
and their implementation is an
extension of the research language \Links. %
Harvey \etal~\cite{harvey2021Multiparty} incorporate~\MPST with explicit connection actions to facilitate multiparty distributed communication,
and develop a code generator based on the actor-like research language \Ensemble to implement their approach.
The work in~\cite{mostrous2018Affine}
remains theoretical,
and both~\cite{mostrous2018Affine,fowler2019Exceptional} are
limited to binary and linear logic-based session types. %
Additionally,
none of these works considers %
time specifications or
addresses the handling of time-related exceptions in their systems, which are key aspects of our work.

\vspace{-1em}
\subparagraph{Session Types Implementations} %
\multicrusty,
extensively compared to~\timedmulticrusty
in this paper,
is a~\Rust implementation 
based on affine \MPST by Lagaillardie \etal~\cite{lagaillardie2022Affine}. 
Their approach relies on  \emph{synchronous} communication, rendering 
time and timeout exceptions unnecessary.

Cutner \etal~\cite{CYV2022}
introduce~\rumpsteak,
a~\Rust implementation
based on the~\CODE{tokio}
\Rust library,
which uses a different design
for asynchronous multiparty communications
compared to~\timedmulticrusty,
relying on the~\crosschan
\Rust library.
The main goal of~\cite{CYV2022}
is to compare the performance
of~\rumpsteak,
mainly designed to analyse
asynchronous
message reordering,
to existing tools
such as the~\kmc~tool developed by Lange and Yoshida~\cite{lange2019Verifying}.
Unlike~\timedmulticrusty,
\rumpsteak lacks formalisation, or handling of timed communications and failures.

\typestate is a~\Rust
library implemented
by Duarte and Ravara~\cite{duarte2022Typestates},
focused on helping
developers to write safer
APIs using typestates
and
their macros~\CODE{#[typestate]},
\CODE{#[automaton]}
and~\CODE{#[state]}.
\timedmulticrusty
and~\typestate
are fundamentally different,
with~\typestate creating a state machine
for checking possible errors in APIs
and not handling affine or timed communications. %
\ferrite,
a~\Rust implementation introduced by Chen \etal~\cite{chen2022Ferrite}, 
is limited to binary session types
and forces the use of linear channels.
The modelling of 
\ferrite is based on 
the shared binary session type
calculus~\texttt{SILL$_{s}$}.

Jespersen \etal~\cite{jespersen2015Session} and Kokke~\cite{kokke2019Rusty} propose
\Rust implementations
of binary session types
for synchronous communication
protocols.
\cite{DBLP:journals/pacmpl/IraciandCHZ2023} extends the framework from~\cite{jespersen2015Session}
to encode the \emph{rate compatibility} relation as a \Rust trait and check whether two types are rate compatible.
Their approach is demonstrated with examples from rate-based systems, 
including~\cite{androidMotionSensors,Pine64,DBLP:journals/tches/WoutersMAGP19}.
Motivated by these applications, we formalise and 
implement the respective timed protocols in \timedmulticrusty,
showcasing the expressivity and feasibility of our system in real-time scenarios.

Neykova \etal~\cite{neykova2017timed} propose a \Python
runtime monitoring framework to handle timed
multiparty communications, guided by~\cite{DBLP:conf/concur/BocchiYY14}.
They use a timed extension of \Scribble~\cite{yoshida2013Scribble} to check the wait-freedom of protocols. %
The \Python toolchain then
\emph{dynamically} checks that
the implementation
has no time violation or
communication mismatch.
However,
their verifications are time-consuming,
and may lead to valid messages
being rejected~\cite[Fig.~16(b)]{neykova2017timed}.
Conversely, 
\timedmulticrusty~\emph{statically} enforces
that each time window is reached
by a clock,
allowing every action to be triggered.
Moreover, through \Rust's affine types  and meshed channels, %
\timedmulticrusty guarantees
linear usage of channels
and label matching by compilation.

\vspace{-1em}
\subparagraph{Conclusion and Future Work}
To address time constraints and timeout exceptions
in asynchronous communication,
we propose \emph{affine timed multiparty session types}~(\ATMP)
along with the toolchain~\timedmulticrusty,
an implementation of \ATMP in \Rust.
Thanks to the incorporation of affinity and failure handling mechanisms,
our approach renders impractical conditions such as \emph{wait-freedom} and \emph{urgent receive} obsolete
while ensuring communication safety, protocol conformance, and deadlock-freedom, even in the presence of~(timeout) failures.
Compared to 
a synchronous toolchain without time,
\timedmulticrusty exhibits negligible overhead
in various complex examples including those from real-time systems, 
while enabling the verification of time constraints under asynchronous communication. %
As future work, we plan to
explore automatic recovery from errors
and timeouts instead of simply terminating 
processes, which will involve extending the analysis of communication causality to timed global types and incorporating reversibility mechanisms into our system.

\bibliography{main}

\iftoggle{full}
{
\newpage 
\appendix
\label{appendix}

\section{Proofs for Association}
\label{sec:app:assoc:proofs}

\subsection{Unfoldings}
\label{sec:app-unfolding}

We define the \emph{unfolding} of timed global types as:

\smallskip%
  \centerline{\(%
\begin{array}{c}
 \unfoldOne{\gtG}
   \;=\;
   \left\{%
    \begin{array}{@{}l@{\hskip 2mm}r@{}}
    \unfoldOne{\gtGi\subst{\gtRecVar}{\gtRec{\gtRecVar}{\gtGi}}}
    &
    \text{if\, $\gtG = \gtRec{\gtRecVar}{\gtGi}$}
    \\
    \gtG
    &
     \text{otherwise}
     \end{array}
     \right.
     \\
     \end{array}
\)}

\smallskip
\noindent
A recursive type $\gtRec{\gtRecVar}{\gtG}$ must be guarded (or
contractive), \ie the unfolding leads to a progressive prefix, \eg a
transmission.
Unguarded types, such as $\gtRec{\gtRecVar}{\gtRecVar}$ and
$\gtRec{\gtRecVar}{\gtRec{\gtRecVari}{\gtRecVar}}$, are not allowed.
Timed local types are subject to the similar definitions and requirements.

Note that any timed global type (or timed local type) we mention in this paper is closed and well-guarded.

\begin{lemma}
\label{lem:unfold-no-rec}
  For a closed, well-guarded timed global type $\gtG$, $\unfoldOne{\gtG}$ can only be of
  form $\gtEnd$, $\gtComm{\roleP}{\roleQ}{}{\cdots}{}{}$, or $ \gtCommTTransit{\roleP}{\roleQ}{}{\cdots}{}{}{}{}$. 
 For a closed, well-guarded timed local type $\stT$, $\unfoldOne{\stT}$ can only be
  of form $\stEnd$, $\stIntSum{\roleP}{}{\cdots}$, or
  $\stExtSum{\roleP}{}{\cdots}$.
\end{lemma}
\begin{proof}
  $\gtRecVar$ will not appear since we require closed types.
  $\gtRec{\gtRecVar}{\gtG'}\subst{\gtRecVar}{\gtRec{\gtRecVar}{\gtG'}} \neq
  \gtRec{\gtRecVar}{\gtG'}$ since we require well-guarded types (recursive
  types are contractive).
  Similar argument for timed local types. 
  \qedhere
\end{proof}

\begin{lemma}
\label{lem:unfold_projection}
If $\gtProj{\gtG}{\roleP} = \stT$, then $\gtProj{\unfoldOne{\gtG}}{\roleP} = \unfoldOne{\stT}$.
\end{lemma}
\begin{proof}
By the definition of projection~(\Cref{def:global-proj}), 
the definitions of unfolding recursive timed global and local types, 
and~\Cref{lem:unfold-no-rec}. 
\qedhere
\end{proof}

\subsection{Subtyping}
\label{sec:app-aat-mpst-subtyping}

\begin{lemma}[Subtyping is Reflexive]
\label{lem:reflexive-subtyping}
  For any closed, well-guarded timed local type $\stT$, $\stT \stSub \stT$ holds.
\end{lemma}
\begin{proof}
We construct a relation $R = \setenum{(\stT, \stT)}$.  
It is trivial to show that $R$ satisfies all clauses 
of~\Cref{def:main_subtyping}, and hence,  $R \subseteq \stSub$. 
\qedhere
\end{proof}

\begin{lemma}[Subtyping is Transitive]\label{transitive-subtyping}
  For any closed, well-guarded timed local type $\stT[1]$, $\stT[2]$, $\stT[3]$,
  if $\stT[1] \stSub \stT[2]$ and $\stT[2] \stSub \stT[3]$ hold, then $\stT[1] \stSub \stT[3]$
  holds.
\end{lemma}
\begin{proof}
By constructing a relation 
$R = \setcomp{(\stT[1], \stT[3])}{\exists \stT[2] \text{ such that } \stT[1] \stSub \stT[2] \text{ and } \stT[2] \stSub \stT[3]}$, and showing that 
$R \subseteq \stSub$. 
\qedhere
\end{proof}

\begin{lemma}
\label{lem:end-subtyping}
For any timed local type $\stT$, $\stT \stSub \stEnd$ if and 
only if $\stEnd \stSub \stT$. 
\end{lemma}
\begin{proof}
By~\Cref{def:main_subtyping}. %
\qedhere 
\end{proof}

\begin{lemma}
\label{lem:unfold-subtyping}
  For any closed, well-guarded timed local type $\stT$,
  \begin{enumerate*}
    \item $\unfoldOne{\stT} \stSub \stT$; and
    \item $\stT \stSub \unfoldOne{\stT}$.
  \end{enumerate*}
\end{lemma}
 \begin{proof}
  \begin{enumerate*}
    \item By $\inferrule{\iruleStSubRecR}$ if $\stT = \stRec{\stRecVar}{\stTi}$. Otherwise, 
      by reflexivity~(\Cref{lem:reflexive-subtyping}). 
    \item By $\inferrule{\iruleStSubRecL}$ if $\stT = \stRec{\stRecVar}{\stTi}$. Otherwise, 
      by reflexivity~(\Cref{lem:reflexive-subtyping}). 
  \end{enumerate*}
  \qedhere
 \end{proof}

\begin{lemma}
\label{lem:merge-subtyping}
  Given a collection of mergable timed local types $\stT[i]$ ($i \in I$).
  For all $j \in I$, $\stT[j] \stSub \stMerge{i \in I}{\stT[i]}$ holds.
\end{lemma}
\begin{proof}
By constructing a relation 
$R =  \setcomp{(\stT[j], \stMerge{i \in I}{\stT[i]})}{j \in I}$, and showing that 
$R \subseteq \stSub$. 
\qedhere
\end{proof}

\begin{lemma}
\label{lem:merge-lower-bound}
  Given a collection of mergable timed local types $\stT[i]$ ($i \in I$).
  If for all $i \in I$, $\stT[i] \stSub \stT$ for some timed local type $\stT$,
  then $\stMerge{i \in I}{\stT[i]} \stSub \stT$.
\end{lemma}
\begin{proof}
By constructing a relation 
$R =  \setenum{(\stMerge{i \in I}{\stT[i]}, \stT)}$, and showing that 
$R \subseteq \stSub$. 
\qedhere
\end{proof}

\begin{lemma}\label{lem:merge-upper-bound}
  Given a collection of mergable timed local types $\stT[i]$ ($i \in I$).
  If for all $i \in I$, $\stT \stSub \stT[i]$ for some timed local type $\stT$,
  then $\stT \stSub \stMerge{i \in I}{\stT[i]}$.
\end{lemma}
\begin{proof}
By constructing a relation 
$R =  \setenum{(\stT, \stMerge{i \in I}{\stT[i]})}$, and showing that 
$R \subseteq \stSub$. 
\qedhere
\end{proof}

\begin{lemma}\label{lem:subtype:merge-subty}
  Given two collections of mergable timed local types $\stU[i], \stT[i]$ ($i \in I$).
  If for all $i \in I$, $\stU[i] \stSub \stT[i]$, then
  $\stMerge{i \in I} {\stU[i]} \stSub \stMerge{i \in I}{\stT[i]}$.
\end{lemma}
\begin{proof}
By constructing a relation 
$R =  \setenum{(\stMerge{i \in I}{\stU[i]}, \stMerge{i \in I}{\stT[i]})}$, and showing that 
$R \subseteq \stSub$. 
\end{proof}

\begin{lemma}[Inversion of Subtyping]\label{lem:subtyping-invert}
  ~
  \begin{enumerate}
    \item If
      \,$\stT \stSub
       \stIntSum{\roleP}{i \in I}{\stTChoice{\stLab[i]}{\stS[i]} {\ccst[i], \crst[i]} \stSeq \stT[i]}
      $, then
      $\unfoldOne{\stT} =
       \stIntSum{\roleP}{j \in J}{\stTChoice{\stLabi[j]}{\stSi[j]} {\ccsti[j], \crsti[j]}\stSeq
       \stTi[j]}
      $, $I \subseteq J$,
      and $\forall k \in I: \stLab[k] = \stLabi[k], \stS[k] \stSub
      \stSi[k]$,  $\ccsti[k] = \ccst[k]$, $\crsti[k] = \crst[k]$, and $\stTi[k] \stSub \stT[k]$.
      \item If
      \,$
       \stIntSum{\roleP}{i \in I}{\stTChoice{\stLab[i]}{\stS[i]} {\ccst[i], \crst[i]} \stSeq \stT[i]}  \stSub \stT
      $, then
      $\unfoldOne{\stT} =
       \stIntSum{\roleP}{j \in J}{\stTChoice{\stLabi[j]}{\stSi[j]} {\ccsti[j], \crsti[j]}\stSeq
       \stTi[j]}
      $, $J \subseteq I$,
      and $\forall k \in J: \stLab[k] = \stLabi[k], \stSi[k] \stSub
      \stS[k]$,  $\ccsti[k] = \ccst[k]$, $\crsti[k] = \crst[k]$, and $\stT[k] \stSub \stTi[k]$.
    \item If
      \,$\stT \stSub
       \stExtSum{\roleP}{i \in I}{\stTChoice{\stLab[i]}{\stS[i]}{\ccst[i], \crst[i]} \stSeq \stT[i]}
      $, then
      $\unfoldOne{\stT} =
       \stExtSum{\roleP}{j \in J}{\stTChoice{\stLabi[j]}{\stSi[j]} {\ccsti[j], \crsti[j]}\stSeq
       \stTi[j]}
      $, $J \subseteq I$,
      and $\forall k \in J: \stLab[k] = \stLabi[k], \stSi[k] \stSub
      \stS[k]$,  $\ccsti[k] = \ccst[k]$, $\crsti[k] = \crst[k]$, and $\stTi[k] \stSub \stT[k]$.
       \item If
      \,$
       \stExtSum{\roleP}{i \in I}{\stTChoice{\stLab[i]}{\stS[i]}{\ccst[i], \crst[i]} \stSeq \stT[i]} \stSub \stT
      $, then
      $\unfoldOne{\stT} =
       \stExtSum{\roleP}{j \in J}{\stTChoice{\stLabi[j]}{\stSi[j]} {\ccsti[j], \crsti[j]}\stSeq
       \stTi[j]}
      $, $I \subseteq J$,
      and $\forall k \in I: \stLab[k] = \stLabi[k], \stS[k] \stSub
      \stSi[k]$,  $\ccsti[k] = \ccst[k]$, $\crsti[k] = \crst[k]$, and $\stT[k] \stSub \stTi[k]$.
  \end{enumerate}
\end{lemma}
\begin{proof}
By~\cref{lem:unfold-no-rec,lem:unfold-subtyping}, the transitivity of subtyping, and \cref{def:main_subtyping} (\inferrule{\iruleStSubIn}, \inferrule{\iruleStSubOut}). 
\end{proof}

\subsection{Semantics of Timed Global Types}
\label{sec:aap-aat-mpst-semantics-gtype}

\begin{lemma}\label{lem:gt-lts-unfold}
  $\gtWithTime{\cVal}{\gtG}
  \gtMove[\stEnvAnnotGenericSym]
  \gtWithTime{\cVali}{\gtGi}$ \;iff
  \;$
  \gtWithTime{\cVal}{\unfoldOne{\gtG}}
  \gtMove[\stEnvAnnotGenericSym]
  \gtWithTime{\cVali}{\gtGi}$.
\end{lemma}
\begin{proof}
  By inverting or applying $\inferrule{\iruleGtMoveRec}$ when necessary.
  \qedhere
\end{proof}

\subsection{Semantics of Typing Environments}
\label{sec:aap-aat-mpst-semantics-confi}
\begin{lemma}
\label{lem:red:trivial-1}
\label{lem:stenv-red:trivial-1-new}
If $\stEnvQ
\stEnvMoveGenAnnot
\stEnvQi$, then $\dom{\stEnvQ} = \dom{\stEnvQi}$.
\end{lemma}
\begin{proof}
 By induction on typing environment reductions. 
 \qedhere
\end{proof}

\begin{lemma}
\label{lem:stenv-red:trivial-2}
  If $\stEnvQ
  \stEnvMoveGenAnnot \stEnvQi$,
  $\stEnvAnnotGenericSym \neq \timeLab$, and
  $\dom{\stEnvQ} = \setenum{\mpS}$, then for all
  $\mpChanRole{\mpS}{\roleP} \in \dom{\stEnvQ}$ with
  $\roleP \neq \ltsSubject{\stEnvAnnotGenericSym}$, we have
      $\stEnvApp{\stEnvQ}{\mpChanRole{\mpS}{\roleP}} =
      \stEnvApp{\stEnvQi}{\mpChanRole{\mpS}{\roleP}}$.
\end{lemma}
\begin{proof}
 By induction on typing environment reductions. 
 \qedhere 
\end{proof}

\begin{lemma}
\label{lem:stenv-red:trivial-2}
  If $\stEnv
  \stEnvMoveGenAnnot \stEnvi$ and $\dom{\stEnv} =
  \setenum{\mpS}$, then 
    for all  $\mpChanRole{\mpS}{\roleP} \in \dom{\stEnv}$ with 
      $\roleP \notin \ltsSubject{\stEnvAnnotGenericSym}$, %
      $\stEnvApp{\stEnv}{\mpChanRole{\mpS}{\roleP}} =
      \stEnvApp{\stEnvi}{\mpChanRole{\mpS}{\roleP}}$.
\end{lemma}
\begin{proof}
By induction on typing context reductions.  
\qedhere 
\end{proof}

\begin{lemma}[Determinism of Typing Environment Reduction]
\label{lem:stenv-red:det}
  If \,$\stEnvQ \stEnvMoveGenAnnot \stEnvQi$ \,and\,
  $\stEnvQ \stEnvMoveGenAnnot
  \stEnvQii$, then $\stEnvQi = \stEnvQii$.
\end{lemma}
\begin{proof}
By induction on typing environment reductions.
\qedhere
\end{proof}

\begin{lemma}
\label{lem:stenv-red:trivial-3}
  If $\stEnv \stEnvMoveWithSession[\mpS]  \stEnvi$, 
  then $\dom{\stEnv} = \dom{\stEnvi}$. 
\end{lemma}
\begin{proof}
Directly from~\Cref{def:aat-mpst-typing-env-reduction,lem:stenv-red:trivial-1-new}.  
\qedhere 
\end{proof}

\begin{lemma}
\label{lem:stenv-red:trivial-4}
  If $\stEnv \stEnvMove  \stEnvi$, 
  then $\dom{\stEnv} = \dom{\stEnvi}$. 
\end{lemma}
\begin{proof}
Directly from~\Cref{def:aat-mpst-typing-env-reduction,lem:stenv-red:trivial-3}.  
\qedhere 
\end{proof}

\begin{lemma}
\label{lem:stenv-red:trivial-5}
~
\begin{enumerate}
\item 
If $\stEnv
        \stEnvQTMoveQueueAnnot{\roleP}{\roleQ}{\stLab[k]}
      \stEnvi
      $, then for any channel with role $\mpC \in \dom{\stEnv}$ with $\mpC \neq \mpChanRole{\mpS}{\roleP}$, 
      $\stEnvApp{\stEnv}{\mpC} = \stEnvApp{\stEnvi}{\mpC}$.

\item  If 
$\stEnv \stEnvQTMoveRecvAnnot{\roleP}{\roleQ}{\stLab[k]}  \stEnvi$,  
 then for any channel with role $\mpC \in \dom{\stEnv}$ 
 with $\mpC \neq \mpChanRole{\mpS}{\roleP}$ and 
 $\mpC \neq \mpChanRole{\mpS}{\roleQ}$, 
        $\stEnvApp{\stEnv}{\mpC} = \stEnvApp{\stEnvi}{\mpC}$.
\end{enumerate}
\end{lemma}
\begin{proof}
By induction on typing environment reductions.  
\qedhere 
\end{proof}

\subsection{Relating Semantics via Association}
\label{sec:proof:relating}

\begin{lemma}[Relating Terminations]
\label{proof:gt-end-assoc}
\label{lem:gt-end-assoc}
If $\gtG = \gtEnd$ and $\stEnvAssoc{\gtWithTime{\cVal}{\gtG}}{\stEnv}{\mpS}$, 
then 
$\forall \mpChanRole{\mpS}{\roleP} \in \dom{\stEnv}: \stEnvApp{\stEnv}{\mpChanRole{\mpS}{\roleP}} = \stMPair{\stCPair{\cVal[\roleP]}{\stEnd}}{\stQEmptyType}$. 
\end{lemma}
\begin{proof}
By the definition of association~(\Cref{def:assoc}), we know that 
$\stEnv =  \stEnv[\gtG] \stEnvComp \stEnv[\Delta] \stEnvComp \stEnv[\stEnd]$, where, by the hypothesis $\gtG = \gtEnd$, 
$\dom{\stEnv[\gtG]} = \dom{\stEnv[\Delta]} = \emptyset$. Hence, $\stEnv =  \stEnv[\stEnd]$, which is the thesis. 
\end{proof}

\begin{lemma}[Inversion of Projection]
\label{lem:inv-proj}
  Given a timed local type $\stT$, which is a supertype of projection from a timed global type
  $\gtG$ on a role $\roleP$, 
  \ie $\gtProj{\gtG}{\roleP} \stSub \stT$,
  then:
  \begin{enumerate}[label=(\arabic*)]
    \item
      If 
      $\unfoldOne{\stT}
      = \stIntSum{\roleQ}{i \in I}{\stTChoice{\stLab[i]}{\stS[i]}{\ccst[i], \crst[i]} \stSeq \stT[i]}$, 
     then either
        \begin{enumerate}
        \item
          $\unfoldOne{\gtG} =
           \gtCommT{\roleP}{\roleQ}{i \in I'}{\gtLabi[i]}{\stSi[i]}{\ccstOi[i], \crstOi[i], \ccstIi[i], \crstIi[i]}{\gtG[i]}$, 
          where $I \subseteq I'$, and for all $i \in I: \stLab[i] = \gtLabi[i]$, 
          $\stS[i] \stSub \stSi[i]$, $\ccst[i] = \ccstOi[i]$, $\crst[i] = \crstOi[i]$, 
          and 
          $\gtProj{\gtG[i]}{\roleP} \stSub \stT[i]$; 
          or,
        \item
          $\unfoldOne{\gtG} =
            \gtCommT{\roleS}{\roleT}{j \in
            J}{\gtLabi[j]}{\stSi[j]}{\ccstOi[j], \crstOi[j], \ccstIi[j], \crstIi[j]}{\gtG[j]}$, or 
         $\unfoldOne{\gtG} = \gtCommTTransit{\roleS}{\roleT}{j \in
            J}{\gtLabi[j]}{\stSi[j]}{\ccstOi[j], \crstOi[j], \ccstIi[j], \crstIi[j]}{\gtG[j]}{k}$, 
          where for all $j \in J: \gtProj{\gtG[j]}{\roleP} \stSub \stT$,
          with $\roleP \neq \roleS$ and $\roleP \neq \roleT$.
        \end{enumerate}
    \item
      If 
      $\unfoldOne{\stT}
      =
      \stExtSum{\roleQ}{i \in I}{\stTChoice{\stLab[i]}{\stS[i]}{\ccst[i], \crst[i]} \stSeq
        \stT[i]}$, then either
      \begin{enumerate}
        \item
          $\unfoldOne{\gtG} =
             \gtCommT{\roleQ}{\roleP}{i \in I'}{\gtLabi[i]}{\stSi[i]}{\ccstOi[i], \crstOi[i], \ccstIi[i], \crstIi[i]}{\gtG[i]}$, or 
             $\unfoldOne{\gtG} =
             \gtCommTTransit{\roleQ}{\roleP}{i \in I'}{\gtLabi[i]}{\stSi[i]}{\ccstOi[i], \crstOi[i], \ccstIi[i], \crstIi[i]}{\gtG[i]}{j}$, 
          where $I' \subseteq I$, and for all $i \in I': \stLab[i] = \gtLabi[i]$, 
           $\stSi[i] \stSub \stS[i]$, $\ccst[i] = \ccstIi[i]$, $\crst[i] = \crstIi[i]$, and  
          $\gtProj{\gtG[i]}{\roleP} \stSub \stT[i]$; 
          or,
        \item
          $\unfoldOne{\gtG} =
            \gtCommT{\roleS}{\roleT}{j \in
            J}{\gtLabi[j]}{\stSi[j]}{\ccstOi[j], \crstOi[j], \ccstIi[j], \crstIi[j]}{\gtG[j]}$, or 
         $\unfoldOne{\gtG} = \gtCommTTransit{\roleS}{\roleT}{j \in
            J}{\gtLabi[j]}{\stSi[j]}{\ccstOi[j], \crstOi[j], \ccstIi[j], \crstIi[j]}{\gtG[j]}{k}$, 
          where for all $j \in J: \gtProj{\gtG[j]}{\roleP} \stSub \stT$,
          with $\roleP \neq \roleS$ and $\roleP \neq \roleT$.
      \end{enumerate}
    \item
      If 
      $\unfoldOne{\stT} = \stEnd$, then $\roleP \notin \gtRoles{\gtG}$. 
  \end{enumerate}
\end{lemma}
\begin{proof}
By the definition of timed global type projection~(\Cref{def:global-proj}); additionally, for cases (b) in 
items (1) and (2), apply~\Cref{lem:merge-subtyping} and the transitivity of subtyping. 
\qedhere 
\end{proof}

\begin{lemma}[Matching Communication Under Projection]
\label{lem:comm-match}
  If two timed local types $\stT, \stU$ are supertype %
  of an internal choice and an
  external choice with matching roles, obtained via projection from a timed global
  type $\gtG$, \ie
  $ \gtProj{\gtG}{\roleP} \stSub 
    \unfoldOne{\stT} 
    =
      \stIntSum{\roleQ}{i \in I_{\roleP}}{\stTChoice{\stLab[i]}{\stS[i]}{\ccst[i], \crst[i]} \stSeq \stT[i]} 
  $
  and
  $ \gtProj{\gtG}{\roleQ} \stSub 
    \unfoldOne{\stU}
    =
    \stExtSum{\roleP}{j \in I_{\roleQ}}{\stTChoice{\stLabi[j]}{\stSi[j]}{\ccsti[j], \crsti[j]} \stSeq \stTi[j]} 
  $, then
      $I_{\roleP} \subseteq I_{\roleQ}$, and
      $\forall i \in I_{\roleP}: \stLab[i] = \stLabi[i]$ and $\stS[i] \stSub \stSi[i]$.
\end{lemma}
\begin{proof}
  By induction on items (1) and (2) of \cref{lem:inv-proj} simultaneously.
  \begin{enumerate}[label=(\alph*)]
    \item
      We have $\unfoldOne{\gtG} = 
      \gtCommT{\roleP}{\roleQ}{i \in I}{\gtLabii[i]}{\stSii[i]}{\ccstOii[i], \crstOii[i], \ccstIii[i], \crstIii[i]}{\gtG[i]}$,
      $I_{\roleP} \subseteq I \subseteq I_{\roleQ}$,
      $\forall i \in I_{\roleP}: \stLab[i] = \gtLabii[i]$ and $\stS[i] \stSub \stSii[i]$, and
      $\forall i \in I: \gtLabii[i] = \stLabi[i]$ and $\stSii[i] \stSub \stSi[i]$.
      We have $I_{\roleP} \subseteq I_{\roleQ}$~(by transitivity of $\subseteq$),
      and $\forall i \in I_{\roleP}: \stLab[i] = \stLabi[i]$~(by transitivity of $=$) and $\stS[i] \stSub \stSi[i]$~(by transitivity of
      $\stSub$).
    \item
      We have $\unfoldOne{\gtG} = \gtCommT{\roleS}{\roleT}
        {j \in J}{\gtLabii[j]}{\stSii[j]}{\ccstOii[i], \crstOii[i], \ccstIii[i], \crstIii[i]}{\gtG[j]}$, or 
        $\unfoldOne{\gtG} = \gtCommTTransit{\roleS}{\roleT}{j \in
            J}{\gtLabii[j]}{\stSii[j]}{\ccstOii[j], \crstOii[j], \ccstIii[j], \crstIii[j]}{\gtG[j]}{k}$, 
      where for all $j \in J: \gtProj{\gtG[j]}{\roleP} \stSub \stT$,
      $\gtProj{\gtG[j]}{\roleQ} \stSub \stU$,
      $\setenum{\roleP, \roleQ} \cap \setenum{\roleS, \roleT} = \emptyset$.
      Apply induction on
      $\gtProj{\gtG[j]}{\roleP} \stSub \stT$ and
      $\gtProj{\gtG[j]}{\roleP} \stSub \stU$ on any $j \in J$.
   \qedhere
  \end{enumerate}
\end{proof}

\begin{lemma}[Matching Communication Under Association]
\label{lem:comm-assoc-match}
Let $\stEnv[\Delta]$ be a typing environment containing only queue types for queries. If $\stEnv[\Delta]$ is associated with a timed global type $\gtG$ for $\mpS$ with 
$\stEnvApp{\stEnv[\Delta]}{\mpChanRole{\mpS}{\roleP}} = \stQCons{
              \stQMsg{\roleQ}{\stLab}{\stS}%
            }{%
              \stQType%
            }$, and a timed local type $\stT[\roleQ]$ is a supertype of an external choice with matching roles, obtained via projection from $\gtG$, \ie  $\gtProj{\gtG}{\roleQ} \stSub \unfoldOne{\stT[\roleQ]} = \stExtSum{\roleP}{i \in I_{\roleQ}}{\stTChoice{\stLabi[i]}{\stSi[i]}{\ccsti[i], \crsti[i]} \stSeq \stTi[i]}$, then $\exists k \in I_{\roleQ}: \stLab = \stLabi[k]$ and $\stS \stSub \stSi[k]$.  
\end{lemma}
\begin{proof}
  By induction on item (2) of \cref{lem:inv-proj} simultaneously.
  \begin{enumerate}[label=(\alph*)]
    \item
      We have either $\unfoldOne{\gtG} = 
      \gtCommT{\roleP}{\roleQ}{i \in I}{\gtLabii[i]}{\stSii[i]}{\ccstOii[i], \crstOii[i], \ccstIii[i], \crstIii[i]}{\gtG[i]}$, or 
      $\unfoldOne{\gtG} = 
      \gtCommTTransit{\roleP}{\roleQ}{i \in I}{\gtLabii[i]}{\stSii[i]}{\ccstOii[i], \crstOii[i], \ccstIii[i], \crstIii[i]}{\gtG[i]}{j}$, where 
      $I \subseteq I_{\roleQ}$, and 
      $\forall i \in I: \stLabi[i] = \gtLabii[i]$ and $\stSii[i] \stSub \stSi[i]$. 
    
    By $\stEnv[\Delta]$ being associated with $\gtG$ for $\mpS$  and  
    $\stEnvApp{\stEnv[\Delta]}{\mpChanRole{\mpS}{\roleP}} = 
\stQCons{
              \stQMsg{\roleQ}{\stLab}{\stS}%
            }{%
              \stQType%
            }$, we know that $\unfoldOne{\gtG} \neq \gtCommT{\roleP}{\roleQ}{i \in I}{\gtLabii[i]}{\stSii[i]}{\ccstOii[i], \crstOii[i], \ccstIii[i], \crstIii[i]}{\gtG[i]}$. Otherwise, $\roleQ \notin  \operatorname{receivers}(\stQCons{
              \stQMsg{\roleQ}{\stLab}{\stS}%
            }{%
              \stQType%
            })$, a desired contradiction.  It follows directly that $\unfoldOne{\gtG} = 
      \gtCommTTransit{\roleP}{\roleQ}{i \in I}{\gtLabii[i]}{\stSii[i]}{\ccstOii[i], \crstOii[i], \ccstIii[i], \crstIii[i]}{\gtG[i]}{j}$, and furthermore, by the association,  $\stLab = \gtLabii[j]$ and $\stS \stSub \stSii[j]$. 
  We have $\stLab = \stLabi[j]$~(by transitivity of $=$) and $\stS \stSub \stSi[j]$~(by transitivity of
      $\stSub$), which is the thesis.  
    \item
      We have $\unfoldOne{\gtG} = \gtCommT{\roleS}{\roleT}
        {j \in J}{\gtLabii[j]}{\stSii[j]}{\ccstOii[i], \crstOii[i], \ccstIii[i], \crstIii[i]}{\gtG[j]}$, or 
        $\unfoldOne{\gtG} = \gtCommTTransit{\roleS}{\roleT}{j \in
            J}{\gtLabii[j]}{\stSii[j]}{\ccstOii[j], \crstOii[j], \ccstIii[j], \crstIii[j]}{\gtG[j]}{k}$, 
      where for all $j \in J: \gtProj{\gtG[j]}{\roleQ} \stSub \stT[\roleQ]$,
      $\setenum{\roleP, \roleQ} \cap \setenum{\roleS, \roleT} = \emptyset$. 
      In the case of $\unfoldOne{\gtG} = \gtCommT{\roleS}{\roleT}
        {j \in J}{\gtLabii[j]}{\stSii[j]}{\ccstOii[i], \crstOii[i], \ccstIii[i], \crstIii[i]}{\gtG[j]}$, $\forall j \in J$, $\stEnv[\Delta]$ is associated with $\gtG[j]$ for $\mpS$. In the case of $\unfoldOne{\gtG} = \gtCommTTransit{\roleS}{\roleT}{j \in
            J}{\gtLabii[j]}{\stSii[j]}{\ccstOii[j], \crstOii[j], \ccstIii[j], \crstIii[j]}{\gtG[j]}{k}$, $\stEnvi[\Delta]$ is associated with $\gtG[k]$ for $\mpS$ with $\stEnvApp{\stEnvi[\Delta]}{\mpChanRole{\mpS}{\roleP}} = \stQCons{
              \stQMsg{\roleQ}{\stLab}{\stS}%
            }{%
              \stQType%
            }$.  
      Apply induction on
      $\gtProj{\gtG[k]}{\roleP} \stSub \stT[\roleQ]$ and $\stEnv[\Delta]$ (or $\stEnvi[\Delta]$) being associated with 
      $\gtG[k]$.  
   \qedhere
  \end{enumerate}
\end{proof}

\begin{lemma}
\label{lem:time_assoc}
If $\stEnvAssoc{\gtWithTime{\cVal}{\gtG}}{\stEnv}{\mpS}$, then $\stEnvAssoc{\gtWithTime{\cVal + t}{\gtG}}{\stEnv + t}{\mpS}$. 
\end{lemma}
\begin{proof}
By the definitions of association~(\Cref{def:assoc}) and $\stEnvQ + t$ (in~\Cref{sec:aat-mpst-typing-system}). 
\qedhere
\end{proof}

\begin{lemma}
\label{lem:stenv_time_action}
$\stEnv \stEnvQTMoveTimeAnnot  \stEnvQ + t$. 
\end{lemma}
\begin{proof}
By the definition of $\stEnvQ + t$, and applying \inferrule{\iruleTCtxTimeSession}, 
\inferrule{\iruleTCtxTimeQ}, \inferrule{\iruleTCtxTimeCombined}, or \inferrule{\iruleTCtxTime} when necessary. 
\qedhere 
\end{proof}

\begin{lemma}
\label{lem:time_assoc_sound}
Given associated timed global type $\gtWithTime{\cVal}{\gtG}$ and typing environment $\stEnvQ$: 
$\stEnvAssoc{\gtWithTime{\cVal}{\gtG}}{\stEnv}{\mpS}$. 
If $\gtWithTime{\cVal}{\gtG} \gtMove[\timeLab] \gtWithTime{\cVali}{\gtGi}$, then there exists $\stEnvi$ such that 
$\stEnv \stEnvQTMoveTimeAnnot  \stEnvQi$ and $\stEnvAssoc{\gtWithTime{\cVali}{\gtGi}}{\stEnvi}{\mpS}$. 
\end{lemma}
\begin{proof}
By \inferrule{\iruleGtMoveTime}, we have $\cVali$ = $\cVal + t$ and $\gtGi = \gtG$. The thesis then holds by~\cref{lem:stenv_time_action,lem:time_assoc}. 
\end{proof}

\begin{lemma}
\label{lem:time_assoc_comp}
Given associated timed global type $\gtWithTime{\cVal}{\gtG}$ and typing environment $\stEnvQ$: 
$\stEnvAssoc{\gtWithTime{\cVal}{\gtG}}{\stEnv}{\mpS}$. 
If $\stEnv \stEnvQTMoveTimeAnnot \stEnvQi$, then there exists $\gtWithTime{\cVali}{\gtGi}$ such that 
$\gtWithTime{\cVal}{\gtG} \gtMove[\timeLab] \gtWithTime{\cVali}{\gtGi}$ and 
$\stEnvAssoc{\gtWithTime{\cVali}{\gtGi}}{\stEnvi}{\mpS}$. 
\end{lemma}
\begin{proof}
By~\Cref{lem:stenv_time_action}, we have $\stEnvi = \stEnv + t$.  
The thesis then holds by \inferrule{\iruleGtMoveTime} and~\Cref{lem:time_assoc}.  
\qedhere 
\end{proof}

\lemSoundProj*
\begin{proof}
By induction on reductions of timed global type $\gtWithTime{\cVal}{\gtG}
  \gtMove[\stEnvAnnotGenericSym]
  \gtWithTime{\cValiii}{\gtGiii}$. 
  We develop several particularly interesting cases in detail. 
 \begin{itemize}[left=0pt, topsep=0pt]
 \item Case~\inferrule{\iruleGtMoveTime}: the thesis holds by~\Cref{lem:time_assoc_sound}. 
 
 \item Case~\inferrule{\iruleGtMoveOut}: from the premise, we have: 
      \begin{gather}
        \stEnvAssoc{\gtWithTime{\cVal}{\gtG}}{\stEnv}{\mpS}
        \label{eq:soundness:send_assoc_pre}
        \\
        \gtG =
           \gtCommT{\roleP}{\roleQ}{i \in I}{\gtLab[i]}{\stS[i]}{\ccstO[i], \crstO[i], \ccstI[i], \crstI[i]}{\gtG[i]}
          \\
          \stEnv = \stEnv[\gtG] \stEnvComp \stEnv[\Delta] \stEnvComp \stEnv[\stEnd]
          \label{eq:soundness:send_assoc_envs}
          \\
       \stEnvAnnotGenericSym = \stEnvQTQueueAnnotSmall{\roleP}{\roleQ}{\gtLab[j]} 
        \\
        j \in I
        \\
        \gtGiii =   \gtCommTTransit{\roleP}{\roleQ}{i \in
          I}{\gtLab[i]}{\stS[i]}{\ccstO[i], \crstO[i], \ccstI[i], \crstI[i]}{\gtG[i]}{j}
          \\
          \cValiii = \cValUpd{\cVal}{\crstO[j]}{0}
      \end{gather}
By association~\eqref{eq:soundness:send_assoc_pre} and~\eqref{eq:soundness:send_assoc_envs}, we have that $\forall \mpChanRole{\mpS}{\roleP} \in \dom{\stEnv[\gtG]}: \stEnvApp{\stEnv[\gtG]}{\mpChanRole{\mpS}{\roleP}} =  \stCPair{\cVal[\roleP]}{\stT[\roleP]}$, and moreover, $\gtProj{\gtG}{\roleP} =  
\stIntSum{\roleQ}{i \in I}{
       \stTChoice{\stLab[i]}{\stS[i]}{\ccstO[i], \crstO[i]} \stSeq (\gtProj{\gtG[i]}{\roleP})
     } \stSub \stT[\roleP]$,  
$\gtProj{\gtG}{\roleQ} =  
\stExtSum{\roleP}{i \in I}{
       \stTChoice{\stLab[i]}{\stS[i]}{\ccstI[i], \crstI[i]} \stSeq (\gtProj{\gtG[i]}{\roleP})
     } \stSub \stT[\roleQ]$, and $\gtProj{\gtG}{\roleR} \stSub \stT[\roleR]$ with $\roleR \neq \roleP$ and $\roleR \neq \roleQ$. 
     Then by~\Cref{lem:subtyping-invert} and~\Cref{lem:comm-match}, we have 
     $\unfoldOne{\stT[\roleP]} = 
    \stIntSum{\roleQ}{i \in I_{\roleP}}{\stTChoice{\stLabi[i]}{\stSi[i]}{\ccsti[i], \crsti[i]} \stSeq \stTi[i]}$  
    and 
    $\unfoldOne{\stT[\roleQ]} = 
    \stExtSum{\roleP}{i \in I_{\roleQ}}{\stTChoice{\stLabii[i]}{\stSii[i]}{\ccstii[i], \crstii[i]}\stSeq \stTii[i]}$, with  
    $I_{\roleP} \subseteq 
    I \subseteq I_{\roleQ}$, and for all $i \in I_{\roleP}: \gtLab[i] = \stLabi[i] = \stLabii[i]$, $\ccsti[i] = \ccstO[i]$, 
    $\crsti[i] = \crstO[i]$, $\ccstii[i] = \ccstI[i]$, 
    $\crstii[i] = \crstI[i]$,  
    $\stSi[i] \stSub \stSii[i]$,   
    $\gtProj{\gtG[i]}{\roleP} \stSub \stTi[i]$, and $\gtProj{\gtG[i]}{\roleQ} \stSub \stTii[i]$.

    Now let us choose some 
     $k \in I_{\roleP}$ such that  
     $\stEnvAnnotGenericSymi = \stEnvQTQueueAnnotSmall{\roleP}{\roleQ}{\gtLab[k]}$. 
     Moreover, we set $\cVali = \cValUpd{\cVal}{C_{\roleP}}{t}$ such that $\cVali \models \ccstO[k]$, 
     and $\stEnvi = \stEnvi[\gtG] \stEnvComp \stEnvi[\Delta] \stEnvi[\stEnd]$ such that $\dom{\stEnv[\gtG]} = \dom{\stEnvi[\gtG]}$, 
     $\forall \mpChanRole{\mpS}{\roleR} \in \dom{\stEnvi[\gtG]}$ with $\roleR \neq \roleP$:  $\stEnvApp{\stEnvi[\gtG]}{\mpChanRole{\mpS}{\roleR}} =  \stEnvApp{\stEnv[\gtG]}{\mpChanRole{\mpS}{\roleR}}$, 
     $\stEnvApp{\stEnvi[\gtG]}{\mpChanRole{\mpS}{\roleP}} =  \stCPair{\cVali[\roleP]}{\stT[\roleP]}$ with 
     $\cVali[\roleP] = \cValUpd{\cVal[\roleP]}{C_{\roleP}}{t}$, $\stEnvi[\Delta] = \stEnv[\Delta]$, and $\stEnvi[\stEnd] = \stEnv[\stEnd]$. It is trivial that 
     $\stEnvAssoc{\gtWithTime{\cVali}{\gtG}}{\stEnvi}{\mpS}$. 
     
      We have 
     $\gtWithTime{\cVali}{\gtG} \gtMove[\stEnvAnnotGenericSymi] \gtWithTime{\cValii}{\gtGii}$ 
     with  $\gtGii =   \gtCommTTransit{\roleP}{\roleQ}{i \in
          I}{\gtLab[i]}{\stS[i]}{\ccstO[i], \crstO[i], \ccstI[i], \crstI[i]}{\gtG[i]}{k}$ and 
          $\cValii = \cValUpd{\cVali}{\crstO[k]}{0}$. 
  We are left to show that there exists $\stEnvii$ such that 
  $\stEnvi
  \stEnvQTMoveQueueAnnot{\roleP}{\roleQ}{\stLab[k]} \stEnvii$ and 
      $\stEnvAssoc{\gtWithTime{\cValii}{\gtGii}}{\stEnvii}{\mpS}$.

Since $\cVali[\roleP] \models \ccstO[k]$, we can apply \inferrule{\iruleTCtxOut} 
on $\mpChanRole{\mpS}{\roleP}$ (and \inferrule{\iruleTCtxCongX}, \inferrule{\iruleTCtxCongCombined} when needed) 
 to get $\stEnvi
  \stEnvQTMoveQueueAnnot{\roleP}{\roleQ}{\stLab[k]} \stEnvii$, such that $\stEnvii = \stEnvii[\gtG] \stEnvComp \stEnvii[\Delta] \stEnvComp \stEnvii[\stEnd]$ with $\stEnvii[\gtG] = \stEnvUpd{\stEnvi[\gtG]}{\mpChanRole{\mpS}{\roleP}}{\stCPair{\cValUpd{\cVali[\roleP]}{\crstO[k]}{0}}{\stTi[k]}}$,  $\stEnvii[\Delta] = \stEnvUpd{\stEnvi[\Delta]}{\mpChanRole{\mpS}{\roleP}}{\stQCons{\stQType}{
            \stQCons{
              \stQMsg{\roleQ}{\stLabi[k]}{\stSi[k]}
            }{
              \stQEmptyType
            }}}$, and $\stEnvii[\stEnd] = \stEnvi[\stEnd]$. 
            
            Finally, we need to show:
            \begin{itemize}[left=0pt, topsep=0pt]
            \item $\stEnvii[\gtG]$ is associated with $\gtWithTime{\cValii}{\gtGii}$ for $\mpS$: 
            \begin{enumerate}[label=(\roman*), left=0pt, topsep=0pt]
            \item trivial;  
            \item trivial; 
            \item by $\gtProj{\gtGii}{\roleP} = \gtProj{\gtG[k]}{\roleP} \stSub \stTi[k]$, 
            $\gtProj{\gtGii}{\roleQ} = \gtProj{\gtG}{\roleQ} \stSub \stT[\roleQ]$, and 
            $\gtProj{\gtGii}{\roleR} = \gtProj{\gtG}{\roleR} \stSub \stT[\roleR]$ with $\roleR \neq \roleP$ and $\roleR \neq \roleQ$; 
            \item by $\cValii = \cValUpd{\cVali}{\crstO[k]}{0}$,  $\cValii[\roleP] = \cValUpd{\cVali[\roleP]}{\crstO[k]}{0}$, and 
            $\cVali = \sqcup_{\roleP \in \gtRoles{\gtG}}\cVali[\roleP]$. 
            \end{enumerate}
            
            \item $\stEnvii[\Delta]$ is associated with $\gtGii$ for $\mpS$: 
           \begin{enumerate}[label=(\roman*), left=0pt, topsep=0pt]
           \item trivial; 
           \item trivial; 
           \item not applicable; 
           \item not applicable; 
           \item since $\stEnvi[\Delta]$ is associated with $\gtG$ for $\mpS$, we have that $\roleQ \notin \operatorname{receivers}(\stEnvApp{\stEnvi[\Delta]}{\mpChanRole{\mpS}{\roleP}})$, and $\forall i \in I$: $\stEnvi[\Delta]$ is associated with $\gtG[i]$ for $\mpS$. Then the thesis  holds by $\stEnvApp{\stEnvii[\Delta]}{\mpChanRole{\mpS}{\roleP}} = 
         \stQCons{\stQType}{
            \stQCons{
              \stQMsg{\roleQ}{\stLabi[k]}{\stSi[k]}
            }{
              \stQEmptyType
            }} \stEquiv  \stQCons{
              \stQMsg{\roleQ}{\stLabi[k]}{\stSi[k]}
            }{\stQCons{\stQType}{
              \stQEmptyType
            }}$ with $\stSi[k] \stSub \stS[k]$, $\stLabi[k] = \gtLab[k]$, and 
          $\stEnvUpd{\stEnvii[\Delta]}{\mpChanRole{\mpS}{\roleP}}{\stQCons{\stQType}{\stQEmptyType}} = \stEnvi[\Delta]$, which is associated with $\gtG[k]$ for $\mpS$.  
           \end{enumerate}
   \item $\stEnvii[\stEnd] = \stEnv[\stEnd]$.          
            \end{itemize}

 \item Case~\inferrule{\iruleGtMoveIn}: 
  from the premise, we have: 
      \begin{gather}
        \stEnvAssoc{\gtWithTime{\cVal}{\gtG}}{\stEnv}{\mpS}
        \label{eq:soundness:rcv_assoc_pre}
        \\
        \gtG =
           \gtCommTTransit{\roleP}{\roleQ}{i \in
          I}{\gtLab[i]}{\stS[i]}{\ccstO[i], \crstO[i], \ccstI[i], \crstI[i]}{\gtG[i]}{j}
         \\
          \stEnv = \stEnv[\gtG] \stEnvComp \stEnv[\Delta] \stEnvComp \stEnv[\stEnd]
          \label{eq:soundness:rcv_assoc_envs}
          \\
       \stEnvAnnotGenericSym = \stEnvQRecvAnnotSmall{\roleQ}{\roleP}{\gtLab[j]}
        \\
        j \in I
        \\
        \gtGi =   \gtG[j]
        \\
          \cVali = \cValUpd{\cVal}{\crstI[j]}{0}
      \end{gather}
By association~\eqref{eq:soundness:rcv_assoc_pre} and~\eqref{eq:soundness:rcv_assoc_envs}, we have that $\forall \mpChanRole{\mpS}{\roleP} \in \dom{\stEnv[\gtG]}: \stEnvApp{\stEnv[\gtG]}{\mpChanRole{\mpS}{\roleP}} =  \stCPair{\cVal[\roleP]}{\stT[\roleP]}$, and moreover, $\gtProj{\gtG}{\roleP} =  
\gtProj{\gtG[j]}{\roleP} \stSub \stT[\roleP]$,  
$\gtProj{\gtG}{\roleQ} =  
\stExtSum{\roleP}{i \in I}{
       \stTChoice{\stLab[i]}{\stS[i]}{\ccstI[i], \crstI[i]} \stSeq (\gtProj{\gtG[i]}{\roleP})
     } \stSub \stT[\roleQ]$, and $\gtProj{\gtG}{\roleR} \stSub \stT[\roleR]$ with $\roleR \neq \roleP$ and $\roleR \neq \roleQ$.  
     Then by~\Cref{lem:subtyping-invert}, we have 
    $\unfoldOne{\stT[\roleQ]} = 
    \stExtSum{\roleP}{i \in I_{\roleQ}}{\stTChoice{\stLabii[i]}{\stSii[i]}{\ccstii[i], \crstii[i]}\stSeq \stTii[i]}$, with  
    $I \subseteq I_{\roleQ}$, and for all $i \in I: \gtLab[i] =  \stLabii[i]$, $\ccstii[i] = \ccstI[i]$, 
    $\crstii[i] = \crstI[i]$,  
    $\stS[i] \stSub \stSii[i]$,   and 
  $\gtProj{\gtG[i]}{\roleQ} \stSub \stTii[i]$.

      We have 
     $\gtWithTime{\cVal}{\gtG} \gtMove[\stEnvAnnotGenericSym] \gtWithTime{\cVali}{\gtGi}$ 
     with  $\gtGi =   \gtG[j]$ and 
          $\cVali = \cValUpd{\cVal}{\crstI[j]}{0}$. 
  We are left to show that there exists $\stEnvi$ such that 
  $\stEnv
  \stEnvQTMoveRecvAnnot{\roleQ}{\roleP}{\stLab[j]} \stEnvi$ and 
      $\stEnvAssoc{\gtWithTime{\cVali}{\gtGi}}{\stEnvi}{\mpS}$. 
      Note that here we only consider the case regarding $\gtGi \neq \gtEnd$, as the case where $\gtGi = \gtEnd$ is trivial.

By the association~\eqref{eq:soundness:rcv_assoc_pre}, we have that $\cVal = \sqcup_{\roleP \in \gtRoles{\gtG}}\cVal[\roleP]$, and $\stEnvApp{\stEnv[\Delta]}{\mpChanRole{\mpS}{\roleP}} = 
         \stQCons{
              \stQMsg{\roleQ}{\stLab[j]}{\stSi[j]}
            }{\stQType}$ with $\stSi[j] \stSub \stS[j]$ and 
          $\stEnvUpd{\stEnv[\Delta]}{\mpChanRole{\mpS}{\roleP}}{\stQType}$ being associated with $\gtG[j]$ for $\mpS$. 
Hence, we can apply \inferrule{\iruleTCtxIn} 
on $\mpChanRole{\mpS}{\roleP}$ and $\mpChanRole{\mpS}{\roleQ}$~(and \inferrule{\iruleTCtxCongX}, \inferrule{\iruleTCtxCongCombined} when needed) 
 to get $\stEnv
  \stEnvQTMoveRecvAnnot{\roleQ}{\roleP}{\stLab[j]} \stEnvi$, 
  such that $\stEnvi = \stEnvi[\gtG] \stEnvComp \stEnvi[\Delta]$ with $\stEnvi[\gtG] = 
  \stEnvUpd{\stEnv[\gtG]}{\mpChanRole{\mpS}{\roleQ}}{\stCPair{\cValUpd{\cVal[\roleQ]}{\crstI[j]}{0}}{\stTii[j]}}$,  $\stEnvi[\Delta] = \stEnvUpd{\stEnv[\Delta]}{\mpChanRole{\mpS}{\roleP}}{\stQType}$, and $\stEnvi[\stEnd] = \stEnv[\stEnd]$.             
            
  Finally, we need to show:
            \begin{itemize}[left=0pt, topsep=0pt]
            \item $\stEnvi[\gtG]$ is associated with $\gtWithTime{\cVali}{\gtGi}$ for $\mpS$: 
            \begin{enumerate}[label=(\roman*), left=0pt, topsep=0pt]
            \item trivial;  
            \item trivial; 
            \item by $\gtProj{\gtGi}{\roleP} = \gtProj{\gtG[j]}{\roleP} \stSub \stT[\roleP]$, 
            $\gtProj{\gtGi}{\roleQ} = \gtProj{\gtG[j]}{\roleQ} \stSub \stTii[j]$, and 
            $\gtProj{\gtGi}{\roleR} = \gtProj{\gtG[j]}{\roleR} \stSub \stMerge{i \in I}{\gtProj{\gtG[i]}{\roleR}} = \gtProj{\gtG}{\roleR} \stSub \stT[\roleR]$ with $\roleR \neq \roleP$ and $\roleR \neq \roleQ$; 
            \item by $\cVali = \cValUpd{\cVal}{\crstI[j]}{0}$,  $\cVali[\roleQ] = \cValUpd{\cVal[\roleQ]}{\crstI[j]}{0}$, and 
            $\cVal = \sqcup_{\roleP \in \gtRoles{\gtG}}\cVal[\roleP]$. 
            \end{enumerate}
            
            \item $\stEnvi[\Delta]$ is associated with $\gtGi$ for $\mpS$: directly from 
             $\stEnvi[\Delta] = \stEnvUpd{\stEnv[\Delta]}{\mpChanRole{\mpS}{\roleP}}{\stQType}$ and  $\stEnvUpd{\stEnv[\Delta]}{\mpChanRole{\mpS}{\roleP}}{\stQType}$ being associated with $\gtG[j]$ for $\mpS$.
             
             \item $\stEnvi[\stEnd] = \stEnv[\stEnd]$.  
                      
            \end{itemize}

 \item Case~\inferrule{\iruleGtMoveRec}: by inductive hypothesis.

 \item Cases~\inferrule{\iruleGtMoveCtx}, \inferrule{\iruleGtMoveCtxi}: 
 similar to cases [GR4], [GR5] from Theorem 3.1 in~\cite{ICALP13CFSM}. 
 \qedhere 
 \end{itemize}

\end{proof}

\lemCompProj*
\begin{proof}
By induction on reductions of typing environment $\stEnv \stEnvQTMoveGenAnnotT \stEnvi$. 
We develop several particularly interesting cases in detail.  
 \begin{itemize}[left=0pt, topsep=0pt]
\item  Case $\stEnv \stEnvQTMoveAnnot{\cUnit} \stEnvi$: the thesis holds by~\Cref{lem:stenv_time_action,lem:time_assoc_comp}.

\item Case \inferrule{\iruleTCtxSend}: 
from the premise, we have: 
\begin{gather}
\stEnvAssoc{\gtWithTime{\cVal}{\gtG}}{\stEnv}{\mpS}
\label{eq:complete:send_assoc_pre}
 \\
 \stEnvAnnotGenericSym = \stEnvQTQueueAnnotSmall{\roleP}{\roleQ}{\stLab[j]} 
 \\
  \stEnv
    \stEnvQTMoveQueueAnnot{\roleP}{\roleQ}{\stLab[j]}    
    \stEnvi
    \label{eq:stenvp-send}
    \\
    \stEnv = \stEnv[\gtG] \stEnvComp \stEnv[\Delta] \stEnvComp \stEnv[\stEnd]
\label{eq:stenvp-sned-type-env}
\end{gather} 
By association~\eqref{eq:complete:send_assoc_pre} and~\eqref{eq:stenvp-sned-type-env}, 
we have that $\forall \mpChanRole{\mpS}{\roleP} \in \dom{\stEnv[\gtG]}: \stEnvApp{\stEnv[\gtG]}{\mpChanRole{\mpS}{\roleP}} =  \stCPair{\cVal[\roleP]}{\stT[\roleP]}$ and $\gtProj{\gtG}{\roleP} \stSub \stT[\roleP]$,  $\forall \mpChanRole{\mpS}{\roleP} \in \dom{\stEnv[\Delta]}: \stEnvApp{\stEnv[\Delta]}{\mpChanRole{\mpS}{\roleP}} =  \stQType[\roleP]$, and $\forall \mpChanRole{\mpS}{\roleP} \in \dom{\stEnv[\stEnd]}: \stEnvApp{\stEnv[\stEnd]}{\mpChanRole{\mpS}{\roleP}} =  \stMPair{\stCPair{\cVal[\roleP]}{\stEnd}}{\stQEmptyType}$. Moreover,  
by inverting \inferrule{\iruleTCtxSend} on~\eqref{eq:stenvp-send} (and inverting and applying \inferrule{\iruleTCtxRec} when necessary), we have $
    \unfoldOne{\stT[\roleP]} =
    \stIntSum{\roleQ}{i \in I_{\roleP}}{\stTChoice{\stLab[
    i]}{\stS[i]}{\ccst[i], \crst[i]} \stSeq \stT[i]}
  $. 

We perform case analysis on~\cref{lem:inv-proj}.
\begin{itemize}[left=0pt, topsep=0pt]
  \item  Item {\bf a} in  Case 1 of \cref{lem:inv-proj}: $\unfoldOne{\gtG} =
           \gtCommT{\roleP}{\roleQ}{i \in I'}{\gtLabi[i]}{\stSi[i]}{\ccstOi[i], \crstOi[i], \ccstIi[i], \crstIi[i]}{\gtG[i]}$, 
          where $I_{\roleP} \subseteq I'$, and for all $i \in I_{\roleP}: \stLab[i] = \gtLabi[i]$, 
          $\stS[i] \stSub \stSi[i]$, $\ccst[i] = \ccstOi[i]$, $\crst[i] = \crstOi[i]$, 
          and 
          $\gtProj{\gtG[i]}{\roleP} \stSub \stT[i]$. By~\eqref{eq:stenvp-send}, it holds that $\cVal[\roleP] \models \ccst[j]$, and hence 
          $\cVal \models \ccstO[j]$. Along with the fact that $j \in I_{\roleP} \subseteq I'$, 
          we can apply \inferrule{\iruleGtMoveOut} on $\gtWithTime{\cVal}{\gtG}$ to get $\gtWithTime{\cVali}{\gtGi}$ with 
  $\cVali = \cValUpd{\cVal}{\crstOi[j]}{0}$ and $\gtGi =  
  \gtCommTTransit{\roleP}{\roleQ}{i \in I'}{\gtLabi[i]}{\stSi[i]}{\ccstOi[i], \crstOi[i], \ccstIi[i], \crstIi[i]}{\gtG[i]}{j}$. We are left to 
  show association: $\stEnvAssoc{\gtWithTime{\cVali}{\gtGi}}{\stEnvi}{\mpS}$, where $\stEnvi = 
  \stEnvUpd{\stEnv[\gtG]}{\mpChanRole{\mpS}{\roleP}}{\stCPair{\cValUpd{\cVal[\roleP]}{\crst[j]}{0}}{\stT[j]}} \stEnvComp   \stEnvUpd{\stEnv[\Delta]}{\mpChanRole{\mpS}{\roleP}}{\stQCons{\stQType[\roleP]}{
            \stQCons{
              \stQMsg{\roleQ}{\stLab[j]}{\stS[j]}
            }{
              \stQEmptyType
            }}} \stEnvComp \stEnv[\stEnd]$. 
             \begin{itemize}[left=0pt, topsep=0pt]
            \item $ \stEnvUpd{\stEnv[\gtG]}{\mpChanRole{\mpS}{\roleP}}{\stCPair{\cValUpd{\cVal[\roleP]}{\crst[j]}{0}}{\stT[j]}}$ is associated with $\gtWithTime{\cVali}{\gtGi}$ for $\mpS$: 
            \begin{enumerate}[label=(\roman*), left=0pt, topsep=0pt]
            \item trivial;  
            \item trivial; 
            \item by $\gtProj{\gtGi}{\roleP} = \gtProj{\gtG[j]}{\roleP} \stSub \stT[j]$, 
            $\gtProj{\gtGi}{\roleQ} = \gtProj{\gtG}{\roleQ} \stSub \stT[\roleQ]$, and 
            $\gtProj{\gtGi}{\roleR} = \gtProj{\gtG}{\roleR} \stSub \stT[\roleR]$ with $\roleR \neq \roleP$ and $\roleR \neq \roleQ$; 
            \item by $\cVali = \cValUpd{\cVal}{\crstOi[j]}{0}$,  $\cVali[\roleP] = \cValUpd{\cVal[\roleP]}{\crst[j]}{0}$, $\crstOi[j] = \crst[j]$, and 
            $\cVal = \sqcup_{\roleP \in \gtRoles{\gtG}}\cVal[\roleP]$. 
            \end{enumerate}
            
            \item $\stEnvUpd{\stEnv[\Delta]}{\mpChanRole{\mpS}{\roleP}}{\stQCons{\stQType[\roleP]}{
            \stQCons{
              \stQMsg{\roleQ}{\stLab[j]}{\stS[j]}
            }{
              \stQEmptyType
            }}}$ is associated with $\gtGi$ for $\mpS$: directly from $\roleQ \notin \operatorname{receivers}(\stQType[\roleP])$, 
            $\stLab[j] = \gtLab[j]$, $\stS[j] \stSub \stSi[j]$, and $\stEnv[\Delta]$ being associated with $\gtG[j]$ (applying $\stEquiv$ when necessary). 
            \item $\stEnv[\stEnd]$ keeps unchanged. 
                     \end{itemize} 
  \item  Item {\bf b} in  Case 1 of \cref{lem:inv-proj}: $\unfoldOne{\gtG} =
            \gtCommT{\roleS}{\roleT}{j \in
            J}{\gtLabi[j]}{\stSi[j]}{\ccstOi[j], \crstOi[j], \ccstIi[j], \crstIi[j]}{\gtG[j]}$, or 
         $\unfoldOne{\gtG} = \gtCommTTransit{\roleS}{\roleT}{j \in
            J}{\gtLabi[j]}{\stSi[j]}{\ccstOi[j], \crstOi[j], \ccstIi[j], \crstIi[j]}{\gtG[j]}{k}$, 
          where for all $j \in J: \gtProj{\gtG[j]}{\roleP} \stSub \stT[\roleP]$,
          with $\roleP \neq \roleS$ and $\roleP \neq \roleT$. 
          The proof proceeds 
          by constructing $\stEnv[k]$ such that $\stEnv[k]  \stEnvQTMoveQueueAnnot{\roleP}{\roleQ}{\stLab[j]}    
          \stEnvi[k]$ and $\stEnvAssoc{\gtWithTime{\cVal[k]}{\gtG[k]}}{\stEnv[k]}{\mpS}$, for any arbitrary $k \in J$. 
         The thesis is then obtained by applying the inductive hypothesis to $\stEnv[k]$ and \inferrule{\iruleGtMoveCtx}~(resp. \inferrule{\iruleGtMoveCtxi}).  
         Further details can be found in the proof of analogous case in~Thm.\@ 11 in~\cite{BHYZ2023}.
\end{itemize}

 \item Case \inferrule{\iruleTCtxRcv}: from the premise, we have: 
\begin{gather}
\stEnvAssoc{\gtWithTime{\cVal}{\gtG}}{\stEnv}{\mpS}
\label{eq:complete:rcv_assoc_pre}
 \\
 \stEnvAnnotGenericSym = \stEnvQRecvAnnotSmall{\roleP}{\roleQ}{\stLab[j]}
 \\
  \stEnv
    \stEnvQTMoveRecvAnnot{\roleP}{\roleQ}{\stLab[j]}  
    \stEnvi
    \label{eq:stenvp-rcv}
    \\
    \stEnv = \stEnv[\gtG] \stEnvComp \stEnv[\Delta] \stEnvComp \stEnv[\stEnd]
\label{eq:stenvp-rcv-type-env}
\end{gather} 
By association~\eqref{eq:complete:rcv_assoc_pre} and~\eqref{eq:stenvp-rcv-type-env}, 
we have that $\forall \mpChanRole{\mpS}{\roleP} \in \dom{\stEnv[\gtG]}: \stEnvApp{\stEnv[\gtG]}{\mpChanRole{\mpS}{\roleP}} =  \stCPair{\cVal[\roleP]}{\stT[\roleP]}$ and $\gtProj{\gtG}{\roleP} \stSub \stT[\roleP]$,  $\forall \mpChanRole{\mpS}{\roleP} \in \dom{\stEnv[\Delta]}: \stEnvApp{\stEnv[\Delta]}{\mpChanRole{\mpS}{\roleP}} =  \stQType[\roleP]$, and $\forall \mpChanRole{\mpS}{\roleP} \in \dom{\stEnv[\stEnd]}: \stEnvApp{\stEnv[\stEnd]}{\mpChanRole{\mpS}{\roleP}} =  \stMPair{\stCPair{\cVal[\roleP]}{\stEnd}}{\stQEmptyType}$. Moreover,  
by inverting \inferrule{\iruleTCtxRcv} on~\eqref{eq:stenvp-rcv}~(and inverting and applying \inferrule{\iruleTCtxRec} when necessary), we have $
    \unfoldOne{\stT[\roleP]} =
    \stExtSum{\roleQ}{i \in I_{\roleP}}{\stTChoice{\stLab[
    i]}{\stS[i]}{\ccst[i], \crst[i]} \stSeq \stT[i]}
  $, and $\stEnvApp{\stEnv[\Delta]}{\mpChanRole{\mpS}{\roleQ}} =  \stQCons{\stQMsg{\roleP}{\stLab[j]}{\stSii}}{\stQType}$ with $\stSii \stSub \stS[j]$.  

We perform case analysis on~\cref{lem:inv-proj}.
\begin{itemize}[left=0pt, topsep=0pt]
  \item  Item {\bf b} in  Case 2 of \cref{lem:inv-proj}: $\unfoldOne{\gtG} =
           \gtCommT{\roleQ}{\roleP}{i \in I'}{\gtLabi[i]}{\stSi[i]}{\ccstOi[i], \crstOi[i], \ccstIi[i], \crstIi[i]}{\gtG[i]}$, or $\unfoldOne{\gtG} =
             \gtCommTTransit{\roleQ}{\roleP}{i \in I'}{\gtLabi[i]}{\stSi[i]}{\ccstOi[i], \crstOi[i], \ccstIi[i], \crstIi[i]}{\gtG[i]}{k}$, 
          where $I' \subseteq I_{\roleP}$, and for all $i \in I': \stLab[i] = \gtLabi[i]$, 
          $\stSi[i] \stSub \stS[i]$, $\ccst[i] = \ccstIi[i]$, $\crst[i] = \crstIi[i]$, 
          and 
          $\gtProj{\gtG[i]}{\roleP} \stSub \stT[i]$. Actually, $\unfoldOne{\gtG}$ cannot be the form of 
           $\gtCommT{\roleQ}{\roleP}{i \in I'}{\gtLabi[i]}{\stSi[i]}{\ccstOi[i], \crstOi[i], \ccstIi[i], \crstIi[i]}{\gtG[i]}$. This is because
           $\stEnv[\Delta]$ is associated with $\gtG$, but by~\eqref{eq:stenvp-rcv}, $\roleP \in \operatorname{receivers}(\stEnvApp{\stEnv[\Delta]}{\mpChanRole{\mpS}{\roleQ}})$, a desired contradiction.  
Moreover, we want to show that $j = k$, which holds due to $\stEnv[\Delta]$ being associated with $\gtG$, 
$\forall i \in I': \stLab[i] = \gtLabi[i]$,  and the requirement that labels in timed local types must be pairwise distinct.

By~\eqref{eq:stenvp-rcv}, it holds that $\cVal[\roleP] \models \ccst[j]$, and hence 
          $\cVal \models \ccstI[j]$. Along with the fact that $j = k \in I'$, 
          we can apply \inferrule{\iruleGtMoveIn} on $\gtWithTime{\cVal}{\gtG}$ to get $\gtWithTime{\cVali}{\gtGi}$ with 
  $\cVali = \cValUpd{\cVal}{\crstIi[j]}{0}$ and $\gtGi =  \gtG[j]$. We are left to 
  show association: $\stEnvAssoc{\gtWithTime{\cVali}{\gtGi}}{\stEnvi}{\mpS}$, where $\stEnvi = 
  \stEnvUpd{\stEnv[\gtG]}{\mpChanRole{\mpS}{\roleP}}{\stCPair{\cValUpd{\cVal[\roleP]}{\crst[j]}{0}}{\stT[j]}} \stEnvComp   \stEnvUpd{\stEnv[\Delta]}{\mpChanRole{\mpS}{\roleQ}}{\stQType} \stEnvComp \stEnv[\stEnd]$. 
 Note that here we only consider the case regarding $\gtGi \neq \gtEnd$, as the case where $\gtGi = \gtEnd$ is trivial.
             \begin{itemize}[left=0pt, topsep=0pt]
            \item $ \stEnvUpd{\stEnv[\gtG]}{\mpChanRole{\mpS}{\roleP}}{\stCPair{\cValUpd{\cVal[\roleP]}{\crst[j]}{0}}{\stT[j]}}$ is associated with $\gtWithTime{\cVali}{\gtGi}$ for $\mpS$: 
            \begin{enumerate}[label=(\roman*), left=0pt, topsep=0pt]
            \item trivial;  
            \item trivial; 
            \item by $\gtProj{\gtG[j]}{\roleP} \stSub \stT[j]$, $\gtProj{\gtG[j]}{\roleQ}  = \gtProj{\gtG}{\roleQ} \stSub \stT[\roleQ]$, and  for any $\roleR$ with $\roleR \neq \roleP$ and $\roleR \neq \roleQ$: $\gtProj{\gtGi}{\roleR} \stSub 
             \stMerge{i \in I'}{\gtProj{\gtG[i]}{\roleR}} = \gtProj{\gtG}{\roleR} \stSub \stT[\roleR]$; 
            \item by $\cVali = \cValUpd{\cVal}{\crstIi[j]}{0}$,  $\cVali[\roleP] = \cValUpd{\cVal[\roleP]}{\crst[j]}{0}$, $\crstIi[j] = \crst[j]$, and 
            $\cVal = \sqcup_{\roleP \in \gtRoles{\gtG}}\cVal[\roleP]$. 
            \end{enumerate}
            
            \item $\stEnvUpd{\stEnv[\Delta]}{\mpChanRole{\mpS}{\roleQ}}{\stQType}$  is associated with $\gtGi$ for 
            $\mpS$: directly from $\stEnv[\Delta]$ being associated with $\gtG$. 
        \item $\stEnv[\stEnd]$ keeps unchanged. 
                     \end{itemize} 
  \item  Item {\bf b} in  Case 2 of \cref{lem:inv-proj}: similar to that  in Case \inferrule{\iruleTCtxSend}. 
   Further details can be found in the proof of analogous case in~Thm.\@ 11 in~\cite{BHYZ2023}.
\end{itemize}

\item Case \inferrule{\iruleTCtxRec}: by inductive hypothesis. 
\qedhere 

 \end{itemize}
\end{proof}

\begin{restatable}{corollary}{colCom}%
  \label{cor:completeness}
  Assume that for any session $\mpS \in \stEnv$, there exists a timed global 
  type $\gtWithTime{\cVal[\mpS]}{\gtG[\mpS]}$ such that $\stEnvAssoc{\gtWithTime{\cVal[\mpS]}{\gtG[\mpS]}}{\stEnv[\mpS]}{\mpS}$. 
  If $\stEnv \stEnvMove \stEnvi$, then for any $\mpS \in \stEnvi$, 
  there exists a timed global 
  type $\gtWithTime{\cVali[\mpS]}{\gtGi[\mpS]}$ such that $\stEnvAssoc{\gtWithTime{\cVali[\mpS]}{\gtGi[\mpS]}}{\stEnvi[\mpS]}{\mpS}$.  
\end{restatable}
\begin{proof}
By induction on $\stEnv \stEnvMove \stEnvi$. 
 \begin{itemize}[left=0pt, topsep=0pt]
\item  Case $\stEnv \stEnvQTMoveTimeAnnot  \stEnvi$: for any $\mpS \in \stEnv$, by~\Cref{lem:stenv_time_action}, 
we have $\stEnv[\mpS] \stEnvQTMoveTimeAnnot  \stEnv[\mpS] + t$. Then the thesis follows by applying~\Cref{lem:time_assoc_comp}. 

\item Cases $\stEnv \stEnvQTMoveQueueAnnot{\roleP}{\roleQ}{\stLab} \stEnvi$ and $\stEnv \stEnvQTMoveRecvAnnot{\roleQ}{\roleP}{\stLab} \stEnvi$: directly from the completeness of association~(\Cref{lem:comp_proj}) and~\Cref{lem:stenv-red:trivial-5}. 
\qedhere
 \end{itemize}
 \end{proof}

\section{Subtyping Properties}
\label{sec:proofs:subtyping-properties}

\lemAssocSubP*
\begin{proof}
From the premise, we have that $\stEnv = \stEnv[\gtG] \stEnvComp \stEnv[\Delta] \stEnvComp 
\stEnv[\stEnd]$ such that $\stEnv[\gtG]$ is associated with $\gtWithTime{\cVal}{\gtG}$ for $\mpS$ and $\stEnv[\Delta]$ is 
associated with $\gtG$ for $\mpS$. Furthermore, since $\stEnv \stSub \stEnvi$,  $\stEnvi$ should of the form: $\stEnvi[\gtG] \stEnvComp \stEnvi[\Delta] \stEnvComp 
\stEnvi[\stEnd]$, where $\dom{\stEnv[\gtG]} = \dom{\stEnvi[\gtG]}$, $\dom{\stEnv[\Delta]} = \dom{\stEnvi[\Delta]}$, 
$\dom{\stEnv[\stEnd]} = \dom{\stEnvi[\stEnd]}$, 
$\forall \mpChanRole{\mpS}{\roleP} \in \dom{\stEnv[\gtG]}:   
        \stEnvApp{\stEnv[\gtG]}{\mpChanRole{\mpS}{\roleP}} = \stCPair{\cVal[\roleP]}{\stT[\roleP]} \stSub 
        \stCPair{\cVal[\roleP]}{\stTi[\roleP]} = \stEnvApp{\stEnvi[\gtG]}{\mpChanRole{\mpS}{\roleP}}$, 
        $\forall \mpChanRole{\mpS}{\roleP} \in \dom{\stEnv[\Delta]}:   
        \stEnvApp{\stEnv[\Delta]}{\mpChanRole{\mpS}{\roleP}} = \stQType[\roleP] \stSub \stQTypei[\roleP] = 
        \stEnvApp{\stEnvi[\Delta]}{\mpChanRole{\mpS}{\roleP}}$, and  
        $\forall \mpChanRole{\mpS}{\roleP} \in \dom{\stEnv[\stEnd]}: 
 \stEnvApp{\stEnv[\stEnd]}{\mpChanRole{\mpS}{\roleP}} = \stMPair{\stCPair{\cVal[\roleP]}{\stEnd}}{\stQEmptyType} \stSub 
 \stMPair{\stCPair{\cVal[\roleP]}{\stEnd}}{\stQEmptyType} = \stEnvApp{\stEnvi[\stEnd]}{\mpChanRole{\mpS}{\roleP}}$. 
 
 We are left to show: 
 \begin{itemize}[left=0pt, topsep=0pt]
 \item $\stEnvi[\gtG]$ is associated with $\gtWithTime{\cVal}{\gtG}$ for $\mpS$: 
 \begin{enumerate}[label=(\roman*), left=0pt, topsep=0pt]
 \item no change here;
 \item trivial; 
 \item since $\forall \roleP \in \gtRoles{\gtG}:  \gtProj{\gtG}{\roleP} \stSub \stT[\roleP] \stSub \stTi[\roleP]$, the thesis holds by the transitivity of subtyping~(\Cref{transitive-subtyping}); 
 \item no change here. 
 \end{enumerate}
 
 \item $\stEnvi[\Delta]$ is associated with $\gtG$ for $\mpS$: 
 \begin{enumerate}[label=(\roman*), left=0pt, topsep=0pt]
 \item no change here;
 \item trivial;  
 \end{enumerate}
 The rest is by induction on the structure of $\gtG$: 
\begin{enumerate}[label=(\roman*), left=0pt, topsep=0pt,  ref={(\roman*)},
          resume, nosep]
 \item $\gtG = \gtEnd$ or $\gtG = \gtRec{\gtRecVar}{\gtGi}$: $\forall \mpChanRole{\mpS}{\roleP} \in \dom{\stEnvi[\Delta]}:   
        \stEnvApp{\stEnv[\Delta]}{\mpChanRole{\mpS}{\roleP}} = \stQEmptyType \stSub 
        \stEnvApp{\stEnvi[\Delta]}{\mpChanRole{\mpS}{\roleP}} =  \stQEmptyType$,  which is the thesis;  
 \item $\gtG =  \gtCommT{\roleP}{\roleQ}{i \in
        I}{\gtLab[i]}{\stS[i]}{\ccstO[i], \crstO[i], \ccstI[i], \crstI[i]}{\gtG[i]}$: 
 \begin{enumerate*}[label={(a\arabic*)}]
        \item since $\roleQ \notin \operatorname{receivers}(\stEnvApp{\stEnv[\Delta]}{\mpChanRole{\mpS}{\roleP}})$ and $\stEnvApp{\stEnv[\Delta]}{\mpChanRole{\mpS}{\roleP}} \stSub \stEnvApp{\stEnvi[\Delta]}{\mpChanRole{\mpS}{\roleP}}$, we have $\roleQ \notin \operatorname{receivers}(\stEnvApp{\stEnvi[\Delta]}{\mpChanRole{\mpS}{\roleP}})$, as desired; 
        \item by inductive hypothesis;  
         \end{enumerate*}
 \item $\gtG = \gtCommTTransit{\roleP}{\roleQ}{i \in
          I}{\gtLab[i]}{\stS[i]}{\ccstO[i], \crstO[i], \ccstI[i], \crstI[i]}{\gtG[i]}{j}$:  
           \begin{enumerate*}[label={(b\arabic*)}]
         \item  since $\stEnvApp{\stEnv[\Delta]}{\mpChanRole{\mpS}{\roleP}} = 
          \stQCons{\stQMsg{\roleQ}{\stLab[j]}{\stSi[j]}}{\stQType}$ with $\stSi[j] \stSub \stS[j]$ and 
          $\stEnvApp{\stEnv[\Delta]}{\mpChanRole{\mpS}{\roleP}} \stSub \stEnvApp{\stEnvi[\Delta]}{\mpChanRole{\mpS}{\roleP}}$, we have $\stEnvApp{\stEnvi[\Delta]}{\mpChanRole{\mpS}{\roleP}} = 
          \stQCons{\stQMsg{\roleQ}{\stLab[j]}{\stSii[j]}}{\stQTypeii}$ with $\stSii[j] \stSub \stSi[j] \stSub \stS[j]$ and 
          $\stQType \stSub \stQTypeii$, as desired;  
  \item  by inductive hypothesis. 
  \end{enumerate*}
 \end{enumerate}
 \item  $\forall \mpChanRole{\mpS}{\roleP} \in \dom{\stEnvi[\stEnd]}: 
\stEnvApp{\stEnvi[\stEnd]}{\mpChanRole{\mpS}{\roleP}} = \stMPair{\stCPair{\cVal[\roleP]}{\stEnd}}{\stQEmptyType}$, which is the thesis. 
\qedhere 
 \end{itemize}
\end{proof}

\begin{lemma}
\label{lem:stenv-time-sub}
If $\stEnv \stSub \stEnvi$, then $\stEnv + t \stSub \stEnvi + t$. 
\end{lemma}
\begin{proof}
It is straightforward that $\dom{\stEnv} = \dom{\stEnvi} = \dom{\stEnvi + t} = \dom{\stEnv = t}$. We are left 
to show that $\forall \mpC \in \dom{\stEnv}:  \stEnvApp{\stEnv \,\tcFmt{+}\, \tcFmt{t}}{\mpC} \stSub \stEnvApp{\stEnvi \,\tcFmt{+}\, \tcFmt{t}}{\mpC}$. 

\begin{itemize}[left=0pt, topsep=0pt]
\item $\stEnvApp{\stEnv}{\mpC} = \stCPair{\cVal}{\stT}$: $\stEnvApp{\stEnvi}{\mpC} = \stCPair{\cVal}{\stTi}$ with 
$\stT \stSub \stTi$. Then the thesis follows directly from $\stEnvApp{\stEnv \,\tcFmt{+}\, \tcFmt{t}}{\mpC} = \stCPair{\cVal + t}{\stT}$ and  $\stEnvApp{\stEnvi \,\tcFmt{+}\, \tcFmt{t}}{\mpC} = \stCPair{\cVal + t}{\stTi}$. 

\item $\stEnvApp{\stEnv}{\mpC} = \stQType$: $\stEnvApp{\stEnvi}{\mpC} = \stQTypei$ with 
$\stQType \stSub \stQTypei$. Then the thesis follows directly from $\stEnvApp{\stEnv \,\tcFmt{+}\, \tcFmt{t}}{\mpC} = \stQType$ and  $\stEnvApp{\stEnvi \,\tcFmt{+}\, \tcFmt{t}}{\mpC} = \stQTypei$. 

\item $\stEnvApp{\stEnv}{\mpC} = \stMPair{\stCPair{\cVal}{\stT}}{\stQType}$: $\stEnvApp{\stEnvi}{\mpC} = \stMPair{\stCPair{\cVal}{\stTi}}{\stQTypei}$ with 
$\stT \stSub \stTi$ and $\stQType \stSub \stQTypei$. Then the thesis follows directly from $\stEnvApp{\stEnv \,\tcFmt{+}\, \tcFmt{t}}{\mpC} = \stMPair{\stCPair{\cVal + t}{\stT}}{\stQType}$ and  $\stEnvApp{\stEnvi \,\tcFmt{+}\, \tcFmt{t}}{\mpC} = \stMPair{\stCPair{\cVal + t}{\stTi}}{\stQTypei}$. 
\qedhere 
\end{itemize}
\end{proof}

\begin{lemma}
\label{lem:stenv-assoc-send-reduction-sub}%
    Assume that $\stEnvAssoc{\gtWithTime{\cVal}{\gtG}}{\stEnv}{\mpS}$ and
    $\stEnv \stSub \stEnvi \stEnvQTMoveQueueAnnot{\roleP}{\roleQ}{\stLab[k]} \stEnvii$. 
    Then, there is $\stEnviii$ such that 
    $\stEnv \stEnvQTMoveQueueAnnot{\roleP}{\roleQ}{\stLab[k]} \stEnviii \stSub \stEnvii$.
    \end{lemma}
\begin{proof}
By association $\stEnvAssoc{\gtWithTime{\cVal}{\gtG}}{\stEnv}{\mpS}$, 
we have that $\forall \mpChanRole{\mpS}{\roleP} \in \dom{\stEnv}: 
\stEnvApp{\stEnv}{\mpChanRole{\mpS}{\roleP}} = \stMPair{\stCPair{\cVal[\roleP]}{\stT[\roleP]}}{\stQType[\roleP]}$. 
Then with the subtyping $\stEnv \stSub \stEnvi$, it holds that $\dom{\stEnv} = \dom{\stEnvi}$, and 
$\forall \mpChanRole{\mpS}{\roleP} \in \dom{\stEnvi}: 
\stEnvApp{\stEnvi}{\mpChanRole{\mpS}{\roleP}} = \stMPair{\stCPair{\cVal[\roleP]}{\stTi[\roleP]}}{\stQTypei[\roleP]}$ with $\stT[\roleP] \stSub \stTi[\roleP]$ and $\stQType[\roleP] \stSub \stQTypei[\roleP]$. 

We now apply and invert \inferrule{\iruleTCtxSend} on $\stEnvi \stEnvQTMoveQueueAnnot{\roleP}{\roleQ}{\stLab[k]} \stEnvii$ (apply and invert \inferrule{\iruleTCtxRec} when necessary) to get 
$\unfoldOne{\stTi[\roleP]} = \stIntSum{\roleQ}{i \in I}{\stTChoice{\stLab[i]}{\stS[i]}{\ccst[i], \crst[i]} \stSeq \stT[i]}$, 
$\stEnvii = \stEnvUpd{\stEnvi}{\mpChanRole{\mpS}{\roleP}}{ \stMPair{
        \stCPair{\mpFmt{\cValUpd{\cVal[\roleP]}{\crst[k]}{0}}}
          {\stT[k]}%
        }{%
          \stQCons{\stQTypei[\roleP]}{%
            \stQCons{%
              \stQMsg{\roleQ}{\stLab[k]}{\stS[k]}%
            }{%
              \stQEmptyType%
            }%
          }}}$, $k \in I$, and $\cVal[\roleP] \models \ccst[k]$. 
          
Since $\stEnv \stSub \stEnvi$, by inversion of subtyping~(\Cref{lem:subtyping-invert}), 
$\unfoldOne{\stT[\roleP]} = \stIntSum{\roleQ}{j \in J}{\stTChoice{\stLabi[j]}{\stSi[j]}{\ccsti[j], \crsti[j]} \stSeq \stTi[j]}$ such that $I \subseteq J$, and $\forall i \in I: \stLab[i] = \stLabi[i], \stS[i] \stSub \stSi[i], \ccst[i] = \ccsti[i], 
\crst[i] = \crsti[i]$, and $\stTi[i] \stSub \stT[i]$. Hence, with $k \in I \subseteq J$ and $\cVal[\roleP] \models \ccst[k] = \ccsti[k]$, we can apply \inferrule{\iruleTCtxSend} on $\stEnv$ 
to get $\stEnviii = \stEnvUpd{\stEnv}{\mpChanRole{\mpS}{\roleP}}{ \stMPair{
        \stCPair{\mpFmt{\cValUpd{\cVal[\roleP]}{\crsti[k]}{0}}}
          {\stTi[k]}%
        }{%
          \stQCons{\stQType[\roleP]}{%
            \stQCons{%
              \stQMsg{\roleQ}{\stLabi[k]}{\stSi[k]}%
            }{%
              \stQEmptyType%
            }%
          }}}$. 

We are left to show $\stEnviii \stSub \stEnvii$: 
\begin{itemize}[left=0pt, topsep=0pt]
\item $\dom{\stEnviii} = \dom{\stEnvii}$: directly from $\dom{\stEnviii} = \dom{\stEnv} = \dom{\stEnvi} = \dom{\stEnvii}$; 
\item $\stEnvApp{\stEnviii}{\mpChanRole{\mpS}{\roleP}} =  \stMPair{
        \stCPair{\mpFmt{\cValUpd{\cVal[\roleP]}{\crsti[k]}{0}}}
          {\stTi[k]}%
        }{%
          \stQCons{\stQType[\roleP]}{%
            \stQCons{%
              \stQMsg{\roleQ}{\stLabi[k]}{\stSi[k]}%
            }{%
              \stQEmptyType%
            }%
          }} \stSub \stEnvApp{\stEnvii}{\mpChanRole{\mpS}{\roleP}} = 
           \stMPair{
        \stCPair{\mpFmt{\cValUpd{\cVal[\roleP]}{\crst[k]}{0}}}
          {\stT[k]}%
        }{%
          \stQCons{\stQTypei[\roleP]}{%
            \stQCons{%
              \stQMsg{\roleQ}{\stLab[k]}{\stS[k]}%
            }{%
              \stQEmptyType%
            }%
          }}$:  straightforward from $\stQType[\roleP] \stSub \stQTypei[\roleP]$, $\stLabi[k] = \stLab[k]$, 
          $\stS[k] \stSub \stSi[k]$, $\crst[k] = \crsti[k]$, and $\stTi[k] \stSub \stT[k]$. 
\item $\stEnvApp{\stEnviii}{\mpChanRole{\mpS}{\roleQ}} \stSub \stEnvApp{\stEnvii}{\mpChanRole{\mpS}{\roleQ}}$ with $\roleQ \neq \roleP$: directly from $\stEnvApp{\stEnviii}{\mpChanRole{\mpS}{\roleQ}} = \stEnvApp{\stEnv}{\mpChanRole{\mpS}{\roleQ}} \stSub \stEnvApp{\stEnvi}{\mpChanRole{\mpS}{\roleQ}} = \stEnvApp{\stEnvii}{\mpChanRole{\mpS}{\roleQ}}$. \qedhere 
\end{itemize}
 \end{proof}

\begin{lemma}
\label{lem:stenv-assoc-rcv-reduction-sub}%
    Assume that $\stEnvAssoc{\gtWithTime{\cVal}{\gtG}}{\stEnv}{\mpS}$ and
    $\stEnv \stSub \stEnvi \stEnvQTMoveRecvAnnot{\roleQ}{\roleP}{\stLab[k]} \stEnvii$. 
    Then, there is $\stEnviii$ such that 
    $\stEnv \stEnvQTMoveRecvAnnot{\roleQ}{\roleP}{\stLab[k]} \stEnviii \stSub \stEnvii$.
    \end{lemma}
\begin{proof}
By association $\stEnvAssoc{\gtWithTime{\cVal}{\gtG}}{\stEnv}{\mpS}$, 
we have that $\forall \mpChanRole{\mpS}{\roleP} \in \dom{\stEnv}: 
\stEnvApp{\stEnv}{\mpChanRole{\mpS}{\roleP}} = \stMPair{\stCPair{\cVal[\roleP]}{\stT[\roleP]}}{\stQType[\roleP]}$. 
Then with the subtyping $\stEnv \stSub \stEnvi$, it holds that $\dom{\stEnv} = \dom{\stEnvi}$, and 
$\forall \mpChanRole{\mpS}{\roleP} \in \dom{\stEnvi}: 
\stEnvApp{\stEnvi}{\mpChanRole{\mpS}{\roleP}} = \stMPair{\stCPair{\cVal[\roleP]}{\stTi[\roleP]}}{\stQTypei[\roleP]}$ with $\stT[\roleP] \stSub \stTi[\roleP]$ and $\stQType[\roleP] \stSub \stQTypei[\roleP]$. 

We now apply and invert \inferrule{\iruleTCtxRcv} on $\stEnvi \stEnvQTMoveRecvAnnot{\roleQ}{\roleP}{\stLab[k]} \stEnvii$ (apply and invert \inferrule{\iruleTCtxRec} when necessary) to get 
$\stQTypei[\roleP] = \stQCons{\stQMsg{\roleQ}{\stLab[k]}{\stS[k]}}{\stQTypei}$, 
 $\unfoldOne{\stTi[\roleQ]} = \stExtSum{\roleP}{i \in I}{\stTChoice{\stLab[i]}{\stSi[i]}{\ccst[i], \crst[i]} \stSeq \stT[i]}$, 
$\stEnvii = \stEnvUpd{\stEnvUpd{\stEnvi}{\mpChanRole{\mpS}{\roleP}}{\stMPair{\stCPair{\cVal[\roleP]}{\stTi[\roleP]}}{\stQTypei}}}{\mpChanRole{\mpS}{\roleQ}}{\stMPair{
        \stCPair{\mpFmt{\cValUpd{\cVal[\roleQ]}{\crst[k]}{0}}}
          {\stT[k]}%
        }{\stQTypei[\roleQ]}}$, $k \in I$, $\cVal[\roleQ] \models \crst[k]$, and $\stS[k] \stSub \stSi[k]$. 
          
Since $\stEnv \stSub \stEnvi$, by inversion of subtyping~(\Cref{lem:subtyping-invert}), 
$\unfoldOne{\stT[\roleQ]} = \stExtSum{\roleP}{j \in J}{\stTChoice{\stLabi[j]}{\stSii[j]}{\ccsti[j], \crsti[j]} \stSeq \stTi[j]}$ such that $J \subseteq I$, and $\forall i \in j: \stLab[i] = \stLabi[i], \stSii[i] \stSub \stSi[i], \ccst[i] = \ccsti[i], 
\crst[i] = \crsti[i]$, and $\stTi[i] \stSub \stT[i]$. By $\stQType[\roleP] \stSub \stQTypei[\roleP]$,  
we know that $\stQType[\roleP] =  \stQCons{\stQMsg{\roleQ}{\stLab}{\stSiii}}{\stQType}$ with $\stLab = \stLab[k]$, $\stS[k] \stSub \stSiii$, and $\stQType \stSub \stQTypei$. Hence, by applying~\Cref{lem:comm-assoc-match}, we have $k \in J$, $\stLab = \stLab[k] = \stLabi[k]$, and $\stSiii \stSub \stSii[k]$. 
Then along with $\cVal[\roleQ] \models \ccst[k] = \ccsti[k]$, we can apply \inferrule{\iruleTCtxRcv} on $\stEnv$ 
to get $\stEnviii = \stEnvUpd{\stEnvUpd{\stEnvi}{\mpChanRole{\mpS}{\roleP}}{\stMPair{\stCPair{\cVal[\roleP]}{\stT[\roleP]}}{\stQType}}}{\mpChanRole{\mpS}{\roleQ}}{\stMPair{
        \stCPair{\mpFmt{\cValUpd{\cVal[\roleQ]}{\crsti[k]}{0}}}
          {\stTi[k]}%
        }{\stQType[\roleQ]}}$. 

We are left to show $\stEnviii \stSub \stEnvii$: 
\begin{itemize}[left=0pt, topsep=0pt]
\item $\dom{\stEnviii} = \dom{\stEnvii}$: directly from $\dom{\stEnviii} = \dom{\stEnv} = \dom{\stEnvi} = \dom{\stEnvii}$; 
\item $\stEnvApp{\stEnviii}{\mpChanRole{\mpS}{\roleP}} =  \stMPair{\stCPair{\cVal[\roleP]}{\stT[\roleP]}}{\stQType}
 \stSub \stEnvApp{\stEnvii}{\mpChanRole{\mpS}{\roleP}} = \stMPair{\stCPair{\cVal[\roleP]}{\stTi[\roleP]}}{\stQTypei}
          $:  straightforward from $\stQType[\roleP] \stSub \stQTypei[\roleP]$ and $\stT[\roleP] \stSub \stTi[\roleP]$. 

\item $\stEnvApp{\stEnviii}{\mpChanRole{\mpS}{\roleQ}} =  \stMPair{
        \stCPair{\mpFmt{\cValUpd{\cVal[\roleQ]}{\crsti[k]}{0}}}
          {\stTi[k]}%
        }{\stQType[\roleQ]} \stSub \stEnvApp{\stEnvii}{\mpChanRole{\mpS}{\roleQ}} = 
          \stMPair{
        \stCPair{\mpFmt{\cValUpd{\cVal[\roleQ]}{\crst[k]}{0}}}
          {\stT[k]}%
        }{\stQTypei[\roleQ]}$:  straightforward from $\stQType[\roleQ] \stSub \stQTypei[\roleQ]$, 
      $\crst[k] = \crsti[k]$, and $\stTi[k] \stSub \stT[k]$.

\item $\stEnvApp{\stEnviii}{\mpChanRole{\mpS}{\roleR}} \stSub \stEnvApp{\stEnvii}{\mpChanRole{\mpS}{\roleR}}$ with $\roleR \neq \roleP$ and $\roleR \neq \roleQ$: directly from $\stEnvApp{\stEnviii}{\mpChanRole{\mpS}{\roleQ}} = \stEnvApp{\stEnv}{\mpChanRole{\mpS}{\roleQ}} \stSub \stEnvApp{\stEnvi}{\mpChanRole{\mpS}{\roleQ}} = \stEnvApp{\stEnvii}{\mpChanRole{\mpS}{\roleQ}}$. 
\qedhere 
\end{itemize}
 \end{proof}

\lemStenvReductionSubAssoc*
  \begin{proof}
 A direct corollary from~\Cref{lem:stenv-time-sub,lem:stenv-assoc-send-reduction-sub,lem:stenv-assoc-rcv-reduction-sub}.  
 \qedhere 
  \end{proof}

\begin{restatable}{lemma}{lemStenvReductionSubAssocGeneral}
    \label{lem:stenv-assoc-reduction-sub-gen}%
    Assume that $\forall \mpS \in \stEnv: \exists \gtWithTime{\cVal[\mpS]}{\gtG[\mpS]}: 
    \stEnvAssoc{\gtWithTime{\cVal[\mpS]}{\gtG[\mpS]}}{\stEnv[\mpS]}{\mpS}$ and
    $\stEnv \stSub \stEnvi \stEnvMoveGenAnnot \stEnvii$. 
    Then, there is $\stEnviii$ such that 
    $\stEnv \stEnvMoveGenAnnot \stEnviii \stSub \stEnvii$.%
  \end{restatable}
\begin{proof}
  Apply~\Cref{lem:stenv-assoc-reduction-sub} (and~\Cref{lem:stenv-red:trivial-5} when necessary). 
  \qedhere 
  \end{proof}

\section{Type System Properties}
\label{sec:app-aat-mpst-type-prop-proof}

\begin{lemma}
\label{lem:end_sub}
If  $\stEnvi \stSub \stEnv$ and $\stEnvEndP{\stEnv}$, then $\stEnvEndP{\stEnvi}$. 
\end{lemma}
\begin{proof}
Straightforward from $\stEnvi \stSub \stEnv$, \inferrule{\iruleMPSub}, and \inferrule{\iruleMPNil}. 
\end{proof}

\begin{lemma}[Narrowing]
\label{lem:narrowing}
If 
$\stJudge{\mpEnv}{\stEnv}{\mpP}$ and 
 $\stEnvi \stSub \stEnv$, then 
  $\stJudge{\mpEnv}{\stEnvi}{\mpP}$.
\end{lemma}
\begin{proof}
 By induction on the derivation of $\stJudge{\mpEnv}{\stEnv}{\mpP}$.  Most
cases hold by inverting the typing $\stJudge{\mpEnv}{\stEnv}{\mpP}$, inserting (possibly vacuous) instances of rule~\inferrule{\iruleMPSub}, 
and then applying the induction hypothesis. 
We develop some interesting cases. 
  \begin{itemize}[left=0pt, topsep=0pt]
 \item  Case \inferrule{\iruleMPSub}: by the transitivity of subtyping~(\Cref{lem:reflexive-subtyping}). 
 
 \item Case \inferrule{\iruleMPNil}:   
 by applying \inferrule{\iruleMPNil} on $\stEnvEndP{\stEnvi}$~(from $\stEnvEndP{\stEnv}$ and \Cref{lem:end_sub}).

 \item Case \inferrule{\iruleMPCall}: by~\Cref{lem:end_sub} and induction hypothesis. 
 
\item Case \inferrule{\iruleMPTime}: by $\stEnvQi \,\tcFmt{+}\, \cUnit \stSub  \stEnvQ \,\tcFmt{+}\, \cUnit$ and induction hypothesis.

 \item Case \inferrule{\iruleMPKill}: by~\Cref{lem:end_sub} and induction hypothesis. 
 
\item Case \inferrule{\iruleMPFailed}: by the fact that $\mpFailedP{\mpP}$ can be typed by any typing environment. 

\item Case \inferrule{\iruleMPQueueEmpty}: by $\stQEmptyType \stSub \stQEmptyType$ and~\Cref{lem:end_sub}. 

\item Case \inferrule{\iruleMPQueue}:   $\stEnv = \stEnv[1] \stEnvComp \stEnv[2]$ such that 
$\stEnv[1] =  \stEnvUpd{\stEnvQ[3]}{\mpChanRole{\mpS}{\roleP}}{%
          \stQCons{\stQMsg{\roleQ}{\stLab}{\stS}}{\stEnvApp{\stEnvQ[3]}{\mpChanRole{\mpS}{\roleP}}}}$, 
          $\stJudge{\mpEnv}{\stEnvQ[3]}{%
        \mpSessionQueueO{\mpS}{\roleP}{%
          \mpQueue%
        }%
      }$, $\stS = \stCPair{\ccst}{\stT}$, $\cVal \models \ccst$, and $\stEnvEntails{\stEnv[2]}{\mpChanRole{\mpSi}{\roleR}}{\stCPair{\cVal}{\stT}}$. Since $\stEnvi \stSub \stEnv$, we have $\stEnvi = \stEnvi[1] \stEnvComp \stEnvi[2]$ such that $\stEnvi[1] =  \stEnvUpd{\stEnvQi[3]}{\mpChanRole{\mpS}{\roleP}}{%
          \stQCons{\stQMsg{\roleQ}{\stLab}{\stSi}}{\stEnvApp{\stEnvQi[3]}{\mpChanRole{\mpS}{\roleP}}}}$ with 
          $\stS = \stDelegate{\ccst}{\stT} \stSub \stSi = \stDelegate{\ccst}{\stTi}$, $\stEnvi[3] \stSub \stEnv[3]$, and 
         $\stEnvi[2] \stSub \stEnv[2]$. By induction hypothesis, we obtain $\stJudge{\mpEnv}{\stEnvQi[3]}{%
        \mpSessionQueueO{\mpS}{\roleP}{%
          \mpQueue%
        }%
      }$. We are left to show that $\stEnvEntails{\stEnvi[2]}{\mpChanRole{\mpSi}{\roleR}}{\stCPair{\cVal}{\stTi}}$, which follows from $\stEnvEntails{\stEnv[2]}{\mpChanRole{\mpSi}{\roleR}}{\stCPair{\cVal}{\stT}}$, $\stT \stSub \stTi$, $\stEnvi[2] \stSub \stEnv[2]$, and \inferrule{\iruleMPSub}. 
      \qedhere 
  \end{itemize}
\end{proof}

\begin{lemma}[Substitution]
\label{lem:substitution}
Assume\, $\stJudge{\mpEnv}{\stEnvQ \stEnvQComp
\stEnvMap{\mpFmt{x}}{\stCPair{\cVal}{\stT}}}{\mpP}$ \,and\,
 $\stEnvEntails{\stEnvQi}{\mpChanRole{\mpS}{\roleP}}{%
         \stCPair{\cVal}{\stT}}$ \,with\, $\stEnvQ \stEnvQComp \stEnvQi$ \,defined.
         Then\, $\stJudge{\mpEnv}{\stEnvQ \stEnvQComp
\stEnvQi}{\mpP\subst{\mpFmt{x}}{\mpChanRole{\mpS}{\roleP}}}$.
\end{lemma}
\begin{proof}
By induction on the derivation of \,$\stJudge{\mpEnv}{\stEnvQ \stEnvQComp
\stEnvMap{\mpFmt{x}}{\stCPair{\cVal}{\stT}}}{\mpP}$.
\end{proof}

\begin{restatable}[Subject Congruence]{lemma}{lemSubjectCongruence}
  \label{lem:subject-congruence}%
  Assume 
  $\stJudge{\mpEnv}{\stEnv}{\mpP}$ and 
  $\mpP \equiv \mpPi$. Then, $\exists \stEnvi$ such that 
  $\stEnv \stEquiv \stEnvi$ and 
  $\stJudge{\mpEnv}{\stEnvi}{\mpPi}$.
\end{restatable}
\begin{proof}
The proof is similar to that of Theorem 3.11 in~\cite{lagaillardie2022Affine}, and in all corresponding cases, we have $\stEnv = \stEnvi$. 
We have three additional cases regarding delay and queue congruence rules. 
 \begin{itemize}[left=0pt, topsep=0pt]
\item $\mpP = \delay{0}{\mpQ} \equiv \mpQ = \mpPi$: 
\begin{itemize}[left=0pt, topsep=0pt]
\item[($\Rightarrow$)]  suppose $\stJudge{\mpEnv}{%
        \stEnvQ%
      }{%
        \delay{0}{\mpQ}}$. By \inferrule{\iruleMPTime}, we have
        $\stJudge{\mpEnv}{%
        \stEnvQ \,\tcFmt{+}\, \tcFmt{0}
      }{%
        \mpQ}$, and hence, $\stJudge{\mpEnv}{%
        \stEnvQ%
      }{\mpQ}$, as desired.
 \item[($\Leftarrow$)]  suppose $\stJudge{\mpEnv}{%
        \stEnvQ%
      }{%
        \mpQ}$, which follows by
        $\stJudge{\mpEnv}{%
        \stEnvQ \,\tcFmt{+}\, \tcFmt{0}
      }{%
        \mpQ}$. Hence, by inversion of \inferrule{\iruleMPTime}, we have
        $\stJudge{\mpEnv}{%
        \stEnvQ%
      }{\delay{0}{\mpQ}}$, as desired.
 \end{itemize}
 \item $\mpP =  \mpRes{\stEnvMap{\mpS}{\stEnvii}}{%
        (
        \mpSessionQueueO{\mpS}{\roleP[1]}{\mpQueueEmpty} \mpPar \cdots \mpPar
        \mpSessionQueueO{\mpS}{\roleP[n]}{\mpQueueEmpty}
        )}
        \equiv 
        \mpNil = \mpPi$: we must have $\stEnvEndP{\stEnvii}$ and $\stEnvEndP{\stEnv}$, and hence, 
        we conclude with $\stEnvi = \stEnv$ by typing rule \inferrule{\iruleMPNil}. 
 \item $\mpP \equiv \mpPi$ holds by the order-swapping congruence on queues:  we apply the same reordering on the queue types of $\stEnv$, getting a congruent typing environment $\stEnvi$  that satisfies the statement.
 \qedhere 
 \end{itemize}
\end{proof}

\section{Proofs for Subject Reduction and Type Safety}
\label{sec:app-aat-mpst-sj-proof}
\begin{lemma}
\label{lem:end-time-pass}
If $\stEnvEndP{\stEnv}$, then 
$\stEnvEndP{\stEnvQ \,\tcFmt{+}\, \cUnit}$.
\end{lemma}
\begin{proof}
From $\stEnvEndP{\stEnv}$, we have that
$\forall \mpC \in \dom{\stEnv}: \stEnvApp{\stEnvQ}{\mpC} =
\stCPair{\cVal[\mpC]}{\stEnd}$, which follows directly that
$\forall \mpC \in \dom{\stEnv \,\tcFmt{+}\, \cUnit} (= \dom{\stEnv}): 
\stEnvApp{\stEnvQ \,\tcFmt{+}\, \cUnit}{\mpC} =
\stCPair{\cVal[\mpC] \,\tcFmt{+}\, \cUnit}{\stEnd}$. Then the thesis holds 
 by applying \inferrule{\iruleMPSub} and
\inferrule{\iruleMPEnd}. 
\end{proof}

\begin{lemma}
\label{lem:subj-equiv}
If $\timePass{\cUnit}{\mpP}$ is defined, then 
$\procSubject{\timePass{\cUnit}{\mpP}} = \procSubject{\mpP}$.
\end{lemma}
\begin{proof}
By induction on the syntactic constructs of $\mpP$. 
The case of $\mpP = \delay{\ccst}{\mpPi}$ 
is excluded due to the undefined nature of $\timePass{\cUnit}{\mpP}$.
\begin{itemize}[left=0pt, topsep=0pt]
\item $\mpP = \mpNil$: $\timePass{\cUnit}{\mpP} = \mpNil = \mpP$. The thesis holds trivially.  

\item $\mpP = \mpCErr$: $\timePass{\cUnit}{\mpP} = \mpCErr = \mpP$. The thesis holds trivially.  

\item $\mpP = \mpQ[1] \mpPar \mpQ[2]$:  $\timePass{\cUnit}{\mpP} = \timePass{\cUnit}{\mpQ[1]} \mpPar 
 \timePass{\cUnit}{\mpQ[2]}$. By the induction hypothesis, we have $\procSubject{\timePass{\cUnit}{\mpQ[1]}} = \procSubject{\mpQ[1]}$ and $\procSubject{\timePass{\cUnit}{\mpQ[2]}} = \procSubject{\mpQ[2]}$. Hence,  
 $\procSubject{\timePass{\cUnit}{\mpQ[1]}} \cup \procSubject{\timePass{\cUnit}{\mpQ[2]}} = \procSubject{\mpQ[1]} \cup \procSubject{\mpQ[2]}$, which is the thesis. 

 \item $\mpP = \mpRes{\mpS}{\mpQ}$: $\timePass{\cUnit}{\mpP} = \mpRes{\mpS}{\timePass{\cUnit}{\mpQ}}$. 
 By the induction hypothesis, we have $\procSubject{\timePass{\cUnit}{\mpQ}} = \procSubject{\mpQ}$. Hence, 
 $ \procSubject{\mpQ} \setminus (\setenum{\mpChanRole{\mpS}{\roleP[i]}}_{i \in I} \cup \setenum{{\mpChanRole{\mpS}{\roleP[i]}}^{\mathscr{Q}}}_{i \in I})$ =  $ \procSubject{\timePass{\cUnit}{\mpQ}} \setminus (\setenum{\mpChanRole{\mpS}{\roleP[i]}}_{i \in I} \cup \setenum{{\mpChanRole{\mpS}{\roleP[i]}}^{\mathscr{Q}}}_{i \in I})$, which is the thesis. 

\item $\mpP = \mpDefAbbrev{ \mpJustDef{\mpX}{\widetilde{x}}{\mpQ}}{\mpQi}$:  $\timePass{\cUnit}{\mpP} = 
\mpDefAbbrev{ \mpJustDef{\mpX}{\widetilde{x}}{\mpQ}}{\timePass{\cUnit}{\mpQi}}$. By the induction hypothesis, 
we have $\procSubject{\timePass{\cUnit}{\mpQi}} = \procSubject{\mpQi}$. 
Hence,  $\procSubject{\mpQi} \cup \procSubject{\mpQ} \setminus \setenum{\widetilde{x}} = 
 \procSubject{\timePass{\cUnit}{\mpQi}} \cup \procSubject{\mpQ} \setminus \setenum{\widetilde{x}}$ with 
 $\procSubject{\mpCall{\mpX}{\widetilde{\mpC}}}
   =
      \procSubject{\mpQ\subst{\widetilde{x}}{\widetilde{\mpC}}}$, which is the thesis.

\item $\mpP = \mpTSel{\mpC}{\roleQ}{\mpLab}{\mpD}{\mpQ}{\mathfrak{n}}$: $\procSubject{\timePass{\cUnit}{\mpP}} = \setenum{\mpC} = \procSubject{\mpP}$ if $\cUnit \leq \mathfrak{n}$, and $\procSubject{\timePass{\cUnit}{\mpP}} = \procSubject{\mpFailedP{\mpP}} =  \procSubject{\mpP} = \setenum{\mpC}$ if $\cUnit > \mathfrak{n}$, as desired. 

\item $\mpP = \mpTBranch{\mpC}{\roleQ}{i \in I}{\mpLab[i]}{x_i}{\mpP[i]}{\mathfrak{n}}{}$: similarly to the case $\mpP = \mpTSel{\mpC}{\roleQ}{\mpLab}{\mpD}{\mpQ}{\mathfrak{n}}$. 

\item $\mpP = \mpCancel{\mpC}{\mpQ}$: $\timePass{\cUnit}{\mpP} = \mpCancel{\mpC}{\timePass{\cUnit}{\mpQ}}$. The thesis holds as $\procSubject{\mpP} = \setenum{\mpC}$ and $\procSubject{\timePass{\cUnit}{\mpP} } = \setenum{\mpC}$. 

\item $\mpP = \delay{\cUniti}{\mpQ}$: $\timePass{\cUnit}{\mpP} = \delay{\cUniti - \cUnit}{\mpQ}$. The thesis holds straightforward  from  $\procSubject{\mpP}  = \procSubject{\mpQ}$ and $\procSubject{\timePass{\cUnit}{\mpP} } =  \procSubject{\mpQ}$. 

\item $\mpP = \trycatch{\mpQ}{\mpQi}$: $\timePass{\cUnit}{\mpP} = \trycatch{\timePass{\cUnit}{\mpQ}}{\timePass{\cUnit}{\mpQi}}$. The thesis holds straightforward  from  $\procSubject{\mpP}  = \procSubject{\mpQ}$, $\procSubject{\timePass{\cUnit}{\mpP} } =  \procSubject{\timePass{\cUnit}{\mpQ}}$, and the induction hypothesis 
$\procSubject{\timePass{\cUnit}{\mpQ}} = \procSubject{\mpQ}$. 

\item $\mpP = \mpFailedP{\mpQ}$: the thesis holds trivially by $\timePass{\cUnit}{\mpP} = \mpP$. 

\item $\mpP = \mpSessionQueueO{\mpS}{\roleP}{\mpQueue}$: the thesis holds directly from  $\timePass{\cUnit}{\mpP} = \mpP$.
\qedhere 
 \end{itemize}
\end{proof}

\lemTimePass*
\begin{proof}
By induction on the derivation $\stJudge{\mpEnv}{\stEnvQ}{\mpP}$.
Note that the rule \inferrule{\iruleMPClock} %
is not applicable due to the undefined nature of $\timePass{\cUnit}{\mpP}$. 
\begin{itemize}[left=0pt, topsep=0pt]
\item Case \inferrule{\iruleMPNil}: trivial by $\timePass{\cUnit}{\mpP} = \mpNil$ and $\stEnvEndP{\stEnvQ \,\tcFmt{+}\, \cUnit}$~(from~\Cref{lem:end-time-pass}).

\item Case \inferrule{\iruleMPTime}:
 $\mpP =  \delay{\cUniti}{\mpPi}$ and $\timePass{\cUnit}{\mpP} = \delay{\cUniti - \cUnit}{\mpPi}$.
 We just need to show $\stJudge{\mpEnv}{%
        \stEnvQ \,\tcFmt{+}\, \cUnit \,\tcFmt{+}\, \tcFmt{(\cUniti - \cUnit)}
      }{%
        \mpPi%
      }$, which follows directly by inverting \inferrule{\iruleMPTime} on
      $\stJudge{\mpEnv}{%
        \stEnvQ}{%
        \delay{\cUniti}{\mpPi}%
      }$.
      
\item Case \inferrule{\iruleMPSel}: $\mpP = \mpTSel{\mpC}{\roleQ}{\mpLab}{\mpD}{\mpPi}{\cUniti}$ and 
$\timePass{\cUnit}{\mpP} =
\mpTSel{\mpC}{\roleQ}{\mpLab}{\mpD}{\mpPi}{\cUniti - \cUnit}$ if $\cUniti \geq \cUnit$, and $\mpFailedP{\mpP}$ if $\cUniti < \cUnit$. 

\begin{itemize}[left=0pt, topsep=0pt]
\item $\timePass{\cUnit}{\mpP} = \mpFailedP{\mpP}$: the thesis holds as $\mpFailedP{\mpP}$ can be typed by any typing environment; 
\item $\timePass{\cUnit}{\mpP} =
\mpTSel{\mpC}{\roleQ}{\mpLab}{\mpD}{\mpPi}{\cUniti - \cUnit}$ with $\cUniti - \cUnit \geq 0$: 
\begin{enumerate}[left=0pt, topsep=0pt]
\item $\forall \cUnitii:    \cUnitii \leq \cUniti - \cUnit \Longrightarrow \cVal + \cUnit + \cUnitii \models \ccst$: directly from 
$\forall \cUnitii: \cUnitii \leq \cUniti \Longrightarrow  \cVal +  \cUnitii \models \ccst$;
\item $\stEnvEntails{\stEnv[1]  \,\tcFmt{+}\, \cUnit}{\mpC}{%
       \stCPair{\cVal \,\tcFmt{+}\, \cUnit}{ \stIntSum{\roleQ}{}{\stTChoice{\stLab}{\stS}{\ccst, \crst} \stSeq \stT}}}$:  directly from
      $\stEnvEntails{\stEnv[1]}{\mpC}{%
       \stCPair{\cVal}{ \stIntSum{\roleQ}{}{\stTChoice{\stLab}{\stS}{\ccst, \crst} \stSeq \stT}}}$; 
\item $\forall \cUnitii \leq \cUniti - \cUnit:  \stJudge{\mpEnv}{%
        \stEnvQ \,\tcFmt{+}\, \cUnit \,\tcFmt{+}\, \cUnitii \stEnvComp \stEnvMap{\mpC}{\stCPair{\cValUpd{\cVal + \cUnit + \cUnitii}{\crst}{0}}{\stT}}%
      }{%
        \mpPi%
      }$: directly from $\forall \cUnitii \leq \cUniti:  \stJudge{\mpEnv}{%
        \stEnvQ \,\tcFmt{+}\, \cUnitii \stEnvComp \stEnvMap{\mpC}{\stCPair{\cValUpd{\cVal + \cUnitii}{\crst}{0}}{\stT}}%
      }{%
        \mpPi%
      }$.
 \item $\stEnvEntails{\stEnvQ[2] \,\tcFmt{+}\, \cUnit}{\mpD}{\stCPair{\cVali + \cUnit}{\stTi}}$: straightforward from the fact that  $\cVali$ can be any clock assignment such that $\cVali \models \ccsti$, allowing us to set it up such that both $\cVali$ and $\cVali \,\tcFmt{+}\, \cUnit$ satisfy $\ccsti$.
\end{enumerate}
\end{itemize}

\item Case \inferrule{\iruleMPBranch}: $\mpP = \mpTBranch{\mpC}{\roleQ}{i \in I}{\mpLab[i]}{x_i}{\mpP[i]}{\cUniti}{}$ and 
 $\timePass{\cUnit}{\mpP} =
\mpTBranch{\mpC}{\roleQ}{i \in I}{\mpLab[i]}{x_i}{\mpP[i]}{\cUniti - \cUnit}{}$ if $\cUniti \geq \cUnit$, and $\mpFailedP{\mpP}$ if $\cUniti < \cUnit$. 

\begin{itemize}[left=0pt, topsep=0pt]
\item $\timePass{\cUnit}{\mpP} = \mpFailedP{\mpP}$: the thesis holds as $\mpFailedP{\mpP}$ can be typed by any typing environment; 
\item $\timePass{\cUnit}{\mpP} =
\mpTBranch{\mpC}{\roleQ}{i \in I}{\mpLab[i]}{x_i}{\mpP[i]}{\cUniti - \cUnit}{}$ with $\cUniti - \cUnit \geq 0$: 
\begin{enumerate}[left=0pt, topsep=0pt]
\item $\forall i \in I: \forall \cUnitii:    \cUnitii \leq \cUniti - \cUnit \Longrightarrow \cVal + \cUnit + \cUnitii \models \ccst[i]$: directly from 
$\forall i \in I: \forall \cUnitii: \cUnitii \leq \cUniti \Longrightarrow  \cVal +  \cUnitii \models \ccst[i]$;

\item $\stEnvEntails{\stEnv[1]  \,\tcFmt{+}\, \cUnit}{\mpC}{%
       \stCPair{\cVal \,\tcFmt{+}\, \cUnit}{ \stExtSum{\roleQ}{i \in I}{\stTChoice{\stLab[i]}{\stS[i]}{\ccst[i], \crst[i]} \stSeq \stT[i]}}}$:  directly from
      $\stEnvEntails{\stEnv[1]}{\mpC}{%
       \stCPair{\cVal}{ \stExtSum{\roleQ}{i \in I}{\stTChoice{\stLab[i]}{\stS[i]}{\ccst[i], \crst[i]} \stSeq \stT[i]}}}$; 

\item $\forall i \in I: \forall \cUnitii \leq \cUniti - \cUnit:  \stJudge{\mpEnv}{%
        \stEnvQ \,\tcFmt{+}\, \cUnit \,\tcFmt{+}\, \cUnitii \stEnvComp \stEnvMap{y_i}{\stCPair{\cVali[i]}{\stTi[i]}} \stEnvComp \stEnvMap{\mpC}{\stCPair{\cValUpd{\cVal + \cUnit + \cUnitii}{\crst[i]}{0}}{\stT[i]}}%
      }{%
        \mpP[i]%
      }$: directly from $\forall i \in I: \forall \cUnitii \leq \cUniti:  \stJudge{\mpEnv}{%
        \stEnvQ \,\tcFmt{+}\, \cUnitii \stEnvComp \stEnvMap{y_i}{\stCPair{\cVali[i]}{\stTi[i]}}  \stEnvComp \stEnvMap{\mpC}{\stCPair{\cValUpd{\cVal + \cUnitii}{\crst[i]}{0}}{\stT[i]}}%
      }{%
        \mpP[i]%
      }$.
\end{enumerate}
\end{itemize}

\item Case \inferrule{\iruleMPPar}: by induction hypothesis. 

\item Case \inferrule{\iruleMPKill}: $\mpP = \timePass{\cUnit}{\mpP} = \kills{\mpS}$. 
Since $\stJudge{\mpEnv}{\stEnvQ}{\kills{\mpS}}$, we have $\stEnv = \stEnvi \stEnvComp \stEnvMap{%
          \mpChanRole{\mpS}{\roleP[1]}%
        }{%
           \stSpecf[1]
        } \stEnvComp \ldots \stEnvMap{%
          \mpChanRole{\mpS}{\roleP[n]}%
        }{%
           \stSpecf[n]
        }$ with  $\stEnvEndP{\stEnvi}$ and $n \geq 0$. 
        Hence,  $\stEnv \tcFmt{+}\, \cUnit = \stEnvi \tcFmt{+}\, \cUnit \stEnvComp \stEnvMap{%
          \mpChanRole{\mpS}{\roleP[1]}%
        }{%
           \stSpecfi[1]
        } \stEnvComp \ldots \stEnvComp \stEnvMap{%
          \mpChanRole{\mpS}{\roleP[1]}%
        }{%
            \stSpecfi[n]
        }$ such that  $\stEnvEndP{\stEnvi \tcFmt{+}\, \cUnit}$~(by~\Cref{lem:end-time-pass}),  
        and $\forall i \leq n: 
           \stSpecfi[i] = \stCPair{\cVal[i]  \tcFmt{+}\, \cUnit}{\stT[i]}$ if $\stSpecf[i] = \stCPair{\cVal[i]}{\stT[i]}$, 
            $\stSpecfi[i] = \stQType[i]$ if $\stSpecf[i] = \stQType[i]$, and 
             $\stSpecfi[i] = \stMPair{\stCPair{\cVal[i]  \tcFmt{+}\, \cUnit}{\stT[i]}}{\stQType[i]}$ if   
             $\stSpecf[i] = \stMPair{\stCPair{\cVal[i]}{\stT[i]}}{\stQType[i]}$. 
             Then the thesis holds directly by applying \inferrule{\iruleMPKill}. 

\item Case \inferrule{\iruleMPTry}: $\mpP = \trycatch{\mpPi}{\mpQ}$.
The thesis holds by applying the induction hypothesis and~\Cref{lem:subj-equiv}.

\item Case \inferrule{\iruleMPCancel}: $\mpP =  \mpCancel{\mpChanRole{\mpS}{\roleP}}{\mpQ}$. We have 
$\stEnv = \stEnvi \stEnvComp \stEnvMap{%
          \mpChanRole{\mpS}{\roleP}%
        }{\stSpecf
        }$ with $\stJudge{\mpEnv}{\stEnvQi}{\mpQ}$, and $\stEnv \tcFmt{+}\, \cUnit$ = $\stEnvi \tcFmt{+}\, \cUnit \stEnvComp \stEnvMap{%
          \mpChanRole{\mpS}{\roleP}%
        }{\stSpecf \tcFmt{+}\, \cUnit
        }$, where $\stSpecf \tcFmt{+}\, \cUnit = \stCPair{\cVal  \tcFmt{+}\, \cUnit}{\stT}$ if $\stSpecf = \stCPair{\cVal}{\stT}$,  $\stQType$ if $\stSpecf = \stQType$, and 
             $\stMPair{\stCPair{\cVal  \tcFmt{+}\, \cUnit}{\stT}}{\stQType}$ if   
             $\stSpecf = \stMPair{\stCPair{\cVal}{\stT}}{\stQType}$.
          Then the thesis holds by applying \inferrule{\iruleMPCancel} on the induction hypothesis $\stJudge{\mpEnv}{\stEnvQi \tcFmt{+}\, \cUnit}{\timePass{\cUnit}{\mpQ}}$.

\item Case \inferrule{\iruleMPFailed}:  $\timePass{\cUnit}{\mpP} = \mpFailedP{\mpP}$. The thesis holds as $\mpFailedP{\mpP}$ can be typed by any environment. 

\item Case \inferrule{\iruleMPResPropG}: $\mpP =  \mpRes{\stEnvMap{\mpS}{\stEnvi}}\mpQ$ and $\timePass{\cUnit}{\mpP} = \mpRes{\stEnvMap{\mpS}{\stEnvi}}{\timePass{\cUnit}{\mpQ}}$. We have 
$\stJudge{\mpEnv}{%
        \stEnvQ \stEnvQComp \stEnvi%
      }{%
        \mpQ
      }$, $\mpS \not\in \stEnvQ$, and $\stEnvAssoc{\gtWithTime{\cVal}{\gtG}}{\stEnvi}{\mpS}$. The thesis holds by applying \inferrule{\iruleMPResPropG} the  induction hypothesis  $\stJudge{\mpEnv}{%
        \stEnvQ \tcFmt{+}\, \cUnit \stEnvQComp \stEnvi \tcFmt{+}\, \cUnit
      }{%
        \timePass{\cUnit}{\mpQ}
      }$, $\mpS \not\in \stEnv \tcFmt{+}\, \cUnit$, and the fact that there exists $\gtWithTime{\cVali}{\gtGi}$ such that $\stEnvAssoc{\gtWithTime{\cVali}{\gtGi}}{\stEnvi \tcFmt{+}\, \cUnit}{\mpS}$~(\Cref{lem:time_assoc_comp,lem:stenv_time_action}). 

 \item Case \inferrule{\iruleMPQueueEmpty}: $\mpP =
 \mpSessionQueueO{\mpS}{\roleP}{\mpQueueEmpty}$ and
 $\timePass{\cUnit}{\mpP} =  \mpSessionQueueO{\mpS}{\roleP}{\mpQueueEmpty}$.
 We have $\stEnvEndP{\stEnv}$, and thus, the thesis holds by $\stEnvEndP{\stEnv \tcFmt{+}\, \cUnit}$. 
 \item Case \inferrule{\iruleMPQueue}: $\mpP = \mpQueueCons{
    \mpQueueOElem{\roleQ}{\mpLab}{\mpChanRole{\mpSi}{\roleR}}
  }{
    \mpQueue
  }$ and  $\timePass{\cUnit}{\mpP}  = \mpP$. We have $\stJudge{\mpEnv}{\stEnvQ}{%
        \mpSessionQueueO{\mpS}{\roleP}{%
          \mpQueue%
        }
      }$, $\stS = \stCPair{\ccst}{\stT}$, $\cVal \models \ccst$, and $\stEnvEntails{\stEnvi}{\mpChanRole{\mpSi}{\roleR}}{\stCPair{\cVal}{\stT}}$. By induction hypotheis, $\stJudge{\mpEnv}{\stEnvQ \tcFmt{+}\, \cUnit}{%
        \mpSessionQueueO{\mpS}{\roleP}{%
          \mpQueue%
        }
      }$. We are left to show $\cVal \tcFmt{+}\, \cUnit \models \ccst$, which is straightforward from the fact that  
      $\cVal$ can be any clock assignment such that $\cVal \models \ccst$, allowing us to set it up such that both 
      $\cVal$ and $\cVal \,\tcFmt{+}\, \cUnit$ satisfy $\ccst$.
\qedhere 
\end{itemize}
\end{proof}

\lemSubjectReductionGlobal*
\begin{proof}
 Let us recap the assumptions:
  \begin{align}
    \label{item:subjred:typed-unused}
    &\stJudge{\mpEnv}{\stEnv}{\mpP}
    \\
    \label{item:subjred:stenv-assoc}
    &\forall \mpS \in \stEnv: \exists
    \gtWithTime{\cVal[\mpS]}{\gtG[\mpS]}: \stEnvAssoc{\gtWithTime{\cVal[\mpS]}{\gtG[\mpS]}}{\stEnv[\mpS]}{\mpS} 
    \\
    \label{item:subjred:no-reliable-crash}
    &\mpP \!\mpMove\! \mpPi
  \end{align}
By \inferrule{\iruleMPRedInstant} and
 \inferrule{\iruleMPRedTimeConsume}, we know that $\mpP \mpMove \mpPi$
 is derived from either $\mpP \mpnonTMove \mpPi$ or $\mpP \mpMoveTime \mpPi$. Therefore,
  the proof can proceed by induction on the derivation of $\mpP \mpnonTMove \mpPi$
  and $\mpP \mpMoveTime \mpPi$.
  When the reduction holds by rule $\inferrule{\iruleMPRedCtx}$, %
  a further structural induction on the reduction context $\mpCtx$ is required.
  When the reduction involves a try-catch context $\mpEtx$, we can observe from the typing rule
$\inferrule{\iruleMPTry}$ that $\mpEtxApp{\mpEtx}{\mpP}$ is typed using the same typing environment
as $\mpP$. This implies that we only need to focus on reductions where $\mpEtx =  \mpCtxHole$.
  Most cases hold %
  by inversion of the typing $\stJudge{\mpEnv}{\stEnv}{\mpP}$, %
  and by applying the induction hypothesis. %
\begin{itemize}[left=0pt, topsep=0pt]
\item Case \inferrule{\iruleMPRedOut}: 
{\small 
 We have:
  \begin{flalign}
    &\label{eq:async-q-red-typ:add:typing}%
    \begin{array}{@{\hskip 0mm}l@{\hskip 0mm}}%
      \stEnvQ = \stEnvQ[\stIntC] \stEnvQComp \stEnvQ[\mpQueue] %
    \end{array}
    \text{ and  }%
    \begin{array}{@{\hskip 0mm}c@{\hskip 0mm}}
      \inference[\iruleMPPar]{%
        \stQJudge{\mpEnv}{%
          \stEnvQ[\stIntC]%
        }{%
        }{%
          \mpTSel{\mpChanRole{\mpS}{\roleP}}{\roleQ}{\mpLab}{%
            \mpChanRole{\mpSi}{\roleQi}%
          }{\mpPi}{\mathfrak{n}}}
        \quad%
        \stQJudge{\mpEnv}{%
          \stEnvQ[\mpQueue]%
        }{%
        }{%
           \mpSessionQueueO{\mpS}{\roleP}{\mpQueue}%
        }%
      }{%
        \stQJudge{\mpEnv}{%
          \stEnvQ%
        }{
        }{%
          \mpP%
        }%
      }%
    \end{array}%
    &\text{%
      (inv.~of \inferrule{\iruleMPPar})%
    }  \quad \quad \quad \quad \quad \quad
    \\[1mm]%
    &\label{eq:async-q-red-typ:add:typing-sel-premise}%
    \stEnvQ[\stIntC] = \stEnv[0] \stEnvComp \stEnv[1] \stEnvComp \stEnv[2]%
    \;\;\;\text{and}\;\;\;%
    \inference[\iruleMPSel]{
    \begin{array}{c}
   \forall \cUnit:   \cUnit \leq \mathfrak{n} \Longrightarrow \cVal + \cUnit \models \ccst
   \\
     \stEnvEntails{\stEnv[1]}{\mpChanRole{\mpS}{\roleP}}{%
       \stCPair{\cVal}{\stIntSum{\roleQ}{}{\stTChoice{\stLab}{\stS}{\ccst, \crst} \stSeq \stT}}%
      }
      \\
      \stEnvEntails{\stEnv[2]}{\mpChanRole{\mpSi}{\roleQi}}{\stCPair{\cVali}{\stTi}}%
      \\%
     \forall \cUnit \leq \mathfrak{n}:  \stJudge{\mpEnv}{%
        \stEnv[0] \,\tcFmt{+}\, \cUnit  
        \stEnvComp \stEnvMap{\mpChanRole{\mpS}{\roleP}}{\stCPair{\cValUpd{\cVal + \cUnit}{\crst}{0}}{\stT}}%
      }{%
        \mpPi}
      \end{array}
    }{%
      \stJudge{\mpEnv}{%
        \stEnv[\stIntC]%
      }{%
        \mpTSel{\mpChanRole{\mpS}{\roleP}}{\roleQ}{\mpLab}{%
            \mpChanRole{\mpSi}{\roleQi}%
          }{\mpPi}{\mathfrak{n}}
      }%
    }
    &\text{%
      (by \eqref{eq:async-q-red-typ:add:typing} %
      and inv.~\inferrule{\iruleMPSel})%
    }%
    \\[1mm]
    &\label{eq:async-q-red-typ:add:envii-exists-sub}%
    \exists \stEnvQii = %
    \stEnvQ[0] \stEnvComp%
    \stEnvMap{\mpChanRole{\mpS}{\roleP}}{\stCPair{\cVal}{
      \stIntSum{\roleQ}{}{\stTChoice{\stLab}{\stS}{\ccst, \crst} \stSeq \stT}}%
    }
    \stEnvComp%
    \stEnv[2]%
    \stEnvQComp{\stEnvQ[\mpQueue]}%
    \quad\text{and}\quad%
    \stEnv \stSub \stEnvii%
    &%
    \hspace{-20mm}%
    \text{%
      (by \eqref{eq:async-q-red-typ:add:typing},
        \eqref{eq:async-q-red-typ:add:typing-sel-premise}, %
      \inferrule{\iruleMPSub} %
      and \Cref{def:aat-mpst-typing-env-syntax})%
    }%
  \end{flalign}

We now prove the existence of $\stEnviii$ such that $\stEnvQii \stEnvQMove \stEnvQiii$: 
  \begin{flalign}
   &\label{eq:async-q-red-typ:add:role-present:q-append-typing}
   \stEnvQi[\mpQueue] =
         \stEnvUpd{\stEnv[\mpQueue]}{\mpChanRole{\mpS}{\roleP}}{
          \stQCons{\stQCons{\stEnvApp{\stEnvQ[\mpQueue]}{\mpChanRole{\mpS}{\roleP}}}{\stQMsg{\roleQ}{\stLab}{\stS}}
        }{\stQEmpty}}
      \stEnvQComp%
      \stEnvQ[2]  \;\;\text{and}\;\;%
      \stQJudge{\mpEnv}{%
        \stEnvQi[\mpQueue]%
      }{%
      }{%
         \mpSessionQueueO{\mpS}{\roleP}{\mpQueueCons{\mpQueue}{%
            \mpQueueOElem{\roleQ}{\mpLab}{%
              \mpChanRole{\mpSi}{\roleQi}}}}}
      &\text{%
        (by \eqref{eq:async-q-red-typ:add:typing}, %
        and \inferrule{\iruleMPQueue}
      }%
      \\
      &\label{eq:async-q-red-typ:add:role-present:enviii-exists-move}%
      \exists \stEnvQiii =%
      \stEnv[0] \stEnvComp%
      \stEnvMap{\mpChanRole{\mpS}{\roleP}}{\stCPair{\mpFmt{\cValUpd{\cVal}{\crst}{0}}}{\stT}}%
      \stEnvQComp \stEnvQi[\mpQueue]%
      \;\;\text{such that}\;\;%
      \stEnvQii \stEnvQMove \stEnvQiii%
      &\text{%
       (by %
        \eqref{eq:async-q-red-typ:add:envii-exists-sub}, %
        \eqref{eq:async-q-red-typ:add:role-present:q-append-typing} %
        and \inferrule{\iruleTCtxSend})%
      }%
  \end{flalign}

  Therefore, using %
  $\stEnvQi[\mpQueue]$ %
  from 
  \eqref{eq:async-q-red-typ:add:role-present:q-append-typing}, 
  $\stEnvii$ from \eqref{eq:async-q-red-typ:add:envii-exists-sub}, and $\stEnvQiii$ %
  from \eqref{eq:async-q-red-typ:add:role-present:enviii-exists-move}, 
  we obtain:
  \begin{align}
     &\label{eq:subj-red:comm:stenvii-assoc}%
    \forall \mpS \in \stEnv: \stEnvAssoc{\gtWithTime{\cVal[\mpS]}{\gtG[\mpS]}}{\stEnvii[\mpS]}{\mpS}
    &\text{%
      (by %
      \eqref{item:subjred:stenv-assoc}, 
      \eqref{eq:async-q-red-typ:add:envii-exists-sub},   %
      and \Cref{lem:stenv-assoc-sub})%
    }%
    \\[1mm]%
     &\label{eq:subj-red:comm:stenviii-assoc}%
    \forall \mpS \in \stEnv: \exists \gtWithTime{{\cVal}^{'''}_{\mpS}}{\gtGiii[\mpS]}: \stEnvAssoc{\gtWithTime{{\cVal}^{'''}_{\mpS}}{\gtGiii[\mpS]}}{\stEnviii[\mpS]}{\mpS}%
   &\text{%
      (by \eqref{eq:subj-red:comm:stenvii-assoc}, %
      \eqref{eq:async-q-red-typ:add:role-present:enviii-exists-move},  %
      and \Cref{cor:completeness}) 
   }%
   \\[1mm]
    &\label{eq:async-q-red-typ:add:typing-conclusion}%
      \inference[\iruleMPPar]{%
          \stJudge{\mpEnv}{%
            \stEnv[0] \stEnvComp%
          \stEnvMap{\mpChanRole{\mpS}{\roleP}}{\stCPair{\cValUpd{\cVal}{\crst}{0}}{\stT}}}{\mpPi}
        &%
        \stQJudge{\mpEnv}{%
        \stEnvQi[\mpQueue]%
      }{%
      }{%
         \mpSessionQueueO{\mpS}{\roleP}{\mpQueueCons{\mpQueue}{%
            \mpQueueOElem{\roleQ}{\mpLab}{%
              \mpChanRole{\mpSi}{\roleQi}}}}}
}{%
        \stQJudge{\mpEnv}{\stEnvQiii}{}{%
          \mpPi \mpPar  \mpSessionQueueO{\mpS}{\roleP}{\mpQueueCons{\mpQueue}{%
            \mpQueueOElem{\roleQ}{\mpLab}{%
              \mpChanRole{\mpSi}{\roleQi}}}}}}
       &\text{%
        (by \eqref{eq:async-q-red-typ:add:typing-sel-premise}, %
        \eqref{eq:async-q-red-typ:add:role-present:q-append-typing}, %
      and %
      \eqref{eq:async-q-red-typ:add:role-present:enviii-exists-move})
    }%
    \\%
    &\label{eq:async-q-red-typ:add:stenvi-exists-move-sub}%
    \exists \stEnvQi:\;%
    \stEnvQ \stEnvQMove \stEnvQi%
    \;\;\text{and}\;\;%
    \stEnvQi \stSub \stEnvQiii%
    &%
    \hspace{-30mm}%
    \text{%
      (by \eqref{eq:async-q-red-typ:add:envii-exists-sub}, %
      \eqref{eq:async-q-red-typ:add:role-present:enviii-exists-move}, 
      and \Cref{lem:stenv-assoc-reduction-sub})%
    }%
    \\%
    &\label{eq:async-q-red-typ:add:stenvi-a-safe}%
   \forall \mpS \in \stEnvi: \exists \gtWithTime{\cVal'_{\mpS}}{\gtGi[\mpS]}:  
   \stEnvAssoc{\gtWithTime{\cVali[\mpS]}{\gtGi[\mpS]}}{\stEnvi[\mpS]}{\mpS}
    &%
    \hspace{-30mm}%
    \text{
      (by \eqref{item:subjred:stenv-assoc}, %
      \eqref{eq:async-q-red-typ:add:stenvi-exists-move-sub} %
      and \Cref{cor:completeness})%
    }%
    \\%
    &\label{eq:async-q-red-typ:add:typing-final}%
   \stQJudge{\mpEnv}{\stEnvQi}{}{%
          \mpPi \mpPar  \mpSessionQueueO{\mpS}{\roleP}{\mpQueueCons{\mpQueue}{%
            \mpQueueOElem{\roleQ}{\mpLab}{%
              \mpChanRole{\mpSi}{\roleQi}}}}}
    &\text{%
      (by \eqref{eq:async-q-red-typ:add:typing-conclusion}, %
      \eqref{eq:async-q-red-typ:add:stenvi-exists-move-sub}, %
      and \Cref{lem:narrowing})%
    }%
  \end{align}
  Hence, we conclude the thesis by %
  \eqref{eq:async-q-red-typ:add:stenvi-exists-move-sub}, %
  \eqref{eq:async-q-red-typ:add:stenvi-a-safe} %
  and \eqref{eq:async-q-red-typ:add:typing-final}.}

\item Case \inferrule{\iruleMPRedIn}: dual to the proof for case \inferrule{\iruleMPRedOut}. 

\item Case \inferrule{\iruleMPRedDet}: $\mpP = \delay{\ccst}{\mpQ}$ and $\mpPi = \delay{\cUnit}{\mpQ}$ with 
$\models \ccstSubt{\ccst}{\cUnit}{C}$. %
 By inverting \inferrule{\iruleMPClock} on $\stJudge{\mpEnv}{%
        \stEnvQ }{\delay{\ccst}{\mpQ}}$, we have $\stJudge{\mpEnv}{%
        \stEnvQ }{\delay{\cUnit}{\mpQ}}$, which is the thesis. 

\item Case \inferrule{\iruleMPRedDelay}: by~\Cref{lem:stenv_time_action,lem:time_passing_typing,cor:completeness}. 

\item Case \inferrule{\iruleMPRedCan}: $\mpP = \mpCancel{\mpChanRole{\mpS}{\roleP}}{\mpQ}$, 
           $\mpPi =  \kills{\mpS} \mpPar \mpQ$, and 
           $\stEnv =  \stEnvMap{%
          \mpChanRole{\mpS}{\roleP}%
        }{%
          \stCPair{\cVal}{\stT} %
        } \stEnvComp \stEnvi$ with $\stJudge{\mpEnv}{%
        \stEnvQi%
      }{%
        \mpQ%
      }$. We want to show $\stJudge{\mpEnv}{\stEnvMap{%
          \mpChanRole{\mpS}{\roleP}%
        }{%
          \stCPair{\cVal}{\stT} %
        } \stEnvComp \stEnvQi}{\kills{\mpS} \mpPar \mpQ}$: apply \inferrule{\iruleMPPar} on $\stJudge{\mpEnv}{\stEnvMap{\mpChanRole{\mpS}{\roleP}}{
      \stCPair{\cVal}{\stT}  %
      }}{\kills{\mpS}}$ and $\stJudge{\mpEnv}{\stEnvQi}{\mpQ}$.

\item Case \inferrule{\iruleMPRedTryFail}: $\mpP = \mpFailedP{\mpQ}$ and $\mpPi = \kills{\mpS}$. Since 
$\mpFailedP{\mpQ}$ can be typed by any environment, we have $\stEnv = \stEnvMap{\mpChanRole{\mpS}{\roleP}}{\stCPair{\cVal}{\stT}}$. Thus, 
by applying \inferrule{\iruleMPKill}, it follows directly that $\stJudge{\mpEnv}{\stEnv}{\kills{\mpS}}$, which is the thesis. 

\item Case \inferrule{\iruleMPRedFailCatch}: $\mpP = \trycatch{\mpFailedP{\mpQ}}{\mpQi}$ and $\mpPi = \mpQi \mpPar \kills{\mpS}$ with  $\stJudge{\mpEnv}{%
        \stEnv
      }{%
     \mpQi
      }$: apply \inferrule{\iruleMPPar} on 
      $\stJudge{\mpEnv}{%
        \stEnv
      }{%
       \mpQi 
      }$ and $\stJudge{\mpEnv}{%
       \stFmt{\emptyset}
      }{%
       \kills{\mpS}
      }$ to get $\stJudge{\mpEnv}{%
        \stEnv
      }{%
       \mpQi \mpPar \kills{\mpS}
      }$, which is the thesis. 
      
\item Case \inferrule{\iruleMPCCat}: $\mpP =  \trycatch{\mpQ}{\mpR} \mpPar \kills{\mpS}$, 
$\mpPi = \mpR \mpPar \kills{\mpS}$, and $\stEnv = \stEnv[1] \stEnvComp \stEnv[2]$ with $\stJudge{\mpEnv}{%
        \stEnv[1]
      }{%
       \trycatch{\mpQ}{\mpR}
      }$ and $\stJudge{\mpEnv}{%
        \stEnv[2]
      }{%
       \kills{\mpS}
      }$. By inverting \inferrule{\iruleMPTry} on $\stJudge{\mpEnv}{%
        \stEnv[1]
      }{%
       \trycatch{\mpQ}{\mpR}
      }$, $\stJudge{\mpEnv}{%
        \stEnv[1]
      }{%
       \mpR
      }$. Furthermore, we apply \inferrule{\iruleMPPar} on 
      $\stJudge{\mpEnv}{%
        \stEnv[1]
      }{%
       \mpR
      }$ and $\stJudge{\mpEnv}{%
        \stEnv[2]
      }{%
       \kills{\mpS}
      }$ to get $\stJudge{\mpEnv}{%
        \stEnv
      }{%
       \mpR \mpPar \kills{\mpS}
      }$, which is the thesis.

\item Case \inferrule{\iruleMPRedCanQ}: 
  $\mpP = 
    \mpSessionQueueO{\mpS}{\roleP}{%
      \mpQueueCons{%
        \mpQueueOElem{\roleQ}{\mpLab}{\mpChanRole{\mpSi}{\roleR}}%
      }{\mpQueue}}
    \mpPar
     \kills{\mpS}$, 
      $\mpPi = 
   \mpSessionQueueO{\mpS}{\roleP}{\mpQueue}
   \mpPar \kills{\mpS} \mpPar \kills{\mpSi}$, 
  and 
    $\stEnv =  \stEnv[1] \stEnvComp \stEnv[2]$ 
    with 
    $\stJudge{\mpEnv}{%
          \stEnvQ[1]
        }{%
           \mpSessionQueueO{\mpS}{\roleP}{%
      \mpQueueCons{%
        \mpQueueOElem{\roleQ}{\mpLab}{\mpChanRole{\mpSi}{\roleR}}%
      }{\mpQueue}}}$  and 
      $\stJudge{\mpEnv}{%
          \stEnvQ[2]
        }{ \kills{\mpS}}$. 
 By applying and inverting \inferrule{\iruleMPQueue} on $\stJudge{\mpEnv}{%
          \stEnvQ[1]
        }{%
           \mpSessionQueueO{\mpS}{\roleP}{%
      \mpQueueCons{%
        \mpQueueOElem{\roleQ}{\mpLab}{\mpChanRole{\mpSi}{\roleR}}%
      }{\mpQueue}}}$, 
      we have $\stEnvQ[1] = \stEnvQi[1] \stEnvComp
      \stEnvQii[1]$ such that $\stEnvApp{\stEnvQi[1]}{\mpChanRole{\mpS}{\roleP}} =
     \stQCons{\stQMsg{\roleQ}{\stLab}{\stS}}{\stQType}
     $,
       $\stJudge{\mpEnv}{%
          \stEnvUpd{\stEnvQi[1]}{\mpFmt{\mpChanRole{\mpS}{\roleP}}}{\stQType}}
        {\mpSessionQueueO{\mpS}{\roleP}{\mpQueue}}$,
      $\stS = \stCPair{\ccst}{\stT}$,
       $\cVal \models \ccst$, and
     $\stEnvEntails{\stEnvii[1]}{\mpChanRole{\mpSi}{\roleR}}{\stCPair{\cVal}{\stT}}$.
     Since $\stJudge{\mpEnv}{%
          \stEnvQ[2]
        }{ \kills{\mpS}}$, we can set
        $\stEnv[2] = \stEnvMap{\mpChanRole{\mpS}{\roleQ}}
        {\stCPair{\cVali}{\stExtSum{\roleP}{i \in I}{\stTChoice{\stLab[i]}{\stSi[i]}{\ccst[i], \crst[i]} \stSeq \stT[i]}}}$
       such that there exists $k \in I$ with $\stLab[k] = \stLab$, $\cVali \models \ccst[k]$, and
         $\stS \stSub \stSi[k]$.
         Then, by applying \inferrule{\iruleTCtxRcv} (and \inferrule{\iruleTCtxCongCombined} when necessary),
         we have $\stEnv = \stEnvi[1] \stEnvComp \stEnvii[1] \stEnvComp \stEnv[2]
          \stEnvQTMoveRecvAnnot{\roleQ}{\roleP}{\stLab[k]} \stEnvUpd{\stEnvQi[1]}{\mpFmt{\mpChanRole{\mpS}{\roleP}}}{\stQType}
          \stEnvComp
          \stEnvii[1] \stEnvComp \stEnvMap{\mpChanRole{\mpS}{\roleQ}}{\stCPair{\cValii}{\stT[k]}} = \stEnvi$.
We are left to show: 
\begin{itemize}[left=0pt, topsep=0pt]
\item $\stJudge{\mpEnv}{\stEnvi}
        { \mpSessionQueueO{\mpS}{\roleP}{\mpQueue}} \mpPar \kills{\mpS} \mpPar\kills{\mpSi}$:  apply \inferrule{\iruleMPKill} and \inferrule{\iruleMPPar}. 
 \item $\forall \mpS \in \stEnvi: 
  \exists \gtWithTime{\cVali}{\gtGi}: 
  \stEnvAssoc{\gtWithTime{\cVali}{\gtGi}}{\stEnvi[\mpS]}{\mpS}$: by~\Cref{cor:completeness}. 
        \end{itemize}

\item Case \inferrule{\iruleMPRedCanIn}: 
$\mpP =
     \mpTBranch{\mpChanRole{\mpS}{\roleP}}{\roleQ}{i \in I}{%
      \mpLab[i]}{x_i}{\mpP[i]}{\mathfrak{n}}{}
      \mpPar
     \mpSessionQueueO{\mpS}{\roleQ}{\mpQueue}
     \mpPar
      \kills{\mpS}$,
     $\mpPi =
      \mpRes{\mpSi}{%
        (\mpP[k]\subst{\mpFmt{x_k}}{\mpChanRole{\mpSi}{\roleR}}%
    \mpPar %
    \kills{\mpSi})}
    \mpPar
     \mpSessionQueueO{\mpS}{\roleQ}{\mpQueue}
       \mpPar %
       \kills{\mpS}$, and 
       $\stEnv =  \stEnv[1] \stEnvComp \stEnv[2] \stEnvComp
    \stEnv[3] \stEnvComp \stEnv[4]$ such that
    $\forall i \in I: \stJudge{\mpEnv}{%
          \stEnvQ[1] \stEnvComp%
          \stEnvMap{y_i}{\stCPair{\cVali[i]}{\stTi[i]}} \stEnvComp%
          \stEnvMap{\mpChanRole{\mpS}{\roleP}}{\stCPair{\cValUpd{\cVal}{\crst[i]}{0}}{\stT[i]}}%
        }{%
          \mpP[i]%
        }$,
       $\forall i \in I:  \stEnvEntails{\stEnv[2]}{\mpChanRole{\mpS}{\roleP}}{%
         \stCPair{\cVal}{ \stExtSum{\roleQ}{i \in I}{\stTChoice{\stLab[i]}{\stS[i]}{\ccst[i], \crst[i]} \stSeq \stT[i]}}%
        }$,
        $\stJudge{\mpEnv}{%
        \stEnvQ[3]
      }{ \mpSessionQueueO{\mpS}{\roleQ}{\mpQueue}
        }$ and
        $
        \stJudge{\mpEnv}{%
        \stEnvQ[4]
      }{ \kills{\mpS}
        }$. Then by~\Cref{lem:substitution}, $\stJudge{\mpEnv}{%
          \stEnvQ[1] \stEnvComp%
          \stEnvMap{\mpChanRole{\mpSi}{\roleR}}{\stCPair{\cVali[k]}{\stTi[k]}} \stEnvComp%
          \stEnvMap{\mpChanRole{\mpS}{\roleP}}{\stCPair{\cValUpd{\cVal}{\crst[k]}{0}}{\stT[k]}}%
        }{%
          \mpP[k]%
        }$. By typing $\kills{\mpSi}$ with an appropriate environment $\stEnvii$ such that there exists $\gtWithTime{\cValii}{\gtGii}$ with $\stEnvAssoc{\gtWithTime{\cValii}{\gtGii}}{\stEnvii \stEnvComp   \stEnvMap{\mpChanRole{\mpSi}{\roleR}}{\stCPair{\cVali[k]}{\stTi[k]}}}{\mpSi}$ and applying \inferrule{\iruleMPResPropG}, we can conclude the thesis.

\item Case \inferrule{\iruleMPRedCall}: apply \inferrule{\iruleMPDef}, \inferrule{\iruleMPCall}, and their respective inversions when necessary. 

\item Case \inferrule{\iruleMPRedCtx}:  
by inversion of the rule and \Cref{def:mpst-proc-context}, 
we have to prove the statement in the following sub-cases:
\begin{enumerate}[leftmargin=*]
  \item
  \label{item:subj-red:ctx:par}
    $\mpP = \mpQ \mpPar \mpR$ \;\;and\;\; $\mpPi = \mpQi \mpPar \mpR$ \;\;and\;\; $\mpQ \mpMove \mpQi$
  \item
  \label{item:subj-red:ctx:res}
    $\mpP = \mpRes{\mpSi}{\mpQ}$ \;\;and\;\; $\mpPi = \mpRes{\mpSi}{\mpQi}$ \;\;and\;\; $\mpQ \mpMove \mpQi$
  \item
  \label{item:subj-red:ctx:def}
    $\mpP = \mpDefAbbrev{\mpDefD}{\mpQ}$ \;\;and\;\; $\mpPi = \mpDefAbbrev{\mpDefD}{\mpQi}$ \;\;and\;\; $\mpQ \mpMove \mpQi$
\end{enumerate}
Cases \ref{item:subj-red:ctx:par} and \ref{item:subj-red:ctx:def} are easily proved using the induction hypothesis.  Therefore, here we focus on case \ref{item:subj-red:ctx:res}.
\begin{flalign}
  \label{eq:subj-red:ctx:res:p-typing}%
  &
  \exists \stEnvi, \gtWithTime{\cVal}{\gtG} \;\;\text{such that}\;\;
  \inference[\iruleMPResPropG]{%
    \begin{array}{@{}l@{}}
      \stEnvAssoc{\gtWithTime{\cVal}{\gtG}}{\stEnvi}{\mpSi}
      \\
      \mpSi \!\not\in\! \stEnv
      \qquad%
      \stJudge{\mpEnv}{
        \stEnv \stEnvComp \stEnvi
      }{
        \mpQ
      }
    \end{array}
  }{%
    \stJudge{\mpEnv}{%
      \stEnv%
    }{%
      \mpP %
    }%
  }%
  &\text{%
    (by \ref{item:subj-red:ctx:res}) %
  }
  \\
  \label{eq:subj-red:ctx:res:stenvi-stenvsi}%
  &\exists \stEnvii, \stEnviii \;\;\text{such that}\;
  \left\{\begin{array}{@{}l@{}}
    \mpSi \!\not\in\! \stEnvii
    \\ %
    \stEnv \stEnvMoveStar \stEnvii
    \\ %
    \stEnvi \stEnvMoveStar \stEnviii
    \\ %
    \forall \mpS \in \stEnvii: \exists \gtWithTime{\cValii[\mpS]}{\gtGii[\mpS]}:  \stEnvAssoc{\gtWithTime{\cValii[\mpS]}{\gtGii[\mpS]}}{\stEnvii[\mpS]}{\mpS}
    \\
    \stJudge{\mpEnv}{\stEnvii \stEnvComp \stEnviii}{\mpQi}
  \end{array}\right\}
  &\text{%
    (by \eqref{eq:subj-red:ctx:res:p-typing}
    and inductive hypothesis)
  }
  \\
  \label{eq:subj-red:ctx:res:stenvi-stenvisi}%
  &
  \exists \gtWithTime{\cVali}{\gtGi} \;\;\text{such that}\;\; 
 \stEnvAssoc{\gtWithTime{\cVali}{\gtGi}}{\stEnviii}{\mpSi}
  &\hspace{-20mm}%
  \text{%
    (by \eqref{eq:subj-red:ctx:res:p-typing}, \eqref{eq:subj-red:ctx:res:stenvi-stenvsi}, 
     and~\Cref{lem:comp_proj})
  }
  \\
  \label{eq:subj-red:ctx:res:pi-typing}%
  &
  \inference[\iruleMPResPropG]{%
    \begin{array}{@{}l@{}}
     \stEnvAssoc{\gtWithTime{\cVali}{\gtGi}}{\stEnviii}{\mpSi}
      \\
      \mpSi \!\not\in\! \stEnvii
      \quad%
      \stJudge{\mpEnv}{
        \stEnvii \stEnvComp \stEnviii
      }{
        \mpQi
      }
    \end{array}
  }{%
    \stJudge{\mpEnv}{%
      \stEnvii%
    }{%
      \mpPi %
    }%
  }%
  &\text{%
    (by \eqref{eq:subj-red:ctx:res:stenvi-stenvsi},
    \eqref{eq:subj-red:ctx:res:stenvi-stenvisi}
    and \ref{item:subj-red:ctx:res})
  }
\end{flalign}
Hence, we obtain the thesis by \eqref{eq:subj-red:ctx:res:stenvi-stenvsi} and \eqref{eq:subj-red:ctx:res:pi-typing}.

\item Cases \inferrule{\iruleMPRedCongr},  \inferrule{\iruleMPRedCongrTime}:  apply~\Cref{lem:subject-congruence} and induction hypothesis. 
\qedhere 
\end{itemize}  
\end{proof}

\lemTypeSafety*
\begin{proof}
 From the hypothesis $\mpP \mpMoveStar \mpPi$,
  we know that $\mpP = \mpP[0] \mpMove \mpP[1] \mpMove \cdots
  \mpMove \mpP[n] = \mpPi$ (for some $n$).
  The proof proceeds by induction on $n$.  
  The base case $n=0$ is immediate: we have $\mpP = \mpPi$,
  hence $\mpPi$ is well-typed -- and since the term $\mpCErr$ is not typeable, $\mpPi$ cannot contain such a term.
  In the inductive case $n = m+1$, we know (by the induction hypothesis) that $\mpP[m]$ is well-typed,
  and we apply~\Cref{lem:sr_global} to conclude that $\mpP[m+1] = \mpPi$ is also well-typed and has no $\mpCErr$ subterms.
\qedhere 
\end{proof}

\section{Proofs for Session Fidelity and Deadlock-Freedom}
\label{sec:app-aat-mpst-sf-proof}

\lemSessionFidelityStrengthGlobal*
\begin{proof}
The proof structure is similar to~\cite[Theorem F.8]{DBLP:journals/pacmpl/ScalasY19}:  
by induction on the derivation of the reduction of $\stEnvQ$, we infer the contents of
$\stEnvQ$ and then the shape of $\mpP$ and its sub-processes $\mpP[\roleP]$ and $\kills{\mpQ}$, showing that they can mimic the reduction of $\stEnvQ$. %
\begin{itemize}[left=0pt, topsep=0pt]
\item Case \inferrule{\iruleTCtxSend}:  in this case, the process $\mpP[\roleP]$ playing role $\roleP$ in session $\mpS$ is a selection on $\mpChanRole{\mpS}{\roleP}$ towards $\roleQ$ (possibly within a process definition).  Therefore, by~\inferrule{\iruleMPRedOut} in~\Cref{fig:aat-mpst-pi-semantics}, 
 $\mpP$ can correspondingly reduce to $\mpPi$ by sending 
 a channel endpoint $\mpChanRole{\mpSi}{\rolePi}$ 
 from $\roleP$ to its message queue in session $\mpS$~(possibly after a finite number of transitions under rule \inferrule{\iruleMPRedCall}). 
 The resulting continuation process $\mpPi$ 
 is typed by $\stEnvi$. 
 The assertion that there exists 
 $\gtWithTime{\cVali}{\gtGi}$ such that $\stEnvAssoc{\gtWithTime{\cVali}{\gtGi}}{\stEnvi}{\mpS}$ 
 follows from 
 $\stEnv \stEnvQTMoveQueueAnnot{\roleP}{\roleQ}{\stLab} \stEnvi$~(apply \inferrule{\iruleTCtxCongCombined} when necessary), 
 $\stEnvAssoc{\gtWithTime{\cVal}{\gtG}}{\stEnv}{\mpS}$, 
 and~\Cref{lem:comp_proj}. 

\item Case \inferrule{\iruleTCtxRcv}: 
in this case, the process $\mpP[\roleP]$ playing role $\roleP$ in session $\mpS$ is a branching on 
$\mpChanRole{\mpS}{\roleP}$ from $\roleQ$ (possibly within a process definition).  
Therefore, by~\inferrule{\iruleMPRedIn} in~\Cref{fig:aat-mpst-pi-semantics}, 
 $\mpP$ can correspondingly reduce to $\mpPi$ by receiving 
 an endpoint $\mpChanRole{\mpSi}{\rolePi}$ 
 from the message queue of $\roleQ$ in session $\mpS$~(possibly after a finite number of transitions under rule \inferrule{\iruleMPRedCall}). 
 The resulting continuation process $\mpPi$ 
 is typed by $\stEnvi$. 
 The assertion that there exists 
 $\gtWithTime{\cVali}{\gtGi}$ such that $\stEnvAssoc{\gtWithTime{\cVali}{\gtGi}}{\stEnvi}{\mpS}$ 
 follows from 
 $\stEnv \stEnvQTMoveRecvAnnot{\roleP}{\roleQ}{\stLab} \stEnvi$~(apply \inferrule{\iruleTCtxCongCombined} when necessary), 
 $\stEnvAssoc{\gtWithTime{\cVal}{\gtG}}{\stEnv}{\mpS}$, 
 and~\Cref{lem:comp_proj}. 

\item Case $\stEnvQ  \stEnvQTMoveTimeAnnot \stEnvQi$: by~\Cref{lem:stenv_time_action}, it holds that 
 $\stEnvQi = \stEnvQ \,\tcFmt{+}\, \cUnit$.
Then only two subcases need to be considered:
\begin{itemize}[left=0pt, topsep=0pt]
\item $\mpP = \delay{\ccst}{\mpQ}$:  by applying \inferrule{\iruleMPRedDet}, we have
$\delay{\ccst}{\mpQ}
     \mpnonTMove%
      \delay{\cUnit}{\mpQ}$ with $\models \ccstSubt{\ccst}{\cUnit}{C}$.
    Then the thesis follows by the case of $\mpP = \delay{\cUnit}{\mpQ}$ in~\Cref{lem:time_passing_typing}.
\item $\mpP \neq \delay{\ccst}{\mpQ}$: directly from~\Cref{lem:time_passing_typing}. 
\end{itemize}
 The assertion that there exists 
 $\gtWithTime{\cVali}{\gtGi}$ such that $\stEnvAssoc{\gtWithTime{\cVali}{\gtGi}}{\stEnvi}{\mpS}$ 
 follows from 
 $\stEnvQ  \stEnvQTMoveTimeAnnot \stEnvQi$, 
 $\stEnvAssoc{\gtWithTime{\cVal}{\gtG}}{\stEnv}{\mpS}$, 
 and~\Cref{lem:comp_proj}. 
\item Case \inferrule{\iruleTCtxRec}: by induction hypothesis. 
\end{itemize}
Regarding the preservation of $\mpPi$ \wrt~\Cref{lem:guarded-definitions}, we need to address the following additional cases when $\kills{\mpS}$ only plays role $\roleP$ in $\mpS$, or $\kills{\mpSi}$ is newly generated. 
\begin{itemize}[left=0pt, topsep=0pt]
\item $\kills{\mpS}$ only plays role $\roleP$ in $\mpS$.   We need to analyse the cases of \inferrule{\iruleMPCCat},  \inferrule{\iruleMPRedCanIn}, and \inferrule{\iruleMPRedCanQ}, which are trivial, since 
if 
$\mpP \mpMove \mpPi$, then $\mpPi$ must be of the form $\mpR \mpPar \kills{\mpS}$, indicating that in $\mpPi$, $\kills{\mpS}$ also plays role $\roleP$ in $\mpS$. 

\item $\kills{\mpSi}$ is newly generated.  
\begin{itemize}[left=0pt, topsep=0pt]
\item $\kills{\mpS}$ is generated by applying \inferrule{\iruleMPRedCan}, \inferrule{\iruleMPRedTryFail}, or 
\inferrule{\iruleMPRedFailCatch}: $\kills{\mpS}$ substitutes for the role of $\roleP$ in $\mpS$, and is typed by 
$\stEnv[\roleP]$.

\item $\kills{\mpSi}$ with $\mpSi \neq \mpS$ is generated by \inferrule{\iruleMPRedCanQ}: since 
$\stJudge{\emptyset}{\stEnvMap{
          \mpChanRole{\mpSi}{\roleP}%
        }{
           \stCPair{\cVali}{\stEnd}
        } \stEnvComp \stEnvMap{
          \mpChanRole{\mpSi}{\roleP}%
        }{
           \stQEmptyType
        }}{\kills{\mpSi}}$, by~\Cref{def:unique-role-proc}, we have that $\kills{\mpSi}$ is included in $\kills{\mpQi}$, as desired. 
        \qedhere 
\end{itemize}
\end{itemize}
\end{proof}

\begin{restatable}{proposition}{lemSingleSessionPersistent}%
    \label{lem:single-session-persistent}%
    Assume $\stJudge{\mpEnvEmpty\!}{\!\stEnv}{\!\mpP}$, %
    where %
    $\stEnvAssoc{\gtWithTime{\cVal}{\gtG}}{\stEnv}{\mpS}$, %
    $\mpP \equiv \mpBigPar{\roleP \in I}{\mpP[\roleP]} \mpPar \kills{\mpQ}$, %
    and $\stEnv = \bigcup_{\roleP \in I}\stEnv[\roleP] \cup \stEnv[0]$ %
    such that, $\stJudge{\mpEnvEmpty}{\stEnvQ[0]}{\kills{\mpQ}}$, and for each $\mpP[\roleP]$, %
    we have $\stJudge{\mpEnvEmpty\!}{\stEnv[\roleP]}{\!\mpP[\roleP]}$. %
    Further, assume that each $\mpP[\roleP]$
    is either $\mpNil$ (up to $\equiv$), %
    or only plays $\roleP$ in $\mpS$, by $\stEnv[\roleP]$. %
    Then, $\mpP \mpMove \mpPi$
    implies $\exists \stEnvi, \gtWithTime{\cVali}{\gtGi}$ %
    such that %
    $\stEnv \!\stEnvMoveWithSessionStar[\mpS]\! \stEnvi$ %
     and %
    $\stJudge{\mpEnvEmpty\!}{\!\stEnvi}{\mpPi}$, %
     with  %
    $\stEnvAssoc{\gtWithTime{\cVali}{\gtGi}}{\stEnvi}{\mpS}$, %
    $\mpPi \equiv \mpBigPar{\roleP \in I}{\mpPi[\roleP]} \mpPar \kills{\mpQi}$, %
    and $\stEnvi = \bigcup_{\roleP \in I}\stEnvi[\roleP] \cup \stEnvi[0]$ %
    such that, $\stJudge{\mpEnvEmpty}{\stEnvQi[0]}{\kills{\mpQi}}$, and for each $\mpPi[\roleP]$, 
    we have $\stJudge{\mpEnvEmpty\!}{\stEnvi[\roleP]}{\!\mpPi[\roleP]}$; %
    furthermore, each $\mpPi[\roleP]$
    is $\mpNil$ (up to $\equiv$),
    or only plays $\roleP$ in $\mpS$, by $\stEnvi[\roleP]$.%
\end{restatable}
\begin{proof}
  Straightforward from the proof of~\Cref{lem:sr_global}, which accounts for all possible transitions from $\mpP$ to $\mpPi$, and in all cases yields the desired properties for its typing environment $\stEnvi$. 
\end{proof}

\begin{lemma}
\label{lem:kill_df}
Assume  $\stJudge{\mpEnvEmpty}{\stEnvQ}{\mpP}$, where $\stEnvAssoc{\gtWithTime{\cVal}{\gtG}}{\stEnvQ}{\mpS}$,  $\mpP \equiv%
  \mpBigPar{\roleP \in I}{%
    \mpP[\roleP]%
  }$, and  $\stEnvQ = \bigcup_{\roleP \in I}\stEnvQ[\roleP]$ such that for each $\mpP[\roleP]$, we have 
  $\stJudge{\mpEnvEmpty}{\stEnvQ[\roleP]}{\mpP[\roleP]}$. Further, assume that each $\mpP[\roleP]$ is either 
  $\mpNil$ (up to $\equiv$), or only plays $\roleP$ in $\mpS$, by $\stEnv[\roleP]$. Then, $\mpP \mpPar \kills{\mpS}$ is deadlock-free. 
\end{lemma}
\begin{proof}
By $\stJudge{\mpEnvEmpty}{\stEnvEmpty}{\kills{\mpS}}$, we have that $\stJudge{\mpEnvEmpty\!}{\!\stEnv}{\!\mpP \mpPar \kills{\mpS}}$, 
    where 
    $\stEnvAssoc{\gtWithTime{\cVal}{\gtG}}{\stEnv}{\mpS}$. 

Consider any $\mpPi$ such that $\mpP \mpPar \kills{\mpS} \!\mpMoveStar\! \mpnonTNotMoveP{\mpPi}$  and 
$\forall \cUnit \geq 0: \timePass{\cUnit}{\mpPi} = \mpPi$. 
  By~\Cref{lem:single-session-persistent}, we know that $\mpPi$ is of the form $\mpBigPar{\roleP \in I}{\mpPi[\roleP]} \mpPar \kills{\mpS} \mpPar \kills{\mpQi}$, and its typing environment $\stEnvi$ is of the form 
  $\bigcup_{\roleP \in I}\stEnvi[\roleP] \cup \stEnvi[0]$ %
    such that, there exists $\gtWithTime{\cVali}{\gtGi}$ with $\stEnvAssoc{\gtWithTime{\cVali}{\gtGi}}{\stEnvi}{\mpS}$, $\stJudge{\mpEnvEmpty}{\stEnvQi[0]}{\kills{\mpS} \mpPar \kills{\mpQi}}$, and for each $\mpPi[\roleP]$, 
    we have $\stJudge{\mpEnvEmpty\!}{\stEnvi[\roleP]}{\!\mpPi[\roleP]}$. 
    Furthermore, $\mpPi$ satisfies the single-session requirements of \Cref{lem:aat-mpst-session-fidelity-global}.  
  
  Since $\mpPi[\roleP]$ is typed by $\stEnvi[\roleP]$, with the association $\stEnvAssoc{\gtWithTime{\cVali}{\gtGi}}{\stEnvi}{\mpS}$, we have that $\mpPi[\roleP]$ is of the form $\mpPii[\roleP] \mpPar \mpSessionQueueO{\mpS}{\roleP}{\mpQueue}$ such that $\mpPii[\roleP]$ is not a message queue. Due to $\mpPi$ satisfying the 
  single-session requirements of \Cref{lem:aat-mpst-session-fidelity-global}, $\procSubject{\mpPii[\roleP]} = \setenum{\mpChanRole{\mpS}{\roleP}}$ if $\mpP \neq \mpNil$. 
  We are left to show that $\mpPii[\roleP] \mpPar \mpSessionQueueO{\mpS}{\roleP}{\mpQueue} \equiv \mpNil \mpPar \kills{\mpQi[\roleP]}$ for each $\roleP$.   
The proof proceeds by induction on  $\mpPii[\roleP]$. 
 \begin{itemize}[left=0pt, topsep=0pt]
 \item $\mpPii[\roleP] = \mpNil$: $\mpQueue$ must be $\mpQueueEmpty$, as otherwise, 
 $\mpSessionQueueO{\mpS}{\roleP}{\mpQueue} \mpPar \kills{\mpS} \mpnonTMove$, a contradiction to $\mpnonTNotMoveP{\mpPi}$. 
 
 \item $\mpPii[\roleP] = \mpTSel{\mpChanRole{\mpS}{\roleP}}{\roleQ}{\mpLab}{\mpD}{\mpR}{\mathfrak{n}}$: 
 $\mpPii[\roleP] \mpPar \mpSessionQueueO{\mpS}{\roleP}{\mpQueue} \mpnonTMove$, 
 a  contradiction to $\mpnonTNotMoveP{\mpPi}$. 
 
\item  $\mpPii[\roleP] = \mpTBranch{\mpChanRole{\mpS}{\roleP}}{\roleQ}{i \in I}{\mpLab[i]}{x_i}{\mpR[i]}{\mathfrak{n}}{}$: 
either  $\mpPii[\roleP] \mpPar \mpSessionQueueO{\mpS}{\roleP}{\mpQueue} \mpnonTMove$, or 
$\mpPii[\roleP] \mpPar \mpSessionQueueO{\mpS}{\roleP}{\mpQueue} \mpPar  \kills{\mpS} \mpnonTMove$
 a  contradiction to $\mpnonTNotMoveP{\mpPi}$. 
 
 \item $\mpPii[\roleP] =  \mpDefAbbrev{\mpDefD}{\mpR}$:  $\mpPii[\roleP] \mpnonTMove$,  a  contradiction to $\mpnonTNotMoveP{\mpPi}$. 
 
 \item $\mpPii[\roleP] = \delay{\ccst}{\mpR}$: $\mpPii[\roleP] \mpnonTMove$,  a  contradiction to 
 $\mpnonTNotMoveP{\mpPi}$. 
 
 \item $\mpPii[\roleP] = \delay{\cUnit}{\mpR}$:  we only consider $\cUnit > 0$ as $\delay{0}{\mpR} \equiv \mpR$. Then $\timePass{\cUniti}{\mpPii[\roleP]} \neq \mpPii[\roleP]$ for any $\cUnitii \leq \cUnit$, a contradiction.
 
 \item  $\mpPii[\roleP] = \mpFailedP{\mpR}$: since $\procSubject{\mpPii[\roleP]} = \setenum{\mpChanRole{\mpS}{\roleP}}$, we have $\mpPii[\roleP] \mpnonTMove$, a  contradiction to $\mpnonTNotMoveP{\mpPi}$. 
 
 \item $\mpPii[\roleP] = \trycatch{\mpR}{\mpRi}$: since $\procSubject{\mpR} = \procSubject{\mpPii[\roleP]} = \setenum{\mpChanRole{\mpS}{\roleP}}$, we have $\mpPii[\roleP] \mpPar \kills{\mpS} \mpnonTMove$, a  contradiction to $\mpnonTNotMoveP{\mpPi}$. 
 
 \item $\mpPii[\roleP] = \mpCancel{\mpChanRole{\mpS}{\roleP}}{\mpR}$: $\mpPii[\roleP] \mpnonTMove$,  a  contradiction to $\mpnonTNotMoveP{\mpPi}$. 
 \qedhere 
 \end{itemize} 

\end{proof}

\begin{lemma}
  \label{lem:aat-mpst-process-df}%
Assume  $\stJudge{\mpEnvEmpty}{\stEnvQ}{\mpP}$, where $\stEnvAssoc{\gtWithTime{\cVal}{\gtG}}{\stEnvQ}{\mpS}$,  $\mpP \equiv%
  \mpBigPar{\roleP \in I}{%
    \mpP[\roleP]%
  }$, and  $\stEnvQ = \bigcup_{\roleP \in I}\stEnvQ[\roleP]$ such that for each $\mpP[\roleP]$, we have 
  $\stJudge{\mpEnvEmpty}{\stEnvQ[\roleP]}{\mpP[\roleP]}$. Further, assume that each $\mpP[\roleP]$ is either 
  $\mpNil$ (up to $\equiv$), or only plays $\roleP$ in $\mpS$, by $\stEnv[\roleP]$. 
 Then, $\mpP$ is deadlock-free. %
  \end{lemma}
\begin{proof}
We write $\stEnvQTypeEndP{\stEnvQ}$
if $\forall \mpChanRole{\mpS}{\roleP} \in \dom{\stEnvQ}: \stEnvApp{\stEnvQ}{\mpChanRole{\mpS}{\roleP}} =  \stMPair{\stCPair{\cVal[\roleP]}{\stEnd}}{\stQEmptyType}$.

 Consider any $\mpPi$ such that $\mpP \!\mpMoveStar\! \mpnonTNotMoveP{\mpPi}$, 
$\forall \cUnit \geq 0: \timePass{\cUnit}{\mpPi} = \mpPi$, and 
  $\mpP = \mpP[0] \!\mpMove\! \mpP[1] \!\mpMove\! \cdots \!\mpMove\! \mpP[n] = \mpNotMoveP{\mpPi}$ (for some $n$)
  with each reduction $\mpP[i] \!\mpMove\! \mpP[i+1]$ ($i \!\in\! 0...n\!-\!1$). 
  By~\Cref{lem:single-session-persistent}, we know that each $\mpP[i]$ is well-typed
  and its typing environment $\stEnv[i]$ is such that $\stEnv \stEnvMoveWithSessionStar[\mpS] \stEnv[i]$;
The proof proceeds by considering the following two cases:

\begin{itemize}[left=0pt, topsep=0pt]
\item $\stEnvQTypeEndP{\stEnvQ[n]}$:  by~\Cref{lem:aat-mpst-session-fidelity-global} (session fidelity), we know that $\mpPi$ is typed by $\stEnvQ[n]$. 
Therefore, by inversion of typing and $\stEnvQTypeEndP{\stEnvQ[n]}$, we have $\mpPi \equiv \mpNil \mpPar \Pi_{\roleP} \mpSessionQueueO{\mpS}{\roleP}{\mpQueueEmpty} \mpPar \kills{\mpQ} \equiv \mpNil \mpPar  \kills{\mpQ}$, which is the thesis.

\item $\stEnvQTypeEndP{\stEnvQ[n]}$  not satisfied: there exists $i < n$ such that starting from $i$, 
$\stEnv[i] $ can only undergo reduction via $\stEnvQTMoveTimeAnnot$  with $\stEnvQTypeEndP{\stEnvQ[i]}$ not satisfied.  Therefore, by~\Cref{lem:aat-mpst-session-fidelity-global}, some $\mpP[j]$ with $j \geq i$ performs a time reduction, resulting in a timeout, and consequently $\kills{\mpS}$. Then,  the thesis holds by~\Cref{lem:kill_df}. 
\qedhere 
\end{itemize}
\end{proof}

\corATMPProcessDfProj*
\begin{proof}
A direct corollary of~\Cref{lem:aat-mpst-process-df}.
\end{proof}

\section{\texorpdfstring{Extended Details for~\Cref{SEC:IMPLEMENTATION:IMPLEMENTATION}}{}}
\label{sec:att-mpst-app-impl}
\label{SEC:ATT-MPST-APP-IMPL}

\subsection{Implementing Time Bounds}
\label{appendix:subsec:implementing_time_bounds}

To demonstrate the implementation of time bounds in \timedmulticrusty,
we focus on the final interaction between $\roleFmt{Sen}$ and ${\roleFmt{Sat}}$,
from the point of view of $\roleFmt{Sat}$.
In this case,
${\roleFmt{Sat}}$ sends a \CODE{Close} message between time units $5$ and $6$ (both inclusive),
following the clock $C_{\roleFmt{Sat2}}$.
The clock $C_{\roleFmt{Sat2}}$ is not reset at the end of this operation.
In \timedmulticrusty,
we design the \CODE{Send} type for sending messages,
this type includes various parameters to describe the aforementioned requirements,
represented as \CODE{Send<[parameter1], [parameter2], ...>}.

Assuming that the \CODE{Close} (payload) type is defined,
to send it using the \CODE{Send} type,
we would begin by writing \CODE{Send<Close, ...>}.
If we abbreviate $C_{\roleFmt{Sat2}}$ as \CODE{'b'},
we would use the clock \CODE{'b'} for the time constraints and
write \CODE{Send<Close, 'b', ...>}.
The parameters for the time constraints in the \CODE{Send} type are
specified after declaring the clock.
In this case,
both bounds are integers,
and are included in the time window.
Therefore,
the \CODE{Send} type would be parameterised
as \CODE{Send<Close, 'b', 0, true, 1, false, ...>}.
Bounds are integers because
generic type parameters are limited to
\CODE{u8},
\CODE{u16},
\CODE{u32},
\CODE{u64},
\CODE{u128},
\CODE{usize},
\CODE{i8},
\CODE{i16},
\CODE{i32},
\CODE{i64},
\CODE{i128},
\CODE{isize},
\CODE{char} and \CODE{bool} \cite{genericRustNoDate}.

Next,
we ensure that the clock \CODE{'b'} is not reset after triggering
the \emph{send} operation.
This is represented as a whitespace \CODE{char} value in the \CODE{Send} type:
\CODE{Send<Close, 'b', 0, true, 1, false, ' ', ...>}.
The last parameter,
known as the \emph{continuation},
describes the operation that follows the sending of the integer.
In this case,
we need to close the connection,
which is done with an \CODE{End} type.
The complete implementation of the sending operation would
be \CODE{Send<Close, 'b', 0, true, 1, false, ' ', End>}.

In a mirrored fashion,
the receiving end can be instantiated
as \CODE{Recv<Close, 'b', 0, true, 1, false, ' ', End>}.
It is imperative to highlight that \CODE{Send}
and \CODE{Recv} inherently mirror each other,
underscoring their \emph{dual-compatibility}.

\Cref{subfig:overview_types:send,subfig:overview_types:recv}
help explain how
\CODE{Send} and \CODE{Recv} work on the inside,
giving us a detailed look at their parameters and features.
The generic type parameters preceded by \CODE{const} within the
\CODE{Send} and \CODE{Recv} types also act as values.
These represent general type categories that \Rust,
naturally supports.
This type-value duality makes it easier to check
if they match correctly during compilation,
ensuring compatibility between two communicating parties.

We also define time units in relation to clocks,
a topic we explore in more details in
\Cref{subsec:implementation:enforcing_time_constraints}.
In \timedmulticrusty,
each \CODE{Send} or \CODE{Recv} type has a single time constraint.
We have done this for several reasons,
mainly to avoid making the types too complicated and because,
in most situations,
one time constraint is enough.

However,
it is essential to acknowledge that certain specialised cases
necessitate the accommodation of multiple clocks
within specific time windows to validate a singular operation.
These situations typically arise when the clocks
were not initially synchronised or have shifted in relation to each other.
Such issues can be addressed through various strategies,
such as introducing delays to specific clocks
or expanding the protocol with additional message exchanges
to synchronise the clocks.

\subsection{Enforcing Time Constraints}
\label{appendix:subsec:implementation:enforcing_time_constraints}

To enforce time constraints,
it is essential to use
reliable clocks and APIs
that can accurately and efficiently instantiate
and compare clocks to time constraints.
\Rust's standard library provides
the \timeRust \cite{standardTimeLibraryWebsite} module,
which enables developers to start and compare clocks
to measure the duration between events.
The \timeRust library employs OS API on all platforms,
and it is important to note that
the time provided by this library
differs from \emph{wall clock time}.
Wall clock time,
which is the time displayed on a physical clock,
remains steady and monotonic,
where consecutive events have consecutive time stamps in the same order,
and each second has the same length.

However,
the \timeRust library has two types of clocks:
the \CODE{Instant} type, which is monotonic but non-steady,
and the \CODE{SystemTime} type, which is steady but non-monotonic.
It is not possible to have a clock that is both steady and monotonic
due to technical reasons associated with the clock device,
virtualisation,
or OS API bugs.
Nonetheless,
this library provides ways to check
and partially handle
such situations if they arise.
In \timedmulticrusty,
we use the \CODE{Instant} type
since we prioritise the correct order of events.

\timedmulticrusty employs the \timeRust library in two key ways
for the management of virtual clocks.
Firstly,
it uses a dictionary (or \CODE{HashMap} in \Rust)
wherein each key corresponds to a character representing a clock's name,
and the associated value is of type \CODE{Instant}.
When handling message transmissions or receptions,
we compare the relevant clocks with their respective time constraints using this dictionary.

To establish an empty \CODE{HashMap} of this nature,
we employ the following code:
\CODE{let clocks = HashMap::<char, Instant>::new();}.
Subsequently,
a new clock is introduced to the dictionary through the operation
\CODE{clocks.insert('a', Instant::now());},
where \CODE{'a'} denotes the clock's name,
and the method \CODE{Instant::now()} retrieves the current time.

\begin{figure}[htbp]
\begin{rustlisting}
pub fn send<T, [...], S>(
  x: T,
  clocks: &mut HashMap<char, Instant>,
  s: Send<T, CLOCK, START, INCLUDE_START, END, INCLUDE_END, RESET, S>,
) -> Result<S, Error>
where
  T: marker::Send,
  S: Session,
{
  // If there is no lower bound
  if s.start < 0 { *@ \label{appendix:line:implementation:send_type:constraint:lower_bound} @*
    // If there is an upper bound
    if s.end >= 0 { *@ \label{appendix:line:implementation:send_type:constraint:upper_bound} @*
      // If this upper bound is included in the time constraint
      if s.include_end { *@ \label{appendix:line:implementation:send_type:constraint:include_upper_bound} @*
        // If the clock is available among all clocks
        if let Some(clock) = clocks.get_mut(&s.clock) { *@ \label{appendix:line:implementation:send_type:constraint:retrieve_clock} @*
          // If the clock respects the time constraint
          if clock.elapsed().as_secs() <= s.end { *@ \label{appendix:line:implementation:send_type:constraint:check_clock} @*
            // If the clock must be reset
            if s.reset != ' ' { *@ \label{appendix:line:implementation:send_type:constraint:check_reset_clock} @*
              let (here, there) = S::new();
              match s.channel.send_timeout( *@ \label{appendix:line:implementation:send_type:constraint:send_payload} @*
                (x, there),
                Duration::from_secs(s.end) - clock.elapsed(),
              ) {
                Ok(()) => { *@ \label{appendix:line:implementation:send_type:constraint:ok} @*
                  let clock_to_reset = clocks.get_mut(&s.reset).unwrap();
                  *clock_to_reset = Instant::now(); *@ \label{appendix:line:implementation:send_type:constraint:reset_clock} @*
                  Ok(here) *@ \label{appendix:line:implementation:send_type:constraint:continuation} @*
                }
                Err(e) => { *@ \label{appendix:line:implementation:send_type:constraint:error} @*
                  cancel(s);
                  panic!("{}", e.to_string())
                }
              }
              (...)
            }
          }
        }
      }
    }
  }
}
\end{rustlisting}
    \caption{Definition of  \CODE{send} primitive.}
    \label{appendix:subfig:implementation:send_type:constraint}
\end{figure}

Within our \CODE{send} and \CODE{recv} primitives,
we implement a cascade of conditions as depicted in
\Cref{appendix:subfig:implementation:send_type:constraint}.
Notably,
we begin by verifying the integrity of the left-hand side
of the time window (referred to as the \emph{lower bound}) in
\Cref{appendix:line:implementation:send_type:constraint:lower_bound}.
Subsequently,
we scrutinise the right-hand side of the time window
(referred to as the \emph{upper bound}) in
\Cref{appendix:line:implementation:send_type:constraint:upper_bound}.
Following this,
we attempt to retrieve the clock \CODE{clock}
from the \CODE{HashMap} \CODE{clocks} in
\Cref{appendix:line:implementation:send_type:constraint:retrieve_clock}
and compare it with the time constraint associated
with the specific communication operation in
\Cref{appendix:line:implementation:send_type:constraint:check_clock}.
Should \CODE{clock} align correctly with the time constraint,
the payload is sent with a timeout barrier using \CODE{send_timeout} in
\Cref{appendix:line:implementation:send_type:constraint:send_payload}.
If the \CODE{send_timeout} operation proves successful in
\Cref{appendix:line:implementation:send_type:constraint:ok},
and if it is determined in
\Cref{appendix:line:implementation:send_type:constraint:check_reset_clock}
that \CODE{clock} requires resetting,
it is indeed reset in
\Cref{appendix:line:implementation:send_type:constraint:reset_clock}
before ultimately returning the continuation of type \CODE{S}
in \Cref{appendix:line:implementation:send_type:constraint:continuation}.

As discerned in
\Cref{appendix:line:implementation:send_type:constraint:check_clock},
where we gauge time using \CODE{clock.elapsed().as_secs()},
we measure time in seconds for the sake of facilitating testing.
However,
it is important to note that this measurement could alternatively be conducted in milliseconds,
microseconds,
or nanoseconds.
Notably,
we have abstained from employing powerful solvers such as
SMT solvers,
as our focus primarily resides in verification of time constraints
for \Rust (local) types.
In \Cref{appendix:subsec:implementation:top_down},
within \timednuscr,
we elucidate that our verification process entails confirming that,
for each clock,
the left-hand side of each time window precedes the right-hand side of the subsequent one.
This ensures,
for the scenarios we aim to address,
that the time constraints are validated at compile time.

\subsection{Implementation of Remote Data Example}
\label{appendix:subsec:implementation:implenting_running_example}

\begin{figure}[htbp]
\begin{subfigure}{.50\textwidth}
\begin{rustlisting}
// The types of the payloads
struct GetData; struct Data; struct Close;
// New behaviour for the Satellite
enum ChoiceToSat {
    Data(...)
    Close(...)
}
// New behaviour for the Sensor
enum ChoiceToSen {
    Data(...)
    Close(...)
}
// Binary types for the Server
type SerToSatData = Send<GetData, 'a', 5, true,*@ \label{line:Server:type:SerToSatData} @*
    5, true, ' ', Recv<Data, 'a', 6, true, 7, true,
    'a', End>>;
type SerToSenData = End; *@ \label{line:Server:type:SerToSenData} @*
type SerToSatStop = Send<Close, 'a', 5, true,*@ \label{line:Server:type:SerToSatStop} @*
    6, true, ' ', End>;
type SerToSenStop = End; *@ \label{line:Server:type:SerToSenStop} @*
\end{rustlisting}
\end{subfigure}
\hfill
\begin{subfigure}{.48\textwidth}
{\lstset{firstnumber=21}\begin{rustlisting}
type SerToSatChoice = Send<ChoiceToSat, 'a',*@ \label{line:Server:type:SerToSatChoice} @*
    5, true, 6, true, ' ', End>>>;
type SerToSenChoice = Send<ChoiceToSen, 'a',*@ \label{line:Server:type:SerToSenChoice} @*
    5, true, 6, true, ' ', End>>>;
// Orderings where RoleBroadcast allows to broadcast
// (make) a choice to every other role
type StackSerData = RoleSat<RoleSat<*@ \label{line:Server:type:StackSerData} @*
    RoleBroadcast>>;
type StackSerStop = RoleSat<RoleEnd>; *@ \label{line:Server:type:StackSerStop} @*
type StackSer = RoleBroadcast; *@ \label{line:Server:type:StackSer} @*
// MeshedChannels
type EndpointSerData = MeshedChannels<*@ \label{line:Server:type:EndpointSerData} @*
    SerToSatData, SerToSenData, StackSerData,
    NameSer>;
type EndpointSerStop = MeshedChannels<*@ \label{line:Server:type:EndpointSerStop} @*
    SerToSatStop, SerToSenStop, StackSerStop,
    NameSer>;
type EndpointSer = MeshedChannels<*@ \label{line:Server:type:EndpointSer} @*
    SerToSatChoice, SerToSenChoice, StackSer,
    NameSer>;
\end{rustlisting}}
\end{subfigure}
\caption{Types for ${\roleFmt{Ser}}$.}
\label{subfig:implementation:types_server}
\end{figure}

\begin{figure}[htbp]
  \begin{subfigure}{0.4\textwidth}
\begin{rustlisting}
// Running the Server role
fn endpoint_ser(*@ \label{appendix:line:Server:function:endpoint_ser} @*
 s: EndpointSer,
 clocks: &mut HashMap<char, Instant>,
) -> Result<(), Error> {
 // Instantiate the clock
 // and add it to clocks
 clocks.insert(*@ \label{appendix:line:Server:function:endpoint_ser:start_clock:a} @*
  'a',
  Instant::now(),
 );

 // Run the protocol for 100 loops
 recurs_ser(100, s, clocks)
}
\end{rustlisting}
\end{subfigure}
\hfill
\begin{subfigure}{0.59\textwidth}
{\lstset{firstnumber=16}\begin{rustlisting}
// Running the Server role
fn recurs_ser(*@ \label{appendix:line:Server:function:recurs_ser} @*
  loops: i32, s: EndpointSer,
  clocks: &mut HashMap<char, Instant>,
) -> Result<(), Error> {
 match loops {
  0 => {let s: EndpointSerStop = choose_ser!(*@ \label{appendix:line:Server:function:recurs_ser:send:branch_stop} @*
    s, clocks, ChoiceToSat::Close, ChoiceToSen::Close);
   let s = s.send(Close {}, clocks)?;*@ \label{appendix:line:Server:function:recurs_ser:send:Stop} @*
   s.close()},*@ \label{appendix:line:Server:function:endpoint_ser:close} @*
  i => { et s: EndpointSerData = choose_ser!(*@ \label{appendix:line:Server:function:recurs_ser:send:branch_data} @*
    s, clocks, ChoiceToSat::Data, ChoiceToSen::Data);
   let s = s.send(GetData {}, clocks)?;*@ \label{appendix:line:Server:function:recurs_ser:send:GetData} @*
   let (_, s) = s.recv(clocks)?;*@ \label{appendix:line:Server:function:recurs_ser:recv:Data} @*
   recurs_ser(i - 1, s, clocks)}}}*@ \label{appendix:line:Server:function:recurs_ser:send:loop} @*
\end{rustlisting}}
\end{subfigure}
\caption{Primitives for ${\roleFmt{Ser}}$.}
\label{appendix:subfig:implementation:primitives_server}
\end{figure}

\Cref{fig:implementation:remote_data} offers an extensive overview
of the timed protocol associated with our running example.
\Cref{subfig:implementation:types_server,appendix:subfig:implementation:primitives_server},
in turn,
elucidate the implementation of ${\roleFmt{Ser}}$
with \timedmulticrusty.
In this context,
\timedmulticrusty uses the \CODE{MeshedChannels} type,
encompassing $n+1$ parameters,
with $n$ representing the count of roles integrated into the protocol.
These parameters include the role's name
associated with each \CODE{MeshedChannels},\linebreak
$n-1$ binary channels facilitating interaction with other roles,
and a stack that governs the sequence of binary channel usage.

The types relevant to ${\roleFmt{Ser}}$ are displayed in
\Cref{subfig:implementation:types_server}.
Within
\Crefrange{line:Server:type:SerToSatData}{line:Server:type:SerToSenData},
we specify the binary types for the first branch,
\CODE{SerToSatData} and \CODE{SerToSenData},
enabling ${\roleFmt{Ser}}$ to engage with ${\roleFmt{Sat}}$ and ${\roleFmt{Sen}}$,
respectively.
To be precise,
the \CODE{SerToSatData} type in \Cref{line:Server:type:SerToSatData}
signifies that ${\roleFmt{Ser}}$,
on this binary channel,
initiates communication by sending
a message labelled \CODE{GetData} to ${\roleFmt{Sat}}$,
before receiving the data in a message labelled \CODE{Data}
from ${\roleFmt{Sat}}$.
Those two operations employ the clock \CODE{a},
and maintain a time window between $5$ and $6$ second (both inclusive)
for the first operation,
and between $6$ and $7$ second (both inclusive) for the second operation.
This clock is reset only within the second operation.

Correspondingly,
\Cref{line:Server:type:SerToSatStop,line:Server:type:SerToSenStop}
define the binary types for ${\roleFmt{Ser}}$ in the second branch,
and
\Cref{line:Server:type:SerToSatChoice,line:Server:type:SerToSenChoice}
enable ${\roleFmt{Ser}}$ to make the choice.
\CODE{StackSerData} in \Cref{line:Server:type:StackSerData}
delineates the order of operations for the first branch,
wherein ${\roleFmt{Ser}}$ interacts twice with ${\roleFmt{Sat}}$ using \CODE{RoleSat},
prior to dispatching a choice with \CODE{RoleBroadcast}.
Similarly,
\Cref{line:Server:type:StackSerStop,line:Server:type:StackSer}
define the stacks for the second branch and for the choice.

\Crefrange{line:Server:type:EndpointSerData}{line:Server:type:EndpointSer}
encapsulate the aforementioned elements,
along with ${\roleFmt{Ser}}$'s name,
\CODE{NameSer},
into \CODE{MeshedChannels}.
Specifically,
\CODE{EndpointSer} initiates the protocol,
while the other two represent ${\roleFmt{Ser}}$'s behaviour in each branch.
\Cref{appendix:subfig:implementation:primitives_server} illustrates how
the \CODE{EndpointSer} type
is employed as an input in the \Rust function \CODE{endpoint_ser}.
Notably,
the output type of this function,
as indicated in \Cref{appendix:line:Server:function:endpoint_ser},
is \CODE{Result<(), Error>},
exemplifying the use of affineness in \Rust.
The \CODE{endpoint_ser} function will yield either a unit type \CODE{()}
or an \CODE{Error} as output.

Initially,
this function initialises the clock \CODE{a} in
\Cref{appendix:line:Server:function:endpoint_ser:start_clock:a},
inserting it into a mutable reference of type
\CODE{HashMap} named \CODE{clocks}
(\Cref{appendix:line:Server:function:recurs_ser}).
It is imperative to note that this virtual clock is role-specific and not shared with other roles.
Subsequently,
\CODE{endpoint_ser} invokes the function \CODE{recurs_ser} and returns its output.
In the event of an error occurring at any point within either of these functions,
the functions panic at this point and
the remaining code within the same scope is dropped.

In \Cref{appendix:line:Server:function:recurs_ser:send:GetData},
variable \CODE{s} of type \CODE{EndpointSerData},
endeavours to send a contentless message labelled \CODE{GetData}.
Here,
by "endeavours",
we imply that the \CODE{send} function will yield
either a value resembling \CODE{EndpointSerData},
excluding the initial \CODE{Send} operation,
or an \CODE{Error}.
If the clock's time fails to adhere to the time constraint
represented in the type \CODE{SerToSatData}
included in the type \CODE{EndpointSerData},
an \CODE{Error} is raised.

Both \CODE{send} and \CODE{recv} functions serve as syntactic sugar for verifying compliance with time constraints.
They achieve this by comparing the relevant clock provided in the type for the time window and,
if necessary,
resetting the clock before transmitting or receiving the message.
The \CODE{clocks} input in those functions
comes in addition to the inputs already
required in their counterparts in
our previous version of \timedmulticrusty.

In the subsequent line,
${\roleFmt{Ser}}$ seeks to receive a message,
and this process repeats until the protocol loops in
\Cref{appendix:line:Server:function:recurs_ser:send:loop}.
In the second branch,
after sending the \CODE{Close} messages in
\Cref{appendix:line:Server:function:recurs_ser:send:Stop},
the connection is terminated in
\Cref{appendix:line:Server:function:endpoint_ser:close}
using the \CODE{close} function,
which shares the same output type as \CODE{endpoint_ser}.

In \timedmulticrusty,
asynchrony is managed by employing the asynchronous primitives
of the \Rust library \crosschan.
The channels used in this library operate as
two sides sharing the same memory cell,
where the sender side can only write to the cell,
and the receiver side can only read from it.
Simulating transfer time in \timedmulticrusty involves blocking
the receiving side for each operation.
Other asynchronous designs used in different \Rust libraries,
such as \tokioRust \cite{web:rust:tokio},
may involve blocking the sending side or restricting
any side from accessing the memory cell for a specific duration.
With our primitives,
it is feasible to attempt to receive a message before it is sent.
Each \CODE{recv} operation will wait until the end of its time window before failing,
assuming no other errors arise.

\subsection{Error Handling}
\label{appendix:subsec:implementation:error_handling}

The error handling capabilities of \timedmulticrusty encompass various errors
that may arise during protocol implementation and execution.
These errors include the misuse of generated types and timeouts.
This demonstrates the versatility of our implementation in verifying
a wide range of communication protocols,
ranging from examples found in literature to more specific instances,
some of which have been studied and implemented in \Cref{sec:evaluation}.

For example,
if Lines \ref{appendix:line:Server:function:recurs_ser:send:GetData} and \ref{appendix:line:Server:function:recurs_ser:recv:Data}
in \Cref{appendix:subfig:implementation:primitives_server} are swapped,
the program will fail to compile because it expects a \CODE{send} primitive in
\Cref{appendix:line:Server:function:recurs_ser:send:GetData},
as indicated by the type of \CODE{s}.
Another error that can be detected at compile-time is when
a payload with the wrong type is sent.
For instance,
sending a \CODE{Data} message instead of
a \CODE{GetData}
in \Cref{appendix:line:Server:function:recurs_ser:send:GetData}
will result in a compilation error.

Furthermore,
\timedmulticrusty is capable of identifying errors at runtime.
If \Crefrange{appendix:line:Server:function:recurs_ser:send:branch_stop}{appendix:line:Server:function:recurs_ser:send:loop}
are replaced with a single \CODE{Ok(())},
the code will compile successfully.
However,
during runtime,
the other roles will encounter failures as they consider
${\roleFmt{Ser}}$ to have failed.
Timeouts are also handled dynamically.
Introducing a 10 seconds delay simulating a time-consuming
and blocking inner task,
using \CODE{sleep(Duration::from_secs(10))},
between
Lines \ref{appendix:line:Server:function:recurs_ser:send:GetData} and \ref{appendix:line:Server:function:recurs_ser:recv:Data}
will cause ${\roleFmt{Ser}}$ to sleep for 10 seconds.
As a result,
the \CODE{recv} operation in \Cref{appendix:line:Server:function:recurs_ser:recv:Data}
will fail due to a violation
of the associated time constraint,
leading ${\roleFmt{Ser}}$ to fail and terminate.
Additionally,
failing to instantiate the clock \CODE{a} in
\Cref{appendix:line:Server:function:endpoint_ser:start_clock:a}
will raise a runtime error since the time constraints cannot be
evaluated without the necessary clocks.

\subsection{Top-Down Approach}
\label{appendix:subsec:implementation:top_down}

Instead of manually writing local types and endpoints in \Rust,
a task that may prove onerous and susceptible to errors in the long term,
an alternative approach emerges:
implementing the corresponding global type
and subsequently projecting it onto local \Rust APIs.
This approach,
known as the Top-Down approach \cite{honda2008Multiparty},
is illustrated in \Cref{fig:theory_methodology}~(right) 
through the architectural design of \timedmulticrusty.
To streamline this process,
we introduce \timednuscr,
an extension of \nuscr \cite{zhouCFSM2021},
with the primary objective of transforming multiparty protocols
expressed in the \Scribble language \cite{yoshida2013Scribble} into global types.

\begin{figure}[htbp]
  \begin{SCRIBBLELISTING}
global protocol RemoteData(role Sen, role Sat, role Ser){ *@ \label{appendix:line:implementation:global_protocol:start} @*
  rec Loop { *@ \label{appendix:line:implementation:global_protocol:line:rec} @*
    choice at Ser { *@ \label{appendix:line:implementation:global_protocol:line:choice} @*
      GetData() from Ser to Sat within [5;6] using a and resetting (); *@ \label{appendix:line:implementation:global_protocol:line:GetData1} @*
      GetData() from Sat to Sen within [5;6] using b and resetting (); *@ \label{appendix:line:implementation:global_protocol:line:GetData2} @*
      Data() from Sen to Sat within [6;7] using b and resetting (b); *@ \label{appendix:line:implementation:global_protocol:line:Data1} @*
      Data() from Sat to Ser within [6;7] using a and resetting (a); *@ \label{appendix:line:implementation:global_protocol:line:Data2} @*
      continue Loop *@ \label{appendix:line:implementation:global_protocol:line:loop} @*
    } or {
      Close() from Ser to Sat within [5;6] using a and resetting (); *@ \label{appendix:line:implementation:global_protocol:line:Stop1} @*
      Close() from Sat to Sen within [5;6] using b and resetting (); *@ \label{appendix:line:implementation:global_protocol:line:Stop2} @*
    }
  }
}
\end{SCRIBBLELISTING}
  \caption{Global protocol for remote data in \timednuscr.}
  \label{appendix:fig:implementation:global_protocol}
\end{figure}

\Cref{appendix:fig:implementation:global_protocol} %
presents the global protocol
relevant to the Remote data protocol in \timednuscr.
The initial interaction within the running example,
in the first branch following the choice,
entails ${\roleFmt{Ser}}$ dispatching a message labelled
\CODE{GetData} to ${\roleFmt{Sat}}$ within the time interval of $5$ to $6$ unit,
governed by the clock \lstscribble{a}.
This is implemented in %
\Cref{appendix:line:implementation:global_protocol:line:GetData1}. %
Notably,
as shown in %
\Cref{appendix:fig:implementation:global_protocol} and
compared to the original \nuscr keywords,
several additional keywords are introduced in \timednuscr:
\lstscribble{within},
\lstscribble{using}
and \lstscribble{and resetting}.
These keywords are employed to denote,
respectively,
the time window (indicated by two integers enclosed in open or closed brackets),
the clock to be used (specified as a character),
and the clock that requires resetting
(depicted within parentheses,
encompassing either nothing or a character).

It is pertinent to underscore that despite sharing the same name,
the virtual clocks assigned to both roles are distinct entities.
To ensure the accuracy of timed protocols,
\timednuscr conducts a series of checks.
Firstly,
it verifies the precise representation of all interactions,
encompassing the five aforementioned components.
Furthermore,
it validates that each time constraint demonstrates a strict increase,
signifying that the lower bound of one operation must not coincide with or occur after the upper bound of the subsequent operation involving the same clock.
Moreover,
each clock,
which should not decrement,
must be capable of satisfying every time constraint,
accounting for resets.
Lastly,
\timednuscr performs standard checks as per the
\MPST protocols \cite{zhouCFSM2021}
and transforms a timed \Scribble protocol into
Communicating Timed Automata \cite{bocchi2015Meeting}.

\section{Extended Details for~\Cref{SEC:EVALUATION}}
\label{appendix:sec:evaluation}

\subsection{Performance Comparison Between~\timedmulticrusty and~\multicrusty}
\label{appendix:subsec:evaluation_long:benchmarks}

\begin{figure}[t!]
    \begin{subfigure}{0.65\textwidth}
        \centering
        \includegraphics[height=10em]{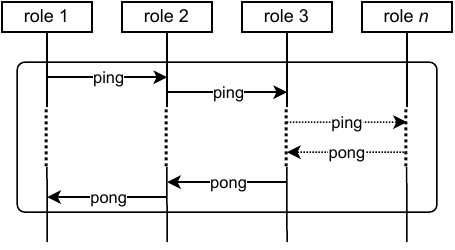}

        \caption{Ring protocol}
        \label{subfig:appendix:ring_protocol}

    \end{subfigure}
    \hfill
    \begin{subfigure}{0.34\textwidth}
        \centering
        \includegraphics[height=10em]{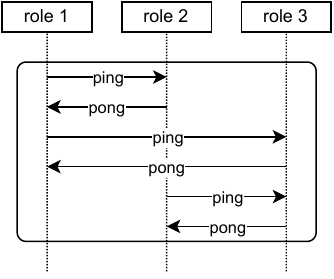}
        \caption{Mesh protocol}
        \label{subfig:appendix:mesh_protocol}
    \end{subfigure}
    \caption{Protocols for benchmarks.}
    \label{fig:appendix:diagram:ben}
\end{figure}

\begin{figure}[t!]%
\begin{subfigure}{0.48\textwidth}
\footnotesize
\[
\begin{array}{r@{\quad}c@{\quad}l@{\quad}l}
\ccst[i] &
\bnfdef  &
i < C < i + 1
\\
\gtG     &
\bnfdef  &
\gtRec{\gtRecVar}{
    \gtCommT{\roleP}{\roleQ}{}{\gtPing}{}{
        \ccst[0],
        \emptyset,
        \ccst[0],
        \emptyset
    }{
        \gtG[0]
    }
}
\\
\gtG[0]  &
\bnfdef  &
\gtCommT{\roleQ}{\roleR}{}{\gtPing}{}{
    \ccst[1],
    \emptyset,
    \ccst[1],
    \emptyset
}{
    \gtG[1]
}
\\
\gtG[1]  &
\bnfdef  &
\gtCommT{\roleR}{\roleS}{}{\gtPing}{}{
    \ccst[2],
    \emptyset,
    \ccst[2],
    \emptyset
}{
    \gtG[2]
}
\\
\gtG[2]  &
\bnfdef  &
\gtCommT{\roleS}{\roleR}{}{\gtPong}{}{
    \ccst[3],
    \{C\},
    \ccst[3],
    \emptyset
}{
    \gtG[3]
}
\\
\gtG[3]  &
\bnfdef  &
\gtCommT{\roleR}{\roleQ}{}{\gtPong}{}{
    \ccst[4],
    \{C\},
    \ccst[4],
    \emptyset
}{
    \gtG[4]
}
\\
\gtG[4]  &
\bnfdef  &
\gtCommT{\roleQ}{\roleP}{}{\gtPong}{}{
    \ccst[5],
    \{C\},
    \ccst[5],
    \{C\}
}{
    \gtRecVar
}
\end{array}
\]
\caption{Ring protocol}
\label{appendix:subfig:ring_protocol}
\end{subfigure}
\hfill
\begin{subfigure}{0.48\textwidth}
\footnotesize
\[
\begin{array}{r@{\quad}c@{\quad}l@{\quad}l}
\ccst[i] &
\bnfdef  &
i < C < i + 1
\\
\gtG     &
\bnfdef  &
\gtRec{\gtRecVar}{
    \gtCommT{\roleP}{\roleQ}{}{\gtPing}{}{
        \ccst[0],
        \emptyset,
        \ccst[0],
        \emptyset
    }{
        \gtG[0]
    }
}
\\
\gtG[0]  &
\bnfdef  &
\gtCommT{\roleQ}{\roleP}{}{\gtPong}{}{
    \ccst[1],
    \emptyset,
    \ccst[1],
    \emptyset
}{
    \gtG[1]
}
\\
\gtG[1]  &
\bnfdef  &
\gtCommT{\roleP}{\roleR}{}{\gtPing}{}{
    \ccst[2],
    \emptyset,
    \ccst[2],
    \emptyset
}{
    \gtG[2]
}
\\
\gtG[2]  &
\bnfdef  &
\gtCommT{\roleR}{\roleP}{}{\gtPong}{}{
    \ccst[3],
    \emptyset,
    \ccst[3],
    \{C\}
}{
    \gtG[3]
}
\\
\gtG[3]  &
\bnfdef  &
\gtCommT{\roleQ}{\roleR}{}{\gtPing}{}{
    \ccst[4],
    \emptyset,
    \ccst[4],
    \emptyset
}{
    \gtG[4]
}
\\
\gtG[4]  &
\bnfdef  &
\gtCommT{\roleR}{\roleQ}{}{\gtPong}{}{
    \ccst[5],
    \{C\},
    \ccst[5],
    \{C\}
}{
    \gtRecVar
}
\end{array}
\]
\caption{Mesh protocol}
\label{appendix:subfig:mesh_protocol}
\end{subfigure}
\caption{Global protocols for benchmarks.}
\label{appendix:fig:diagram}
\end{figure}

To compare~\timedmulticrusty
and~\multicrusty,
we choose two different protocols shown in~\Cref{fig:appendix:diagram:ben,appendix:fig:diagram}:
the~\emph{ring} protocol,
and the~\emph{mesh} protocol.
In the former,
the participants pass a
message from one to the other,
and in the latter,
each participant
sends a message
to each other.
In both protocols,
the loop is run 100 times.
To fairly compare
both libraries,
each operation
must happen within
the same time window:
between $0$ and $10$ seconds.

\begin{figure}[t!]
    \centering
    \begin{subfigure}[t]{0.026\textwidth}
        \includegraphics[width=\textwidth]{pdf/time_seconds.pdf}
    \end{subfigure}
    \begin{subfigure}[t]{0.22\textwidth}
        \includegraphics[width=\textwidth]{pdf/graph_mesh_compile_time_100.pdf}
        \caption{\emph{mesh} - compilation}
        \label{appendix:fig:benchmark_results_long:compile:mesh}
    \end{subfigure}
    \begin{subfigure}[t]{0.22\textwidth}
        \includegraphics[width=\textwidth]{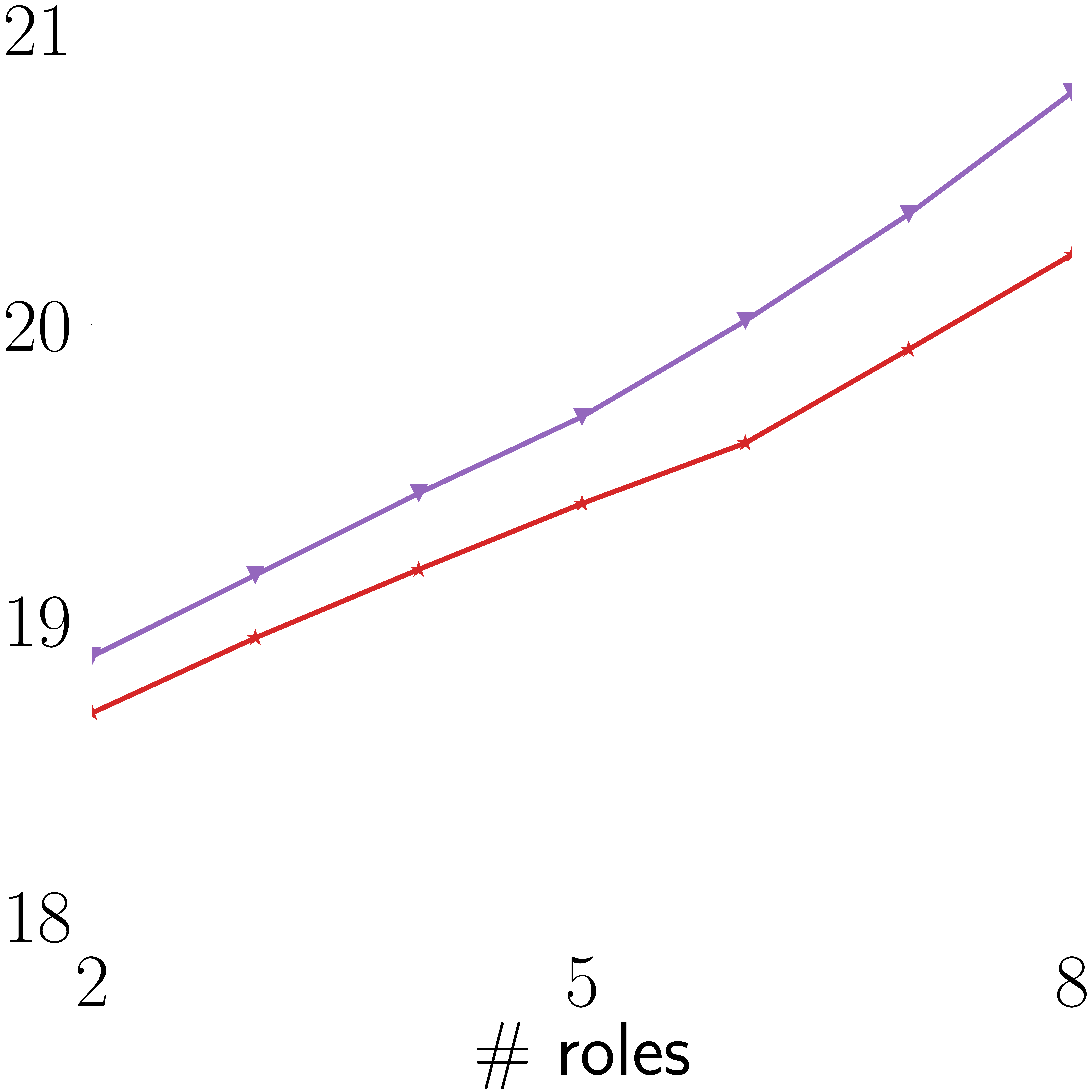}
        \caption{\emph{ring} - compilation}
        \label{appendix:fig:benchmark_results_long:compile:ring}
    \end{subfigure}\hfill%
    \begin{subfigure}[t]{0.026\textwidth}
        \includegraphics[width=\textwidth]{pdf/time_millisecond.pdf}
    \end{subfigure}
    \begin{subfigure}[t]{0.22\textwidth}
        \includegraphics[width=\textwidth]{pdf/graph_mesh_runtime_100.pdf}
        \caption{\emph{mesh} - runtime}
        \label{appendix:fig:benchmark_results_long:running:mesh}
    \end{subfigure}
    \begin{subfigure}[t]{0.22\textwidth}
        \includegraphics[width=\textwidth]{pdf/graph_ring_runtime_100.pdf}
        \caption{\emph{ring} - runtime}
        \label{appendix:fig:benchmark_results_long:running:ring}
    \end{subfigure}

    \caption{Benchmark results.}
    \label{appendix:fig:benchmark_results_long}
\end{figure}

\Cref{appendix:fig:benchmark_results_long}
displays our benchmark results,
running from 2 roles to 8 roles.
In all graphs,
the red line with stars
is linked to~\multicrusty (\AMPST in the legend),
and the purple line with triangles
is linked to~\timedmulticrusty (\ATMP in the legend).
The~\emph{X}-axis corresponds
to the number of roles
in the protocols;
and the~\emph{Y}-axis,
to the time needed for the related task.
In all~\emph{compilation} graphs,
the~\emph{Y}-axis scale
is in seconds (s);
and in all~\emph{running time} graphs,
the~\emph{Y}-axis scale
is in milliseconds (ms).

\vspace{-1em}
\subparagraph{Ring Protocol}
is a protocol
where a message is sent
from the first role
to the last role through
each role,
exactly once,
then goes the other
way around,
as shown
in~\Cref{subfig:appendix:ring_protocol,appendix:subfig:ring_protocol}.
The results of the~\emph{compile-time} benchmarks are shown
in~\Cref{appendix:fig:benchmark_results_long:compile:ring},
and for the~\emph{runtime} benchmarks are shown
in~\Cref{appendix:fig:benchmark_results_long:running:ring}.
At 2 roles,
\multicrusty is compiled less than 2\%
faster than~\timedmulticrusty and
less than 5\% faster at 8 roles,
ranging between 19 s and 20.8 s
for~\timedmulticrusty.
This overhead
is also visible for
running time,
where~\multicrusty is 15\%
faster than~\timedmulticrusty
at 2 roles,
and then 5\% faster
at 8 roles.
The overhead
does not seem to increase
as the number of roles
increases:
there is less than
0.5 ms of difference
between~\multicrusty and~\timedmulticrusty
at 6, 7 and 8 roles.

\vspace{-1em}
\subparagraph{Mesh Protocol}
is a protocol
where all roles
send and receive
a message to each other,
as shown
in~\Cref{subfig:appendix:mesh_protocol,appendix:subfig:mesh_protocol}.
The results of the~\emph{compile-time} (resp.~\emph{runtime}) benchmarks are shown
in~\Cref{appendix:fig:benchmark_results_long:compile:mesh} (resp.~\Cref{appendix:fig:benchmark_results_long:running:mesh}).
The curves of the benchmark results
follow a similar trend
as for the~\emph{ring} protocol:
\multicrusty is compiled faster than
\timedmulticrusty by less than 1\%
at 2 roles,
and 4\% at 8 roles,
and is running faster
by less than 1\% at 2 roles
and 15\% at 8 roles.
The compilation times
for~\timedmulticrusty
range between 18.9 s
and 26 s,
and the running times
are between 0.9 ms
and 11.9 ms.
For both graphs,
the gap exponentially increases,
in absolute values,
due to the exponential
number of enforced time constraints:
there are $2*\binom n2=n*(n-1)$
\CODE{Send} operations,
where $n$ is the number of roles.

\vspace{-1em}
\subparagraph{Results Summary}
All graphs show
\timedmulticrusty has
overheads compared
to~\multicrusty,
which increases
with the number of roles
in absolute values.
This can be explained
by the light additions
we made to~\multicrusty
for creating~\timedmulticrusty:
the~\CODE{Send} and~\CODE{Recv}
types in~\multicrusty
have 2 parameters,
resp. for the payload
and the continuation;
and in~\timedmulticrusty,
they have 8 parameters,
2 for the payload
and the continuation,
and 6 for the time constraints.
Unlike~\multicrusty,
\timedmulticrusty's
\CODE{send} and~\CODE{recv} primitives
use a cascade of~\CODE{if}/\CODE{else}
to check the time constraints before all.
Those results show that,
with additional key features
for handling asynchronous
and timed communications
compared to~\multicrusty,
we were able to maintain
low overheads overall.

\subsection{Case Studies}
\label{appendix:subsec:evaluation_long:examples}

Apart from the protocols sourced from existing scientific literature,
we have included five additional protocols that use the asynchronous and
timed features of~\timedmulticrusty
extracted from industrial examples.
Besides the Remote Data protocol detailled earlier,
there are three protocol extracted from
\cite{DBLP:journals/pacmpl/IraciandCHZ2023} and one protocol extracted from
\cite{servoWebEngineBuggy,boquin2020Understanding}.
In particular,
the PineTime Heart Rate protocol
\cite{Pine64},
extracted from \cite{DBLP:journals/pacmpl/IraciandCHZ2023}
is a protocol used in the industry
to help the PineTime smartwatch retrieve data from
a heart rate sensor
at a specific pace.

We have not compared~\timedmulticrusty with the~\crosschan library
as~\cite{lagaillardie2022Affine} has already conducted such a comparison between their library,
\multicrusty,
and the~\crosschan library.
Consequently,
we can infer similar comparison results between~\timedmulticrusty and~\crosschan.
For the~\emph{Mesh} protocol,
\multicrusty exhibits superior performance compared to the~\crosschan library
in terms of the number of roles,
indicating that~\timedmulticrusty would likewise outperform the~\crosschan library.
However,
for other protocols,
the~\crosschan library demonstrates better performance than~\multicrusty,
and by extension,
\timedmulticrusty.

}{
}

\end{document}